\nonstopmode

\input{preamble.sty}
\usepackage{macros}

\begin{document}
    \begin{frontmatter}
        \title{Comparing Call-by-Name and Call-by-Value Reduction and \\ Reduction Strategies in Calculi for Classical Logic} 


        \Paper{%
            \author{Steffen van Bakel and David Davies}
            \email{s.vanbakel@imperial.ac.uk,david.davies17@imperial.ac.uk}
        }
        \Journal{%
            \author{Steffen van Bakel}
            \ead{svanbakel@ic.ac.uk}
            \author{David Davies}
            \ead{david.davies17@ic.ac.uk}
        }
        \address{Department of Computing, Imperial College London, 180 Queen's Gate, London SW7 2BZ, UK}

        \begin{abstract}


            We define call-by-name and call-by-value reduction and reduction strategies for the calculi $\nclmu$ (symmetric $\lmu$), $\lmmt$, and $\Xs$ ($`X$ with implicit substitution).
            We establish a strong relation between these notions through defining a single interpretation from $\nclmu$ to $\lmmt$ that respects normal reduction, as well as the call-by-name and call-by-value reduction in $\nclmu$ within the their counterpart in $\lmmt$; for the strategies, we will show similar, but weaker results.
            We also define a single mapping from $\lmmt$ to $\Xs$, and show that this also respects all three notions.
            We then continue with studying the natural encoding of $\Xs$ into $\lmmt$, and show that only full reduction is respected, but that reduction steps are needed to model substitution, so the {\CBN} and {\CBV} strategies cannot be respected. 
            We conclude with studying the combination of our efforts and define an interpretation of $\slmu$ into $\Xs$, and show that  {\CBN} and {\CBV} reduction are respected.
            This result underlines that $\Xs$ and $\lmmt$ are similar, but different calculi, and that the nature of $\slmu$ makes that any encoding into either can never fully respect the strategies.
        \end{abstract}

        \Paper{%
            \keywords{classical logic, call by name, call by value, interpretations}%
        }
        \Journal{%
            \begin{keyword}
                classical logic, call by name, call by value, interpretations 
            \end{keyword}%
        }
 \end{frontmatter}


 \arraycolsep2\point
 \abovecaptionskip 2\point
 \belowcaptionskip 3\point


 \def \mycite{\cite}


\section*{Introduction}

The {\LC} \cite{Church'36,Barendregt'84} has long served as a foundation for functional programming languages through its Call-by-Name (\CBN) and Call-by-Value (\CBV) subsystems.
The former models a `lazy' notion of reduction, in which computation of an argument is delayed until used within a function; the latter employs `eager' evaluation, which always reduces an argument before being given to a function.
Although both are subsystems of the same calculus, their semantics can vary greatly; seen that {\CBN} and {\CBV} as reduction strategies make different choices about which redex to contract, implementation of these onto a machine or a calculus normally differ, \ie yielding different {\CBN} and {\CBV} interpretations.
This becomes evident, for example, in the Krivine abstract machines (KAM) \cite{Krivine'07}; for {\CBN} they need only have operations for pushing and popping arguments from the stack, whereas a {\CBV} KAM also requires a notion of stack-frames \cite{LepigreESOP'16} in which the function is frozen onto the stack whilst its argument is evaluated.\footnote{This feature is also present in $\lmmt$'s rule $`l$; see \Def\ref{lmmt reduction}.}
Such semantics give an operational meaning to functional languages, and the differences in the semantics can explain the difference between the behaviour of said languages.

\paragraph{(Classical) Logic and Computation}
It is well known that the {\LC} also serves as a proof-term syntax for (an implicative fragment of) intuitionistic logic. 
This link between computation and logic is known as the Curry-Howard correspondence  \cite{Howard'80}, where a proof of a proposition $A$ is represented by a program of type $A$.
This supports calling intuitionistic logic `constructive', as a proof corresponds to a program that will compute an answer.

Classical logic is known to be non-constructive.
It was thus thought that only intuitionistic logic enjoyed a computational counterpart, but \namecite{Griffin'90} discovered an extension of the Curry-Howard correspondence by typing control operators, which allow for manipulation of the current program continuation.
In particular, Griffin presented a typed {\LC} extended with \posscite{Felleisen'87} $\mathcal{C}$-operator for which Griffin gave the type of double-negation elimination, $\neg\neg A \rightarrow A$ (or better: $((A \arrow `B) \arrow `B ) \arrow A$).
In this calculus, a continuation variable $k : \neg A$ represents a continuation expecting a term of type $A$.
The logical correspondence means a notion of `backtracking' realises true classical propositions such as double-negation elimination and Pierce's law, $((A\arrow B)\arrow A)\arrow A$.
Naturally, this led to the exploration of many `classical calculi' -- those with control operators -- of which we highlight a few.

One of the earliest was \posscite{Parigot'92} $\lmu$ , which extends the {\LC} with the $\mu$-operator and a separate set of $\mu$-variables used to denote continuations.
$\lmu$ allows for terms of the form $\Mu `a.[`b]M$, which represents a switch from the current context, which $\mu$ labels by $`a$, to the context $`b$.
With this operator, $\lmu$ implements minimal classical logic \cite{Ariola-Herbelin'03} through the Curry-Howard principle.
A key difference compared with $\lambda C$ is that $\lmu$ has a separate domain of variables for dealing with continuations, the $\mu$-variables, that are (implicitly) typed with negated types, enabling it to express the `proof by contradiction' proof step.
This allows for Parigot's alternate interpretation of continuation variables as representing the `alternative conclusions' of a sequent, rather than negated propositions, bringing it closer to the notation of derivable sequents as used in \posscite{GentzenG'35} {sequent calculus {\LK}} .

A next significant step was \posscite {Curien-Herbelin'00} {calculus $\lmmt$} (see also \cite{Herbelin'95}), a proof-term syntax for a logic that can be seen as a variant of \posciteauthor{GentzenG'35} {\LK} with focus, which expanded on $\lmu$'s alternative conclusions.
\citeauthor{Curien-Herbelin'00} observed the term syntax needed to reflect the symmetry of the sequents, in particular the ability to manipulate and combine assumptions.
Continuations were thus made first-class as an explicit part of the syntax, where the \emph{command} $\cell<t|e>$ denotes the running of \emph{term} $t$ in the \emph{continuation} $e$.
Where terms encode the proof tree of a formula on the right of a sequent, continuations (or co-terms) encode that of a formula on the left (see \Def\ref{basic system});
computation in $\lmmt$ then induces a cut-elimination in the corresponding proof tree by letting the term on the left interact with the continuation on the right.
The link with cut-elimination in {\LK} also marks a step away from confluence.

The logic behind the type system of $\lmmt$ is quite different to {\LK}, due to the management of the \emph{focus} of a sequent - an approach inherited from $\lmu$ - which states if it is concerned with a formula on the left or right of the sequent.
Indeed there are two versions of the axiom rule (for focusing on the left or right), and the typing rules for $\mu$ and $\mutilde$ only shift the focus of the sequent from a neutral position: logically they perform a no-op.
The shift in focus is necessary as the formula under focus is not explicitly named (assigned to a variable) by the proof-term.
It would seem a true proof-term syntax for {\LK} could be achieved by explicitly naming all formulae.

This observation, along with Urban's investigations of the computational aspects of cut-elimination in {\LK} \cite{Urban'00,Lengrandthesis02}, led to \posscite{vBLL-ICTCS'05} $\X$ (see also \cite{Bakel-Lescanne-MSCS'08}), which constitutes a term calculus for the implicative fragment of \LK. 
Whereas terms in $\lmmt$ represent the derivation of a particular formula, terms (or nets) in $\X$ represent the entire sequent, such that all formulae in the sequent are named by free variables and names (called plugs and sockets); this means that no particular formula is given precedence.
Reduction is expressed by local manipulations of the proof trees, as by Gentzen's cut-elimination, so not through global, substitution-like operations.
$\X$ turns out to be not just a direct representation of {\LK}, but it also provides a fine-grained view of computation, including having the expressivity to implement calculi with explicit substitutions, like \posscite{Bloo-Rose'95} $\Lx$ .
In fact, $\X$ represents {\LK} rather too closely, in that it allows for an arbitrary non-strongly normalising kind of cut-elimination, but of course restrictions can be made to counteract this, such as {\CBN} and {\CBV} reduction strategies, or forcing active cuts to finish proceeding through the tree before activating another cut, as we do here with $\Xs$ (see \Sect \ref{Xs def}).

\paragraph{Sub-reduction versus strategies}
A distinction can be made between between \textit{(sub) reduction} relations and \textit{strategies}; a reduction relation states with (sub) expressions can be reduced, where which one to choose is unrestricted, whereas a strategy picks one candidate (and always the same) between the reducible sub-expressions, normally by not allowing reduction in certain sub-terms.

For the $`l$-calculus, {\CBV} reduction was first defined by \citet{Plotkin'75}; he defines a $`l$-value as a $`l$-term that is not a ``combination'' (\ie not an application), so is a variable or an abstraction; since the latter correspond to functions, a value can/need not be further evaluated.
The {\CBV} reduction then only contracts a redex $`@ (`Lx.M) N $ if the argument $N$ is a value; if it is not, the argument is evaluated (through {\CBV} reduction) until it becomes a value.
This might not terminate, and terms might not have a {\CBV} normal form where normal reduction would be normalising, as for $ `@ (`Lxy.y) ({`@ `D `D }) $, where $`D = `La.aa$.

Although each of $`l$, $\lmu$, $\lmmt$ and $\X$ have notions of {\CBN} and {\CBV} reduction, these are expressed in very different ways, making a direct comparison obscure. 
One method of defining {\CBN} or {\CBV} involves specifying the reducible expressions (redexes) considered valid for each; for the $`l$ calculus, these are, respectively $`@ (`Lx.M) N $ for \CBN, and $`@ (`Lx.M) V $ for \CBV, where $V$ is a \emph{value}.%
{\CBN} allows for the contraction of a redex irrespective of the shape of the argument $N$, whereas {\CBV} forces the evaluation of the argument (to a value) before allowing the redex to contract.
\footnote{For the \LC, {\CBV} effectively becomes a sub-reduction system of \CBN: all {\CBV} redexes are {\CBN} redexes, whereas the term $`@ (`Lx.x) (yy) $ is a {\CBN} redex, but not one for \CBV.}
Another method is to also specify uniquely which out of a possible multitude of redexes can be contracted, by limiting the evaluation contexts; we can then define {\CBN} and {\CBV} \emph{reduction strategies}, \ie deterministic sub-reduction systems. 

The great advantage of {\CBV} reduction, and an explanation for its popularity in programming, is that it is, in general, exponentially more efficient than \CBN, since it can significantly reduce the number of redundant expression evaluations that occur in a program:
since an operand is run to completion before being used, it is completely evaluated before a redex is contracted, only finished computations are distributed through parameter call, and no unnecessary duplication of reductions takes place.
This is not the case for \CBN: if the variable $x$ occurs twice in $M$, then the redex contraction $ `@ (`Lx.M) N \reduc M \tsubst[N/x] $ places the term $N$ twice, and then evaluating $ M \tsubst[N/x] $ risks running $N$ twice, thus duplicating reductions.

For \posscite{Parigot'92} $\lmu$-calculus, the situation is slightly different. 
\citeauthor{Parigot'92}'s original presentation of $\lmu$ was {\CBN}, as is that of the $`l$-calculus.
It is possible to define {\CBV} reduction on $\lmu$ by restricting the contraction of $`b$-redexes to $`@ (`Lx.M) V $, as also done by \namecite{Curien-Herbelin'00}, but, perhaps surprisingly, \citet{Ong-Stewart'97} define a {\CBV} variant for $\lmu$ that requires \emph{extending} $\lmu$'s notion of reduction.
This variant, denoted $\lmuV$, adds to the reduction relation of $\lmu$ a rule allowing for a $\mu$-abstraction to pull in a term on the \textit{left} of an application \cite{Ong-Stewart'97} (similar to an idea used in \cite{Barbanera-Berardi'96}), combined with the usual restrictions of arguments as values, so introducing  $`@ V (`M `a .`b M) $ as a redex.
Adding that rule to standard (\CBN) $\lmu$ would break confluence (this was already noted by \citet{Parigot'92}), a property that $\lmu$ is set up to satisfy.
But this problem disappears once the redexes are restricted to those with a value for arguments (in fact, reduction in $\lmuV$ is confluent \cite{Py-PhD'98,Bakel-LMCS'23}), and the extra rule can be added safely.
So for $\lmu$, {\CBV} is \emph{not} a sub-reduction system of \CBN; in fact, they are both sub-reduction systems for $\nclmu$, \posscite{David-Nour-TLCA05} \emph{symmetric $\lmu$} , which is the non-confluent system obtained from $\lmu$ by adding that extra rule.
For these two notions it is also possible to define {\CBN} and {\CBV} \emph{strategies}, by limiting the evaluation contexts, and these have been studied extensively, as in \cite{Py-PhD'98}.

For \posscite{Herbelin'95} $\lmmt$-calculus , the situation is again different.
Reduction in $\lmmt$ is not confluent; it has a critical pair in the term $\wcell<`M`a.c_1|\mt x.c_2>$, that can be contracted in both directions, to both $c_1 \tsubst[\mt x.c_2/`a] $ and $c_2 \tsubst[`M`a.c_1/x]$ (see \Def\ref{lmmt reduction}), with possibly different results.
\namecite{Curien-Herbelin'00} define {\CBN} reduction by not allowing the first contraction, and {\CBV} by not allowing the second; they do not consider strategies for either. 
They show that their interpretation of {\CBN} $\lmu$ into $\lmmt$ (see \Def\ref{CBN CBV Herbelin}) respects reduction in the {\CBN} and {\CBV} sub-systems (modulo $`m$-expansion); moreover, the {\CBV} reduction considered in \cite{Curien-Herbelin'00} for $\lmu$ just restricts $`b$-reduction to terms of the shape $`@ (`Lx.M) V $.
Similar results are obtained in \cite{Rocheteau'05}, but for an extended notion of $\lmu$.

A similar situation exists for $`X$, which has a critical pair in $\cut P `a + x Q $ that in certain circumstances can reduce to both $\cutL P `a + x Q $ and $\cutR P `a + x Q $, terms that (can) run to different results, and {\CBN} and {\CBV} reduction can be defined by blocking one or the other.
\Citeauthor{Bakel-Lescanne-MSCS'08} define an interpretation of $\lmu$ and $\lmmt$ into $`X$, and show that reduction is respected modulo $`m$-expansion; it deals with {\CBN} and {\CBV} reduction only for the interpretations of the {\LC} and $\Lx$ into $`X$.

Notions of {\CBN} and {\CBV} sub-reduction systems for each of the calculi $\lmu$, $\lmmt$, and $`X$ are thus obtained as restrictions on their respective notions of reduction, that differ strongly in origin and character, but are shown to correspond through various interpretations.
This now naturally leads to the following questions: how do these notions compare? and what if reduction \emph{strategies} are considered? what then, if any, is the relation between these different notions? are these the same and correct restrictions?
More broadly, can one be sure that the notions of {\CBN} and {\CBV} reduction and strategies are compatible between the calculi?
Moreover, we will see through the results of our paper that the standard mapping from $\lmmt$ to $\X$ as defined by \citet{Lengrand'03}, respects full reduction, as well as {\CBN} and {\CBV} reduction and strategies. 
So, in particular, there is no need for a separate {\CBN} or {\CBV} interpretation.
Seen that both calculi have clearly defined {\CBN} and {\CBV} reduction and strategies, this seems natural, but that then raises the question: how about $\slmu$ and $\lmmt$?

Notions of {\CBN} and {\CBV} sub-reduction system for each of the calculi $\lmu$, $\lmmt$, and $\X$ are thus obtained as restrictions on their respective notions of reduction, that differ strongly in origin and character, but are shown to correspond through various interpretations.
This now naturally leads to the following questions: what if reduction \emph{strategies} are considered? what then, if any, is the relation between these different notions? are these the same and correct restrictions? 
More broadly, can one be sure that the notions of {\CBN} and {\CBV} strategies are compatible between the calculi?
 
In this paper, we mainly focus on {\CBN} and {\CBV} \emph{reduction strategies}, for all calculi we consider. 
Although these are defined in very different ways, in \Def\ref{our lmu to lmmt} we will present $\lmulmmt[`.]{`.}$, a single, natural interpretation from $\nclmu$ to $\lmmt$, that respects reduction.
As to respecting reduction and strategies, we will argue in \Rem\ref{Herbelin not enough} that it is not possible to show that an interpretation fully respects of $\nclmu$; because of the fact that $\lmu$'s $`m$-reduction adds applications, and $\lmmt$'s does not, $\lmu$'s $`m$-substitution cannot simply be represented by term substitution in $\lmmt$; we will state the provable results that \emph{can} be shown.

To achieve a similar result when comparing $\lmmt$ and $`X$, the situation is slightly more complex, but also easier.
The added complexity lies in the fact that $`X$ is a calculus defined without substitution; reduction is defined by explicitly moving a term through the syntactic structure of another in small steps, much like explicit substitution is defined for \posscite{Bloo-Rose'95} $\Lx$ .
Defining a reduction strategy for a calculus with explicit substitutions is a different story altogether: for the $`l$-calculus, it would only propagate towards the head-variable, not the other variables, as done, for example, in \cite{Bakel-Vigliotti-CONCUR'09}. 
Is it therefore not straightforward to compare {\CBN} and {\CBV} strategies for $\lmmt$ and $\X$; to remedy this, here we will define $\Xs$, a version of $`X$ that uses (implicit) substitution, for which we will follow \namecite{Summers-PhD'08} and define {\CBN} and {\CBV} strategies, and establish the relation between $\lmmt$ and $\Xs$.
We will define $\lmmtX(`.){`.}$, an interpretation from $\lmmt$ to $\Xs$, and show that it respects all reductions and strategies.
We will also define $\Xlmmt[`.]$, a mapping from $\X$ to $\lmmt$ and show that reduction is preserved, but the strategies are not.

We will then compose the various interpretations and define one for $\slmu$ into $\X$, and show in what way it preserves full, \CBN, and {\CBV} reduction; in view of the shortcomings of $\lmulmmt[`.]$, we will not state results about the strategies.

\paragraph{Subsystems of Classical Calculi} Previous work has shown that each of the {\CBN} and {\CBV} subsystems agree -- that is, for example, there is a mapping of call-by-name $\lmu$ into call-by-name $\lmmt$, and $\lmuV$ into call-by-value $\lmmt$ \cite{Curien-Herbelin'00} -- but the mappings are \emph{per evaluation discipline}, and thus do not answer if the restrictions on reduction are analogous.
For example, in $\lmuV$ we have $x(\Mu `a.\Cmd) \reduc \Mu `g.\Cmd [\musub x.`g/`a] $.
However, the {\CBN} translation of \citet{Curien-Herbelin'00}, which gives $\Mu `d.<x|(\Mu `a.c)`.`d>$, does not reduce to to the translation of $\Mu `g.c [\musub x.`g/`a] $.

The question we address in this paper is: is there is a single mapping from $\nclmu$ into $\lmmt$, and a single mapping from $\lmmt$ into $\X$, such that equality is preserved under {\CBN} and {\CBV} reduction, and after restricting the constituent systems, through {\CBN} and {\CBV} reduction strategies?
Concretely, if $M$ and $N$ are a $\nclmu$-terms such that $M \reduc N$, then we wish not just for 
$ \Sem[M] \eqlmmt \Sem[N] $, but also if $M \redCBV N$ then $ \Sem[M] \eqlmmtV \Sem[N] $, and if $M \redCBN N$ then ${\Sem[M] \eqlmmtN \Sem[N]}$, where the latter two equations mean the two terms are equal through paths using only by-value and by-name reductions, respectively, but not necessarily in the same direction.
This latter fact can be understood by seeing that, as in \cite{Curien-Herbelin'00,Bakel-Lescanne-MSCS'08}, it is impossible to show these results with respect to just reduction, since we have to consider reductions going in the opposite direction.
These will be limited to reductions of redexes that are created by the interpretation whilst dealing with the added applications that are the result of $\slmu$'s $`m$-reduction.

This paper presents two such novel mappings, satisfying a subtly stronger property, that a single reduction $M \reduc N$ in the source calculus induces at least one reduction in the target.
More specifically, for the $\lmu$ to $\lmmt$ translation, there is a $t$ such that $\Sem[M] \tcredlmmt t$ and $\Sem[N] \rtcredlmmt t$, and similarly for $\redlmmtV$ and $\redlmmtN$.
The key in the mapping from $\lmu$ to $\lmmt$ is that it allows for an extra possible reduction which allows the head of an application to capture its context in {\CBV}.

\paragraph{Overview}
The first three sections introduce the three calculi of concern: $\nclmu$, $\lmmt$, and $\Xs$.
The first mapping, $\lmulmmt[`.]$, from $\nclmu$ to $\lmmt$ is given in \Sect\ref{sec:TranslateLmuLmmt}, while 
\Sect\ref{sec:InterpLmmtX} defines the second mapping, $\lmmtX(.)$, from $\lmmt$ into $`X$.
Both the mappings $\lmulmmt[`.]$ and $\lmmtX(.)$ respect equality, assignable types, and call-by-name and call-by-value equality, and furthermore respect that a reduction in the source calculus gives rise to at least one reduction in the target.
In \Sect\ref{Xi in lmmt} we will study the natural encoding $\Xlmmt[`.]$ of $`X$'s terms into $\lmmt$.
Optimising that encoding slightly by avoiding to create too many cuts, we will show that $ \lmmtX(.) $ is the right-inverse of $\Xlmmt[`.]$, but that $\Xlmmt[`.]$ is only $ \lmmtX(.) $'s left-inverse up to extensionality.
We will show that $\Xlmmt[`.]$ respects reduction in $\Xs$ by equality in $\lmmt$: to simulate $\Xs$'s substitution in $\lmmt$ under the encoding, in the image rules $(`m)$ and $(\MuT)$ are needed, so this encoding cannot respect the two strategies.
In \Sect \ref{slmu in X} we will combine the efforts of \Sect\ref{sec:TranslateLmuLmmt} and \ref{sec:InterpLmmtX} to define a mapping $\slmuX(`.)`. $ from $\slmu$ to $\Xs$ and prove that full, \CBN, and {\CBV} reduction are preserved.




\setbox81=\hbox{\textrm{C}}
\renewqsymbol{`C}{\copy81}

\def \reductionfigure{ \begin{figure*}[t]
\point.9pt	%
 \[ \def \Turnlmu{\Turn} 
\Inf	[\arrE]
	{\Inf	[`m]
		{\Inf	{\Inf	[`m]
			{\InfBox{ \derlmu `G |- M : A\arr B | `a`:A\arr B,`g`:D,`D }
			}{~ \derlmu `G |- {`M`g . [`a] M } : D | `a`:A\arr B,`D ~ }
		}{\InfBox<\Choose{126}{146}>{\derlmu `G |- \Cont[{`M`g . [`a] M }] : C | `a`:A\arr B,`D } }
		}{\derlmu `G |- {`M`a . [`b] \Cont[{`M`g . [`a] M }]} : A\arr B | `D }
\InfBox{ \derlmu `G |- N : A | `D }
	}{ \derlmu `G |- {`@ (`M`a . [`b] \Cont[{`M`g . [`a] M }]) N } : B | `D }
 \kern-10mm
\Inf	[`m]
	{\Inf	{\Inf	[`m]
		{\Inf	[\arrE]
			{\InfBox{ \derlmu `G |- M : A\arr B | `d`:B,`g`:D,`D }
	 \Inf	[\Weak]
			{\InfBox{ \derlmu `G |- N : A | `D }
			}{\derlmu `G |- N : A | `d`:B,`g`:D,`D }
			}{\derlmu `G |- {`@ M N } : B | `d`:B,`g`:D,`D }
		}{ \derlmu `G |- {`M`g . [`d] `@ M N } : D | `d`:B,`D }
		}{\InfBox<\Choose{116}{134}>{\derlmu `G |- \Cont[{`M`g . [`d] `@ M N }] : C | `d`:B,`D } }
	}{\derlmu `G |- {`M`d . [`b] \Cont[{`M`g . [`d] `@ M N }]} : B | `D }
 \]
 \[ \def \Turnlmu{\Turn} \def \ContS{\copy169}
\Inf	[\arrE]
	{\InfBox{ \derlmu `G |- N : A\arr B | `D }
	 \Inf	[`m]
		{\Inf	{\Inf	[`m]
			{\InfBox{ \derlmu `G |- M : A | `a`:A,`g`:D,`D }
			}{ \derlmu `G |- {`M`g . [`a] M } : D | `a`:A,`D }
		}{\InfBox<\Choose{108}{128}>{\derlmu `G |- \Cont[{`M`g . [`a] M }] : C | `a`:A,`D } }
		}{\derlmu `G |- {`M`a . [`b] \Cont[{`M`g . [`a] M }]} : A | `D }
	}{ \derlmu `G |- {`@ N (`M`a . [`b] \Cont[{`M`g . [`a] M }]) } : B | `D }
 \kern-4mm
\Inf	[`m]
	{\Inf	{\Inf	[`m]
		{\Inf	[\arrE]
			{\Inf	[\Weak]
			{\InfBox{ \derlmu `G |- N : A\arr B | `D }
			}{\derlmu `G |- N : A\arr B | `d`:B,`g`:D,`D }
	 \InfBox{ \derlmu `G |- M : A | `d`:B,`g`:D,`D }
			}{\derlmu `G |- {`@ N M } : B | `d`:B,`g`:D,`D }
		}{\derlmu `G |- {`M`g . [`d] `@ N M } : D | `d`:B,`D }
		}{\InfBox<\Choose{114}{134}>{\derlmu `G |- \Cont[{`M`g . [`d] `@ N M }] : C | `d`:B,`D } }
	}{\derlmu `G |- {`M`d . [`b] \Cont[{`M`g . [`d] `@ N M }]} : B | `D }
 \]
 \caption{An illustration of structural reduction in $\lmu$.}
 \label{Reduction in lmu}
 \end{figure*}}

 \section{{\nclmuPDF}: The Symmetric {\lmuPDF}-Calculus} 
 \label{lmu}
 \label{sec:LMu}

 \citet{Parigot'92} introduced $\lmu$ as an extension of the {\CBN} $\lambda$-calculus by adding an operator $\mu$, along with its associated reduction rules,
 \[ \begin{array}{rcl}
 `@ (`M`a . [`a] M ) N & \reduc & `M`g . [`g] `@ (M\mursubst[N.`g/`a]) N, \\
 `@ (`M`a . [`b] M ) N & \reduc & `M`g . [`b] M \mursubst[N.`g/`a],
 \end{array} \]
in which $`g$ is \textit{fresh}, and the `structural substitution' $\mursubst[N.`g/`a]$ recursively replaces each of $M$'s sub-terms of the form $[`a]P$ by $[`a]PN$ (see Definition \ref{def:Lmu:StructuralSubstitution}).
Unlike $`b$-reduction, $\mu$-reduction keeps the $\mu$-binding on the reduct:
this allows $\lmu$ to encode \emph{control operators} -- those that manipulate the (applicative) evaluation context, or continuation, of terms, by dealing with a term at the time.

The variant of $\lmu$ considered in this paper extends that defined by Parigot by including the `left-$\mu$' reduction rules (discussed at the end of Section 3.2 in \cite{Parigot'92}),
 \[ \begin{array}{rcl}
 `@ N (`M`a . [`a] M ) & \reduc & `M`g . [`g] `@ N (\mulsubst[N.`g/`a]M), \\
 `@ N (`M`a . [`b] M ) & \reduc & `M`g . [`b] \mulsubst[N.`g/`a]M,
 \end{array}\]
which ends up incorporating and generalising the \CBV-variant as defined by Ong and Stewart \mycite{Ong-Stewart'97}.
The resulting calculus corresponds to David and Nour's \emph{symmetric $\lmu$} ($\nclmu$) \cite{David-Nour-TLCA05}, and has a non-confluent reduction system, a property shared by the two other calculi with which we wish to compare $\nclmu$.

 \subsection{Proof-Theoretic Foundation of {\lmuPDF}}
 \label{sec:Lmu:ProofTheoreticFoundation}
An example of an approach for representing classical proofs, one can view $\lmu$ as a natural deduction system \cite{Prawitz'65} in which there is one main conclusion that is being manipulated and possibly several alternative ones.
It is a terms-as-proofs representation of an implicative classical logic with focus.
The formulas for this logical system $\TurnF$ are
 \[ \begin{array}{rcl}
	A,B &::=& `v \mid A\arrow B
 \end{array} \]
and the contexts $`G$ and $`D$ are multi-sets of formulas. 
It has the following inference rules:
{ \def \derF #1 |- #2 | #3 {#1 \Turn #2 \Mid #3 }
 \[ \begin{array}{rl@{\dquad}rl@{\dquad}rl}
(\Ax) : &
\Inf	{\derF `G,A |- A | `D }
	&
(\arrI) : &
\Inf	{\derF `G,A |- B | `D
	}{ \derF `G |- A\arrow B | `D }
	&
(\arrE) : &
\Inf	{\derF `G |- A\arrow B | `D
 \quad
 \derF `G |- A | `D
	}{ \derF `G |- B | `D }
 \end{array} \]
 \[ \begin{array}{rl@{\dquad}rl}
(\Act) : &
\Inf	{\derF `G |- `B | A,`D
	}{ \derF `G |- A | `D }
	&
(\Pass) : &
\Inf	{\derF `G |- A | A,`D
	}{ \derF `G |- `B | A,`D }
 \end{array} \] }

The intention of this system is to express (a) classical logic, for which it encapsulates the `Proof by Contradiction' inference rule, $(\PbC)$.
Viewing the formulas in $`D$ as negated assumptions, any statement $ \derF `G |- A | `D $ can be seen as $ \deriv `G,\neg `D |- A $ (where $\neg `D$ lists the negated versions of all formulas in $`D$).
With this understanding, the rules $(\Act)$ and $(\Pass)$ correspond to allowing the following variants of rules $(\PbC)$ and $(\negE)$:
 \[ \begin{array}{c@{\qquad}c}
\Inf	[\PbC]
	{ \derlog `G,\neg `D,\neg A |- `B
	}{ \derlog `G,\neg `D |- A }
	&
\Inf	[\negE]
	{\Inf	[\Ax]
		{ \derlog `G,\neg `D,\neg A |- \neg A }
 \quad
 \derlog `G,\neg `D,\neg A |- A
	}{ \derlog `G,\neg `D,\neg A |- `B }
 \end{array} \]

$\lmu$ seeks to provide a proof-term syntax for this deduction system, in which each logical rule corresponds to a term-constructor.
Typing annotates in judgements $ \derF `G |- A | `D $ each formulae in $`G$ and $`D$ with term-variables $x,y,z,\dots$ and $\mu$-variables (or names) $`a,`b,`g,\dots$, respectively, and the formula $A$ by a term $t$ whose variables/names occur in the annotated $`G$ and $`D$.
Variables, abstractions, and application encode $(\Ax)$, $(\arrI)$, and $(\arrE)$, as usual.
The two new rules, $(\Act)$ and $(\Pass)$, require term-constructors that explicitly interact with the set of conclusions $\Delta$.
These are the naming constructor $[`a]M$, denoted by $`C$, and name binding $`m`a.`C$:

{\def\Turnlmu{\Turn}%
 \[ \begin{array}{c@{\qquad}c}
(\Act):
\Inf	{ \derlmu `G |- C : `B | `a`:A,`D }
	{ \derlmu `G |- `m`a.`C: A | `D }
 &
(\Pass):
\Inf	{ \derlmu `G |- M : A | `a`:A,`D }
	{ \derlmu `G |- [`a]M : `B | `a`:A,`D }
 \end{array} \] %
}

This is of course only half of the work needed; the real difficulty is in determining how proofs using $(\Act)$ and $(\Pass)$ should contract, as these contractions will induce the notion of reduction on $\lmu$.
 \namecite{Parigot'92} bases the reductions of $\lmu$ on the cut-elimination procedure of earlier work on free deduction~\cite{Parigot'91}; we will cover the former in detail in the next section. 

 \subsection{Syntax and Reduction of {\nclmuPDF}}
 \label{sec:Lmu:Syntax}
 \label{sec:Lmu:Reduction}

The terms of $\nclmu$ are precisely those of $\lmu$ \cite{Parigot'92}.
 \begin{definition}[Syntax of $\nclmu$] \label{lm-terms}
	The following defines the syntax of $\lmu$, over \emph{variables} $x,y,\dots$ and \emph{names} $`a,`b,\dots$:
 \[\begin{array}{rcl@{\quad}l}
M,N &::=&  V \mid `@ M N \mid `m`a. `C & (\emph{terms})  \\
V &::=& x \mid `Lx.M & (\emph{values})  \\
`C &::=& [`b]M & (\emph{named terms, commands})
 \end{array} \] 

$`l$ and $`m$ bind variables and names, respectively. 
The notions of free and bound names/variables follow the usual definition, and we accept Barendregt's convention \cite{Barendregt'84} to keep free and bound names/variables distinct, using (silent) $`a$-conversion whenever necessary.

We write $x \ele M$ ($`a \ele M$) if $x$ ($`a$) occurs in $M$, either free of bound, and call a term \emph{closed} if it has no free names or variables.
For brevity, we will often treat commands $[`a]M$ as terms.
 \end{definition}

Morally, $`m`a.`C$ denotes a term that captures a fragment of its applicative context, and sends it to the $`a$-named sub-commands of $`C$, $[`a]M$.
In turn, $[`a]M$ denotes the term $M$ in the context to which $`a$ binds.
Whilst this contextual understanding of $\lmu$ aids intuition, one must remember that applicative contexts are not a primitive part of the syntax of $\lmu$, and that $\mu$ captures such contexts `piecewise' -- that is -- capturing only the portions $M"[]"$ or $"[]"N$ of the wider context, one term at a time.
The particular feature is due to $`m$-contractions re-introducing $`m$ abstractions, and will play a role in our encodings.

The reduction rules for terms represent contractions of their corresponding proofs in $\TurnF$.
The rules $\lmu$ inherits from the {\LC} are the \emph{logical} reductions, \emph{i.e.}~they deal with the removal of an introduction-elimination pair for a type constructor.
As logical reductions require substituting axioms by proofs of the formulae of said axioms, $\lmu$ uses (implicit) substitution for the corresponding proof-terms.
This is the usual notion of substitution in  ereby $M\tsubst[N/x]$ stands for the (instantaneous) substitution of all occurrences of $x$ in $M$ by $N$.
The system $\TurnF$ also features \emph{structural} cuts, that change the focus of a proof.
Understanding the alternative conclusions as negated premises,\footnote{Only names can have negated types, not terms; for a variant of $\lmu$ with negation as a type constructor, see \cite{Bakel-LMCS'23}.} the structural rules essentially deal with introduction-elimination pairs for negation with these premises, 
and swap to one of the other conclusions.
These structural rules then have proof contractions, which redirect the right-hand premise of the structural cut to multiple logical cuts in the left-hand premise, replacing occurrences of axioms using the main formula of the right-hand premise.

$\lmu$ uses the constructs $\mu`a.M$ and $[`a]M$ to witness \emph{activation} and \emph{passivation} in $\TurnF$, which together express that the focus of the derivation changes (see \Def\ref{tas lmu}) -- a \emph{context switch}.
In this sense, terms of the form $`@ (`m`a.`C_1 ) M_1 $ and $`@ M_2 (`m`b.`C_2 ) $ represent structural cuts.
It follows that $\lmu$ then needs a notion of `replacing structural cuts with many logical cuts' in order to match the same concept in $\TurnF$; this is the purpose of \emph{structural substitution}.

The aforementioned terms $`@ (`m`a.`C_1 ) M_1 $ and $`@ M_2 (`m`b.`C_2 ) $ must contract by generating the appropriate logical cuts in the subcommands of $`C_1$ and $`C_2$.
Concretely, the first must reduce to a term $(\mu`a.`C'_1)$, in which every subcommand of the form $[`a]N$ is replaced by $[`a]`@ N M_1 $ (\ie a use of the axiom $`a$), and the second to a term $(\mu`b.`C'_2)$, replacing $[`b]P$ by $[`b]`@ M_2 P $.
Looking towards typing, as the type of $`a$ must match that of $N$ in $[`a]N$, the type of $`a$ must change after the structural substitution; if we have $N : A\arrow B$, then $`@ N M_1 : B$, meaning $`a$ must have type $A\arrow B$ in the former, and type $B$ in the latter.
It will make more sense to instead replace $`a$ by a new, fresh variable, so that we need not worry about these changing types.

The first notion of structural substitution, defined by \citet{Parigot'92}, is then written $M \mursubst[N.`g/`a]$, denoting the term obtained from $M$ in which every command of the form $[`a]P$ is replaced by $[`g] P N$, with $`g $ a fresh name.
The second notion, suggested by \citet{Ong-Stewart'97} with a different notation, is $\mulsubst[N.`g/`a] M$, in which every $[`a]P$ in $M$ is replaced by $[`g] N P$.

 \begin{definition}[Structural substitution\footnotemark]
 \label{def:Lmu:StructuralSubstitution}
 \emph{Right-structural substitution}, $M \mursubst[N.`g/`a]$, and \emph{left-structu\-ral substitution}, $\mulsubst[N.`g/`a] M $, are defined inductively over pseudo-terms:
 \[\begin{array}{c@{\dquad}c}	
 \begin{array}[t]{@{}r@{\dsk}l@{~}c@{~}l}
x & \mursubst[N.`g/`a] & \ByDef & x \\
(`Lx.M) & \mursubst[N.`g/`a] & \ByDef & `Lx . (M \mursubst[N.`g/`a] ) \\
(PQ) & \mursubst[N.`g/`a] & \ByDef & `@ (P \mursubst[N.`g/`a]) (Q \mursubst[N.`g/`a]) \\ {}
[`a]M & \mursubst[N.`g/`a] & \ByDef & [`g] `@ (M \mursubst[N.`g/`a]) N \\ {}
[`b]M & \mursubst[N.`g/`a] & \ByDef & [`b] ( M \mursubst[N.`g/`a] ) \quad (`b \not= `a) \\
(`m`b.`C ) & \mursubst[N.`g/`a] & \ByDef & `m`d. `C \mursubst[N.`g/`a]
 \end{array}
	&
 \begin{array}[t]{r@{\dsk}l@{~}c@{~}l@{}}
 \mulsubst[N.`g/`a] &x & \ByDef & x \\
 \mulsubst[N.`g/`a] & (`Lx.M) & \ByDef & `Lx . (\mulsubst[N.`g/`a] M) \\
 \mulsubst[N.`g/`a] & (PQ) & \ByDef & `@ (\mulsubst[N.`g/`a] P ) (\mulsubst[N.`g/`a] Q ) \\
 \mulsubst[N.`g/`a] & [`a]M & \ByDef & [`g] `@ N (\mulsubst[N.`g/`a] M ) \\
 \mulsubst[N.`g/`a] & [`b]M & \ByDef & [`b] \mulsubst[N.`g/`a] M \quad (`b \not= `a) \\
 \mulsubst[N.`g/`a] & `m`b.`C & \ByDef & `M`d . \mulsubst[N.`g/`a] `C
 \end{array}
 \end{array} \]
 \end{definition}
 \footnotetext{%
 \citet{Parigot'92} only defines the first variant of these notions of structural substitutions (so does not use the prefix `right'); \citet{Ong-Stewart'97} define the two notions together using a notion of contexts.
}

With structural substitution in place, we can now explore reduction of terms using $\mu$.
Different variants of $\mur$ and $\mul$, such as call-by-name and -value, exist in the literature; we follow the definitions of \citet{Ariola-Herbelin'03}.

 \begin{definition}[$\nclmu $ reduction] 
 \label{mu reduction definition}
 \label{def:Lmu:Reduction}
 \begin{enumerate}
 \firstitem
The reduction rules of $\nclmu$ are:
 \[\begin{array}{rrcl@{\quad}l}
 \textit{logical } (`b) : & `@ (`Lx.M ) N & \reduc & M \tsubst[N/x] \\
 \textit{right-structural } (\mur) : & `@ (`m`a . `C ) N & \reduc & `M`g . `C \mursubst[N.`g/`a] & (`g\textit{ fresh}) \\
 \textit{left-structural } (\mul) : & `@ M (`m`a.`C) &\reduc& `M`g . \mulsubst[M.`g/`a] `C & ( `g\textit{ fresh}) \\
 \textit{renaming } (\Rename) : & [`b] `m`g.`C & \reduc & `C \tsubst[`b/`g] \\
 \textit{erasing } (\Erase) : & `M`a . [`a] M & \reduc & M & (`a\notele M)
 \end{array}\]

 \item
We write $ \rednclmu$ for the relation that is the compatible closure of these rules, and $ \rtcrednclmu$ for the reflexive and transitive closure of $\rednclmu$.

 \end{enumerate}
 \end{definition}
Term substitution $\tsubst[N/x]$ and the $(`b)$ rule are the same as for the \LC.
The $\mul$ and $\mur$ reductions encode the structural proof contractions of {$\TurnF$};
computationally, they capture a single-term applicative context to a name $`a$, and place every $`a$-named term in this context.
The rules $(\Rename)$ and $(\Erase)$ are best understood computationally.
The first says that binding to $`g$ the context that $`b$ denotes, is the same as just having each $`g$-named sub-term directly use the context $`b$.
The $(\Erase)$ rule is an $\eta$-like rule for $\mu$; if all we do with the current context $`a$ is run $M$ in that context, 
then we may as well not bother capturing said context.
Note that $(\Erase)$ is the only way to eliminate the top-most $\mu$ in a term of the form $`m`a .`C$.   
 
We highlight how $\mul$ and $\mur$ interact with commands.
For $\mur$ reduction, we have:
 \[ \begin{array}{rcl@{\quad}l}
	`@ (`M`a.[`b]M) N &\reduc& `M`g.[`b]M \mursubst[N.`g/`a] \quad (`b \not=`a) & \textrm{and} \\
	`@ (`M`a.[`a]M) N &\reduc& `M`g.[`g] `@ (M \mursubst[N.`g/`a]) N
 \end{array} 
 \]
and for $\mul$ reduction, we have:
 \[ \begin{array}{rcl@{\quad}l}
`@ N (`M`a.[`b]M) &\reduc& `M`g.[`b] \mulsubst[N.`g/`a] M \quad (`b \not=`a) & \textrm{and} \\
`@ N (`M`a.[`a]M) &\reduc& `M`g.[`g] `@ N (\mulsubst[N.`g/`a] M)
 \end{array} \]
In each reduction, the structural substitution for a name $`a$ only applies to an $`a$-named term; it otherwise skips the substitution and recurses on sub-terms.
Ultimately, each reduction deposits $N$ at every $`a$-named term, in either head (via $\mul$) or argument (via $\mur$) position.

The reduction rules have two critical pairs:
 \[ \begin{array}{r@{\dquad}c@{\dquad}l}
	`@ (`M`a . [`b] M ) (`M`g . [`d] N ) & \textrm{and} &
	`@ (`Lx . M) (`M`g . [`d] N ) .
 \end{array} \]
The first reduces to both
$ `M`s .[`b] M \mursubst[{(`M`g . [`d] N )}.`s/`a]$ and
$ `M`t . [`d] \mulsubst[{(`M`a . [`b] M)} . `t / `g] N $
(where all names are distinct), and these need not be the same term.
Similarly, the second reduces to both
$ M \tsubst[{(`M`g . [`d] N )}/x] $ and
$ `M `t . [`d] \mulsubst[{(`Lx . M)} . `t / `g] N $.
It follows that $\rednclmu$ is not confluent.

 \begin{remark} \label{added application} 
Notice that in the \LC, since $`b$-reduction is based on term-substitution, \ie replacing term variables by terms, it cannot create applications or abstractions: a variable $x$ can be replaced by an application $`@ M N $, but in the tree that represents the $`l$-term, no new application nodes can be created.
 
\reductionfigure

This is no longer true in $\nclmu$ (or $\lmu$, for that matter) where a $`m$-contraction can create applications; \Fig \ref{Reduction in lmu} illustrates the reduction steps
 \[ \begin{array}{r@{\quad}rcl}
	(\mur): & `@ (`M`a . [`b] \Cont[{`M`g . [`a] M }]) N &\rednclmu& `M`d . [`b] \Cont[{`M`g . [`d] `@ M N }] \\
	(\mul): & `@ N (`M`a . [`b] \Cont[{`M`g . [`a] M }]) &\rednclmu& `M`d . [`b] \Cont[{`M`g . [`d] `@ N M }]
 \end{array} \]
 (where $ `b`:C \ele `D $ and $`a$ does not occur in $M$).
Here the application $`@ M N $ in the first, and $`@ N M $ in the second, are new. 
This is a special feature of $\lmu$, and is not part of either $\lmmt$ nor $\X$, even when adapting to the different formalisms, and creates anomalies when interpreting $\nclmu$ into those other calculi.
Since both the left and right structural substitution create an application, a simulation would have to account for that. 
We will see this addressed in the sections that follow.
 \end{remark}

Historically, removing the critical pairs has led to the definition of {\CBN} and {\CBV} sub-reduction systems.
For \CBN, the standard restriction simply deletes the rule $(\mul)$, which then yields Parigot's original $\lmu$ calculus.\footnote{It is worthwhile noticing that $\lmu$ still has a critical pair in (assuming all names are distinct)
 \[ \begin{array}{rcl}
 `@ ( `M `a . `b . `M `g . `d P ) Q & \redlmu & 
 	\begin{cases}
`M `s . `b `M `g . `d P \mursubst[Q.`s/`a] \\
 `@ ( `M `a . `d P \tsubst[`b/`g] ) Q 
	\end{cases} 
 \end{array} \]
However, this does not create a problem; \citet{Py-PhD'98} shows that $\lmu$-reduction is confluent.
Notice that the first reduction step destroys the outer-most redex; this created a problem with an earlier attempt to prove confluence using Tait and Martin-L\"of's technique (see \cite{Barendregt'84,Pfenning'92});
Py solved this problem by using \posscite{Aczel'78} variant of the definition of parallel reduction. 
}
For \CBV, many different approaches exist: one approach eliminates the critical pairs by limiting the rules $(\mul)$ and $(`b)$ to only apply in case the operand is a value $V$ (\ie a variable, or an abstraction):
 \[ \begin{array}{rcl@{\quad}l}
	`@ (`Lx.M ) V & \reduc & M \tsubst[V/x] \\
	`@ V (`m`a.`C) &\reduc& `M`g . \mulsubst[V.`g/`a] `C & (`g\textit{ fresh})
 \end{array} 
 \]
The term $ `@ (`M`a . [`b] M ) (`M`g . [`d] N ) $ can then only be a $\mur$-redex, and $ `@ (`Lx . M) (`M`g . [`d] N ) $ only a $\mul$-redex.
This notion of {\CBV} relies upon not considering $`m$-abstractions values, although one could argue that it can be seen as a meaningful term.
This suggests allowing the contraction of \emph{any} redex only when the argument is a value and changing rule $(\mur)$ as well:
 \[ \begin{array}{rcl@{\quad}l}
	`@ (`m`a . `C ) V & \reduc & `M`g . `C \mursubst[V.`g/`a] & (`g\textit{ fresh})
 \end{array} \]
as for example \citet{Rocheteau'05} does, but not every {\CBV} reduction for $\lmu$ is defined in this way.

We focus on a more specific sort of sub-reduction system -- \emph{reduction strategies} -- in which any term has at most one redex.

 \begin{definition}[{\CBN} and {\CBV} reduction and strategies for $\nclmu$] \label{cbn cbv nclmu}

 \begin{enumerate}
 \firstitem \label{cbn def}
Call-by-name reduction $\rednclmun$ is defined by eliminating rule $(\mul)$ from $\rednclmu$.

 \item
The {\CBN} \emph{evaluation strategy} $\rednclmuN$ is a restriction of $\rednclmun$ by limiting the contextual closure to:
 \[\begin{array}{rcl}
P \reduc Q &\Then&
 \begin{cases}{lcl@{\dquad}l}
`@ P M &\reduc& `@ Q M & ( P \not= `m`g.`C ) \\
`M`a . [`b] P &\reduc& `M`a . [`b] Q & ( P \not= `m`g.`C ) \\
 \end{cases}
 \end{array}\]

 \item 
Call-by-value reduction on $\slmu$, $\rednclmuv$ is defined by restricting the rules $(`b)$ and $(\mul)$:
 \[ \begin{array}[t]{rrcl@{\quad}l}
(`b_{\vsubscr}) : & `@ ( `Lx.M ) V &\reduc& M\tsubst[V/x]
 \\
(\mulv) : & `@ V (`m`a.`C) &\reduc& `M`g . \mulsubst[V.`g/`a] `C & ( `g\textit{ fresh})
 \end{array} \]

 \item
The {\CBV} \emph{evaluation strategy} $\rednclmuV$  is a restriction of $\rednclmuv$ by limiting the contextual closure to: %
 \[\begin{array}{rcl}
P \reduc Q &\Then&
	\begin{cases}{lcl@{\dquad}l}
		`@ P M &\reduc& `@ Q M & ( P \not= `m`g.`C ) \\
		`@ V P &\reduc& `@ V Q & ( P \emph{ is an application}) \\
		`M`a . [`b] P &\reduc& `M`a . [`b] Q & ( P \not= `m`g.`C ) \\
	\end{cases}
 \end{array}\]

 \end{enumerate}
 \end{definition}

Both $\rednclmuN$ and $\rednclmuV$ are reduction strategies as they pick precisely one $\rednclmu$-redex to contract.
These strategies are much stricter than $\rednclmu$ reduction: a term may be in either {\CBN} or \CBV-normal form (\emph{i.e.}~reduction has stopped), but need not be a $\rednclmu$-normal form.

 \subsection{Typing and Soundness for {\nclmuPDF}}
 \label{sec:Lmu:Typing}

Since terms of $\nclmu$ are the terms of $\lmu$, type assignment is exactly the same for both.
Judgements are of the shape ${\derlmu `G |- M : A | `D }$, in which $`D$ consists of pairs of names and types; the left-hand context $`G$, as for the $`l$-calculus, contains pairs of variables and types, and represents the types of the free term variables of $M$.
The right-hand context $`D$ represents the \emph{alternative conclusions}, and contains the free names of $M$; each typed name $`b `: B$ should match the the types of all terms $N$ such that $[`b]N$ is a subterm of $M$.
The formula $A$ in the judgement is the \emph{main}, or \emph{active}, conclusion, labelled by the term $M$.
 \begin{definition}[Typing for $\lmu$ \cite{Parigot-Brno'93}] \label{tas lmu} \label{lmu rules}
 \begin{enumerate}
 \firstitem
	With $`v$ ranging over a countably infinite set of type-variables,
	types have the grammar:
 \[ \begin{array}{rcl}
		A,B &::=& `v \mid (A\arrow B)
 \end{array} \]
We will remove brackets in types when possible, using $ A \arrow B \arrow C = (A \arrow (B \arrow C)) $ 
 \item
	A \emph{context} (of term variables) $`G$ is a partial mapping from term variables to types, denoted as a finite set of \emph{statements} $x`:A$, such that the \emph{subjects} of the statements ($x$) are distinct.
	We write $`G_1,`G_2$ for the \emph{compatible} union of $`G_1$ and $`G_2$ (if $x`:A_1 \ele `G_1$ and $x`:A_2 \ele `G_2$, then $A_1 = A_2$); $`G, x`:A$ for $`G, \Set{x`:A}$; $x \notele `G$ if there exists no $A$ such that $x`:A \ele `G$; and $`G\Except x$ for $`G\Except \Set{x`:A}$.

 \item
	A \emph{context of names} $`D$ (or \emph{co-context}) is a partial mapping from \emph{names} to types, denoted as a finite set of \emph{statements} $`a`:A$, such that the \emph{subjects} of the statements ($`a$) are distinct.
	The notions $`D_1,`D_2$, as well as $`D, `a`:A$ and $`a \notele `D$ are analogous to those for contexts.

 \item
	A \emph{judgement} is an expression of the shape $\derlmu `G |- M : A | `D $;
	we extend the notion of free and bound variables and names to judgements $\derlmu `G |- M : A | `D $ and consider the term variables appearing in $`G$ and the names occurring in $`D$ as binding the free occurrences in $M$.

 \item
		The type assignment rules for $\lmu$ are:
		{\def\Turnlmu{\Turn}
 \[
	 \begin{array}{rl@{\quad}rl@{\quad}rl}
			(\Ax) : &
	 \Inf	{ \derlmu `G,x`:A |- x : A | `D }
			&
			(\arrI) : &
	 \Inf	{ \derlmu `G,x`:A |- M : B | `D
			}{ \derlmu `G |- {`Lx.M} : A\arrow B | `D }
			&
			(\arrE) : &
	 \Inf	{ \derlmu `G |- M : A\arrow B | `D
		 \quad
		 \derlmu `G |- N : A | `D
			}{ \derlmu `G |- { M N } : B | `D }
 \end{array} \]
 \[
	 \begin{array}{rl@{\quad}l}
			(`m) : &
	 \Inf	{ \derlmu `G |- M : B | `a`:A,`b`:B,`D
			}{ \derlmu `G |- {`M`a.[`b]M} : A | `b`:B,`D }
	 \quad
	 \Inf	{ \derlmu `G |- M : A | `a`:A,`D
			}{ \derlmu `G |- {`M`a.[`a]M} : A | `D }
 \end{array} \] }
		We will write $\derlmu `G |- M : A | `D $ for judgements derivable in this system.

 \end{enumerate}
 \end{definition}
The rules $(\Ax)$, $(\arrI)$, and $(\arrE)$ are as usual for the \LC, and do not interact with the co-context.
The $(`m)$ rules are best understood in two parts:
if $M$ is a term of type $A$, and $`a$ is a name of type $A$, then $[`a]M$ is a `well-formed' command;
then, if we have a well-formed command $`C$ with an alternate conclusion $`b `: B$, we can activate this conclusion.
One can think of $[`a]M$ as storing the type of $M$ amongst the alternative conclusions by giving it the name $`a $; $\mu`a.`C$ will then retrieve the type of $M$ from these conclusions.

Computationally, one can interpret names as references to continuations (\ie applicative contexts), so that $`a `: A$ occurs in the co-context means that the term has access to a continuation that accepts a term of type $A$.
The well-formed command part of the $(\mu)$ reduction rule then says that, if we also have a term $M$ of the same type $A$, we can send it to the continuation $`a$ in the command $[`a]M$.
The typing of $\mu$-binding represents capturing such a continuation; if a continuation expects a term of type $A$, and a command $`C$ has sent terms of type $A$ to $`a$, then, by assigning the type $A$ to $\mu`a.`C$, we ensure that this term will only appear in continuations expecting terms of type $A$, thus making them valid to capture.

The type system is strongly related to $\TurnF$: if we erase all term information from the inference rules, we get the rules from $\TurnF$, except for the variants of rule $(`m)$, which we can infer:
 \[
\Inf	[\Act]
	{\Inf	[\Pass]
	{\InfBox{\derlog `G |- B | A,B,`D }
	}{ \derlog `G |- `B | A,B,`D }
	}{ \derlog `G |- A | B,`D }
 \qquad
\Inf	[\Act]
	{\Inf	[\Pass]
	{\InfBox{\derlog `G |- A | A,`D }
	}{ \derlog `G |- `B | A,`D }
	}{ \derlog `G |- A | `D }
 \]
Thus, the erased $(`m)$ rules are derivable in $\TurnF$.

The following result is standard and of use in the proofs below.

 \begin{lemma} [Weakening and thinning for $\Turnlmu$] \label{lmn thinning lemma} \label{lmn weakening lemma}
	The following rules for \emph{weakening} and \emph{thinning} are admissible for $\Turnlmu$:
 \[ \def \Turnlmu {\Turn}
 \begin{array}{rl@{\dquad}rl}
		(\Weak) : &
\Inf	[`G\subseteq `G',`D\subseteq `D']
		{\derlmu `G |- M : A | `D
		}{\derlmu `G' |- M : A | `D' }
		&
		(\Thin) : &
\Inf	[<45mm>{`G' = \Set{x`:B \ele `G \mid x \ele \fv(M)}, ~ \\ ~ `D' = \Set{`a`:B \ele `D \mid `a \ele \fn(M)}}]
		{\derlmu `G |- M : A | `D
		}{\derlmu `G' |- M : A | `D' }
 \end{array} \]
 \end{lemma}
 \begin{proof}
	Standard.\qed
 \end{proof}

We will now show that type assignment is sound, \ie is closed under reduction.
This result might itself be as expected, and is presented here mostly for completeness.
 \Comment{
	We use the results for the three notions of term substitution from \cite{Bakel-PPDP'19}.

 \begin{proposition}[Substitution lemma \cite{Bakel-PPDP'19}]
 \label{term substitution lemma} \label{lmu substitution lemma}
 \label{right structural substitution lemma} \label{left structural substitution lemma}

 \begin{enumerate}

 \firstitem
		If $\derlmu `G,x`:B |- M : A | `D $ and $\derlmu `G |- L : B | `D $, then $\derlmu `G |- M\tsubst[L/x] : A | `D $.

 \item
		If $\derlmu `G |- M : A | `a`:B\arrow C,`D $ and $ \derlmu `G |- L : B | `D $, then $\derlmu `G |- M\mursubst[L.`g/`a] : A | `g`:C,`D $.

 \item
		If $ \derlmu `G |- L : B\arrow C | `D $ and $ \derlmu `G |- M : A | `a`:B,`D $, then $ \derlmu `G |- {\mulsubst[L.`g/`a] M} : A | `g`:C,`D $.

 \end{enumerate}
 \end{proposition}
}

First we show results for the three notions of term substitution.

 \begin{lemma}[Substitution lemma]
 \label{lem:Lmu:Types:Substitution}
 \label{term substitution lemma} 
 \label{lmu substitution lemma}
 \label{right structural substitution lemma} 
 \label{left structural substitution lemma}

 \begin{enumerate}

 \firstitem If $\derlmu `G,x`:B |- M : A | `D $ and $\derlmu `G |- L : B | `D $, then $\derlmu `G |- M\tsubst[L/x] : A | `D $.

 \item
If $\derlmu `G |- M : A | `a`:B\arrow C,`D $ and $ \derlmu `G |- L : B | `D $, then $\derlmu `G |- M\mursubst[L.`g/`a] : A | `g`:C,`D $.

 \item
If $ \derlmu `G |- L : B\arrow C | `D $ and $ \derlmu `G |- M : A | `a`:B,`D $, then $ \derlmu `G |- {\mulsubst[L.`g/`a] M} : A | `g`:C,`D $.

 \end{enumerate}
 \end{lemma}

 \begin{Proof}{By induction on the definition of the three kinds of substitution.}{\Lmm \ref{lmu substitution lemma}}
%
\begin {enumerate}
\itemsep 4\point

\firstitem
By induction on the definition of term substitution.

\Longer
{
\begin{description} \itemsep 2pt

        \item [{$ x\tsubst[L/x] = L $}]
            If $ \derlmu `G,x`:B |- x : A | `D $, then $B = A$, so $ \derlmu `G |- L : A | `D $, so also $ \derlmu `G |- x[L/x] : A | `D $.

        \item [{$ y \tsubst[L/x] \same y $ $(y \not= x)$}]
            Then $y`:A \ele `G$ and by rule $(\Ax)$ we have $ \derlmu `G |- y : A | `D $.

        \item [{$ `ly.(N\tsubst[L/x]) = (`ly.N)\tsubst[L/x] $}]
            Then $A = C\arrow D$.
            If $ \derlmu `G,x`:B |- `ly.N : C\arrow D | `D $, then by rule $(\arrI)$, $ \derlmu `G,x`:B,y`:C |- N : D | `D $.
            Then by induction, $ \derlmu `G,y`:C |- N\tsubst[L/x] : D | `D $, so by $(\arrI)$, $ \derlmu `G |- `ly.(N\tsubst[L/x]) : C\arrow D | `D $.

        \item [{$ (PQ)\tsubst[L/x] \equiv `@ P\tsubst[L/x] Q\tsubst[L/x] $}]
            If $ \derlmu `G,x`:B |- PQ : A | `D $, then, by rule $(\arrE)$ there exist $C$ such that both $\derlmu `G,x`:B |- P : C\arrow A | `D $ and $\derlmu `G,x`:B' |- Q : C | `D $.
            Then, by induction, $\derlmu `G |- P\tsubst[L/x] : C\arrow A | `D $ and $\derlmu `G |- Q\tsubst[L/x] : C | `D $; the result follows by rule $(\arrE)$.

        \item [{$ (`M`a . [`b] N)\tsubst[L/x] \same `M`a . [`b] N\tsubst[L/x] $}]
            If $ \derlmu `G,x`:B |- `M`a.[`b]N : A | `D $, then, by rule $(`m)$ there exist $C$ such that $\derlmu `G,x`:B |- N : C | `a`:A,`b`:C,`D' $ with $`D = `b`:C,`D'$.
            Then, by induction, $\derlmu `G |- N\tsubst[L/x] : C | `a`:A,`b`:C,`D' $, and by rule $(`m)$ we have $\derlmu `G |- N\tsubst[L/x] : C | `a`:A,`b`:C,`D' $.

    \end{description}
}

\item
By induction on the definition of right-structural substitution\Shorter{; we only show the interesting cases}.

\begin{description}

        \Longer
        {
        \item [{$ x\mursubst[L.`g/`a] \ByDef x $}]
            Then $x`:A \ele `G$, and by rule $(\Ax)$ we have $\derlmu `G |- x : A | `g`:C,`D $.

        \item [{$ (`Lx.N)\mursubst[L.`g/`a] \ByDef `Lx . (N\mursubst[L.`g/`a] ) $}]
            Then $A = D\arrow E$ and, by rule $(\arrI)$, $\derlmu `G,x`:D |- N : E | `a`:B\arrow C,`D $.
            Then by induction we have $\derlmu `G,x`:D |- {N\mursubst[L.`g/`a]} : E | `g`:C,`D $, so by rule $(\arrI)$ also $\derlmu `G |- { `Lx . {N\mursubst[L.`g/`a]} } : D\arrow E | `g`:C,`D $.

        \item [{$ (`@ P Q )\mursubst[L.`g/`a] \ByDef `@ P\mursubst[L.`g/`a] Q\mursubst[L.`g/`a]
            $}]
            Then by rule $(\arrE)$ there exists $D$ such that $ \derlmu `G |- P : D\arrow A | `a`:B\arrow C,`D $ and $ \derlmu `G |- Q : D | `a`:B\arrow C,`D $.
            Then by induction we can assume both $ \derlmu `G |- P\mursubst[L.`g/`a] : D\arrow A | `g`:C,`D $ and $ \derlmu `G |- Q\mursubst[L.`g/`a] : D | `g`:C,`D $; the result follows by rule $(\arrE)$.
        }

    \item [{$ `M`d.[`a]N\mursubst[L.`g/`a] \ByDef `M`d.[`g](N\mursubst[L.`g/`a] L) $}]
        Then by rule $(`m)$
        $ \derlmu `G |- N : B\arrow C | `d`:A,`a`:B\arrow C,`D $, and by induction
        $ \derlmu `G |- N\mursubst[L.`g/`a] : B\arrow C | `d`:A,`g`:C,`D $.
        Since $`a$, $`d$ and $`g$ all do not occur (free) in $N$, we can construct
        \[ \def\Turnlmu{\Turn}
            \Inf [`m]
            {\Inf [\arrE]
                {\InfBox { \derlmu `G |- N\mursubst[L.`g/`a] : B\arrow C | `d`:A,`g`:C,`D }
                    \quad
                    \Inf    [\Weak]
                    {\InfBox { \derlmu `G |- L : B | `D }
                    }{ \derlmu `G |- L : B | `d`:A,`g`:C,`D }
                }{ \derlmu `G |- { `@ (N\mursubst[L.`g/`a]) L } : C | `d`:A,`g`:C,`D }
            }{ \derlmu `G |- { `M`d.[`g] `@ (N\mursubst[L.`g/`a]) L } : A | `g`:C,`D }
        \]

    \item [{$ (`M`d.[`b]N)\mursubst[L.`g/`a] \ByDef `M`d.[`b] ( N\mursubst[L.`g/`a] ) ~ (`b \not= `a) $}]
        Then by rule $(`m)$ there exists $D$ such that $`D = `b`:D,`D'$, and
        $ \derlmu `G |- N : D | `d`:A,`b`:D,`a`:B\arrow C,`D' $, and by induction
        $ \derlmu `G |- N\mursubst[L.`g/`a] : D | `d`:A,`b`:D,`g`:C,`D' $.
        But then, by rule $(`m)$, also $ \derlmu `G |- `M`d.[`b] N\mursubst[L.`g/`a] : A | `b`:D,`g`:C,`D' $.

\end{description}

\item
By induction on the definition of left-structural substitution\Shorter{; we only show the interesting cases}.

    \begin{description}

            \Longer
            {
            \item [{$ \mulsubst[L.`g/`a] x \ByDef x $}]
                Then $x`:A \ele `G$, and by rule $(\Ax)$ we have $\derlmu `G |- x : A | `g`:C,`D $.

            \item [{$ \mulsubst[L.`g/`a] (`Lx.N) \ByDef `Lx . (\mulsubst[L.`g/`a] N) $}]
                Then $A = D\arrow E$ and, by rule $(\arrI)$, $\derlmu `G,x`:D |- N : E | `a`:B,`D $.
                Then by induction we have $\derlmu `G,x`:D |- {\mulsubst[L.`g/`a] N} : E | `g`:C,`D $, so by rule $(\arrI)$ also $\derlmu `G |- { `Lx . \mulsubst[L.`g/`a] N } : D\arrow E | `g`:C,`D $.

            \item [{$ \mulsubst[L.`g/`a] (`@ P Q ) \ByDef {\mulsubst[L.`g/`a] P} \, {\mulsubst[L.`g/`a] Q}
                $}]
                Then by rule $(\arrE)$ there exists $D$ such that $ \derlmu `G |- P : D\arrow A | `a`:B,`D $ and $ \derlmu `G |- Q : D | `a`:B,`D $.
                Then by induction both $ \derlmu `G |- {\mulsubst[L.`g/`a] P} : D\arrow A | `g`:C,`D $ and $ \derlmu `G |- {\mulsubst[L.`g/`a] Q} : D | `g`:C,`D $; the result follows by rule $(\arrE)$.
            }

        \item [{$ \mulsubst[L.`g/`a ] `M`d.[`a]N \ByDef `M`d.[`g] `@ L ({\mulsubst[L.`g/`a] N}) $}]
            Then by rule $(`m)$
            $ \derlmu `G |- N : B | `d`:A,`a`:B,`D $, and by induction
            $ \derlmu `G |- {\mulsubst[L.`g/`a] N} : B | `d`:A,`g`:C,`D $.
            Since $`d$ and $`g$ do not occur (free) in $L$, we can construct
            \[ \def\Turnlmu{\Turn}
                \Inf [`m]
                {\Inf [\arrE]
                    {\Inf    [\Weak]
                        {\InfBox{ \derlmu `G |- L : B\arrow C | `D }
                        }{ \derlmu `G |- L : B\arrow C | `d`:A,`g`:C,`D }
                        \InfBox{ \derlmu `G |- {\mulsubst[L.`g/`a] N} : B | `d`:A,`g`:C,`D }
                    }{\derlmu `G |- { `@ L {\mulsubst[L.`g/`a] N} } : C | `d`:A,`g`:C,`D }
                }{ \derlmu `G |- { `M`d.[`g] `@ L {\mulsubst[L.`g/`a] N} } : A | `g`:C,`D }
            \]

        \item [{$ \mulsubst[L.`g/`a] (`M`d.[`b]N) \ByDef `M`d.[`b] ( {\mulsubst[L.`g/`a] N} ) ~ (`b \not= `a) $}]
            Then by rule $(`m)$ there exists $D$ such that $`b`:D,`D' = `D$, and
            $ \derlmu `G |- N : D | `d`:A,`a`:B,`b`:D,`D' $.
            Then we have
            $ \derlmu `G |- {\mulsubst[L.`g/`a] N} : D | `d`:A,`g`:C,`b`:D,`D' $ by induction.
            But then, by rule $(`m)$, also $ \derlmu `G |- `M`d.[`b]{\mulsubst[L.`g/`a] N} : A | `g`:C,`b`:D,`D' $.
            \qed

    \end {description}
\end {enumerate}

 \end {Proof}

 \begin{theorem} [Soundness] 
 \label{soundness lmu}
 \label{lem:Lmu:Types:SubjectReduction}
	If $ \derlmu `G |- M : A | `D $, and $M \rednclmuN N$, then $ \derlmu `G |- N : A | `D $.
 \end{theorem}
 \begin{Proof}{By induction on the definition of $\rednclmuN$.}{\Thm \ref{soundness lmu}}
By induction on the definition of $\rednclmuN$.

\begin{description}

    \item [{$ `@ ( `Lx . M ) N \rednclmuN M \tsubst [N/x] $}]
        The derivation for $ \derlmu `G |- {`@ (`L x . M ) N } : A | `D $ is shaped like
        \[ \def \Turnlmu{\Turn}
            \Inf    [\arrE]
            {\Inf    [\arrI]
                {\InfBox{ \derlmu `G,x`:B |- M : A | `D }
                }{ \derlmu `G |- `Lx.M : B\arrow A | `D }
                \InfBox{ \derlmu `G |- N : B | `D }
            }{ \derlmu `G |- { `@ (`Lx . M ) N } : A | `D }
        \]

        Then, by \Lmm\,\ref{term substitution lemma}, we have $ \derlmu `G |- M\tsubst[N/x] : A | `D $.

    \item [{$ `@ (`M`a.[`a]M) N \rednclmuN `M`g . [`g] `@ (M\mursubst[N.`g/`a]) N $}]
        The derivation for $ `@ (`M`a.[`a]M) N $ is shaped like
        \[ \def \Turnlmu{\Turn}
            \Inf    [\arrE]
            {\Inf    [`m]
                {\InfBox{ \derlmu `G |- M : B\arrow A | `a`:B\arrow A,`D }
                }{ \derlmu `G |- `M`a.[`a]M : B\arrow A | `D }
                \InfBox{ \derlmu `G |- N : B | `D }
            }{ \derlmu `G |- {`@ (`M`a.[`a]M) N } : A | `D }
        \]
        Then by \Lmm\,\ref{term substitution lemma}, we have $\derlmu `G |- M\mursubst[N.`g/`a] : B\arrow A | `g`:A,`D $.
        Since $`g$ is fresh, by weakening also $ \derlmu `G |- N : B | `g`:A,`D $, and we can construct
        \[ \def \Turnlmu{\Turn}
            \Inf    [`m]
            {\Inf    [\arrE]
                {\InfBox{ \derlmu `G |- M\mursubst[N.`g/`a] : B\arrow A | `g`:A,`D }
                    \quad
                    \InfBox{ \derlmu `G |- N : B | `g`:A,`D }
                }{ \derlmu `G |- {`@ M\mursubst[N.`g/`a] N } : A | `g`:A,`D }
            }{ \derlmu `G |- `M`g.[`g]{`@ (M\mursubst[N.`g/`a]) N } : A | `D }
        \]

    \item [{$ `@ (`M`a.[`d]M) N \rednclmuN `M`g . [`d] M\mursubst[N.`g/`a] $, with $`a \not= `d $}]
        The derivation for $ `@ (`M`a.[`d]M) N $ is shaped like
        \[ \def \Turnlmu{\Turn}
            \Inf    [\arrE]
            {\Inf    [`m]
                {\InfBox{ \derlmu `G |- M : C | `a`:B\arrow A,`d`:C,`D' }
                }{ \derlmu `G |- `M`a.[`d]M : B\arrow A | `d`:C,`D' }
                \InfBox{ \derlmu `G |- N : B | `d`:C,`D' }
            }{ \derlmu `G |- {`@ (`M`a.[`d]M) N } : A | `d`:C,`D' }
        \]
        with $`D = `d`:C,`D'$.
        Then by \Lmm\,\ref{term substitution lemma}, we have $ \derlmu `G |- M\mursubst[N.`g/`a] : C | `g`:A,`d`:C,`D' $, and we can construct
        \[ \def \Turnlmu{\Turn}
            \Inf    [`m]
            {\InfBox{ \derlmu `G |- M\mursubst[N.`g/`a] : C | `g`:A,`d`:C,`D' }
            }{ \derlmu `G |- `M`g.[`d]M\mursubst[N.`g/`a] : A | `d`:C,`D' }
        \]

    \item [{$ `@ M (`M`a.[`a]N) \rednclmuN `M`g.[`g]`@ M (\mulsubst M.`g/`a N) $}]
        The derivation for $ `@ M (`M`a.[`a]N) $ is shaped like
        \[ \def \Turnlmu{\Turn}
            \Inf    [\arrE]
            {\InfBox{ \derlmu `G |- M : B\arrow A | `D }
                \quad
                \Inf    [`m]
                {\InfBox{ \derlmu `G |- N : B | `a`:B,`D }
                }{ \derlmu `G |- { `M`a.[`a]N } : B | `D }
            }{ \derlmu `G |- { `@ M (`M`a . [`a] N ) } : A | `D }
        \]
        Then by \Lmm \ref{
        term substitution lemma},
        we have $ \derlmu `G |- {\mulsubst M.`g/`a N} : B | `g`:A,`D $, and
        we can construct
        \[ \def \Turnlmu{\Turn}
            \Inf    [`m]
            {\Inf    [\arrE]
                {\Inf    [\Weak]
                    {\InfBox{ \derlmu `G |- M : B\arrow A | `D }
                    }{ \derlmu `G |- M : B\arrow A | `g`:A,`D }
                    \InfBox{ \derlmu `G |- {\mulsubst[M.`g/`a] N} : B | `g`:A,`D }
                }{ \derlmu `G |- {`@ M (\mulsubst[M.`g/`a] N) } : A | `g`:A,`D }
            }{ \derlmu `G |- { `M`g . `g `@ M (\mulsubst[M.`g/`a] N ) } : A | `D }
        \]

    \item [{$ `@ M (`M`a.[`d]N) \rednclmuN `M`g . [`d] {\mulsubst M.`g/`a N }$, with $`a \not= `d $}]
        The derivation for $ `@ M (`M`a.[`d]N) $ is shaped like
        \[ \def \Turnlmu{\Turn}
            \Inf    [\arrE]
            {\InfBox{ \derlmu `G |- M : B\arrow A | `d`:C,`D' }
                \quad
                \Inf    [`m]
                {\InfBox{ \derlmu `G |- N : C | `a`:B,`d`:C,`D' }
                }{ \derlmu `G |- { `M`a.[`d]N } : B | `d`:C,`D' }
            }{ \derlmu `G |- {`@ M (`M`a.[`d]N) : A } | `d`:C,`D' }
        \]
        with $`D = `d`:C,`D'$.
        Then by \Lmm \ref{
        term substitution lemma},
        we have $ \derlmu `G |- {\mulsubst[M.`g/`a] N} : C | `g`:A,`d`:C,`D' $, and by rule $(`m)$ we have
        $ \derlmu `G |- { `M`g . [`d] \mulsubst[M.`g/`a] N } : A | `d`:C,`D' $.

        \Comment{
        \item [{$ `M`a.[`a]M \rednclmuN M $}]
            The derivation for $ `M`a.[`a]M $ is shaped like
            \[ \def \Turnlmu{\Turn}
                \Inf    [`m]
                {\InfBox{ \derlmu `G |- M : A | `a`:A,`D }
                }{ \derlmu `G |- `M`a.[`a]M : A | `D }
            \]
            Since $`a$ does not occur in $M$, we can thin $`a`:A,`D$ and obtain $ \derlmu `G |- M : A | `D $.
        }

    \item [{$ `M`a.[`b] `M`g.[`d]M \rednclmuN `M`a.([`d]M) \tsubst[`b/`g] $}]
        The derivation for $ `@ (`M`a.[`d]M) N $ is shaped like
        \[ \def \Turnlmu{\Turn}
            \Inf    [`m]
            {\Inf    [`m]
                {\InfBox{ \derlmu `G |- M : D | `a`:A,`b`:B,`g`:B,`d`:D,`D' }
                }{ \derlmu `G |- `M`g.[`d]M : B | `a`:A,`b`:B,`d`:D,`D' }
            }{ \derlmu `G |- `M`a.[`b] `M`g.[`d]M : A | `b`:B,`d`:D,`D' }
        \]
        So in particular, replacing all occurrences of $`g$ by $`b$, we obtain a derivation for $ \derlmu `G |- M\tsubst[`b/`g] : D | `a`:A,`b`:B,`d`:D,`D' $.
        Now either:

        \begin{description}
            \item [$`d \not= `g$]
                Then we can construct:
                \[ \def \Turnlmu{\Turn}
                    \Inf    [`m]
                    {\InfBox{ \derlmu `G |- M\tsubst[`b/`g] : D | `a`:A,`b`:B,`d`:D,`D' }
                    }{ \derlmu `G |- `M`a.[`d]M\tsubst[`b/`g] : A | `b`:B,`d`:D,`D' }
                \]

            \item [$`d = `g$]
                Then $D = B$ as well, and we can construct:
                \[ \def \Turnlmu{\Turn}
                    \Inf    [`m]
                    {\InfBox{ \derlmu `G |- M\tsubst[`b/`g] : B | `a`:A,`b`:B,`D' }
                    }{ \derlmu `G |- `M`a.[`b]M\tsubst[`b/`g] : A | `b`:B,`D' }
                \]

        \end{description}

    \item[{$`M`a . [`a] M \reduc M $}, with $a\notele M$]
        The derivation for $ `M`a . [`a] M $ is shaped like
        \[ \def \Turnlmu{\Turn}
            \Inf    [`m]
            {\InfBox{ \derlmu `G |- M : A | `a`:A,`D }
            }{ \derlmu `G |- `M`a.[`a]M : A | `D }
        \]
        Since $a\notele M$, by thinning we get $ \derlmu `G |- M : A | `D $.

\end{description}
The contextual rules follow by induction.\QED
 \end{Proof}

This result of course also holds for {\CBV} and {\CBN}-reduction as a simple corollary.

\Comment{
}



 \section{The Calculus {\lmmtPDF}} 
 \label{lmmt section}
 \label{sec:LMMT}

This section will give a short summary of $\lmmt$, as introduced by \citet{Curien-Herbelin'00}.
\Sect \ref{sec:LMu} explored how $\lmu$ is simultaneously a proof-term syntax for (a kind of) classical logic, and a language of control, as illustrated in \cite{Bakel-PPDP'19}, where a $`l$-calculus extended with exception handling was introduced that can be implemented in $\lmu$.
Although $\lmu$ achieves the latter without a primitive notion of continuation, this comes at the cost of a rather complex notion of reduction, namely through structural substitution.
The internal classical logic itself imposes a restriction on the shape of proof trees in which negation elimination only applies to negated assumptions, forcing the negated proposition to be an instance of the $(\Ax)$ rule.
Informally, the issue is that natural deduction is not well-suited to classical logic. 

\citet{GentzenG'35} says as much at the inception of natural deduction, and argues that the sequent calculus {\LK} is better placed for reasoning about classical logic; the latter allows for the negated assumptions in a $(\Cut)$ to be arbitrary proof trees, not just $(\Ax) $.
As the type system of $\lmu$ follows the spirit of natural deduction, it also inherits these difficulties in its logical interpretation.
It seems, then, a calculus for classical logic would benefit from a type system in the style of {\LK}.

Natural deduction defines logical connectives through their \emph{introduction} and \emph{elimination} rules, which always modify the single proposition on the right-hand side of a judgement.
{\LK} on the other hand, features multiple propositions on both sides of judgements (called \emph{sequents}), and only \emph{introduction} rules for logical connectives;
unlike in natural deduction, these rules apply to either the left or right of sequents.
The only way to eliminate a connective is to eliminate the whole formula in which it appears, via the $(\Cut)$-rule.
Sequents in {\LK} take the form $A_1,\ldots,A_n \vdash B_1,\ldots,B_m $, where the conjunction of $A_1,\ldots,A_n$ is understood to imply the disjunction of $B_1,\ldots,B_m$; so if all of the $ A_1,\ldots,A_n $ hold, then so must at least one of the $ B_1,\ldots,B_m $.

We give the definition for the implicative fragment of the sequent calculus \LK,  known as Kleene's $G_3$ \cite{Kleene'52}, as it plays a major role in this paper.

 \begin{definition} [Implicative \LK] \label{LK}
Let $`G$ and $`D$ be multi-sets of formulas, defined by the grammar:
 \[ \begin{array}{rcl}
A,B &::=& `v \mid A \Arrow B 
 \end{array} \]
The implicative fragment of the sequent calculus \LK,  known as Kleene's $G_3$ \cite{Kleene'52}, si defined through the rules:
 \[ \def \TurnLK {\Turn} \begin {array}{rl@{\quad}rl@{\quad}rl@{\quad}rl}
(\textit{Ax})`: &
 \Inf	{ \derLK `G,A |- A,`D }
&
(\ArrL)`: &
 \Inf	{ \derLK `G |- A,`D 
	 \quad
	 \derLK `G,B |- `D 
	}{ \derLK `G, A\Arrow B |- `D }
& 
(\ArrR)`: &
 \Inf	{ \derLK `G,A |- B,`D 
 	}{ \derLK `G |- A\Arrow B,`D }
&
(\textit{cut})`: &
 \Inf	{ \derLK `G |- A,`D 
	 \quad
	 \derLK `G,A |- `D 
	}{ \derLK `G |- `D }
 \end {array} \]
 \end{definition}
It is well known that {\LK} enjoys a symmetry that natural deduction lacks, and both propositions on the left and on the right can be shown through a proof.

This is not the case for for $\TurnF$ or $\lmu$, where negation can only be shown through an assumption (inhabited by a name $`a$ in $\lmu$ in the co-context).
Moving to design a term calculus to represent the proofs for {\LK} will inhabit sequents with more complex proof trees for negation, leading to use more complex expressions, or terms.
As a name $`a `: A$ refers to a continuation expecting a term of type $A$, these terms should be able to construct these continuations.
$\lmmt$ is a calculus that does that: it takes both terms and continuations as primitive.
Where $\lmu$ has commands of the form $[`a]M$ representing $M$ in the continuation $`a$,
$\lmmt$ allows commands $c$ of the form $ \cell<t|e> $, 
representing $t$ in the continuation $e$. 
Examples of such continuations are names, $`a$ (and then $\cell<M|`a>$ corresponds to $[`a]M$), and \emph{stacks}, $t`.e$, the latter corresponding with the context $e("[]"t)$.
A command $\cell<{`M `a .c}|e> $ represents that the entire continuation of ${`M `a .c}$ is $e$, and can be reduced by substitution
$ 
 \cell<{`M `a .c}|e>
 \reduc
 \tsubst[e/`a],
$ 
thereby avoiding the need for structural substitution.

In $\lmu$, if $`@ (`m`a . `C_1 ) (`m`a . `C_2 )$ reduces by $(\mur)$, it represents $`m`a . `C_1 $ taking control over the current context $"[]"(`m `a.`C_2 )$.
We can dually say that the $(\mul)$ reduction of the same term represents the context $"[]"(`m `a.`C_2)$ taking control over the current term, $`m`a . `C_1 $; 
this means $\lmmt$ must feature a construct allowing continuations to do as such.
This is the purpose of $\mt$, which captures the term in a command, and substitutes it into a subcommand:
$ 
\begin{array}{rcl}
 \cell< t | \mt x.c >
	&\reduc&
 c\tsubst[t/x]
 \end{array} $. 
One may notice the symmetry between the $\mu$ and $\mutilde$ reductions, and their leading to non-confluence in the command
$ 
 \cell< {`M `a .c_1} | {\mt x . c_2} >
$. 
Remarkably, a systematic preference of $\mu$ or $\mutilde$ reductions in such commands leads to {\CBV} or {\CBN} reduction, respectively.

 \paragraph{Overview}
\Sect \ref{sec:LMMT:Syntax} defines the syntax and reduction for {$\lmmt$}.
\Sect \ref{sec:LMMT:Typing} provides the type system, with the expected soundness theorem, \ie assignable types are preserved under reduction. 
Finally, \Sect \ref{sec:LMMT:NonConfluence} explores the inherent non-confluence of reduction in {$\lmmt$}, and obtaining {\CBV} or {\CBN} as confluent subsystems.

 \subsection{Syntax and Reduction}
 \label{sec:LMMT:Syntax}
 \begin{definition}[Commands, Terms, and Contexts \cite{Curien-Herbelin'00}]
 Let $x, y, z, \ldots $ range over a countably infinite set of \emph{term variables} and $`a, `b, `g, \ldots$ range over a countably infinite set of \emph{environment variables} (or \emph{names}). 
 The following defines the three syntactic categories of $\lmmt$:
 \[ \begin{array}{rcl@{\hspace*{2cm}}l}
 c &::=& \cell<t|e> & (\textit{commands}) \\
 t &::=& x \mid`L x.t \mid `M`b.c & (\textit{terms}) \\
 e &::=& `a \mid t`.e \mid \mt x.c & (\textit{environments}) \\
 \end{array} \]

 \observe
The operators $`l$ and {$\mt$} bind variables, and $\mu$ binds names. 
The notion of free or bound term and environment variables is defined as usual.

 \end{definition}

The environment $t`.e$ corresponds with the continuation $e["[]"\:t]$ in the \LC;
an environment of the form $t_1`. \dots t_n`.`a $ is called a \emph{stack}. 
$\lmmt$ inherits $`M`a .c$ from $\lmu$, which represents a term that captures its environment and binds it to $`a$.
Symmetrically, $\mt x.c$ is an environment that captures the active term, and binds it to $x$; 
the {\LC} analogue of this environment is the continuation \textrm{\sf let $x = "[]"$ in $c$}.
A command $\cell<t|e>$ represents a term $t$ running in the environment $e$, which may look something like $e[t]$ in the {\LC}.
Commands of the form $\cell<t|`a>$ correspond to the naming construct $[`a]t$ of $\lmu$, giving the name $`a$ to the output of $t$.
From the grammar, it follows that each environment is a sequence of terms, ending with either a name or an environment of the shape $ \mt x.c $:
 \[\begin{array}{rcl}
 e &::=&
 \begin{cases}
 t_1 `.~ {\dots} `.~ t_n`.`a \\
 t_1 `.~ {\dots} `.~ t_n`. \mt x.c \\
 \end{cases}
 \end{array}\]

Commands are the computational units of $\lmmt$, whose reduction is of a weak-head form.

 \begin{definition}[Reduction in $\lmmt$ \cite{Curien-Herbelin'00,Herbelin'05}] \label{lmmt reduction}
 Let $ c\tsubst[e/`b] $ stand for the implicit substitution of the free occurrences of the environment variable $`b$ in the command $c$ by the environment $e$, and $ c\tsubst[t/x] $ for that of $x$ by the term $t$.
 The reduction rules are defined by:
 \[ \begin{array}{c}
 \begin{array}{rrcl}
 \multicolumn{4}{c}{\textit{logical rules}} \\
 (`l): & \cell<`L x.t_1 | t_2`.e >	&\reduc& \cell< t_2 | {\mt x.< t_1|e>} > \\
 (`m): & \cell< `M`b.c | e > & \reduc & c\tsubst[e/`b] \\
 (\mt): & \cell< t | \mt x.c > & \reduc & c\tsubst[t/x]
 \end{array}
 \quad
 \begin{array}{rrcll}
 \multicolumn{5}{c}{\textit{extensional rules}} \\
 (`h): &`L x . `M`b.< t | x`.`b >	&\reduc& t ~ & (x,`b \notele \FV{t}) \\
 (`h`m): & `M`a.< t | `a > & \reduc & t & (`a \notele \FV{t}) \\
 (`h\mt): & \mt x.< x | e > & \reduc & e & (x \notele \FV{e})
 \end{array}
 \\ [10mm]
 \textit{contextual rules} \\ [-2mm]
 \begin{array}{r@{~}c@{~}l}
 t \reduc t' & \Imply &
 \begin{cases}{r@{\,}c@{\,}l}
  \cell< t | e >	&\reduc& \cell< t' | e > \\
  `L x . t	&\reduc& `L x . t' \\
  t`.e	&\reduc& t'`.e
 \end{cases}
 \end{array}
 \quad
 \begin{array}{r@{~}c@{~}l}
 e \reduc e' & \Imply &
 \begin{cases}{r@{\,}c@{\,}l}
  \cell< t | e >	&\reduc& \cell< t | e' > \\
  t`.e	&\reduc& t`.e'
 \end{cases}
 \end{array}
 \quad
 \begin{array}{r@{~}c@{~}l}
 c \reduc c' & \Imply &
 \begin{cases}{r@{\,}c@{\,}l}
  `M`b . c	&\reduc& `M`b . c' \\
  \mt x . c	&\reduc& \mt x . c'
 \end{cases}
 \end{array}
 \end{array} \]
We use $\redlmmt$ for this notion of reduction and $\eqlmmt$ for the induced equality, and write $ M \crlmmt N $ if there exists $P$ such that $ M \rtcredlmmt P $ and $ N \rtcredlmmt P $.

We say that the reductions $ \cell< `M`b.c | e > \redlmmt c\tsubst[e/`b] $ and $ \cell< t | \mt x.c > \redlmmt c\tsubst[t/x] $ \emph{take place over $`b$}, respectively $x$, and write $ c \redlmmt (n) c' $ when the reduction step from $c$ to $c'$ takes place over $n$, mainly as an aid to the reader.

\observe We will use $\downarrow_{`.}$ also for other notions of reduction.

 \end{definition}

The logical rules operate on commands, rules $(`h)$ and $(`h`m)$ on terms, and rule $(`h\mt)$ on environments.
We will have little need of the extensional rules, only using $(`h`m)$ in \Thm \ref{interpretation respects beta reduction} and when dealing with $\slmu$'s rule $(\Rename)$.

In \Sect \ref{sec:LMMT:Typing}, we will see that commands are the syntax for cuts; reduction of commands thus induces a cut-elimination in the corresponding proof-tree.
Unlike cuts, not all commands are reducible: \eg $\cell<x|`a>$, $\cell<`Lx.t|`a>$ and $\cell<x|t`.e>$ are irreducible; this is one of the differences between {\LK} and $\lmmt$.

Although $\lmmt$ has abstraction, it does not have application, as that corresponds to an \emph{elimination rule}, the likes of which are not a part of \LK.
In fact, abstraction, representing \LK's $(\arrR)$, has a counterpart in \emph{environment construction} $t`.e$, representing \LK's $(\arrL)$, where a term with a hole is built, offering the operand $t$ and the continuation $e$.
The main operators are $`m$ and $\mutilde$-abstraction, which, in a sense correspond to (delayed) substitution (parameter call) and to environment call.

$\lmmt$ has both \emph{explicit} and \emph{implicit} variables: the implicit variables are for example in $t`.e$, where the hole `$`.$' (which acts as input) does not have an identity, and in $`L x.t$ where the environment (output) is anonymous.
We can make these variables explicit by \emph{naming}, respectively, $\mt y.<y|t`.e>$ and $`M`a.<`L x.t|`a>$; when $y$ ($`a)$ is fresh, these terms are $`h$ redexes, but, in general, the implicit variable can be made to correspond to one that already occurs.

 \begin{example} \label{lmmt prep}
We have the following reduction:

 \[ `M`g . < t_1 | { \mt x . < t_2 | { \mt y . < x | y `. `a > } > } > 
 	\redlmmtn 
  `M`g . < t_2 | { \mt y . < t_1 | y `. `a > } > 
 	\redlmmtn 
  `M`g . < t_1 | t_2 `. `a >  
 \]
that builds a stack $ t_2 `. `a $ for $ t_1 $.
This will prove useful in some of the proofs; in order to be able to identify when we use this property, we will write $ `M`g . < t_1 | { \mt x . < t_2 | { \mt y . < x | y `. `a > } > } > \redlmmtS  `M`g . < t_1 | t_2 `. `a > $
for this reduction.
 \end{example}

The calculus expresses elegantly the duality of \LK's left- and right introduction in a very symmetric syntax.
However, this duality notwithstanding, $\lmmt$ does not fully represent {\LK}.
The {\LK} proof
 \[ \def \TurnLK{\Turn}
 \Inf [\Cut]
	{\Inf [\arrR]
		{ \InfBox { \derLK {`G,A} |- B,`D }
		}{ \derLK `G |- A\arr B,`D }
	 \Inf [\arrL]
		 { \InfBox { \derLK `G |- A,`D }
		 \quad
		 \InfBox { \derLK `G,B |- `D }
		}{ \derLK `G,A\arr B |- `D }
	}{ \derLK `G |- `D }
 \]
naturally reduces to both
 \[ \def \TurnLK{\Turn}
\Inf [\Cut]
	{ \InfBox { \derLK `G |- A,`D }
 ~
	\Inf [\Cut]
		{\InfBox { \derLK {`G,A} |- B,`D }
 ~
 	\Inf [\Weak]
		{ \InfBox { \derLK `G,B |- `D }
  			}{ \derLK `G,A, B |- `D }
 		}{ \derLK `G,A |- `D }
 	}{ \derLK `G |- `D }
 \dquad\textrm{and}\dquad
\Inf [\Cut]
	{\Inf [\Cut]
		{\Inf [\Weak]
 			{ \InfBox { \derLK `G |- A,`D }
 			}{ \derLK `G |- A, B,`D }
		 \InfBox { \derLK `G,A |- B,`D }
	 	}{ \derLK `G |- B,`D }
	 \InfBox { \derLK `G,B |- `D }
 	}{ \derLK `G |- `D }
 \]
The rule $(`l)$ represents the reduction to the first derivation, but the second has no such associated rule.
This is a fundamental difference with the calculus $`X$ (see \Sect \ref{sec:XCalculus}), which is able to represent both reductions;
this implies that there does not exist a full reduction-preserving interpretation of $`X$ into $\lmmt$.
One remedy would add the second alternative to $\lmmt$ as well, which would result in the reduction rule
 \[ \begin{array}{rcl}
\cell<`Ly.t | t'`.e >	&\reduc&
 \begin{cases}{ll@{\dquad}l}
 \cell< t' | {\mt y.<t|e>} > \\ 
 \cell <{`M`g . < t' | {\mt y.<t|`g> } > }|e> & (`g \textit{ fresh}) 
 \end{cases}
 \end{array}\]
Since $`g$ is fresh, the latter can reduce to the former,
 \[\begin{array}{lclclcl}
 \cell <{`M`g . < t' | {\mt y.<t|`g> } > }|e> & \redlmmt (`g) & 
 \cell < t' | {\mt y.<t|e> } >
 \end{array}\]
but the second alternative allows $e$ to take control of reduction;
if $e$ had the form $\mt z.c$, there is the new possible reduction,
 \[\begin{array}{rcl}
 \cell <{`M`g . < t' | {\mt y.<t|`g> } > }| {\mt z.c} > 
 & \redlmmt &
 c \tsubst{`M`g . \cell < t' | {\mt y.<t|`g> } > }/z
 \end{array}\]
We will look in more detail at this in \Sect \ref{add `m_l to lmmt};
 \Def\ref{lmmtE} implements this extension, which allows for a proper embedding of $\Xs$ into $\lmmt$.

Adding this alternative would extend the expressivity of $\lmmt$, since now we would have also the reduction:
 \[ \begin{array}{rclcl}
 \cell<`Ly.t | t'`.\mt z.c >	&\reduc&
 \cell <{`M`g . < t' | {\mt y.<t|`g> } > }|\mt z.c> 
	&\reduc&
 c \tsubst[ `M`g . < t' | {\mt y.<t|`g> } > / z ] \\
 \end{array} \]

As to the encoding of \CBV-reduction, little is gained by adding this rule; the positioning of sub-terms in $ \cell <{`M`b . < t' | {\mt z.<t|`b> } > }|e> $ and $ \cell < t' | {\mt z.<t|e> } > $ is very similar.


 \subsection{Typing} \label{sec:LMMT:Typing}

{
In its typed version, $\lmmt$ is a proof-term syntax for a classical sequent calculus that treats a logic with focus, and can be seen as an extension of Parigot's $\lmu$ and a variant of Gentzen's \LK, restricted to implication (\Def \ref{LK}), by  splitting the rule $(\Ax)$ into two variants, and adding activation and deactivation rules.

\[ \def \TurnLmmt {\Turn}
\begin{array}{c@{\quad}c@{\quad}c}
 \begin{array}{rl@{\quad}rl}
 \multicolumn{4}{c}{
(\Cut): ~
\Inf {\derlog `G |- A \mid `D \quad \derlog `G \mid A |- `D
 }{\derlog `G |- `D }
} \\ [4mm] 
(\AxT): &
\Inf {\derlog `G,A |- A \mid `D }
 &
(\AxC): &
\Inf {\derlog `G \mid A |- A,`D }
 \end{array}
 &
 \begin{array}{rl}
(\arrR): &
\Inf {\derlog `G,A |- B \mid `D
 }{\derlog `G |- A\arr B \mid `D }
 \\ [5mm]
(\arrL):&
\Inf { \derlog `G |- A \mid `D \quad \derlog `G \mid B |- `D
 }{ \derlog `G \mid A\arr B |- `D }
 \end{array}
 &
 \begin{array}{rl}
(\ActR): &
\Inf {\derlog `G |- A,`D
 }{\derlog `G |- A \mid `D }
 \\ [5mm]
(\ActL): &
\Inf {\derlog `G,A |- `D
 }{\derlog `G \mid A |- `D }
 \end{array}
\end{array} \]
}

The type system $\Turnlmmt$ inhabits these rules with terms, combining the left-and-right presentation of {\LK} with the active/inactive formulae of types in $\lmu$.
As there are three syntactic categories, it derives three different judgements:
 \begin{flatenumerate}
 \item [$\derLmmt c : `G |- `D $] $c$ is a well-formed command;
 \item [$\derLmmt `G |- t : A | `D $] $t$ has type $A$; and 
 \item [$\derLmmt `G | e : A |- `D $] $e$ has type $A$.
 \end{flatenumerate}
The type system is defined with respect to these judgements.
 \begin{definition}[Typing for $\lmmt$ \cite{Curien-Herbelin'00}] 
 \label{def:LMMT:Typing:TypeSystem}
 \label{basic system}
 Using the notion of types, and contexts of variables and names of \Def \ref{tas lmu}, type assignment for $\lmmt$, $\Turnlmmt$, is defined via the rules:
 \[ \def \Turnlmmt {\Turn}
\begin{array}{rl@{\quad}rl}
 \multicolumn{4}{c}{
(\Cut):~
\Inf {\derLmmt `G |- t : A | `D \quad \derLmmt `G | e : A |- `D
 }{\derLmmt \cell<t|e> : `G |- `D } \qquad
\vspace*{3mm}
}
 \\ 
(\AxT): &
\Inf {\derLmmt `G,x{:}A |- x : A | `D }
 &
(\AxC): &
\Inf {\derLmmt `G | `a : A |- `a{:}A,`D }
 \\ [4mm]
(\arrR): &
\Inf {\derLmmt `G,x`:A |- t : B | `D
 }{\derLmmt `G |- {`L x.t } : A\arr B | `D }
 &
(\arrL):&
\Inf { \derLmmt `G |- t : A | `D \quad \derLmmt `G | e : B |- `D
 }{ \derLmmt `G | t`.e : A\arr B |- `D }
 \\ [5mm]
(`m): &
\Inf {\derLmmt c : `G |- `a`:A,`D
 }{\derLmmt `G |- `M`a.c : A | `D }
 &
(\mt): &
\Inf {\derLmmt c : `G,x`: `A |- `D
 }{\derLmmt `G | \mt x.c : A |- `D }
 \end{array} \]
 \end{definition}

 \Comment{
As the type system should have an {\LK} flavour, it should avoid explicit elimination rules, preferring instead right and left introduction rules.
This imposes on the syntax of $\lmmt$.
Introduction rules in natural deduction correspond with left rules in {\LK}, so $\lambda$-abstractions also implement the rule $(\arrL)$, but application $ `@ t u $, the witness of $(\arrE)$, 

 \[ \begin{array}{c@{\qquad}c}
\Inf {\derlog `G |- A\arr B 
  \quad
  \derlog `G |- A
 }{\derlog `G |- B }
 \end{array} \]
is not a primitive construct.
In {\LK}, the rules $(\arrL)$ and $(\Cut)$,
 \[ \begin{array}{c@{\qquad}c}
\Inf {\derlog `G |- A,`D
  \quad 
  \derlog `G,B |- `D 
 }{\derlog `G, A\arr B |- `D }
&
\Inf {\derlog `G |- A,`D
  \quad 
  \derlog `G,A |- `D 
 }{\derlog `G |- `D }
 \end{array} \]
together take the role of $(\arrE)$ 
by forming an eliminable cut with $(\arrR)$ and $(\arrL)$:
 \[
 \Inf[\Cut]
 {\Inf[\arrR]
 {\derlog `G, A |- B,`D 
 }{\derlog `G |- A\arrow B,B }
  \quad
  \Inf[\arrL] 
 {\derlog `G |- A.B 
 \quad
 \Inf [\Ax]
  {\derlog `G,B |- B } 
 }{\derlog `G,A\arrow B |- B }
 }{\derlog `G |- B }
 \]
In general, the eliminators of types induce syntactic constructs of continuations, just as elimination rules induce right-rules in {\LK}.
We would expect that a $\lambda$-abstraction $\lambda x.t$ reduces against a stack $u`.e$;
it follows that stacks should inhabit $(\arrR)$:
 \[ \def \Turnlmmt {\Turn}%
 \Inf {\Inf { \derLmmt `G,x`:A |- t : B | `D 
 }{ \derLmmt `G |- {`L x.t } : A\arr B | `D }
 \quad
 \Inf { \derLmmt `G |- u : A | `D 
 \quad
 \derLmmt `G | e : B |- `D 
 }{ \derLmmt `G | u`.e : A\arr B |- `D }
 }{ \derLmmt {\cell< `L x.t | u`.e >} : `G |- `D }
 \]
}

We will now show soundness, \ie that type assignment is respected by reduction; first we show a substitution lemma.

 \begin{lemma} [Substitution]
 \label{lem:LMMT:Typing:Substitution}
 \label{lmmt subst lemma}
 \begin{enumerate}
 \firstitem If $ \derLmmt c : `G,x`:A |- `D $, and $ \derLmmt `G |- t : A | `D $, then $ \derLmmt c\tsubst[t/x] : `G |- `D $.
 \item If $ \derLmmt c : `G |- `b`:A,`D $, and $ \derLmmt `G | e : A |- `D $, then $ \derLmmt c\tsubst[e/`b] : `G |- `D $.
 \end{enumerate}
 \end{lemma}
 \begin{proof}
Straightforward by induction over the structure of derivations in $\Turnlmmt$. \QED
 \end{proof}

We can now show that assignable types are preserved under reduction.

 \begin{theorem}[Soundness of type assignment]
\label{Subject reduction nclmu}
\label{lem:LMMT:Typing:SubjectReduction}
 \begin{enumerate}
 \firstitem If $ \derLmmt c : `G |- `D $, and $c \redlmmt c'$, then $ \derLmmt c' : `G |- `D $.
 \item If $ \derLmmt `G |- t : A | `D $, and $t \redlmmt t'$, then $ \derLmmt `G |- t' : A | `D $.
 \item If $ \derLmmt `G | e : A |- `D $, and $e \redlmmt e'$, then $ \derLmmt `G | e' : A |- `D $.
 \end{enumerate}
 \end{theorem}

 \begin{Proof}{Simultaneous by induction on the definition of $\rtcredlmmt$.}{\Thm \ref{Subject reduction nclmu}}

Simultaneous by induction on the definition of $\rtcredlmmt$; we will only show the base cases.
\begin{description}
    \item[$ \cell<`L x.t_1 | t_2`.e > \reduc \cell< t_2 | {\mt x.< t_1|e>} > $]
If $ \derLmmt \cell<`Lx.t_1|t_2`.e> : `G |- `D $, then the derivation is shaped like on the left; regrouping the sub-derivations, we can construct the one on the right.
\[ \kern -5mm \def \Turnlmmt{\Turn}
    \Inf    {\Inf    {\InfBox{ \derLmmt `G,x`:A |- t_1 : B | `D }
}{ \derLmmt `G |- `Lx.t_1 : A\arr B | `D }
\Inf    {\InfBox{ \derLmmt `G |- t_2 : A | `D }
    \quad
    \InfBox{ \derLmmt `G | e : B |- `D }
}{ \derLmmt `G | t_2`.e : A\arr B |- `D }
    }{ \derLmmt \cell<`Lx.t_1|t_2`.e> : `G |- `D }
    \dquad
    \Inf    {\InfBox{ \derLmmt `G |- t_2 : A | `D }
\kern-10mm
\Inf    {\Inf    {\InfBox{ \derLmmt `G,x`:A |- t_1 : B | `D }
\quad
\Inf    [\Weak]
{\InfBox{ \derLmmt `G | e : B |- `D }
}{ \derLmmt `G,x`:A | e : B |- `D }
    }{ \derLmmt \cell<t_1|e> : `G,x`:A |- `D }
}{ \derLmmt `G | \mt x . <t_1|e> : A |- `D }
    }{ \derLmmt \cell<t_2|{\mt x.<t_1|e>}> : `G |- `D }
\]

    \item[$ \cell < `M`b.c | e > \reduc c\tsubst e/`b $]
If $ \derLmmt \cell<`M`b.c|e> : `G |- `D $, then the derivation is shaped like
\[ \def \Turnlmmt{\Turn}
\Inf	{\Inf	{\InfBox{ \derLmmt c : `G |- `b`:A,`D }
		}{ \derLmmt `G |- `M`b.c : A | `D }
\quad
	 \InfBox{ \derLmmt `G | e : A |- `D }
	}{ \derLmmt \cell<`M`b.c|e> : `G |- `D }
\]
By \Lmm \ref{lmmt subst lemma} we get $ \derLmmt c\tsubst[e/`b] : `G |- `D $.

    \item[$ \cell<t|\mt x.c> \reduc c\tsubst t/x $]
If $ \derLmmt \cell<t|\mt x.c> : `G |- `D $, then the derivation is shaped like
\[ \def \Turnlmmt{\Turn}
    \Inf    {\InfBox{ \derLmmt `G |- t : A | `D }
\quad
\Inf    {\InfBox{ \derLmmt c : `G,x`:A |- `D }
}{ \derLmmt `G | \mt x.c : A |- `D }
    }{ \derLmmt \cell<t|\mt x.c> : `G |- `D }
\]
By \Lmm \ref{lmmt subst lemma} we get $ \derLmmt c\tsubst[t/x] : `G |- `D $.

    \item[$ `L x . `M`b.\cell< t | x`.`b > \reduc t $, $ x,`b \notele \FV{t} $]
If $ \derLmmt `G |- `Lx.`M`b.<t|x`.`b> : A | `D $, then the derivation is shaped like
\[ \def \Turnlmmt{\Turn}
    \Inf    {\Inf    {\Inf    {\InfBox{ \derLmmt `G,x`:A |- t : A\arr B | `b`:B,`D }
\quad
\Inf    {\Inf    { \derLmmt `G,x`:A |- x : A | `b`:B,`D }
    \Inf    { \derLmmt `G,x`:A | `b : B |- `b`:B,`D }
}{ \derLmmt `G,x`:A | x`.`b : A\arr B |- `b`:B,`D }
    }{ \derLmmt \cell<t|x`.`b> : `G,x`:A |- `b`:B,`D }
}{ \derLmmt `G,x`:A | `M`b.<t|x`.`b> : B |- `D }
    }{ \derLmmt `G |- `Lx.`M`b.<t|x`.`b> : A\arr B | `D }
\]
From $ \derLmmt `G,x`:A |- t : A\arr B | `b`:B,`D $ and $ x,`b \notele \FV{t} $, by Thinning we get $ \derLmmt `G |- t : A\arr B | `D $.

    \item[$ `M`a.< t | `a > \reduc t $, $`a \notele \FV{t} $]
If $ \derLmmt `G |- `M`a.<t|`a> : A | `D $, then the derivation is shaped like
\[ \def \Turnlmmt{\Turn}
    \Inf    {\Inf    {\InfBox{ \derLmmt `G |- t : A | `a`:A,`D }
    \quad
    \Inf{ \derLmmt `G | `a : A |- `a`:A,`D }
}{ \derLmmt \cell<t|`a> : `G |- `a`:A,`D }
    }{ \derLmmt `G |- `M`a.<t|`a> : A | `D }
\]
From $ \derLmmt `G |- t : A | `a`:A,`D $ and $ `a \notele \FV{t} $, by Thinning we get $ \derLmmt `G |- t : A | `D $.

    \item[$ \mt x.< x | e > \reduc e $, $ x \notele \FV{e} $]
If $ \derLmmt `G |- `M`a.<t|`a> : A | `D $, then the derivation is shaped like
\[ \def \Turnlmmt{\Turn}
    \Inf    {\Inf    {\Inf{ \derLmmt `G,x`:A |- x : A | `D }
    \quad
    \InfBox{ \derLmmt `G,x`:A | e : A |- `D }
}{ \derLmmt \cell<x|e> : `G,x`:A |- `D }
    }{ \derLmmt `G | \mt x . <x|e> : A |- `D }
\]
From $ \derLmmt `G,x`:A | e : A |- `D $ and $ x \notele \FV{e} $, by Thinning we get $ \derLmmt `G | e : A |- `D $.
\QED
\end{description}
\observe
We can extend this last result also for the second alternative to rule $(`l)$, since we can derive:
\[ \def \Turnlmmt{\Turn}
    \Inf    {\Inf    {\Inf    {\InfBox{ \derLmmt `G |- t' : A | `D }
\kern-7mm
\Inf    {\Inf    {\InfBox{ \derLmmt `G,y`:A |- t : B | `D }
\Inf    { \derLmmt `G,y`:A |- `g : B | `g`:B,`D }
    }{ \derLmmt {\cell<t|`g>} : `G,y`:A |- `g`:B,`D }
}{ \derLmmt `G | {\mt y.<t|`g>} : A |- `g`:B,`D }
    }{ \derLmmt { \cell< t' | {\mt y.<t|`g>} >} : `G |- `g`:B,`D }
}{ \derLmmt `G |- {`M`g . < t' | {\mt y.<t|`g> } > } : B | `D }
\kern-10mm
\InfBox{ \derLmmt `G | e : B |- `D }
    }{ \derLmmt \cell <{`M`g . < t' | {\mt y.<t|`g> } > }|e> : `G |- `D }
\] 
 \end{Proof}


 \subsection{Non-Confluence, {\CBNPDF} and {\CBVPDF}}
 \label{sec:LMMT:NonConfluence}

Since \Cut-elimination of the classical sequent calculus is not confluent, neither is reduction in $\lmmt$; it has a critical pair in the command $\cell < `M`a.c_1 | \mt x . c_2 > $, which reduces to both $ c_1\tsubst[\mt x.c_2/`a]$ and $ c_2\tsubst[`M`a.c_1/x]$.
For example, in {\LK} the proof (where $(W)$ is the admissible weakening rule)
 \[ \def \TurnLK {\Turn}
 \Inf [\Cut]
 {\Inf [W]
 { \InfBox{\D_1}{\derLK `G |- `D }
 }{ \derLK `G |- A,`D }
 \Inf [W]
 { \InfBox{\D_2}{\derLK `G |- `D }
 }{ \derLK `G,A |- `D }
 }{\derLK `G |- `D }
 \]
reduces to both $\D_1$ and $\D_2$, different proofs, albeit for the same sequence; likewise, in $\TurnLmmt$ we can derive (where $`a$ does not appear in $c_1$, and $x$ does not appear in $c_2$):
 \[ \def \Turnlmmt {\Turn}
 \Inf [\Cut]
 {\Inf [`m]
 {\Inf [W]
 { \InfBox{\derLmmt c_1 : `G |- `D }
 }{ \derLmmt c_1 : `G |- `a{:}A,`D }
 }{ \derLmmt `G |- `M`a.c_1 : A | `D }
 \Inf [\mt]
 {\Inf [W]
 { \InfBox{\derLmmt c_2 : `G |- `D }
 }{ \derLmmt c_2 : `G,x{:}A |- `D }
 }{ \derLmmt `G | \mt x.c_2 : A |- `D }
 }{\derLmmt {\cell< `M`a.c_1 | \mt x.c_2 >} : `G |- `D }
 \]
and $\cell< `M`a.c_1 | \mt x.c_2 >$ reduces to both $c_1$ and $c_2$: witnesses to the same sequent, but not necessarily representing the same proof; this is related to Lafont's example.

On the other hand, the term $ `M `g . < `Lx.t | `M`a.c `. `g > $ is \emph{not} a $\lmmt$ critical pair, whereas its $\nclmu$-counterpart $ `@ (`Lx.M) (`M`a.\Cmd) $ (see \Sect\ref{nclmu in lmmt}) \emph{is} a $\nclmu$ critical pair.
We will come back to this at the end of \Sect\ref{nclmu in lmmt}. 

The $\lmmt$-calculus expresses the duality of \LK's left and right introduction in a very symmetric syntax.
But the duality goes beyond that: for instance, the symmetry of the reduction rules displays syntactically the duality between the {\CBV} and {\CBN} evaluations (see also \cite{WadlerDual'03}).


In \cite{Curien-Herbelin'00} the {\CBV} sub-reduction is not defined as a strategy but is obtained by forbidding a $\mt$-reduction when the command is also a $`m$-redex, whereas the {\CBN} sub-reduction forbids a $`m$-reduction when the redex is also a $\mt$-redex; there is no other restriction defined in \cite{Curien-Herbelin'00,Herbelin'05} in terms of not permitting certain contextual rules in the definition of {\CBV} and {\CBN}.
Since we want to stude {\CBN} and {\CBV} (left-most outer-most) reduction \emph{strategies} in the sense that each term has at most one contractable cut, we will define those here.

 \begin{definition}[{\CBN} and {\CBV} reduction strategies for $\lmmt$] \label{CBN and CBV in lmmt}
 \begin{enumerate}
 \firstitem
\emph{Values} $V$ are defined by the grammar $V ::= x \mid`Lx.t $, and \emph{stacks}\footnote{In \cite{Herbelin'05}, stacks are called \emph{linear evaluation contexts}.} $S$ are defined by $S ::= `a \mid t `. S$.

 \item
The {\CBN}-reduction $\redlmmtn$ is defined by limiting rule $(`m)$ as on the left, to be only applicable to stacks;
the {\CBV}-reduction $\redlmmtv$ is defined by limiting rule $(\mt)$ as on the right, to be only applicable to values.

 \[ \begin{array}{rrcl}
(`l): & \cell<`L x.t_1 | t_2`.e >	&\reduc& \cell< t_2 | {\mt x.< t_1|e>} > \\
(\mun): & \cell< `M`b.c | S > & \reduc & c \tsubst[S/`b] \\
(\mt): & \cell< t | \mt x . c > & \reduc & c \tsubst[t/x] \\
(`h`m): & `M`a.< t | `a > & \reduc & t \quad (`a \notele \FV{t})
 \end{array} 
 \qquad
 \begin{array}{rrcl}
(`l): & \cell<`L x.t_1 | t_2`.e >	&\reduc& \cell< t_2 | {\mt x.< t_1|e>} > \\
(`m): & \cell< `M`b.c | e > & \reduc & c\tsubst[e/`b] \\
(\mtv): & \cell< V | \mt x . c > & \reduc & c \ntsubst[V/x] \\
(`h`m): & `M`a.< t | `a > & \reduc & t \quad (`a \notele \FV{t})
 \end{array} \]

 \item
The {\CBN}-reduction strategy $\redlmmtN$ is defined by restricting $\redlmmtn$ through limiting the contextual rules to: 
 \[ \begin{array}{rcll}
c \reduc c' & \Imply & `M`b . c \reduc `M`b . c' 
 \end{array} \]
The {\CBV}-reduction strategy $\redlmmtV$ is defined from $\redlmmtv$ the same way.

 \end{enumerate}
 
 \observe
Both reduction strategies only reduce terms or commands, never environments, and only the reduction rules $(`m)$ and $(\mt)$ are changed or restricted.
Since terms are either variables, abstractions, or $`m$-terms, and reduction under a $`l$ is not allowed, only $`m$-terms can be reduced.

 \end{definition}

Of course \Thm \ref{Subject reduction nclmu} holds for the {\CBN} and {\CBV} strategies as well.


\section{{\XiCalcPDF}: {\XCalcPDF} with (Implicit) Substitutions}
\label{Xs def}
\label{sec:Xi}

    \label{sec:XCalculus}
    \label{sec:Xsyntax}

This section explores the $\Xs$-calculus, a variant of the calculus $\X$ which in \cite{vBLL-ICTCS'05,Bakel-Lescanne-MSCS'08} has been proven to be a fine-grained implementation model for various well-known calculi, including the $`l$-calculus, $\Lx$, $\lmu$, and $\lmmt$.
The calculus $`X$ looks to represent proofs in (the implicative fragment of) the sequent calculus \LK,  known as Kleene's $G_3$ \cite{Kleene'52}.
The variant of $\X$ we consider was defined by \namecite{Summers-PhD'08}; it uses substitution to express reduction.

 \subsection{The principle of {\XCalcPDF}} \label{sec:XCalculus:Intro}

The variant of {\LK} we consider offers a natural presentation of the classical propositional calculus with implication.
{\LK} has no notion of active formulae, thus in $`X$, terms denote the entire derivation.
This is not the case for  the proof systems underpinning $\lmu$ and $\lmmt$ deviate from that, in that a proof focusses on one formula, and terms (and environments) represent a proof of that formula. 
In contrast, a term $P$ in $`X$ is a witness of a type judgement ${ \derX P : `G |- `D }$, whereby $`G$ is a context of inputs (called \emph{sockets}), $x_1 `:A_1,\dots, x_m`:A_m$, and $`D$ a context of outputs (called \emph{plugs}), $`a_1`:B_1,\dots,`a_n`:B_n$.
The type system of $`X$ can be seen to annotate ${\LK}$ derivations by labelling these inputs and outputs.

For the rule $(\Ax)$, we must have syntax for a sequent of the form ${ \derLK `G, x `: A |- `a `: A, `D }$.
The only sensible choice is to explicitly associate $x$ and $`a$, which forms the `capsule' ${\caps<x.`a> }$
, that explicitly states which assumptions the rule uses,
 \[ \def  \TurnX {\Turn}
\Inf	{\derX \caps<x.`a> : `G,x`:A |- `a`:A,`D } 
 \]
The syntax ${\caps<x.`a> }$ denotes connecting the input $x$ to the output $`a$.

Similarly, the rule $(\Cut)$ must denote which occurrences of the principal formula to remove.
Given derivations of ${ \derLK `G |- `a `: A, `D }$ and ${ \derLK `G, x `: A |- `D }$, a $(\Cut)$ will remove $x$ and $`a$ from the context; this implies that a cut \emph{binds} these variables.
 \[ \def  \TurnX {\Turn}
\Inf	{\derX P : `G |- `a `: A, `D 
	\quad
	\derX Q : `G, x `: A |- `D 
	}{ \derX \cut P `a + x Q : `G |- `D }
 \]
Here, $\hat{\cdot}$ denotes the binding of a variable.
The term ${\cut P `a + x Q }$ then explicitly associates $`a$ with $\hat{x}\,Q$ in $P$, and, symmetrically, $x$ with $P\,\hat{`a}$ in $Q$.
Notice that the notation stresses the view of terms in $\X$: variables correspond to \textit{inputs} and are bound on the left, whereas names correspond to \textit{outputs} and are bound on the right. 

We can continue this process for the left and right arrow rules, giving each formula a label which the corresponding term uses,
 \[ \def  \TurnX {\Turn} \begin{array}{c@{\qquad}c}
\Inf	{\derX P : `G,x`: A |- `a`: B,`D 
	}{ \derX \exp x P `a . `b : `G |- `b`: A\arr B,`D } 
    &
\Inf	{\derX P : `G |- `a`: A,`D 
	 \quad 
	 \derX Q : `G,x`: B |- `D 
	}{ \derX \imp P `a [y] x Q : `G,y`: A\arr B |- `D } 
 \end{array} \]
The first rule must discharge $x$ and $`a$ from the context, therefore it binds them and (potentially) introduces the output $`b$ of a function.
The term ${\exp x P `a . `b }$ then exports a function of type $A\arr B$ through $`b$, and defines this function by sending its input along $x$ in $P$, and retrieving its output from $`a$.

The second rule similarly binds $`a$ and $x$ in the underlying term.
As the resulting formula has an input of type ${A\arr B}$, the rule gives this input a name $y$ and adds it to the context;
the name $y$ may already appear in $`G$, in which case it must have already had the type ${A\arr B}$.
The term ${\imp P `a [y] x Q }$ therefore expects to receive a function (get connect to via a cut) on $y `: A\arr B$, to which it links the function's input $A$ with the output $`a$ of $P$, and the output $B$ with the input $x$ of $Q$.

This then leads to: 

 \begin {definition}[Syntax for $`X$ \cite{Bakel-Lescanne-MSCS'08}] \label{X terms}
The terms (or \emph{nets}) of the $`X$-calculus are defined by the following syntax, where $x,y$ range over the infinite set of \textsl{term variables} (also called \emph{sockets}), and $`a, `b$ over the infinite set of \textsl{context variables} (also called \emph{plugs}); the common name for both is \emph{connector}.
 \[ \begin{array}{rcccccccc}
P,Q &::=& \caps<x.`a> & \mid & \exp y P `b . `a & \mid & \imp P `b [y] x Q & \mid & \cut P `a + x Q
\\
&& \textsl{capsule} & & \textsl{export} & & \textsl{import} & & \textsl{cut}
 \end{array} \]

 \observe 
We borrow the terminology from $\lmmt$, and call $P$ in $ \cut P `a + x Q $ or $ \imp P `b [y] x Q $ a term, and $Q$ a context.
 \end {definition}
The $\hat{`.}$ symbolises that the connector underneath is bound\footnote{Using the \emph{``hat''}-notation, we are keeping in line with the old tradition of \emph{Principia Mathematica}~\cite{Whitehead-Russel'25,Whitehead-Russel'97}; it is rumoured that this notation morphed into $`l$ because of type setting problems.} in the adjacent term.
The notion of bound and free connector is defined as usual, and we will identify terms that only differ in the names of bound connectors, as usual.
We adopt Barendregt's convention in that free and bound connectors of terms will be different.

The elimination of rule $(\Cut)$ plays a major role in {\LK}, as cut-free proofs enjoy nice proof-theoretic properties.
Gentzen proposed proof reductions by cut-elimination, and showed that the inner-most, left-most strategy terminates  ; the general strategy involves moving cuts upward through a proof tree, applying the latter cuts to sub-formulae of the principal formula of the original cut.
These reductions become the fundamental principle of computation in $`X$.

A prototypical eliminable cut is that with a matching left and right rule as premises,
 \[ \def \TurnLK {\Turn}
\Inf	[\Cut]
	{\Inf	{\InfBox{ \derLK `G,A |- B,`D }
		}{ \derLK `G |- A\Arrow B,`D }
	 \quad
	 \Inf	{\InfBox	{\derLK `G |- A,`D }	
		 \quad
		 \InfBox { \derLK `G,B |- `D }
		}{ \derLK `G, A\Arrow B |- `D }
	}{ \derLK `G |- `D }
\]
Naturally, we can eliminate this use of $(\Cut)$ on $A\Arrow B$ by instead $(\Cut)$-ing on $A$ and $B$,
although we have a choice on the order of these propagated cuts:
 \[ \def\TurnLK {\Turn} \begin{array}{c@{\qquad}c}
\Inf	[\Cut]
	{\Inf	[\Cut]
		{\InfBox { \derLK `G |- A, `D }
		\quad
		\InfBox { \derLK `G,A |- B,`D }
		}{ \derLK `G |- B,`D }
	 \quad
	 \InfBox { \derLK `G,B |- `D }
	}{ \derLK `G |- `D }
    &
\Inf	[\Cut]
	{ \InfBox {\derLK `G |- A, `D }
	 \quad
	 \Inf	[\Cut]
	 	{\InfBox { \derLK `G,A |- B,`D }
		\quad
		\InfBox { \derLK `G,B |- `D }
		}{ \derLK `G,A |- `D }
	}{ \derLK `G |- `D }
\end{array}\]
Although each resulting proof-tree has more cuts than the former, the latter two cut on only strict subformulae of the original principal formula ${A\Arrow B }$ which froms the basic principle for Gentzen's termination proof.
In $`X$, cut eliminations of matching left/right rules induce \emph{logical} reductions.
The original derivation corresponds with a term of the form, ${\cut { \exp y P `b . `a } `a + x { \imp Q `g [x] z R } }$, and the cut-elimination to the following reductions,
\[\begin{array}{rcl}
    {\cut { \exp y P `b . `a } `a + x { \imp Q `g [x] z R } }
	&\reduc&
    \begin{cases}
	\cut Q `g + y { \cut P `b + z R }
	\\
	\cut { \cut Q `g + y P } `b + z R
    \end{cases}
\end{array}\]
Recalling our understanding of $(\arrL)$ and $(\arrR)$, the term ${\exp y P `b . `a }$ defines a function with input $y$ and output $`b$, and ${\imp Q `g [x] z R }$ feeds a function to arrive in $x$ with input from $`g$, and sends its output to $z$.
When cutting these two terms together, the above reduction associates these inputs and outputs accordingly;
the input $y$ of $P$ connects with $`g$, and the output $`b$ with the input $z$ of $R$.
As the $(\Cut)$ rule only makes one such connection at a time, the order in which one makes these connections induces one of the two reducts -- which need not yield the same set of normal forms (reduction is not confluent, so normal forms can be plenty).

In general, the premises of a cut may not be matching left and right rules,
 \[ \def \TurnLK {\Turn}
\Inf	[\Cut]
	{\InfBox{ \derLK `G |- A, `D }
	 \quad
	 \Inf	[\ArrL]
	 	{\InfBox{ \derLK `G, A |- B, `D }
		 \quad
		 \InfBox{ \derLK `G, A, C |- `D }
		}{ \derLK `G, A, B\Arrow C |- `D }
	}{ \derLK `G, B\Arrow C |- `D }
 \]
In this case, we wish to push the $(\Cut)$ on $A$ up the proof-tree, to find the occurrence of the left and right rules introducing $A$,
 \[ \def \TurnLK {\Turn}
\Inf	[\ArrL]
	{\Inf	[\Cut]
		{\InfBox { \derLK `G |- A, `D }
		 \quad
		 \InfBox { \derLK `G, A |- B, `D }
		}{ \derLK `G |- B , `D }
	 \quad
	 \Inf	[\Cut]
	 	{\InfBox { \derLK `G |- A, `D }
		 \quad
		 \InfBox { \derLK `G, A, C |- `D }
		}{ \derLK `G, C |- `D }
	}{ \derLK `G, B\Arrow C |- `D }
\]
In $`X$, such cut eliminations  correspond to the \emph{propagating} reductions.





 \subsection{Type System}

Type assignment for $`X$ is defined as follows:

\begin{definition}[Typing for $`X$ \cite{vBLL-ICTCS'05}]
    \label{Typing for X}
Using the notion of types, and contexts of variables and names of \Def \ref{tas lmu}, type assignment for $`X$ is defined as follows:
 \begin{enumerate}
 \item
\emph{Type judgements} are expressed via the ternary relation $\derX P : `G |- `D $, where $`G$ is a context of \textsl{sockets} and $`D$ is a context of \textsl{plugs}, and $P$ is a term, called the \emph{witness}.

 \item
 \emph{Context assignment for} $`X$ is defined by the following rules:
 \[ \def \TurnX {\Turn} 
\begin {array}{rl@{\dquad}rl} 
(\Cap): & 
\Inf	{\derX \caps<y.`b> : `G,y`: A |- `b`: A,`D } 
	&
(\Imp): & 
\Inf	{\derX P : `G |- `a`: A,`D 
	 \quad 
	 \derX Q : `G,x`: B |- `D 
	}{ \derX \imp P `a [y] x Q : `G,y`: A\arr B |- `D } 
	\\ [5mm]
(\Exp): & 
\Inf	{\derX P : `G,x`: A |- `a`: B,`D 
	}{ \derX \exp x P `a . `b : `G |- `b`: A\arr B,`D } 
	&
(\Cut): & 
\Inf	{\derX P : `G |- `a`: A,`D 
	 \quad 
	 \derX Q : `G,x`: A |- `D 
	}{ \derX \cut P `a + x Q : `G |- `D } 
 \end {array} \]

We write $\derX P : `G |- `D $ if there is a derivation that has this judgement in the bottom line. 

 \end {enumerate}
 \end {definition}

Notice that $`G$ and $`D$ carry the types of the free connectors in $P$, as unordered sets.
There is no notion of type for $P$ itself, instead the derivable statement shows how $P$ is connectable.


 \def \cutleft#1#2#3#4{ {#1} \sk \{ {#2} \DaggerL \XIn{#3}{#4} \nhsk \} \nsk }
 \def \cutright#1#2#3#4{ \{ \XOut{#1}{#2} \DaggerR {#3} \} \sk {#4}}
      
 \subsection{$\Xs$: $\X$ with implicit substitutions}

Van Bakel and Lescanne \cite{Bakel-Lescanne-MSCS'08} argue that $\X$ is a calculus with explicit substitution, which makes it suitable to encode calculi such as the $`l$-calculus, $\Lx$, $\lmu$ and $\lmmt$; these encodings are with respect to full reduction. 
Since in this paper we look to model similar results for restrictions of those calculi to {\CBN} or {\CBV} strategies, this explicit character of $\X$ poses a problem.
In particular, to show the preservations results in \Sect\ref{lmmt in X} for our encoding of $\lmmt$ into $`X$, we need to be able to simulate the implicit substitution of $\lmmt$; this is non-problematic for full reduction, as shown in \cite{Bakel-Lescanne-MSCS'08}, but when modelling the {\CBN} and {\CBV} strategies, in the proofs of \Lmm\ref{mu simulation} and \ref {mut simulation} we would be forced to propagate the active cuts (\ie execute reduction steps to model substitution) all through the terms, even where not permitted under {\CBN} or {\CBV} strategies for $`X$, so would not be able to show full simulation.

Take for example $\lmmt$'s reduction rule $(`m)$, then we would need to show a property like \[ \Tran[\cell<`M`b.c | e >] \rtcredX \Tran[{c \tsubst[e/`b]}] , \] so need to show that $\lmmt$'s substitution is preserved under the translation function $\Tran[`.]$.
Since substitution is not part of the definition of $\X$, the only thing that is possible to show, as is done in \cite{Bakel-Lescanne-MSCS'08}, is that the interpretation of the substitution of $\lmmt$ gets \emph{executed} through reduction in $\X$, mainly through the propagation rules, in all sub-terms.
This implies that, in a proof for this result, to fully achieve $\Tran[{c \tsubst[e/`b]}]$, these steps need to be executed \emph{in full}, irrespective of the restriction of reduction in evaluation contexts that the {\CBN} or {\CBV}-reduction strategy imposes, thus forcing us, in practice, to allow for reduction, at least of the propagation rules, to take place everywhere, thus violating the reduction strategies.
No longer considering active cuts as reducible terms, but rather expressing propagation of active cuts through substitution, avoids that problem.

We therefore change our definition to that of $\Xs$, a variant of $\X$ defined by \namecite{Summers-PhD'08}§ (called $`X^i$ there) that has a notion of substitution; notice that part of the justification of this definition lies in \Lmm\ref{X renaming lemma}.
\citet{Summers-PhD'08} defines $`X^i$ as a variant or $`X$ (adding also negation), that replaces the propagation rules by a substitution-like operation $ P `a \leftrightarrowtriangle x $ and $ `a \leftrightarrowtriangle x Q $; the benefit of this change is that activated and unactivated cuts cannot interfere, and makes it more clear that the intention of reduction is that activated cuts run to completion. 
We agree with Summers that separating the rules that express propagation of activated cuts from the logical reduction rules gives a notion of reduction that is easier to deal with; it should be noted, however, that treating activated cuts as implicit substitution, rather than explicit substitution, as was the case in \cite{Bakel-Lescanne-MSCS'08}, severely restricts $`X$'s notion of reduction and possible reducts.


In the definition of reduction, it is important to know when a connector is `introduced', that is, its first appearance in a term; this plays a crucial role in the reduction rules, both for $\X$ and for $\Xs$.
\begin{definition}[Connector Introduction \cite{Bakel-Lescanne-MSCS'08}]~
    \label{def:XCalculus:Syntax:Introduces}
    \begin{enumerate}
	\item 
	$P$ \emph{introduces} $x$ when either $P = \imp Q `b [x] y R $ with $x \notin \FS{Q,R}$, or $P = \caps<x.`a>$.
    
	\item 
	$P$ \emph{introduces} $`a$ either $P = \exp x Q `b . `a $ and $`a \notin \FP{Q }$, or $P = \caps<x.`a>$. 
    \end{enumerate}
\end{definition}
Essentially, $P$ introduces a connector when it contains the first free use of said connector.
In the following, the logical reductions (\Def \ref{Xs Reduction}) only apply for terms that introduce their bound connector in a cut; 
otherwise, activating a cut implies that reduction must propagate this cut to all uses of the connectors before eliminating it.

The reduction system is based on cut-elimination, as discussed in Section \ref{sec:XCalculus:Intro}.
The core rules are the logical ones, and handle eliminating instances of $(\Ax)$ and left/right rule pairs.
For details, we refer to \cite{Bakel-Lescanne-MSCS'08}].

 \resetqsymbol{`X}{\Xs} 

 \begin {definition}[Substitution on $\Xs$]
The terms of $\Xs$ are those of $\X$:
 \[ \begin{array}{rcccccccc}
P,Q &::=& \caps<x.`a> & \mid & \exp y P `b . `a & \mid & \imp P `b [y] x Q & \mid & \cut P `a + x Q
 \end{array} \]

\Comment{
 \[ \begin{array}{lclcl}
	\cut {\cut \caps<y.`a> `g + z \caps<z.`a> } `a + x \caps<v,`d> & \reduc \\
	\cutL {\cut \caps<y.`a> `g + z \caps<z.`a> } `a + x \caps<v,`d> & \reduc \\
	\cut {\cutL \caps<y.`a> `a + x \caps<v,`d> } `g + z { \cutL \caps<z.`a> `a + x \caps<v,`d> } & \reduc \\
	\cutL \caps<y.`a> `a + x \caps<v,`d> & \reduc \\
	\cut \caps<y.`a> `a + x \caps<v,`d> & \reduc \\
	\cutR \caps<y.`a> `a + x \caps<v,`d> & \reduc & \caps<v,`d> \\
 \end{array} \]
$  |\DaggerLeft| \quad |\DaggerRight| $

}

\emph{Right substitution} on $\Xs$ is defined through (where we use `$\substo$' for `$=$' to indicate the `flow' of the substitution):
 \[ \begin {array}{rrcl@{\quad}l}
(\deactL): & \cutL \caps<y.`a> `a + x P &\substo& P \tsubst[y/x] \\
(\Li): & \cutL Q `a + x P &\substo& Q & (`a \notele \fn(Q) ) \\
(\Lii): & \cutL { \exp y Q `b . `a } `a + x P &\substo&
	\cut { \exp y { \cutL Q `a + x P } `b . `g } `g + x P & (
	`g\textit{ fresh}) \\
(\Liii): & \cutL { \exp y Q `b . `g } `a + x P &\substo& \exp y { \cutL Q `a + x P } `b . `g & (`g \not= `a) \\
(\Liv): & \cutL { \imp Q `b [z] y R } `a + x P &\substo& \imp { \cutL Q `a + x P } `b [z] y { \cutL R `a + x P } \\
(\Lv): & \cutL { \cut Q `b + y R } `a + x P &\substo& \cut { \cutL Q `a + x P } `b + y { \cutL R `a + x P }
 \end {array} \]

\noindent
and \emph{left substitution} on $\Xs$ through:
 \[ \begin {array}{rlcl@{\quad}l}
(\deactR): & \cutR P `a + x \caps<x.`b> &\substo& P \tsubst[`b/`a] \\
(\Ri): & \cutR P `a + x Q &\substo& Q & (x \notele \fv(Q) ) \\
(\Rii): & \cutR P `a + x { \exp y Q `b . `g } &\substo& \exp y { \cutR P `a + x Q } `b . `g \\
(\Riii): & \cutR P `a + x { \imp Q `b [x] y R } &\substo& 
	\cut P `a + z { \imp { \cutR P `a + x Q } `b [z] y { \cutR P `a + x R } } 
	 & z\textit{ fresh}) \\
(\Riv): & \cutR P `a + x { \imp Q `b [z] y R } &\substo& \imp { \cutR P `a + x Q } `b [z] y { \cutR P `a + x R } \hspace*{3mm} & (z \not= x) \\
(\Rv): & \cutR P `a + x { \cut Q `b + y R } &\substo& \cut { \cutR P `a + x Q } `b + y { \cutR P `a + x R }
 \end {array} \]

\observe
We will often use $ \XOut{P}{`a}$ for $\XOut{P\tsubst[`b/`a]}{`b}$ when this term gets created through reduction and $`b \notele P$.
Similarly for  $\XIn{x}{Q}$ instead of $\XIn{y}{Q\tsubst[y/x]}$ with $y \notele Q$.
 \end{definition}

In $\X$, `$\substo$' are reduction steps; although here it is an equality, and `$=$' could be used as well, we hope that using the arrow notation will help the reader read the proofs of this paper.


 \begin{definition}[Reduction on $`X$] \label{Xs Reduction}
The single step reduction steps for $`X$ are defined through the logical rules (were $`a$ and $x$ are both introduced):
 \[ \begin {array}{rr@{\,}c@{\,}lcll}
(\Cap): &
\caps<y.`a> &\dag{`a}{x}& \caps<x.`b> &\reduc& \caps<y.`b>
\\
(\Exp): &
( \exp y P `b . `a ) &\dag{`a}{x}& \caps<x.`g> &\reduc&
 \exp y P `b . `g
\\
(\Imp): &
\caps<y.`a> &\dag{`a}{x}& ( \imp Q `b [x] z R ) &\reduc&
 \imp Q `b [y] z R
\\
(\Ins): &
( \exp y P `b . `a ) &\dag{`a}{x}& ( \imp Q `g [x] z R ) &\reduc&
 \begin{cases}
 \cut Q `g + y { \cut P `b + z R }
\\
 \cut { \cut Q `g + y P } `b + z R \\
 \end{cases} 
\\
\Span{5}{\textrm{\kern -14mm The substitution activation rules are:}} \\ [2mm]
(\actL): & \multicolumn{3}{c}{ \cut P `a + x Q \hspace*{5mm} } &\reduc& \cutL P `a + x Q \quad (\textit{$P$ does not introduce $`a$})
 \\
(\actR): & \multicolumn{3}{c}{ \cut P `a + x Q  \hspace*{5mm} } &\reduc& \cutR P `a + x Q \quad (\textit{$Q$ does not introduce $x$})
 \end {array} \]

We write $P \redXs Q$ if $P$ reduces to $Q$ using one of the above logical or substitution activation rules, and use $\rtcredXs$ for the reflexive, transitive, compatible reduction relation generated by $\redXs$.

 \end {definition}
Notice that, by these rules, in case both $`a$ is not introduced in $P$ and $x$ is not introduced in $Q$, activation can take place in \emph{both} directions, so then the cut $\cut P `a + x Q $ again forms a critical pair.

Notice that the substitution rules $(\Lii)$ and $(\Riii)$ introduce new cuts: defining, for example, $(\Lii)$  as $ \cutL { \exp y Q `b . `a } `a + x P \substo \cutL { \exp y { \cutL Q `a + x P } `b . `g } `g + x P
$ would ignore that perhaps $x$ is introduced in $P$, so a logical rule would be applicable ($`g$ is introduced, since fresh), but, more importantly, would lead to looping substitutions.

We will now define the {\CBN} and {\CBV} reduction and strategies for $`X$.

 \begin {definition}[Call By Name and Call-by-Value reduction and strategies for $`X$] ~
 \begin{enumerate}
 \item
For $`X$, the \CBN-reduction $\redXsn$ is defined by limiting $\redXs$ through replacing rule $(\actL)$ with $(\actLN)$:
 \[ \begin {array}{rlcll}
(\actLN): & \cut P `a + x Q & \redCBN & \cutL P `a + x Q , & (\textit{$P$ does not introduce $`a$ and $Q$ introduces $x$)}. \\
(\actR): & \cut P `a + x Q &\reduc& \cutR P `a + x Q & (\textit{$Q$ does not introduce $x$}) 
 \end {array} \]

As in \cite{Lengrand'03}, we only allow one variant of $(\Ins)$:
 \[ \begin{array}{rcl}
\cut { \exp y P `b . `a } `a + x { \imp Q `g [x] z R } & \redCBN & \cut Q `g + y { \cut P `b + z R }
 \end{array} \]

 \item
The \CBV-reduction $\redXsv$ is defined by limiting $\redXs$ replacing rule $(\actR)$ with $(\actRV)$:
 \[ \def\arraystretch{1} \begin {array}{rlcll}
(\actL): & \cut P `a + x Q &\reduc& \cutL P `a + x Q & (\textit{$P$ does not introduce $`a$}) \\
(\actRV): & \cut P `a + x Q & \redCBV & \cutR P `a + x Q , & (\textit{$P$ introduces $`a$ and $Q$ does not introduce $x$}).
 \end {array} \]

As for \CBN, we only allow the first variant of $(\Ins)$:
 \[ \begin{array}{rcl}
\cut { \exp y P `b . `a } `a + x { \imp Q `g [x] z R } & \redCBV & \cut Q `g + y { \cut P `b + z R }
 \end{array} \]
 
  \item 
The {\CBN} reduction strategy $\redXsN$ is defined from $\redXsn$ by 
not allowing any contextual rules; only terms that are cuts are reducible; the {\CBV} reduction strategy $\redXsV$ is defined from $\redXsv$ in the same way.

 
 \end {enumerate}
 \end {definition}
So, for {\CBN}, when a logical rule cannot be applied to $\cut P `a + x Q $, go left only if $x$ is introduced in $Q$; otherwise, go right; 
for {\CBV}, in that case go right only if $`a$ is introduced in $P$; otherwise, go left. 

This way, we obtain two notions of reduction that are locally confluent ($\redXsn$ and $\redXsv$) because of the absence of critical pairs;
$\redXsN$ and $\redXsV$ are both confluent since they are strategies.
Notice that the only difference between {\CBN} and {\CBV} reduction lies in activation and that both strategies do not allow for reduction in contexts, nor inside substitutions. 

The soundness result of simple type assignment with respect to reduction is stated as usual:

\begin{theorem}[Witness reduction]
    \label{Witness reduction Xs}
    If $\derX P : `G |- `D $, and $P \redXs Q$, then $\derX Q : `G |- `D $.
\end{theorem}
The proof for this result is basically that for the similar result for $\redX$ in \cite{Bakel-Lescanne-MSCS'08}.

We have the following results. 
These are already shown in \cite{Bakel-Lescanne-MSCS'08}, but for $\X$; since here the exact steps that are needed in the reduction must be known when modelling {\CBN} or \CBV, we give the proofs in detail.

 \begin{lemma} 
 \label{renaming lemma} \label{garbage collection} \label{X renaming lemma}
 \begin{enumerate}
 \firstitem $ \cutL P `a + x \caps<x.`b> = P \tsubst[`b/`a] $, and $ \cut P `a + x \caps<x.`b> \rtcredXs P \tsubst[`b/`a] $.
 \item $ \cutR \caps<y.`a> `a + x P = P \tsubst[y/x] $, and $ \cut \caps<y.`a> `a + x P \rtcredXs P \tsubst[y/x] $.
 \item $ (\Lgc)$: $ \cutL Q `a + x P = Q $, $`a \notele \fn(Q)$. 
 \item $(\Rgc)$: $ \cutR P `a + x Q = Q $, $x \notele \fv(Q)$.
 \end{enumerate}
 \end{lemma}
 \begin{Proof}{By induction on the structure of nets. 
Substitution activation plays no role in this proof}{\Lmm \ref{garbage collection}}
By induction on the structure of nets.
 \begin{enumerate} \itemsep 4\point
 \item 
 \begin{description}

 \firstmyitem[$P = \caps<y.`a>$]
\cutL \caps<y.`a> `a + x \caps<x.`b> &\substo& \bsp (\deactL)
\Tran[V y]{`b} &\ByDef&
\caps<y.`a> \tsubst[`b/`a] 
 \end{array}$
 
 \myitem [$P = \caps<y.`g>$, $`g\not=`a$]
\cutL \caps<y.`g> `a + x \caps<x.`b> &\substo& \bsp (\Li)
\caps<y.`g> &\ByDef&
\caps<y.`g> \tsubst[`b/`a] 
 \end{array}$

 \myitem [$P = \exp y Q `g . `a $]
\cutL { \exp y Q `g . `a } `a + x \caps<x.`b> &\substo& \bsp(\Lii)
\exp y { \cutL Q `a + x \caps<x.`b> } `g . `b &\substo& \bsp (\IH) 
\exp y { Q \tsubst[`b/`a]} `g . `b &\ByDef& \\
( \exp y Q `g . `a ) \tsubst[`b/`a] 
 \end{array}$

 \myitem [$P = \exp y Q `g . `d $, $`d\not=`a$]
\cutL { \exp y Q `g . `d } `a + x \caps<x.`b> &\substo& \bsp(\Liii)
\exp y { \cutL Q `a + x \caps<x.`b> } `g . `d &\substo& \bsp (\IH) 
\exp y Q\tsubst[`b/`a] `g . `d &\ByDef& \\ 
( \exp y Q `g . `d ) \tsubst[`b/`a] 
 \end{array}$

 \myitem [{$P = \imp Q `g [y] z R $}]
\cutL { \imp Q `g [y] z R } `a + x \caps<x.`b> &\substo& \bsp (\Liv)
\imp { \cutL Q `a + x \caps<x.`b> } `g [y] z {\cutL R `a + x \caps<x.`b> } &\substo& \bsp (\IH) \\
\imp Q\tsubst[`b/`a] `g [y] z R\tsubst[`b/`a] &\ByDef& 
( \imp Q `g [y] z R ) \tsubst[`b/`a] 
 \end{array}$

 \myitem [$P = \cut Q `g + z R $]
\cutL { \cut Q `g + z R } `a + x \caps<x.`b> &\substo& \bsp (\Lv)
\cut { \cutL Q `a + x \caps<x.`b> } `g + z {\cutL R `a + x \caps<x.`b> } &\substo& \bsp (\IH) \\
\cut Q\tsubst[`b/`a] `g + z R\tsubst[`b/`a] &\ByDef& 
( \cut Q `g + z R ) \tsubst[`b/`a] 
 \end{array}$ \\ 

 \end{description}

 \noindent
For the second part, if $`a$ is introduced in $P$, the result follows by rules $(\Cap)$ or $(\Exp)$; 
otherwise $ \cut P `a + x \caps<x.`b> \redX (\actL) \cutL P `a + x \caps<x.`b> $, and the result follows by the first part.

 \item 
 \begin{description}

 \firstmyitem[$P = {\Tran[V x]{`b}} $]
\cutR \caps<y.`a> `a + x {\Tran[V x]{`b}} &\substo& \bsp (\deactR)
{\Tran[V x]{`b}} &\ByDef&
\Tran[V y]{`b} \tsubst[y/x] 
 \end{array}$
 
 \myitem [$P = \caps<z.`b>$, $z\not=x$]
\cutR \caps<y.`a> `a + x \caps<z.`b> &\substo& \bsp (\Ri)
\caps<z.`b> &\ByDef&
\caps<z.`b> \tsubst[y/x] 
 \end{array}$

 \myitem [$P = \exp z Q `g . `b $]
\cutR \caps<y.`a> `a + x { \exp z Q `g . `b } &\substo& \bsp (\Rii)
\exp z { \cutR \caps<y.`a> `a + x Q } `g . `b &\substo& \bsp (\IH) 
\exp z Q\tsubst[y/x] `g . `b &\ByDef&  \\
( \exp z Q `g . `b ) \tsubst[y/x] 
 \end{array}$

 \myitem [{$P = \imp Q `g [x] z R $}]
\cutR \caps<y.`a> `a + x { \imp Q `g [x] z R } &\substo& \bsp (\Riii) 
\imp { \cutR \caps<y.`a> `a + x Q } `g [y] z {\cutR \caps<y.`a> `a + x R }  \\ 
~=(\IH) ~ \imp Q\tsubst[y/x] `g [y] z R\tsubst[y/x] &\ByDef& 
( \imp Q `g [x] z R ) \tsubst[y/x] 
 \end{array}$

 \myitem [{$P = \imp Q `g [v] z R $, $v \not= x$}]
\cutR \caps<y.`a> `a + x { \imp Q `g [v] z R } &\substo& \bsp (\Riv) \\
\imp { \cutR \caps<y.`a> `a + x Q } `g [v] z {\cutR \caps<y.`a> `a + x R } &\substo& \bsp (\IH) \\
\imp Q\tsubst[y/x] `g [v] z R\tsubst[y/x] &\ByDef& 
( \imp Q `g [v] z R ) \tsubst[y/x] 
 \end{array}$

 \myitem [$P = \cut Q `g + z R $]
\cutR \caps<y.`a> `a + x { \cut Q `g + z R } &\substo& \bsp (\Rv) 
\cut { \cutR \caps<y.`a> `a + x Q } `g + z {\cutR \caps<y.`a> `a + x R } &\substo& \bsp (\IH) \\
\cut Q\tsubst[y/x] `g + z R\tsubst[y/x] &\ByDef& 
( \cut Q `g + z R ) \tsubst[y/x] 
 \end{array}$ \\

 \end{description}

 \noindent
For the second part, if $x$ is introduced in $P$, the result follows by rules $(\Cap)$ or $(\Imp)$; 
otherwise $ \cut \caps<y.`a> `a + x P \redX \cutR \caps<y.`a> `a + x P $, and the result follows by the first part.

 \item 
 \begin{description}

 \firstmyitem[$Q = \caps<y.`g>$, $`g\not=`a$]
\cutL \caps<y.`g> `a + x P &\substo& \bsp (\Li)
\caps<y.`g> 
 \end{array}$

 \myitem [$Q = \exp y R `g . `d $, $`d\not=`a$]
\cutL { \exp y R `g . `d } `a + x P &\substo& \bsp (\Liii)
\exp y { \cutL R `a + x P } `g . `d &\substo& \bsp (\IH) 
\exp y R `g . `d 
 \end{array}$

 \myitem [{$Q = \imp R `g [y] z S $}]
\cutL { \imp R `g [y] z S } `a + x P &\substo& \bsp (\Liv) 
\imp { \cutL R `a + x P } `g [y] z {\cutL S `a + x P } &\substo& \bsp (\IH) 
\imp R `g [y] z S 
 \end{array}$

 \myitem [$Q = \cut R `g + z S $]
\cutL { \cut R `g + z R } `a + x P &\substo& \bsp (\Lv) 
\cut { \cutL R `a + x P } `g + z {\cutL S `a + x P } &\substo& \bsp (\IH) 
\cut R `g + z S 
 \end{array}$

 \end{description}

 \item 
 \begin{description}

 \firstmyitem[$P = \caps<z.`b>$, $z\not=x$]
\cutR P `a + x \caps<z.`b> &\substo& \bsp (\Ri)
\caps<z.`b> 
 \end{array}$

 \myitem [$P = \exp z R `g . `b $]
\cutR P `a + x { \exp z R `g . `b } &\substo& \bsp (\Rii)
\exp z { \cutR P `a + x R } `g . `b &\substo& \bsp (\IH)
\exp z R `g . `b 
 \end{array}$

 \myitem [{$P = \imp Q `g [v] z R $, $v \not= x$}]
\cutR P `a + x { \imp Q `g [v] z R } &\substo& \bsp (\Riv) 
	\imp { \cutR P `a + x Q } `g [v] z {\cutR P `a + x R } &\substo& \bsp (\IH)  \\
\imp Q `g [v] z R
 \end{array}$

 \myitem [$P = \cut Q `g + z R $]
\cutR P `a + x { \cut Q `g + z R } &\substo&\bsp (\Rv)
\cut { \cutR P `a + x Q } `g + z {\cutR P `a + x R } &\substo& \bsp (\IH) 
\cut Q `g + z R 
 \end{array}$
\qed 

 \end{description}
 \end{enumerate}
 
 \observe
Notice that substitution activation plays no role in this proof.

 \end{Proof}

\Comment{
Type assignment for $\Xs$ is defined as for $\X$, and the following soundness result of type assignment with respect to $\redXs$ reduction is stated as usual and is easy to show.

 \begin {theorem}[Witness reduction for $\redXs$]
\label{Witness reduction}
If $\derX P : `G |- `D $, and $P \redXs Q$, then $\derX Q : `G |- `D $.
 \end {theorem}
 \begin{proof}
By \Thm\ref{Witness reduction X} and \ref{X implements Xs}.\QED
 \end{proof}
}




\section{Embedding {\nclmuPDF} into {\lmmtPDF}}
\label {nclmu in lmmt}
\label{sec:TranslateLmuLmmt}

In this sectiom, we will define an interpretation from $\slmu$ to $\lmmt$, and show that it preserves reduction, $\redlmmtn$, $\redlmmtv$, $\redlmmtN$ upto $\redlmmtn$ reductions necessitated by the interpretation of $`m$-substitutions, and, $\redlmmtV$ upto $\redlmmtv$ reductions created by the interpretation.

\subsection{Reduction-Specific Embeddings}
One has many choices when it comes to embedding $\nclmu$ into $\lmmt$.
The most immediate is to interpret application ${`@ M N }$ as the term $M$ operating against a stack of the form ${N `. e}$, where $e$ is the ambient evaluation context. 
\citet{Curien-Herbelin'00} provide such an interpretation, $\cdot^{>}$, which we relabel as $ \Sem[ `. ] $:
 \[ \begin{array}{rcl}
\Sem[x] & \ByDef & x \\
\Sem[`Lx.M] & \ByDef & `L x . \Sem[M] \\
\Sem[`@ M N ] & \ByDef & `M`a . < \Sem[M] | \Sem[N] `. `a > \\
\Sem[`M`a . `b M] & \ByDef & `M`a . < \Sem[M] | `b >
 \end{array} \]
Using this interpretation, and the observation that
 \[ \begin{array}{@{}lclclcl}
\Sem [`@ {`@ P Q } R ] &\ByDef& 
`M`a . < \Sem[`@ P Q ] | \Sem[R] `. `a > &\ByDef& 
`M`a . < {`M`b . < \Sem[P] | \Sem[Q] `. `b >} | \Sem[R] `. `a > &\redlmmt& 
`M`a . < \Sem[P] | \Sem[Q] `. \Sem[R] `. `a >
 \end{array}\]
we can point out the fundamental difference between $`m$-reduction in $\lmu$ and $\lmmt$.

 \begin{remark}[{$\Sem[`.]$} does not respect ($\redlmmtN$/$\redlmmtV$) reduction] \label{no reduction}
\resetqsymbol{`P}{\Sem [P]}
\resetqsymbol{`Q}{\Sem [Q]}
In $\lmu$, as discussed above, the intention of $`m$-reduction is to redirect an applicative context, but it has to do that `one term at the time'.
 \[ \begin{array}{@{}lclclcl}
`@ {`@ ({ `M`a.[`t] `@ x (`M`s.[`a] M)}) P } Q &\rednclmuN& 
`@ ({ `M`b.[`t] `@ x (`M`s.[`b]{`@ M P })}) Q &\rednclmuN& 
`M`g .[`t] `@ x (`M`s.[`g]{`@ {`@ M P } Q } ) 
 \end{array} \]
and therefore $`m$-reduction has to be self-replicating in nature. 
We would like to show that 
 \[ \begin{array}{@{}lclclcl} 
\Sem[{ `@ {`@ ({ `M`a.[`t] `@ x (`M`s.[`a] M)}) P } Q }] &\rtcredlmmtN& 
\Sem [{ `M`g .[`t] `@ x (`M`s.[`g]{`@ {`@ M P } Q } ) }]
 \end{array} \]
but cannot.

In $\lmmt$ we have the ($\redlmmtN$/$\redlmmtV$) reduction
 \[ \begin{array}{@{}lclclcl}
\Sem[ `@ {`@ ({ `M`a.[`t] `@ x (`M`s.[`a] M)}) P } Q ] &\ByDef& \\
%
`M`g . < {`M`b . < { `M`a . < { `M`d . < x | { `M`s . <\Sem[{M}] | `a >} `. `d >} | `t >} | `P `. `b >} | `Q `. `g > 
	&\redlmmt& (`b) \\
`M`g . < { `M`a . < { `M`d . < x | { `M`s . <\Sem[{M}] | `a >} `. `d >} | `t >} | `P `. `Q `. `g > 
	&\redlmmt& (`a) \\
`M`g . < { `M`d . < x | { `M`s . <\Sem[{M}] | `P `. `Q `. `g >} `. `d >} | `t > 
	&\redlmmt& (`d) \\
`M`g . < x | { `M`s . <\Sem[{M}] | `P `. `Q `. `g >} `. `t > 
 \end{array} \]
where the whole environment $ `P `. `Q `. `b $ gets pulled in the final step (notice that the first sequence of steps does not deal with the $`m`a$-redex contraction, but just prepares the environment, contracting the $`m$-redexes that are generated by the interpretation).
For the other part we have reduction:
 \[ \begin{array}{@{}lclclcl}
\Sem [{`M`g .[`t] `@ x ({`M`s.[`g]{`@ {`@ M P } Q } }) }] &\ByDef& \\
`M`g . < { `M`d . < x | { `M`s . < { `M`p . < { `M`r . < \Sem [M] | `P `. `r >} | `Q `. `p >} | `g >} `. `d >} | `t > &\redlmmt& (`d) \\
`M`g . < x | { `M`s . < { `M`p . < { `M`r . < \Sem [M] | `P `. `r >} | `Q `. `p >} | `g >} `. `t > &\redlmmt& (`r) \\
`M`g . < x | { `M`s . < { `M`p . < \Sem [M] | `P `. `Q `. `p >} | `g >} `. `t > 
 \end{array} \]
(Notice that the last reduction step, over $`r$, takes place in the continuation, so in not allowed in either $\redlmmtN$ or $\redlmmtV$.)
This illustrates, as already mentioned in \cite{Curien-Herbelin'00}, that we cannot show that the reduction strategies are preserved under $\Sem[`.]$, but at most can show that the terms involved share a reduct under $	\redlmmtn$ and $\redlmmtv$.
This is caused by the fact that the interpretation has to add additional $`m$-redexes to represent the applications added by $\lmu$'s $`m$-reduction (see \Rem \ref{added application}). 
In fact, the two $`m$ abstractions are fundamentally different: in $\lmu$, $`m$-reduction reconstructs the $`m$-abstraction, whereas in $\lmmt$ it disappears.
\citet{Curien-Herbelin'00} say that their reduction preservation result holds `modulo $`m$-expansions', but we will be more precise than that.

Here we like to distinguish the occurrences of $`m$ in our interpretations into $\lmmt$ that are `computational', \ie correspond to $`m$ abstractions already occurring in $\lmu$, and those that are generated by the interpretations; we will write $\UL{`m}$ for the latter, and $\redlmmtmo$ for reduction steps using these.
 \end{remark}

\Comment{
 \begin{remark}
 In $\lmu$, as discussed above, the intention of $`m$-reduction is to redirect an applicative context, but it has to do that `one term at the time'.
 \[
 \begin{array}{@{}lclclcl}
 `@ {`@ {`@ ({ `M`a.[`t] `@ x (`M`s.[`a] M)}) P } Q } R &\rednclmu& 
 `@ {`@ ({ `M`g .[`t] `@ x (`M`s.[`g]{`@ M P })}) Q } R &\rednclmu& \\ &&
 `@ ({ `M`d.[`t] `@ x (`M`s.[`d]{`@ {`@ M P } Q } )}) R &\rednclmu& \\ &&
 ~ `M`b.[`t] `@ x (`M`s.[`b]{`@ {`@ {`@ M P } Q } R } )
 \end{array}
 \]
 and therefore has to be recursive in nature.
 This is not the case for $\lmmt$, where we have the reduction
 \[
 \begin{array}{@{}lclclcl}
 \Sem[ `@ {`@ {`@ ({ `M`a.[`t] `@ x (`M`s.[`a] M)}) P } Q } R ] &\ByDef& \\
 `M`b . < {`M`g . < {`M`d . < {`M`a . <
 {`M`r . <x | {`M`s . <\Sem[M]|`a>} `. `r >}
 | `t >} | `P `. `d >} | `Q `. `g >} | `R `. `b > &\rtcredlmmt& \\
 `M`b . < {`M`a. <
 {`M`r . <x | {`M`s . <\Sem[M]|`a>} `. `r >} 
 | `t >} | `P `. `Q `. `R `. `b > &\redlmmt& \bsp (`a) \\
 `M`b . <
 {`M`r . <x | {`M`s . <\Sem[M] | `P `. `Q `. `R `. `b >} `. `r >}
 | `t > &\ByDef& 
 \end{array}
 \]
 where the whole environment $ `P `. `Q `. `R `. `b $ gets pulled in in one step (notice that the first sequence of steps does not deal with the $`m`a$-redex contraction, but just prepares the environment, contracting the $`m$-redexes that are generated by the interpretation).
 We can also reduce the term as follows:
 \[
 \begin{array}{@{}lclclcl}
 \Sem[ `@ {`@ {`@ ({ `M`a.[`t] `@ x (`M`s.[`a] M)}) P } Q } R ] &\ByDef& \\
 `M`b . < {`M`g . < {`M`d . < {`M`a . <
 {`M`r . <x | {`M`s . <\Sem[M]|`a>} `. `r >}
 | `t >} | `P `. `d >} | `Q `. `g >} | `R `. `b >
 &\redlmmt& \bsp (`a) \\
 `M`b . < {`M`g . < {`M`d . <
 {`M`r . <x | {`M`s . < \Sem[M] | `P `. `d >} `. `r >} | `t >}
 | `Q `. `g >} | `R `. `b >
 &\redlmmt& \bsp (`d) \\
 `M`b . < {`M`g . <
 {`M`r . <x | {`M`s . < \Sem[M] | `P `. `Q `. `g >} `. `r >} | `t >}
 | `R `. `b >
 &\redlmmtV& \bsp (`g) \\
 `M`b . < {`M`r . <x | {`M`s . < \Sem[M] | `P `. `Q `. `R `. `b >} `. `r >} | `t >
 \end{array}
 \]
 which `pulls in one term at the time', but does not do that using $`a$ but rather the $`m$-redexes generated by the interpretation for applications.
 In fact, in a way the $`m$-abstractions added by the interpretation implement the repetitive character of $\lmu$'s $`m$-reduction; for each surrounding application, a $`m$-abstraction is inserted, that will be used to execute one of the recursive steps.
 \end{remark}
}

 \begin{example}[\cite{Bakel-Lescanne-MSCS'08}]
 It is worthwhile to remark that reduction in the image of $\Sem{ `. }$, even when restricted to the $`l$-calculus, is not confluent.
 In fact, we have both:
 \[
 \begin{array}{@{}lcl}
 \Sem{`@ ({`lz.zz}) ({pp}) } &\ByDef& \\
 `M`a . < \Sem{`lz.zz} | \Sem{pp} `. `a > &\ByDef& \\
 `M`a . < `Lz.(zz)^n | \Sem{pp} `. `a > &\redlmmt& \bsp (`l) \\
 `M`a . < \Sem{pp} | {\mt z . <(zz)^n | `a >} > &\ByDef& \\
 `M`a . < {`M`b . <p|p `. `b>} | {\mt z . <(zz)^n | `a >} > &\redlmmt& \bsp (`m) \\
 `M`a . < p|p `. {\mt z . <(zz)^n | `a >} > & \ByDef& \\
 `M`a . < p|p `. {\mt z . <{`M`g . <z|z `. `g>} | `a >} >
 \end{array}
\dquad
 \begin{array}{lcl}
 \Sem{`@ ({`lz.zz}) ({pp}) } &\redlmmt& \bsp (`l) \\
 `M`a . < \Sem{pp} | {\mt z . <(zz)^n | `a >} > &\ByDef& \\
 `M`a . < \Sem{pp} | {\mt z . <{`M`g . <z|z `. `g>} | `a >} > &\redlmmt& \bsp (\MuT) \\
 `M`a . < {`M`g . <\Sem{pp}|\Sem{pp} `. `g>} | `a > &\redlmmt& \bsp (`m) \\
 `M`a . < \Sem{pp} | \Sem{pp} `. `a > &\ByDef& \\
 `M`a . < {`M`b . <p|p `. `b>} | \Sem{pp} `. `a > &\redlmmt& \bsp (`m) \\
 `M`a . <p|p `. \Sem{pp} `. `a > &\ByDef& \\
 `M`a . <p|p `. {`M`b . <p|p `. `b>} `. `a >
 \end{array}
 \]
 This holds for all the interpretations from $\nclmu$ to $\lmmt$ we discuss in this paper.
 Notice that the right reduction is in $\redlmmtN$.
 The left is {$\redlmmtV$}; it can be extended in $\redlmmtv$ with
 \[
 \begin{array}{lcl}
 `M`a . < p|p `. {\mt z . <{`M`g . <z|z `. `g>} | `a >} > & \redlmmt (`m) &
 `M`a . < p|p `. {\mt z . <z|z `. `a>} >
 \end{array}
 \]
but in this step reduction takes place in the continuation.
Notice that the terms $ `M`a . < p|p `. {\mt z . <z|z `. `a>} > $ and $ `M`a . <p|p `. {`M`b . <p|p `. `b>} `. `a > $ are both in $\redlmmt$ normal form, and are very different.
 \end{example}

 The above interpretation $\Sem[ `. ]$ works well when modelling the {\CBN} strategy, but creates problems for the one for \CBV.
 As we will see below, in the implementation of both these $\lmu$-reduction strategies in $\lmmt$, no reduction will take place in an environment, so not in $e$ in $\cell<t|e>$; environments are seen as resources, where data can be taken from, not places where evaluation takes place; this corresponds to the approach of KAM. 
This implies that when interpreting the reduction $`@ V M \rednclmuV `@ V N $, in the interpretation we would like to run $`M`a . < \Sem[V] | \Sem[M] `. `a > $ to $`M`a . < \Sem[V] | \Sem[N] `. `a > $ in $\redlmmtV$, but cannot, since we cannot run in the environment.

In \cite{Curien-Herbelin'00}, as is common when encoding {\CBN} and {\CBV} reduction, Curien and Herbelin define \emph{two} separate encodings for $`l`m$ into $\lmmt$: one, $\SEMcbn{ `. }$, to model {\CBN} and another, $\SEMcbv{ `. }$, to model \CBV:

 \begin {definition}[\cite{Curien-Herbelin'00}] \label{CBN CBV Herbelin}
The interpretations $ \SEMcbn{ `. } $ and $ \SEMcbv{ `. }$ of $`l`m$ into $\lmmt$ are defined by:
 \[ \begin {array}{rcl}
 \SEMcbn x & \ByDef & x \\
 \SEMcbn{`Lx.M} & \ByDef & `Lx.\SEMcbn{M} \\
 \SEMcbn{`@ M N } & \ByDef & `M`a . < \SEMcbn{M} | \SEMcbn{N} `. `a > \\
 \SEMcbn{`m`b.\Cmd} & \ByDef & `m`b.\SEMcbn{\Cmd} \\
 \SEMcbn{[`a]M} & \ByDef & \cell< \SEMcbn{ M} | `a >
 \end {array}
 \qquad
 \begin {array}{rcl}
 \SEMcbv x & \ByDef & x \\
 \SEMcbv{`Lx.M} & \ByDef & `Lx.\SEMcbv{ M } \\
 \SEMcbv{`@ M N } & \ByDef & `M`a . < \SEMcbv{N} | {\mt x . < \SEMcbv{M} | x `. `a >} > \\
 \SEMcbv{`M`b.\Cmd} & \ByDef & `m`b.\SEMcbv{\Cmd} \\
 \SEMcbv{[`a]M} & \ByDef & \cell< \SEMcbv{ M} | `a >
 \end {array} \]
 \end {definition}

Observe that $\SEMcbn{ `. }$ is Herbelin's interpretation $ \Sem[ `. ] $ we mentioned above, and that these interpretations only differ in the case for application; remark that we have:
 \[ \begin{array}{cccccccc}
`M`a . < t_1 | {\mt x . < t_2 | x `. `a >} > & &\redlmmtN& \bsp (x) & `M`a . < t_2 | t_1 `. `a >
 \end{array} \]
so, essentially, $ \SEMcbv{`@ M N } \redlmmt \SEMcbn{`@ M N } $, 
and in effect, Herbelin only considers $\SEMcbv{ `. }$ in \cite{Curien-Herbelin'00};
that paper only deals with $(`b)$ and $(\mur)$-reduction, and defines {$\redlmmtV$} reduction as we define $\redlmmtv$, by limiting the operands in those rules to values, so does not deal with {\CBV} $\lmu$.

Curien and Herbelin \cite{Curien-Herbelin'00} state reduction preservation results for their encodings (formulated as if $\lmu$-reduction is represented through $\lmmt$-reduction, modulo $(`m)$-expansion), but give very little detail.
Their results are stated with respect to the notion of {\CBN} and \CBV-reduction for $\lmmt$ that just remove the $\Pair<(`m),(\MuT)>$ critical pair, so are the sub-reduction systems $\redlmmtn$ and $\redlmmtv$, not the reduction strategies $\redlmmtN$ and $\redlmmtV$ we have defined here (see \Def \ref{CBN and CBV in lmmt}), but for $\lmu$.

Given that we are interpreting one calculus with well-defined {\CBN} and {\CBV} reduction strategies into another, we wanted to investigate if there could be a single interpretation that respects both {\CBN} and {\CBV} strategies for $\nclmu$ as well, rather than using the standard separation into a {\CBN} and a {\CBV} interpretations.
We start by investigating if Herbelin's encodings already serve this purpose.
We have already seen in \Rem \ref{no reduction} that this is not possible for $\SEMcbn{ `. }$, but we can also make the following observations:

 \begin{remark}
 \label{Herbelin not enough}
 \begin{itemize}
 \firstitem
 The interpretation $\SEMcbn{ `. }$ creates problems when interpreting the {$\slmu$} $(\mul)$-reduction
 \[ \begin{array}{rcl}
`@ N (`M`b . [`b] M) &\rednclmu& `M `g . [`g] `@ N (\mulsubst N.`g/`b M )
 \end{array} \]
We would like to show that this reduction step is preserved by $\SEMcbn{ `. }$, but:
 \[ \begin{array}{rcl}
 \SEMcbn{ `@ N (`M`b . [`b] M) }
	&\ByDef& \\
`M`a . < \SEMcbn{N} | { `M`b . <\SEMcbn{M}|`b>} `. ~ `a >
	&(?)&
`M`a . < \SEMcbn{N} | \SEMcbn{\mulsubst[N.`g/`b]M} `. ~ `a >
 \\
	&\lmmtderV & \bsp (`g)
 `M`g . < {`M `a . < \SEMcbn{N} | \SEMcbn{\mulsubst[N.`g/`b]M} `. ~ `a >} | `g > \\
	&\ByDef& 
\SEMcbn{ `M `g . [`g] `@ N (\mulsubst N.`g/`b M ) }
 \end{array} \]
but cannot: the term $ `M`a . < \SEMcbn{N} | { `M`b . <\SEMcbn{M}|`b>} `. ~ `a > $ is not a $(`m)$-redex (over $`b$).
Actually, $\lmmt$ lacks a reduction rule corresponding to $\lmu$'s $(\mul)$;
in \Sect\ref{add `m_l to lmmt} we will discuss adding this kind of rule to $\lmmt$.

 \item
The interpretation $\SEMcbn{ `. }$ also does not deal well with the contextual reduction rules for $\redlmmtV$.
For the rule $ M \redlmmtV N \Implies `@ V M \redlmmtV `@ V N $ we have:
 \[ \begin{array}{lclclclc}
\SEMcbn{ `@ V M } & \ByDef &
`M`a . < \SEMcbn{V} | \SEMcbn{M} `. `a > &\crlmmtV& \bsp (?)
`M`a . < \SEMcbn{V} | \SEMcbn{N} `. `a > & \ByDef &
\SEMcbn{ `@ V N }
 \end{array} \]
which asks for reduction in the environment.

 \Comment
 {
 We could add the reduction rule $ \cell < t_1 | { `M`b . <t_2|`b>} `. ~ `a > \reduc \cell< t_1 | t_2 `. `a > $.

 \[ \def \Turnlmmt {\Turn}
 \Inf [`m]
 {\Inf [\Cut]
 {\InfBox{ \derlmmt `G |- t_1 : B\arr A | `a`:A,`D }
 \kern-10mm
 \Inf [\arrL]
 {\Inf [\Cut]
 {\Inf [`m]
 {\InfBox{ \derlmmt `G |- t_2 : B | `b`:B,`a`:A,`D }
 \quad
 \Inf [\AxC]
 { \derlmmt `G | `b : B |- `b`:B,`a`:A,`D }
 }{ \derlmmt \cell<t_2|`b> : `G |- `b`:B,`a`:A,`D }
 }{ \derlmmt `G |- `M`b.\cell<t_2|`b> : B | `a`:A,`D }
 \kern-10mm
 \Inf [\AxL]
 { \derlmmt `G |- `a : A | `a`:A,`D }
 }{ \derlmmt `G | {`M`b . <t_2|`b> `. ~ `a } : B\arr A |- `a`:A,`D }
 }{ \derlmmt { \cell < t_1 | {`M`b . <t_2|`b> `. ~ `a } >} : `G |- `a`:A,`D }
 }{ \derlmmt `G |- {`M`a . < t_1 | { `M`b . <t_2|`b> `. ~ `a } >} : A | `D }
 \]
 \[ \def \Turnlmmt {\Turn}
 \Inf [`m]
 {\Inf [\Cut]
 {\InfBox{ \derlmmt `G |- t_1 : B\arr A | `a`:A,`g`:A,`D }
 \quad
 \Inf [\arrL]
 {\InfBox{ \derlmmt `G |- M : B | `a`:A,`g`:A,`D }
 \quad
 \Inf [\AxC]
 { \derlmmt `G |- `a : A | `a`:A,`g`:A,`D }
 }{ \derlmmt `G | t_2 `. `a : B\arr A |- `a`:A,`g`:A,`D }
 }{ \derlmmt {\cell < t_1 | t_2 `. `a>} : `G |- `a`:A,`g`:A,`D }
 }{ \derlmmt `G |- {`M`a . < t_1 | t_2 `. `a>} : A | `g`:A,`D }
 \]
 }

 \item
That $\SEMcbv{ `. }$ deals correctly with $(\mul)$ is illustrated by:
 \[ \begin{array}{rcl}
 \SEMcbv{ {`@ N (`m `b . [`b] M ) } }
 & \ByDef & \\
 `M`a . < {`M `b . < \SEMcbv{M} | `b >} | {\mt x . < \SEMcbv{N} | x `. `a >} >
 & \redlmmt & \bsp (`m)
 `M`a . < \SEMcbv{M}\tsubst[{\mt x . < \SEMcbv{N} | x `. `a >}/`b] | {\mt x . < \SEMcbv{N} | x `. `a >} >
 \\ & \rtclmmtder &
 `M`a . < \SEMcbv{\mulsubst [N.`a/`b] M } | {\mt x . < \SEMcbv{M} | x `. `a >} >
 \\ & \ByDef &
 \SEMcbv{ {`@ M (\mulsubst [N.`a/`b] M ) } } \\
 \end{array} \]
provided of course that we verify that $ \SEMcbv{M}\tsubst[{\mt x . < \SEMcbv{N} | x `. `a >}/`b] = \SEMcbv{ \mulsubst [N.`a/`b] M } $; we will do so in \Lmm \ref{interpretation respects lmu substitution} and in \Sect\ref{add `m_l to lmmt}.

 \item
There is a problem in showing $ M \rednclmuV N \Implies \SEMcbv{M} \crlmmtV \SEMcbv{N} $ when dealing with the contextual reduction rules.
 The first, $ M \reduc N \Implies `@ V M \reduc `@ V N $ now follows easily, since we have:
 \[ \begin{array}{lclclclc}
 \SEMcbv{ `@ V M } & \ByDef &
 `M`a . < \SEMcbv{M} | {\mt x . < \SEMcbv{V} | x `. `a >} > &\crlmmtV& \bsp (\IH)
 `M`a . < \SEMcbv{N} | {\mt x . < \SEMcbv{V} | x `. `a >} > & \ByDef &
 \SEMcbv{ `@ V N }
 \end{array} \]
benefitting from the swap between the terms, but for the second $ M \reduc N \Implies `@ M P \reduc `@ N P $ we now have:
 \[ \begin{array}{lclclclc}
 \SEMcbv{ `@ M P } & \ByDef &
 `M`a . < \SEMcbv{P} | {\mt x . < \SEMcbv{M} | x `. `a >} > &(?)&
 `M`a . < \SEMcbv{P} | {\mt x . < \SEMcbv{N} | x `. `a >} > & \ByDef &
 \SEMcbv{ `@ V N }
 \end{array} \]
for which we need to allow for reduction to take place inside a $\MuT$-term, so inside the environment.

 \item
When modelling $\CBN$ reduction under this interpretation, there is no need to reduce in the environment, since we can then contract the $\MuT$-redexes:
 \[ \begin{array}{lclclclc}
 \SEMcbv{ `@ M P } & \ByDef &
 `M`a . < \SEMcbv{P} | {\mt x . < \SEMcbv{M} | x `. `a >} > &\rtcredlmmtN& \bsp (x)
 `M`a . < \SEMcbv{M} | \SEMcbv{P} `. `a > &\crlmmtN& \bsp (\IH) \\ &&
 `M`a . < \SEMcbv{N} | \SEMcbv{P} `. `a > &\crlmmtN& \bsp (x)
 `M`a . < \SEMcbv{P} | {\mt x . < \SEMcbv{N} | x `. `a >} > & \ByDef &
 \SEMcbv{ `@ V N }
 \end{array} \]
 This is not allowed for $\redlmmtV$, since $\SEMcbv{P}$ need not be a value.

 \end{itemize}
 \end{remark}

 \begin{remark}
 We could argue that the encoding $\SEMcbv{ `. }$ actually represents a \CBV-reduction strategy variant on $\nclmu$ with the contextual rules:
 \[ \begin{array}{rcl}
 \ContV
 &::=&
 "[]" \mid `@ {\ContV} V \mid `@ M {\ContV} \mid `M`a . `b \ContV
 \end{array} \]
 which would force the evaluation of the parameter until it becomes a value, after which the term in function position gets reduced; this corresponds to a reduction like (where we assume that each $P_i$ runs to a value $V_i$):
 \[ \begin{array}{lclclclcl}
 `@ {`@ {`@ {`@ (`Lx.M) P_1 } P_2 \dots} P_{n-1} } P_n & \rtcredCBV &
 `@ {`@ {`@ {`@ (`Lx.M) P_1 } P_2 \dots} P_{n-1} } V_n & \rtcredCBV &
 `@ {`@ {`@ {`@ (`Lx.M) P_1 } P_2 \dots} V_{n-1} } V_n & \rtcredCBV & \\
 `@ {`@ {`@ {`@ (`Lx.M) P_1 } V_2 \dots} V_{n-1} } V_n & \redCBV &
 `@ {`@ {`@ {`@ (`Lx.M) V_1 } V_2 \dots} V_{n-1} } V_n & \redCBV &
 `@ {`@ {`@ {M\tsubst[V_1/x] } V_2 \dots} V_{n-1} } V_n &
 \end{array} \]
 which would perhaps be too great a deviation from a `normal' {\CBV} strategy.

 We would then have:
 \[ \begin{array}{lclclclc}
 \SEMcbv{ `@ P M } & \ByDef &
 `M`a . < \SEMcbv{M} | {\mt x . < \SEMcbv{P} | x `. `a >} > &\crlmmtV& \bsp (\IH)
 `M`a . < \SEMcbv{N} | {\mt x . < \SEMcbv{P} | x `. `a >} > & \ByDef &
 \SEMcbv{ `@ P N } \\ [2mm]
 \SEMcbv{ `@ M V } & \ByDef &
 `M`a . < \SEMcbv{V} | {\mt x . < \SEMcbn{M} | x `. `a >} > &\rtcredlmmtN& \bsp (x)
 `M`a . < \SEMcbv{M} | \SEMcbn{V} `. `a > &\crlmmtN& \bsp (\IH) \\ &&
 `M`a . < \SEMcbv{N} | \SEMcbn{V} `. `a > &\crlmmtN& \bsp (x)
 `M`a . < \SEMcbv{V} | {\mt x . < \SEMcbn{N} | x `. `a >} > & \ByDef &
 \SEMcbv{ `@ N V }
 \end{array} \]
 without the need to reduce inside the environment.%
 \footnote{This might be well suited to model reduction in the Call by Push Value calculus \cite{Levy'99}, where reduction inside parameters is not permitted.}

 \end{remark}

\Comment{
 Curien and Herbelin's result is:
 \begin {theorem}[Simulation of $`l`m$ in $\lmmt$ \cite{Curien-Herbelin'00}] \label{simul:lmlmmt}
 \begin {enumerate} \itemsep0pt
 \firstitem If $M\redCBV N $ then $\SEMcbv{M} \redCBV \SEMcbv{N} $ up to $`m$-expansion.
 \item If $M\redCBN N $ then $\SEMcbn{M} \redCBN \SEMcbn{N} $ up to $`m$-expansion.
 \end {enumerate}
 \end {theorem}
 This result is stated with respect to the notion of {\CBN} and \CBV-reduction for $\lmmt$ as defined in \cite{Curien-Herbelin'00}, that just remove the $`m,\MuT$ critical pair, not the strategies we have defined here.
}

\subsection{Non-Confluent Embedding}
We will now show that we can strengthen the results of \cite{Curien-Herbelin'00}, and show that we can define \emph{one} interpretation with which we can successfully represent all three notions of reduction and strategy.
We will essentially show that our interpretation can be used to represent not only the (traditional) $\lmu$ calculus, but also $\nclmu$, and not only for the {$\redlmmtV$} reduction strategy, but also $\redlmmtN$, as well as unrestricted reduction.

We will first define our interpretation.

 \begin{definition} [{Interpretation $ \lmulmmt[ `. ] $ of $\nclmu$ into $\lmmt$}]
 \label{our lmu to lmmt}
\label{def:TranslateLmuLmmt:Translation}
We define the interpretation $ \lmulmmt[ `. ] $ of $\nclmu$ into $\lmmt$ inductively through:
 \[ \begin{array}{rcl@{\quad}l}
\lmulmmt[x] & \ByDef & \lmulmmt[v x] 
 \\
\lmulmmt[`Lx.M] & \ByDef & \lmulmmt[l x . M] 
 \\
\lmulmmt[`@ M N ] & \ByDef & \lmulmmt[a M N] 
 \\
\lmulmmt[{`M`b.[`g]M}] & \ByDef & \lmulmmt[m `b . `g M]
 \end{array} \]

 \observe
As mentioned above, in the case for application we use $\UL{`m}$ for $`m$-abstractions  and $\UL{\mutilde}$ for $\mutilde$-abstractions generated by the interpretation, and $\redlmmtmo$ for reduction steps using these.
 
 \end{definition}
Notice that this interpretation also is a mapping from the $`l$-calculus to $\lmmt$.

It is straightforward to show that this interpretation respects assignable types:

 \begin{theorem}
 \label{thm:TranslateLmuLmmt:TranslationWellTyped}
 If $ \derLmu `G |- M : A | `D $, then $ \derLmmt `G |- \lmulmmt[M] : A | `D $.
 \end{theorem}
 \begin{Proof}{By induction on the definition of type assignment.}{\Thm \ref{thm:TranslateLmuLmmt:TranslationWellTyped}}
\begin{description}
 \item[$\Ax$]
Then $ M \same x $ and $ x`:A \ele `G $; since $ \lmulmmt[x] = x $, also $ \derLmmt `G |- x : A | `D $ by rule $(\Ax)$.

 \item[$\arrI$]
	Then $ M \same `L x . N $, $ A \same B \arr C $, and $ \derLmu `G,x`:B |- N : C | `D $.
	By induction, $ \derLmmt `G,x`:B |- \lmulmmt[N] : C | `D $, and by rule $(\arrI)$, $ \derLmmt `G |- `Lx.\lmulmmt[N] : A | `D $, and $`Lx.\lmulmmt[N] = \lmulmmt[`Lx.N] $.

 \item[$\arrE$]
	Then $ M \same `@ P Q $, and there exists $B$ such that $ \derLmu `G |- P : B\arr A | `D $ and $ \derLmu `G |- Q : B | `D $.
	Then by induction, $ \derLmmt `G |- \lmulmmt[P] : B\arr A | `D $ and $ \derLmmt `G |- \lmulmmt[Q] : B | `D $; we also have $ \derLmmt `G |- \lmulmmt[P] : B\arr A | `a`:A,`D $ and $ \derLmmt `G,x`:B\arr A |- \lmulmmt[Q] : B | `a`:A,`D $ by weakening, and we can construct (where $`G' = `G,x`:B\arr A,y`:B$):
 \[ \def \Turnlmmt {\Turn} 
\Inf	{\Inf	{\InfBox { \derLmmt `G |- \lmulmmt[P] : B\arr A | `a`:A,`D }
	\kern-12mm
		 \Inf	{\Inf	{\InfBox { \derLmmt `G,x`:B\arr A |- \lmulmmt[Q] : B | `a`:A,`D }
	\kern-10mm
				 \Inf	{\Inf	{\Inf	{ \derLmmt `G' |- x : B\arr A | `a`:A,`D }
						 \Inf	{\Inf	{ \derLmmt `G' |- y : B | `a`:A,`D }
							 \Inf	{ \derLmmt `G' | `a : A |- `a`:A,`D }
							}{ \derLmmt `G' | { y `. `a } : B\arr A |- `a`:A,`D }
						}{ \derLmmt {\cell < x | y `. `a > } : `G' |- `a`:A,`D }
					}{ \derLmmt `G,x`:B\arr A | {\mt y . < x | y `. `a > } : B |- `a`:A,`D }
				}{ \derLmmt {\cell < \lmulmmt[Q] | {\mt y . < x | y `. `a > } > } : `G,x`:B\arr A |- `a`:A,`D }
			}{ \derLmmt `G | {\mt x . < \lmulmmt[Q] | {\mt y . < x | y `. `a > } > } : B\arr A |- `a`:A,`D }
		}{ \derLmmt { \cell < \lmulmmt[P] | {\mt x . < \lmulmmt[Q] | {\mt y . < x | y `. `a > } > } > } : `G |- `a`:A,`D }
	}{ \derLmmt `G |- \lmulmmt[a P Q] : A | `D }
 \]
and $ \lmulmmt[a P Q] = \lmulmmt[{`@ P Q }] $

 \item[$`m$]
	We have two cases: $ M \same `M`a.[`b]N $, $`D = `b`:B,`D'$, and $ \derlmu `G |- N : B | `a`:A,`b`:B,`D' $; then by induction we have $ \derLmmt `G |- \lmulmmt[N] : B | `a`:A,`b`:B,`D' $.
	We can construct:
 \[ \def \Turnlmmt {\Turn}
\Inf	{\Inf	{\InfBox { \derLmmt `G |- \lmulmmt[N] : B | `a`:A,`b`:B,`D' } ~
		 \Inf	{ \derLmmt `G | `b : B |- `a`:A,`b`:B,`D' }
		}{ \derLmmt {\cell <\lmulmmt[N] | `b > } : `G |- `a`:A,`b`:B,`D' }
	}{ \derLmmt `G |- \lmulmmt[m `a . `b N] : A | `b`:B,`D' }
 \]
and $ \lmulmmt[m `a . `b N] = \lmulmmt[`M `a . `b N] $.

	Or $ M \same `M`a.[`a]N $ and $ \derlmu `G |- N : A | `a`:A,`D $; then by induction we have $ \derLmmt `G |- \lmulmmt[N] : A | `a`:A,`D $.
	We can construct:
 \[ \def \Turnlmmt {\Turn}
\Inf	{\Inf	{\InfBox{ \derLmmt `G |- \lmulmmt[N] : A | `a`:A,`D } ~
		 \Inf	{ \derLmmt `G | `a : A |- `a`:A,`D }
		}{ \derLmmt {\cell <\lmulmmt[N] | `a > } : `G |- `a`:A,`D }
	}{\derLmmt `G |- \lmulmmt[m `a . `a N] : A | `D }
 \]
and $ \lmulmmt[m `a . `a N] = \lmulmmt[`M `a . `a N] $.
\QED

 \end{description}
 \end{Proof}

Remark that we have:{
 \[ \begin{array}{c@{~}c@{~}c@{~}c@{~}c@{~}c@{~}c}
\lmulmmt[`@ P Q ] & \ByDef &
\lmulmmt[a P Q] & \reduc & 
\lmulmmt[b P Q] & \reduc &
\Mo`a . < `P | `Q `. `a > \\
\SEMcbv {`@ P Q } & \ByDef &&&
`M`a . < \SEMcbv{Q} | {\mt y . < \SEMcbv{P} | y`.`a >} > & \reduc & 
`M`a . < \SEMcbv{P} | \SEMcbv{Q} `. `a > \\
\SEMcbn {`@ P Q } &\ByDef & && &&
`M`a . < \SEMcbn{P} | \SEMcbn{Q} `. `a >
 \end{array} \] }
a relation that will be useful when doing the proofs.

We will show that our encoding respects the three notions of reduction by showing, in \Thm \ref{beta mu reduction preserved}, \ref{CBN reduction preserved}, and \ref{CBV reduction preserved} that the encoding respects reduction (through equality):
 \begin{eqnarray}
 M \rtcrednclmu N &\Implies& \SemM \crlmmt `N 
 	\label{full reduction} \\
 M \rtcrednclmuN N &\Implies& \Exists {t,t' \ele \lmmt} ( `N \rtcredlmmtN t ~\&~ t \redlmmtn t' ~\&~ \SemM \rtcredlmmtN t' {}^{1} )
 	\label{CBN reduction} \\
 M \rtcrednclmuV N &\Implies& \Exists {t,t' \ele \lmmt} ( `N \rtcredlmmtV t ~\&~ t \redlmmtv t' ~\&~ \SemM \rtcredlmmtV t'  {}^{2} )
 	\label{CBV reduction} 
 \end{eqnarray}
 \begin{description}
 \item [$1$] in $\rtcredlmmtn$, only ($\redlmmtn$) redexes are contracted to complete the substitutions.
 \item [$2$] in $\rtcredlmmtv$, only ($\redlmmtv$) redexes are contracted to complete the substitutions.
 \end{description}
We will see in the proofs of \Lmm \ref{interpretation respects mur substitution CBN} and \Thm \ref{interpretation respects mur reduction CBN} that in order to model $(\mur)$-reduction
, some reverse reduction steps are needed as well, as can be expected in view of \Rem \ref{added application}.

For the basic steps in reduction this will be shown through \Thm \ref{interpretation respects beta reduction} that shows that the encoding respects the $(`b)$-reduction rule:
 \[ \begin{array}{rcl}
\lmulmmt[ `@ (`Lz.M) N ] &\rtcredlmmt& \lmulmmt[M\tsubst N/z ]
 \end{array}
 \]

\Thm \ref{interpretation respects mur reduction} shows it respects $(\mur)$-reduction:
 \[ \begin{array}{rcl}
\lmulmmt[{ `@ (`M`d.[`b]M) N }] & \eqlmmt& \lmulmmt[{`M`g .[`b]M \mursubst[N.`g/`d] }] , \dquad (`b \not=`d) \\
\lmulmmt[{ `@ (`M`d.[`d]M) N }] &\eqlmmt& \lmulmmt[{ `M`g .[`g]`@ (M \mursubst[N.`g/`d]) N }]
 \end{array}
 \]
and \Thm \ref{interpretation respects lmu reduction} shows it respects $(\mul)$-reduction:
 \[ \begin{array}{rcl}
\lmulmmt[ `@ N (`M`a. `d M) ] &\eqlmmt& \lmulmmt[ `M`g . `d {\mulsubst N.`g/`a } M ] , \dquad (`a \not= `d) \\
\lmulmmt[ `@ N (`M`a. `a M) ] &\eqlmmt& \lmulmmt[ `M`g . `g {`@ N (\mulsubst N.`g/`a M )}]
 \end{array}
 \]

We will, for each of these three results, argue that all the contractions that take place in the proofs would still be allowed when restricting to the {$\redlmmtV$} and $\redlmmtN$ reduction strategies, so prove the three results $(\ref{full reduction})$, $(\ref{CBV reduction})$, and $(\ref{CBN reduction})$ simultaneously.
The only exception to this is \Thm \ref{interpretation respects mur reduction}, the proof of which is not sound for $\redlmmtN$; \Thm \ref{interpretation respects mur reduction CBN} will show that result for $\redlmmtN$.
Note that, for $\redlmmtV$, we need to check that:%
 \begin{flatenumerate}
 \item
 only values are substituted through $(\MuT)$-reduction steps;
 \item
 no reduction takes place in environments,
 \end{flatenumerate}%
and for $\redlmmtN$ that:%
 \begin{flatenumerate}
 \item
 only stacks are substituted through $(\mur)$-reduction steps;
 \item
 no reduction takes place in environments.
 \end{flatenumerate}

\Comment{
 \[
 \begin{array}{lclcl}
 \lmulmmt[`@ N ({`M`a . [`b] \Cont[{`M`g . [`a] M }]})] & \ByDef & \\
 \lmulmmt[a <`d,v> N {m `a . `b {\LCont{`M`g . [`a] M }\RCont }}] & \ByDef & \\
 `M`d . < `N | {\mt u . <{`M`a . <{\LCont{ \lmulmmt[m `g . `a M]}\RCont } | `b >} | {\mt v . <u|v `. `d>} >} >
 \\ \\
 `M`d . <{\LCont \lmulmmt[m `g . `d {a N M }] \RCont } | `b > & \ByDef & \\
 \lmulmmt[m `d . `b {\LCont {`M`g . [`d] `@ N M } \RCont }] & \ByDef & \\
 \lmulmmt[{`M`d . [`b] \Cont[{`M`g . [`d] `@ N M }]}]
 \end{array}
 \]

 We will show that our interpretation deals with \emph{both} $\redlmmtN$ and $\redlmmtV$, by showing first that it encodes full non-confluent $\ncredbmu$ reduction, and then show that \CBN-$\lmu$ is captured in \CBN-$\lmmt$, and \CBV-$\lmu$ in \CBV-$\lmmt$.
 To illustrate this result for the $\mul$-reduction step, we give a motivating example.
 In Example\ref{leftsub proof example} we will show how the derivations are transformed.

 \begin{example}
We have $ `@ R ({`@ N (`M`b . `b `@ P Q )}) \rednclmu `@ R (`M`d . `d {`@ N (`@ P Q )}) \rednclmu `M`g . `g `@ R ({`@ N (`@ P Q )}) $ (assuming $`b$ does not occur in $P$ or $Q$) so need to relate $ \lmulmmt[ `@ R ({`@ N (`M`b . `b `@ P Q )})] $ and $ \lmulmmt[`M`g . `g {`@ R ({`@ N (`@ P Q )})}] $:
 \[ \begin{array}{lclcl}
 \lmulmmt [`@ R ({`@ N ({`M`b . `b (`@ P Q )})})] &\ByDef& \\
 \lmulmmt [a R {a <`d> N {m `b . `b {`@ P Q }}}] &\redlmmt& \bsp (`b) \\
 `M`a . < {`M`d . <{\lmulmmt[`@ P Q ]} | {\mt u . <`N|u `. `d>} >} | {\mt y . <`R|y`.`a>} > &\redlmmt& \bsp (u) \\
 `M`a . < {`M`d . <`N|{\lmulmmt[`@ P Q ]} `. `d>} | {\mt y . <`R|y`.`a>} > &\redlmmt& \bsp (u) \\
 \\
 \lmulmmt[m `g . `g {a <`d> R {a N {`@ P Q }}}] &\ByDef& \\
 \lmulmmt[m `g . `g {`@ R ({`@ N (`@ P Q )})}] &\ByDef& \\
 \lmulmmt[`M`g . `g ({`@ R ({`@ N (`@ P Q )})})]
 \end{array}
 \]

 \end{example}
}

 \subsection{The interpretation respects general reduction}
We will now show that the interpretation respects reduction, principally by showing that it respects the rules $(`b)$, $(\mur)$, and $(\mul)$.

 \subsubsection*{$`b$ reduction}
We start by showing that the interpretation respects $`b$-reduction, for which we first need to show it respects term substitution.

 \begin{lemma} [{$\lmulmmt[`.]$ respects $\slmu$'s term substitution}]
 \label{tsubst lemma} \label{interpretation respects term substitution}
 $ \SemM \tsubst [`N/z] = \lmulmmt[{M \tsubst N/z }] $.
 \end{lemma}

 \begin{Proof}{By induction on the structure of terms.}{\Lmm \ref{tsubst lemma}}
By induction on the structure of terms.
\addqsymbol{`N}{\lmmtSem [N]\!}
 \begin{description}
 \myitem[$M = z$]
\lmulmmt[z] \tsubst [`N/z] &\ByDef& 
\lmulmmt[v z] \tsubst [`N/z] &\ByDef& 
`N \ByDef \lmulmmt[{z \tsubst N/z }]
 \end{array} $

 \myitem[$M = y$, $y\not=z$]
\lmulmmt[y] \tsubst [`N/z] &\ByDef& 
\lmulmmt[v y] \tsubst [`N/z] &\ByDef& 
y \ByDef
\lmulmmt[y] \ByDef
\lmulmmt[{y \tsubst N/z }]
 \end{array} $

 \myitem[$M = `Ly.P$]
\lmulmmt[`Ly.P] \tsubst [`N/z] &\ByDef& 
\lmulmmt[l y . P] \tsubst [`N/z] = \! (\IH) ~
\lmulmmt[l y . {P\tsubst N/z }] &\ByDef & 
\lmulmmt[`L y . {P\tsubst N/z }]
 \end{array} $

 \myitem[$M = PQ$]
\lmulmmt[`@ P Q ] \tsubst [`N/z] ~\ByDef~ 
\lmulmmt[a P Q] \tsubst [`N/z] ~\ByDef~ \\
`M`a.<`P \tsubst [`N/z] | {\mt x . < `Q \tsubst [`N/z] | { \mt y.< x | y`.`a >} >} > 
=\! (\IH) ~\\
\lmulmmt[a {P \tsubst N/z } {Q \tsubst N/z }] 
 ~\ByDef~
\lmulmmt[{(P \tsubst [N/z]) (Q \tsubst [N/z])}] ~\ByDef~ 
\lmulmmt[{(PQ) \tsubst N/z }]
 \end{array} $

 \myitem[{$M = `M`b.[`g]P$}]
\lmulmmt[`M`b. `g P] \tsubst [`N/z] &\ByDef& 
\lmulmmt[m `b . `g P] \tsubst [`N/z] &\ByDef& 
`M`b.< `P \tsubst [`N/z] | `g > &=& \bsp (\IH) \\
\lmulmmt[m `b . `g {P \tsubst N/z }] &\ByDef& 
\lmulmmt[{`M`b.[`g]P\tsubst N/z }] &\ByDef& 
\lmulmmt[{(`M`b.[`g]P)\tsubst N/z }]
 \end{array} $

 \arrayqed[2\point]
 \end{description}
 \end{Proof}%
Reduction does not play a role in this result.

 \begin{theorem} [{$\lmulmmt[`.]$ respects $\slmu$'s $`b$-reduction}]
 \label {interpretation respects beta reduction}
$\lmulmmt[ `@ (`Lz.M) N ] \tcredlmmt \lmulmmt[M\tsubst N/z ] $.
 \end{theorem}

\setbox65=\hbox{\textit{\proofname:}\,}
 \begin{proof}
$ \kern-7mm \begin{array}[t]{lclclcl} \kern7mm
\lmulmmt[`@ (`Lz.M) N ] &\ByDef&
\Span{3}{
\lmulmmt[a {l z . M} N ] ~ \redlmmt \! (x) } 
	\\
\lmulmmt[b {l z . M} N ] & \redlmmt& \bsp (y) 
\Mo`a.< \lmulmmt[l x . M] | `N `. `a > & \redlmmt& \bsp (`l) 
\Mo`a.< `N | {\mt z.< \SemM | `a >} > 
	&\redlmmt& \bsp (z) \\
\Mo`a.< \SemM \tsubst [`N/z ] | `a > &=& \bsp (\ref{tsubst lemma}) 
\Mo`a.< \lmulmmt[M \tsubst N/z ] | `a > & \redlmmt& \bsp (`h`m) 
\lmulmmt[{M\tsubst N/z }]
 \end{array} $
\arrayQED
 \end{proof}
%
All reduction steps are good for $\redlmmtN$ and also for $\redlmmtV$, since then $N$ is a value, and so is $`N$; notice that $`lz.M$ and $`lz.\SemM$ are values as well.
Notice that all reductions are in $\redlmmtmo$, except for the $(`l)$ step.

 \subsubsection*{Right-\TeXorPDFstring{$`m$}{μ} Reduction}

The preservation of right-structural reduction is not as straightforward.
It is not the case that, for example,
 \[ \lmulmmt[{ `@ (`M`d.[`d]M) N }] \tcredlmmt \lmulmmt[{ `M`g .[`g](M \mursubst[N.`g/`d])N }] . \]
A proof of such would require that $\SemM{\cdot}$ commutes with right-structural substitution,
 \[ \SemM \mursubst [`N.`g/`d] \tcredlmmt \lmulmmt[M \mursubst N.`g/`d ], \]
an analogue of \Lmm \ref{interpretation respects term substitution}.
In fact, it is the \emph{converse} that holds,
 \[ \lmulmmt[M \mursubst N.`g/`d ] \tcredlmmt \SemM \mursubst [`N.`g/`d] . \]
In the corresponding proof, the case that necessitates a $\redlmmt$-reduction is when $M$ has the form $`M`b.[`d]P$:
 \[ \begin{array}{rcl}
\lmulmmt[{(`M`b . `d P ) \mursubst N.`g/`d }]
	&\ByDef& 
\lmulmmt[`M`b . `g `@ (P \mursubst N.`g/`d ) N ]
	\\ &\ByDef& 
\lmulmmt[m `b . `g {a {P \mursubst N.`g/`d } N }] 
	\\ &\redlmmtmo&\bsp (`a)
`M`b . < \lmulmmt[P \mursubst N.`g/`d ] | \Ng >
	\\ &\tcredlmmtmo& \bsp (\IH)
`M`b . < `P \tsubst [\Ng/`d] | \Ng > 
	\\ &=&
\lmulmmt[`M`b. `d P] \tsubst [\Ng/`d].
 \end{array} \]
The $\UL{`m}$-reduction step is of the shape $ \cell <\Mo `a.c | `g> \redlmmt c \tsubst[`g/`a] $, which we call a \emph{renaming cut}.

The issue is (again, as mentioned in \Rem \ref{added application}) that the structural substitution $ (`M`b.[`d]P) \mursubst[N.`g/`d] = `M`b . [`g] `@ (P \mursubst[N.`g/`d] ) N $ adds an application, and creates an extra $\UL{`m}$-redex in the translation;
$\lmulmmt[(P \mursubst N.`g/`d ) N]$ must generate a $\UL{`m}$-binding to encode application, but we know that this will immediately bind to the context to which $\gamma$ denotes.
This is precisely why \citet[Proposition 2.3]{Curien-Herbelin'00} state that their translations are `up to $`m$-expansion'.

The proofs for the preservation of right-structural reduction come in two parts.
First, in \Lmm \ref{interpretation respects mur substitution CBV} we show that right structural substitution is preserved modulo $\UL{`m}$-expansion.
As we will point out, the proof of \Thm \ref{interpretation respects mur reduction} will make $(\MuT)$-reduction steps that are not allowed in $\redlmmtn$.
However, in \Thm \ref{interpretation respects mur reduction CBN}, we will show that the interpretation preserves $\redlmmtN$.

To show that the interpretation respects full $(\mur)$-reduction, we first show a result for right substitution.
To improve readability, we will write $\Ng$ for $  \mto x . < `N | { \mto y . < x | y`.`g>} > $ when the exact structure of that term is not important; notice that then $\lmulmmt[`@ P Q ] = \lmulmmt[A P Q] $.

 \begin{lemma} [{$\lmulmmt[`.]$} respects $\slmu$'s right structural substitution up to $\UL{`m}$ expansion] \label{interpretation respects mur substitution CBV} 
 \[ \lmulmmt[{}M \mursubst N.`g/`d ] \rtcredlmmtmo \SemM \tsubst [\Ngf/`d] = \SemM \tsubst [\Ng/`d] \]
 \end{lemma}

 \begin{proof}
By induction on the structure of terms:
 \begin{description}

 \myitem [$M = z$]
\lmulmmt[{z \mursubst N.`g/`d }] =
\lmulmmt[z] \ByDef
\lmulmmt[v z] =
\lmulmmt[v z] \tsubst [\Ng/`d] &\ByDef& 
\lmulmmt[z] \tsubst [\Ng/`d]
 \end{array} $

 \myitem[$M = `Lz.P$]
\lmulmmt[{(`Lz.P) \mursubst N.`g/`d }] &=& 
\lmulmmt[{`Lz.P \mursubst N.`g/`d }] &\ByDef& 
\lmulmmt[l z . {P\mursubst N.`g/`d }] &\rtcredlmmtmo& \bsp (\IH) \\
`L z . \lmulmmt[P] \tsubst [\Ng/`d] &=& 
(\lmulmmt[l z . P]) \tsubst [\Ng/`d] &\ByDef& 
\lmulmmt[`Lz.P] \tsubst [\Ng/`d]
 \end{array} $

 \myitem[$M = PQ$] 
\lmulmmt[{(`@ P Q ) \mursubst N.`g/`d }] 
	~=~ 
\lmulmmt[{`@ (P \mursubst N.`g/`d ) (Q \mursubst N.`g/`d ) }] 
	& \ByDef \\
\Mo `a . < \lmulmmt[{P\mursubst[N.`g/`d]}] | {\mto x . < \lmulmmt[{Q\mursubst[N.`g/`d]}] | {\mto y. <x|y`.`a>} >} > 
	&\rtcredlmmtmo& (\IH) \\
\Mo `a . < `P \tsubst [\Ng/`d] | {\mto x . < `Q\tsubst[\Ng/`d] | {\mto y. <x|y`.`a>} >} > 
	&=& \\
\lmulmmt[a P Q] \tsubst [\Ng/`d] 
	&\ByDef& 
\lmulmmt[`@ P Q ] \tsubst [\Ng/`d] 
 \end{array} $

 \myitem[{$M = `M`b.[`d]P$}]{lclcl}
\lmulmmt[{(`M`b . `d P ) \mursubst N.`g/`d }] 
	&\ByDef& 
\lmulmmt[`M`b . `g `@ (P \mursubst N.`g/`d ) N ] 
	&\ByDef& \\
\lmulmmt[m `b . `g {A {P \mursubst N.`g/`d } N}] 
	&\redlmmtmo& \bsp (`a) 
`M`b . < \lmulmmt[P \mursubst N.`g/`d ] | \Ng > 
	&\rtcredlmmtmo& \bsp (\IH) \\
`M`b . < `P \tsubst [\Ng/`d] | \Ng > 
	&=& 
\lmulmmt[m `b . `d P] \tsubst [\Ng/`d] 
	&\ByDef& 
\lmulmmt[`M`b. `d P] \tsubst [\Ng/`d]
 \end{array} $ 

 \myitem[{$M = `M`b.[`t]P$, $`t\not=`d$}]
&& \kern-10mm
\lmulmmt[{(`M`b.[`t]P)\mursubst N.`g/`d }] &=& 
\lmulmmt[{`M`b.[`t]P\mursubst N.`g/`d }] &\ByDef& \\
\lmulmmt[m `b . `t {P \mursubst N.`g/`d }] &\rtcredlmmtmo& \bsp (\IH) 
`M`b . < `P \tsubst [\Ng/`d] | `t > &=& 
\lmulmmt[m `b . `t P] \tsubst [\Ng/`d] &\ByDef& \\
\lmulmmt[`M`b. `t P] \tsubst [\Ng/`d]
 \end{array} $
\arrayqed

 \end{description}

\observe
The only contraction takes place for $ M = `M`b.[`d]P $, where a $(\UL{`m})$-step takes place (over $`a$), a renaming step that replaces $`a$ by $`g$, a step permitted in both $\redlmmtn$ and $\redlmmtv$.

 \end{proof}

So modelling right substitution requires some renaming expansions, but only of $\UL{`m}$-redexes generated through the interpretation.
This last lemma is stated without considering {\CBN} or \CBV; and in fact, since it allows induction to be applied under an abstraction and, in the application case, also on the term occurring in the context, this proof is not valid in $\redlmmtN$ or $\redlmmtV$ (but is in both $\redlmmtn$ and $\redlmmtv$).
Substitution is a free operation, in the sense that it is not restricted in either {\CBN} or {\CBV} reduction strategies, but by this result its implementation would not be free, since involving reduction everywhere.

With this result we can show that our encoding deals with $\mur$-reduction through equality, by showing they have a common reduct.


 \begin{theorem} [{$\lmulmmt[`.]$} respects $\mur$-reduction] 
 \label {interpretation respects mur reduction}

 \begin{enumerate}
 \firstitem $ \lmulmmt[{ `@ (`M`d.[`d]M) N }] \crlmmt \lmulmmt[{ `M`g .[`g]`@ M\mursubst[N.`g/`d] N }] $.
 \item $ \lmulmmt[{ `@ (`M`d.[`b]M) N }] \crlmmt \lmulmmt[{`M`g .[`b] (M \mursubst[N.`g/`d]})] $, $`b \not=`d$.
 \end{enumerate}
 \end{theorem}


 \begin{proof} 
 \begin{enumerate}
\prooffirstitem
\lmulmmt[{ `@ (`M`d.[`d]M) N }] 
	&\ByDef& 
\lmulmmt[A <`g> {m `d . `d M} N ] 
	&\redlmmt& \bsp (`d)	\\
\Mo `g . < \SemM \tsubst [\Ng/`d] | \Ng >   
	&\rtcredlmmtmo& \bsp (\ref{interpretation respects mur substitution CBV}) 
\Mo `g . < \lmulmmt[M\mursubst N.`g/`d ] | \Ng  >
	& \lmmtdermo & \bsp (`a) \\
\lmulmmt[m `g . `g {A {M\mursubst N.`g/`d } N }]
	&\ByDef& 
\lmulmmt[{`M`g .[`g]M\mursubst N.`g/`d \hsk N }]
 \end{array} $

 \item 
{\def \qsk{} \def \hsk{}
$ \begin{array}[t]{lclclcl}
\lmulmmt[{ `@ (`M`d.[`b]M) N }] 
	&\ByDef& 
\lmulmmt[A <`g> {m `d . `b M} N ] 
	&\redlmmt& \bsp (`d) 
 `M`g . < \SemM \tsubst [`N`.`g/`d] | `b >   
	&\rtcredlmmtmo& \bsp (\ref{interpretation respects mur substitution CBV}) \\
`M`g . < \lmulmmt[M\mursubst N.`g/`d ] | `b   >
	&\ByDef& 
\lmulmmt[m `g . `g {A {M\mursubst N.`g/`d } N }]
	&\ByDef& 
\lmulmmt[{`M`g .[`b]M \mursubst[N.`g/`d] }]
 \end{array} $
\arrayqed
}

 \end{enumerate}

\observe
Notice that the $(`m)$-contractions over $`d$ pull in $\Ng = \mt x . < `N | { \mt y . < x | y`.`g>} > $, which is not a stack, so these steps would not be allowed in $\redlmmtN$; all steps are allowed in \CBV, apart from those in $\redlmmtmo$, which are in $\redlmmtv$.
Also, because the result is shown with some steps going against reduction, it is stated using $\crlmmt$.

 \end{proof}
 
 \subsubsection*{Left-\TeXorPDFstring{$`m$}{μ} Reduction}

We will now turn our focus on $\nclmu$'s left-structural substitution. The following lemma shows that we can implement that in $\lmmt$ without extending the reduction relation.
In the following two results, we will now write $\Ng$ for $ \mt y . < `N | y `. `g> $. 

 \begin{lemma} [{$\lmulmmt[`.]$ respects $\slmu$'s left structural substitution}] 
 \label{interpretation respects lmu substitution} 
 \[ \lmulmmt[\mulsubst N.`g/`d M] ~\rtcredlmmt~ \SemM \tsubst [{\mt y . < `N | y`.`g >} / `d] \ByDef \SemM \tsubst [\Ng/`d] \]
 \end{lemma}

 \begin{proof}
By induction on the structure of terms\Shorter{; as in the proof for \Lmm\ref{interpretation respects mur substitution CBV}, the proof is straightforward, except for (the other steps need no reduction outside the inductive hypothesis):}\Longer{.}

 \begin{description}

\Longer
{
 \myitem [$M = x$]
\lmulmmt[\mulsubst N.`g/`d x] = 
\lmulmmt[x] &\ByDef& 
\lmulmmt[v x] &=& 
\lmulmmt [v x] \tsubst [\Ng / `d] &\ByDef&
\lmulmmt [x] \tsubst [\Ng / `d]
 \end{array} $

 \myitem [$M = `Lx.P$]
\lmulmmt[\mulsubst N.`g/`d (`Lx.P)] &=& 
\lmulmmt[`L x . \mulsubst N.`g/`d P] &\ByDef& 
\lmulmmt[l x . \mulsubst N.`g/`d P] &\rtcredlmmt& \bsp (\IH) \\
\lmulmmt [l x . P] \tsubst [\Ng / `d] &=&
(`Lx.\lmulmmt [P]) \tsubst [\Ng / `d] &\ByDef&
\lmulmmt [`Lx.P] \tsubst [\Ng / `d]
 \end{array} $

 \myitem [$M = `@ P Q $]
\lmulmmt[\mulsubst N.`g/`d (`@ P Q ) ] ~=~ 
\lmulmmt[`@ (\mulsubst N.`g/`d P) (\mulsubst N.`g/`d Q)]
&\ByDef& \\
\lmulmmt[a {\mulsubst N.`g/`d P} {\mulsubst N.`g/`d Q}] &\rtcredlmmt& \bsp (\IH) \\
`M`a . < `P \tsubst [\Ng / `d] | {\mt x . < `Q \tsubst [\Ng / `d] | {\mt y . < x | y`.`a >} >} > &=& \\
\lmulmmt [a P Q] \tsubst [\Ng / `d] &\ByDef& 
\lmulmmt [`@ P Q ] \tsubst [\Ng / `d]
 \end{array} $

 \myitem [$M = `M`a . `b P $, $`d \not=`b$] 
\lmulmmt[\mulsubst N.`g/`d (`M`a . `b P )] ~=~ 
\lmulmmt[`M`a . `b {\mulsubst N.`g/`d P} ] ~\ByDef~ 
\lmulmmt[m `a . `b {\mulsubst N.`g/`d P}] &\rtcredlmmt& \bsp (\IH) \\
`M`a . <`P \tsubst [\Ng / `d] | `b > ~ =~ 
\lmulmmt [m `a . `b P] \tsubst [\Ng / `d] ~\ByDef~ 
\lmulmmt [`M`a . `b P] \tsubst [\Ng / `d] 
 \end{array} $
}

 \myitem [$M = `M`a . `d P $]
\lmulmmt[\mulsubst N.`g/`d (`M`a . `d P )] 
	&\ByDef& 
\lmulmmt[`M`a . `g `@ N (\mulsubst N.`g/`d P) ] 
	&\ByDef& \\
\Span{3}
{ \begin{array}{@{}lclcl}
\lmulmmt[m `a . `g {a <`b,x> N {\mulsubst N.`g/`d P}}]  
	&\redlmmt& \bsp (x) \\
\Mo `a . <{ `M `b . < \lmulmmt [\mulsubst N.`g/`d P ] | {\mt y . <`N|y`.`b>} > }| `g > 
	&\redlmmtmo& \bsp (`b) 
 \end{array} } \\
`M`a . < \lmulmmt[\mulsubst N.`g/`d P] | \Ng > 
	&\rtcredlmmt& \bsp (\IH) 
`M`a . < `P \tsubst [\Ng / `d] | \Ng > 
	&=&  \\
\lmulmmt[m `a . `d P] \tsubst [\Ng / `d] 
	&\ByDef& 
\lmulmmt [`M`a . `d P] \tsubst [\Ng / `d]
 \end{array} $
\arrayQED

 \end{description}

 \end{proof}


Also here, by induction reduction takes place under an abstraction and, in the case for application, in the continuation.
As can be seen from the proof, the reduction required is more involved than just the renaming $\UL{`m}$ reduction (over $`b$); also the $(\MuT)$-step (over $x$) pulls in $`N$, which would be a value in $\redlmmtv$. 
This is needed to be able to write the resulting term as a substitution (to $`d$).
We do not need to check that $\redlmmtN$ is respected, since left structural substitution is not part of $\redlmmtN$.

With this result, we can now show that $\slmu$'s reduction rule $(\mul)$ can be simulated in $\lmmt$.

 \begin{theorem} [{$\lmulmmt[`.]$ respects $\slmu$'s $\mul$}] 
 \label {interpretation respects lmu reduction}
 \begin{enumerate}
 \firstitem
$ \lmulmmt[ `@ N (`M`a. `a M) ] ~\crlmmt~ \lmulmmt[ `M`g . `g {`@ N (\mulsubst N.`g/`a M )}] $.
 \item
$ \lmulmmt[ `@ N (`M`a. `d M) ] ~\crlmmt~ \lmulmmt[ `M`g . `d {\mulsubst N.`g/`a } M ] $, with $`a \not= `d$.
 \end{enumerate}
 \end{theorem}

 \begin{proof}
 \begin{enumerate}

\prooffirstitem
\lmulmmt[ `@ N (`M`a . `a M) ] 
	&\ByDef& 
\lmulmmt[a <`g,x> N {m `a . `a M}] 
	&\redlmmtmo & \bsp (x) \\ 
\lmulmmt[b <`g,y> N {m `a . `a M}] 
	&\redlmmt & \bsp (`a) 
\Mo `g . < \SemM \tsubst [\Ng / `a ] | \Ng > 
	&\rtclmmtdermo& \bsp (\ref{interpretation respects lmu substitution}) \\
\Mo `g . < \lmulmmt[{\mulsubst[N.`g/`a]M}] | \Ng > 
	&\lmmtdermo& \bsp (x) 
\Mo `g . < `N | {\mto x . < \lmulmmt[{\mulsubst[N.`g/`a]M}] | {\mto y . < x | y`.`g >} >} > 
	&\lmmtdermo& \bsp (`d) \\
\Span{3}{ \lmulmmt[m `g . `g {a <`d,x> N {\mulsubst N.`g/`a M}}] 
	~\ByDef~ 
\lmulmmt[`M`g . `g `@ N {(\mulsubst N.`g/`a M)} ] }
 \end{array} $

 \item
$ \begin{array}[t]{@{}lclrl}
\lmulmmt[ `@ N (`M`a . `d M) ] &\ByDef& 
\lmulmmt[a <`g,x> N {m `a . `d M}] &\redlmmtmo& \bsp (x) \\
\lmulmmt[b <`g,y> N {m `a . `d M}] &\redlmmtmo& \bsp (`a) 
\Mo `g . < \SemM \tsubst [\Ng / `a ] | `d > &\rtclmmtdermo& \bsp (\ref{interpretation respects lmu substitution}) \\
\lmulmmt[m `g . `d {\mulsubst N.`g/`a M}] &\ByDef& 
\lmulmmt[`M`g . `d {\mulsubst N.`g/`a M}]
 \end{array} $
\arrayQED

 \end{enumerate}

\observe
The $(\MuT)$-reductions on $x$ are permitted in $\redlmmtV$, since then $N$ and $`N$ will be values.
The reductions over $`a$ and $`d$ are also OK for $\redlmmtV$.
The reduction steps in $\rtclmmtdermo$ (\Lmm \ref{interpretation respects lmu substitution}) are in $\redlmmtv$.

 \end{proof}
Notice that all steps in this proof are in $\redlmmtmo$, except the $`m$-contractions over $`a$.

We can now state the main results for our encoding.

 \begin{theorem}[Preservation of $\slmu$ reduction] \label{beta mu reduction preserved}
If $M \rtcrednclmu N$, then $\SemM \crlmmt `N$.
 \end{theorem}
 \begin{proof} \def \rednclmu {\reduc}
 \begin{description}

 \item[{$ `@ (`Lz.M) N \rednclmu M\tsubst N/z $}]
By \Thm \ref{interpretation respects beta reduction}.

 \item[{$ `@ (`m `a . `C) N \rednclmu `m `g . `C \mursubst N.`g/`a $}]
By \Thm \ref{interpretation respects mur reduction}.

 \item[{$ `@ N (`m `a . `C) \rednclmu `M `g . \mulsubst N.`g/`a `C $}]
By \Thm \ref{interpretation respects lmu reduction}.

 \myitem[{$`M`a . [`b] { `M`g .[`d] M } \reduc `M`a . [`d] M \tsubst[`b/`g], ~ `g \not= `d $}]
\lmulmmt[`M`a . `b { `M`g . `d M }] &\ByDef & 
\lmulmmt[m `a . `b {m `g . `d M }] &\redlmmt & \\
`M `a . < \SemM \tsubst[`b/`g] | `d > ~ \ByDef ~
\lmulmmt[m `a . `d M{\tsubst[`b/`g]}] &\ByDef& 
\lmulmmt[{`M`a . [`d] M \tsubst[`b/`g] }]
 \end{array} $

 \myitem[{$`M`a . [`b] { `M`g .[`g] M } \reduc `M`a . [`b] M \tsubst[`b/`g] $}]
\lmulmmt[`M`a . `b { `M`g . `g M }] &\ByDef & 
\lmulmmt[m `a . `b {m `g . `g M }] &\redlmmt & \\
`M `a . < \SemM \tsubst[`b/`g] | `b > ~ \ByDef ~
\lmulmmt[m `a . `b M{\tsubst[`b/`g]}] &\ByDef& 
\lmulmmt[{`M`a . [`b] M \tsubst[`b/`g] }]
 \end{array} $

 \myitem[{$`M`a . [`a] M \rednclmu M ~ (`a\notele M) $}]
\lmulmmt[`M`a . `a M] &\ByDef & 
\lmulmmt[m `a . `a M] &\redlmmt (`h`m) & 
\lmulmmt[M]
 \end{array} $

 \myitem[$P \rednclmu Q \Then `Lx.P \rednclmu `Lx.Q$]
\lmulmmt[`Lx.P] &\ByDef & 
\lmulmmt[l x . P] &\crlmmt& \bsp (\IH)
\lmulmmt[l x . Q] &\ByDef & 
\lmulmmt[`Lx.Q] 
 \end{array} $

 \myitem[$P \rednclmu Q \Then `@ P R \rednclmu `@ Q R $]
\lmulmmt[`@ P R ] &\ByDef & 
\lmulmmt[a P R ] &\crlmmt& \bsp (\IH) \\
\lmulmmt[a Q R ] &\ByDef & 
\lmulmmt[`@ Q R ] 
 \end{array} $

 \myitem[$P \rednclmu Q \Then `@ R P \rednclmu `@ R Q $]
\lmulmmt[`@ R P ] &\ByDef & 
\lmulmmt[a R P ] &\crlmmt& \bsp (\IH) \\
\lmulmmt[a R Q ] &\ByDef & 
\lmulmmt[`@ R Q ] 
 \end{array} $

 \myitem[$P \rednclmu Q \Then `M `a . `b P \rednclmu `M `a . `b Q $]
\lmulmmt[`M `a . `b P ] &\ByDef & 
\lmulmmt[m `a . `b P ] &\crlmmt& \bsp (\IH)
\lmulmmt[m `a . `b Q ] &\ByDef & 
\lmulmmt[`M `a . `b Q ]
 \end{array} $
\arrayQED[-10pt]
 \end{description}
 \end{proof}

 \subsection{Results for \CBN}
We start by observing that the proof of \Thm \ref{interpretation respects beta reduction} also holds for \CBN:

 \begin{corollary} [{$\lmulmmt[`.]$ respects $\slmu$'s {\CBN} $`b$-reduction}]
 \label {interpretation respects CBN beta reduction}
$\lmulmmt[ `@ (`Lz.M) N ] \tcredlmmtN \lmulmmt[M\tsubst N/z ] $.
 \end{corollary}

As for $(\mur)$, for $\redlmmtN$, we have observed that we cannot use the proof of \Thm \ref{interpretation respects mur reduction}; the situation is different from that of $\redlmmtV$, in that the reductions over $x$ in that proof are allowed in $\redlmmtV$.
However, we can state that our interpretation maps a {$\redlmuN$} reduction in $\slmu$ into a {$\redlmmtN$} reduction in $\lmmt$, but for reductions that implement substitutions and prepare a stack; these are still {$\redlmmtn$} reductions, but can take place in continuations, and are not arbitrary, but specifically generated by the interpretation or necessitated by the strategy.
 
We need a separate proof for that property of \Thm \ref{interpretation respects mur reduction} for \CBN.
In fact, we need to transform the terms such that the $`m$-contractions over $`d$ pull in \textit{a stack}, as demanded by \CBN; this asks for a different lemma that expresses the effect of right structural substitution, one that works with stacks, as an alternative to \Lmm \ref{interpretation respects mur substitution CBV}.

 \begin{lemma} 
 \label{interpretation respects mur substitution CBN} 
 $ \lmulmmt[{}M \mursubst N.`g/`d ] \rtcredlmmtmo \SemM \mursubst [`N.`g/`d] $.
 \end{lemma}

 \begin{proof} \def \Ng {\mursubst [`N.`g/`d] } \def \tsubst[#1/#2] {#1}
By induction on the structure of terms\Shorter{; as before, all cases are straightforward, but for}:
 \begin{description}

\Longer{
 \myitem [$M = z$]
\lmulmmt[{z \mursubst N.`g/`d }] =
\lmulmmt[z] \ByDef
\lmulmmt[v z] =
\lmulmmt[v z] \tsubst [\Ng/`d] &\ByDef& 
\lmulmmt[z] \tsubst [\Ng/`d]
 \end{array} $

 \myitem[$M = `Lz.P$]
\lmulmmt[{(`Lz.P) \mursubst N.`g/`d }] &=& 
\lmulmmt[{`Lz.P \mursubst N.`g/`d }] &\ByDef& 
\lmulmmt[l z . {P\mursubst N.`g/`d }] &\rtcredlmmtmo& \bsp (\IH) \\
`L z . \lmulmmt[P] \tsubst [\Ng/`d] &=& 
(\lmulmmt[l z . P]) \tsubst [\Ng/`d] &\ByDef& 
\lmulmmt[`Lz.P] \tsubst [\Ng/`d]
 \end{array} $

 \myitem[$M = PQ$] 
\lmulmmt[{(`@ P Q ) \mursubst N.`g/`d }] 
	~=~ 
\lmulmmt[{`@ (P \mursubst N.`g/`d ) (Q \mursubst N.`g/`d ) }] 
	& \ByDef \\
\Mo `a . < \lmulmmt[{P\mursubst[N.`g/`d]}] | {\mt x . < \lmulmmt[{Q\mursubst[N.`g/`d]}] | {\mt y. <x|y`.`a>} >} > 
	&\rtcredlmmtmo& (\IH) \\
\Mo `a . < `P \tsubst [\Ng/`d] | {\mt x . < `Q\tsubst[\Ng/`d] | {\mt y. <x|y`.`a>} >} > 
	&=& \\
\lmulmmt[a P Q] \, \tsubst [\Ng/`d] 
	&\ByDef& 
\lmulmmt[`@ P Q ] \tsubst [\Ng/`d] 
 \end{array} $

}

 \myitem[{$M = `M`b.[`d]P$}]
\lmulmmt[{(`M`b . `d P ) \mursubst N.`g/`d }] 
	&\ByDef& 
\lmulmmt[`M`b . `g `@ (P \mursubst N.`g/`d ) N ] 
	&\ByDef& \\
{\lmulmmt[m `b . `g {A {P \mursubst N.`g/`d } N}] }
	&\redlmmtmo& \bsp (`a) 
`M`b . < \lmulmmt[P \mursubst N.`g/`d ] | { \mt x . < `N | { \mt y . < x | y`.`g > } > } > 
	&\redlmmtS& \bsp (\ref{lmmt prep}) \\
`M`b . < \lmulmmt[P \mursubst N.`g/`d ] | `N`.`g > 
	&\rtcredlmmtmo& \bsp (\IH) 
`M`b . < `P \tsubst [\Ng/`d] | `N`.`g > 
	&=& \\
\lmulmmt[m `b . `d P] \tsubst [\Ng/`d] 
	&\ByDef& 
\lmulmmt[`M`b. `d P] \tsubst [\Ng/`d]
 \end{array} $ 

\Longer{
 \myitem[{$M = `M`b.[`t]P$, $`t\not=`d$}]
\lmulmmt[{(`M`b.[`t]P)\mursubst N.`g/`d }] ~=~ 
\lmulmmt[{`M`b.[`t]P\mursubst N.`g/`d }] ~\ByDef~ 
\lmulmmt[m `b . `t {P \mursubst N.`g/`d }] &\rtcredlmmtmo& \bsp (\IH) \\
`M`b . < `P \tsubst [\Ng/`d] | `t > ~=~ 
\lmulmmt[m `b . `t P] \tsubst [\Ng/`d] ~\ByDef~ 
\lmulmmt[`M`b. `t P] \tsubst [\Ng/`d]
 \end{array} $
}

\arrayqed

 \end{description}

\observe
As before, the only contractions \Longer{take place for $ M = `M`b.[`d]P $; these }are a renaming cut in ${\redlmmtmo} \subseteq {\redlmmtn} $ over $`a$ introduced by the interpretation to deal with the added application, and the steps $ {\redlmmtS} 
\subseteq {\redlmmtmo} \subseteq {\redlmmtn} $ to build the stack $`N`.`g$.
And again, under induction reduction can take place everywhere.

 \end{proof}

Notice that this proof works with $`N`.`g$, not $\Ng = \mto x . < `N | { \mto y . < x | y`.`g >} > $ as 
was the case with \Lmm \ref{interpretation respects mur substitution CBV}.

 \begin{definition}
We will write $\lmulmmt[M] \crlmmtNn `N$ when there exists $t,t' \ele \lmmt$ such that $\lmulmmt[M] \rtcredlmmtN t \rtcredlmmtn t' \rtclmmtderN `N $, and the reduction steps in $\redlmmtn$ only effectuate $\lmu$'s structural substitution, and similar for $\crlmmtVv$; notice that this includes the case that $t'=t$, and also when $\lmulmmt[M] \crlmmtn `N$
 \end{definition}

 \begin{theorem} [{$\lmulmmt[`.]$ respects $\mur$-reduction in \CBN}] 
 \label {interpretation respects mur reduction CBN}

 \begin{enumerate}
 \firstitem $ 
\lmulmmt[{ `M`g .[`g]`@ M\mursubst[N.`g/`d] N }] 
\crlmmtNn 
\lmulmmt[{ `@ (`M`d.[`d]M) N }] )
$.

 \item $ 
\lmulmmt[{`M`g .[`b] (M \mursubst[N.`g/`d]})] 
\crlmmtNn 
\lmulmmt[{ `@ (`M`d.[`b]M) N }] )
$, $`b \not=`d$.
 \end{enumerate}
 \end{theorem}


 \begin{proof} 
 \begin{enumerate}
 \firstitem \label{first part CBN}
$ \kern-7.3mm \begin{array}[t]{lclclcl} \kern7.3mm
\lmulmmt[{`M`g .[`g](M\mursubst N.`g/`d) \hsk N }]
	&\ByDef& 
\lmulmmt[m `g . `g {A {M\mursubst N.`g/`d } N }]
	&\redlmmtmo & \bsp 
	(`a)  && \\  
`M`g . < \lmulmmt[m `d . `d {M\mursubst N.`g/`d }] | \Ng >
	& \redlmmtN& 
\bsp (`d) \lmulmmt[a <`g,x> {M\mursubst N.`g/`d } N ] 
	& \redlmmtS& \bsp (\ref{lmmt prep}) \\
`M`g . < \lmulmmt[M \mursubst N.`g/`d ]  | `N`.`g >   
	&\rtcredlmmtn& \bsp (\ref{interpretation respects mur substitution CBN}) 
`M`g . < \SemM \tsubst [`N`.`g/`d] | `N`.`g >   
	&\lmmtderN & \bsp 
	(`d) && \\
`M`g . < \lmulmmt[m `d . `d M] | `N`.`g >  
	&\lmmtderS& 
\Span{3}{
	\bsp (\ref{lmmt prep}) 
\lmulmmt[A {m `d . `d M} N ] 
	~\ByDef~   
\lmulmmt[{ `@ (`M`d.[`d]M) N }] 
}
 \end{array} $

 \item 
{\def \qsk{} \def \hsk{}
$ \begin{array}[t]{lclclcl}
\lmulmmt[{`M`g .[`b]M \mursubst[N.`g/`d] }]
	&\ByDef& 
\Span{3}{
\lmulmmt[m `g . `b {M \mursubst N.`g/`d }] 
	~\rtcredlmmtN (\ref{interpretation respects mur substitution CBN}) ~ 
`M`g . < \SemM \tsubst [\Ng/`d] | `b > 
}
	&\redlmmt& \bsp (`d) \\
`M`g . < \lmulmmt[m `d . `b M] | \Ng > 
	&\lmmtderS& \bsp (\ref{lmmt prep}) 
\lmulmmt[A <`g
> {m `d . `b M} N ] 
	&\ByDef& 
\lmulmmt[{ `@ (`M`d.[`b]M) N }] 
 \end{array} $
\arrayqed
}

 \end{enumerate}

\observe
Notice that the $(`m)$-contractions are over $`d$ and $`a$ and pull in $`N`.`g$ and $`g$, respectively, which are stacks, so these steps are allowed in $\redlmmtN$.
Also, we need extra $\redlmmtS$ reductions (from \Exm \ref{lmmt prep}) in both directions to create the stack, which are in $\redlmmtN$.
 
 \end{proof}
 
We do not need to check results for $(\mul$) reduction, since that is not part of \CBN.
In total, we can state:
 
 \begin{theorem}[Preservation of \CBN-strategy] 
If $M \rtcrednclmuN N$, then $\lmulmmt[M] \crlmmtNn `N $.
 \end{theorem}
 \begin{proof} 
 \begin{description}

 \item[{$ `@ (`Lz.M) N \rednclmuN M\tsubst N/z $}]
By \Cor \ref{interpretation respects CBN beta reduction}.

 \item[{$ `@ (`M `a . \Cmd) N \rednclmuN `M `g . \Cmd \mursubst N.`g/`a $}]
By \Thm \ref{interpretation respects mur reduction CBN}.

 \myitem[{$`M`a . [`b] { `M`g .[`d] M } \rednclmuN `M`a . [`d] M \tsubst[`b/`g], ~ `g \not= `d $}]
\lmulmmt[`M`a . `b { `M`g . `d M }] &\ByDef & 
\lmulmmt[m `a . `b {m `g . `d M }] &\redlmmtN & \\
`M `a . < \lmulmmt[M] \tsubst[`b/`g] | `d > ~ \ByDef ~
\lmulmmt[m `a . `d M{\tsubst[`b/`g]}] &\ByDef& 
\lmulmmt[{`M`a . [`d] M \tsubst[`b/`g] }]
 \end{array} $

 \myitem[{$`M`a . [`b] { `M`g .[`g] M } \rednclmuN `M`a . [`b] M \tsubst[`b/`g] $}]
\lmulmmt[`M`a . `b { `M`g . `g M }] &\ByDef & 
\lmulmmt[m `a . `b {m `g . `g M }] &\redlmmtN & \\
`M `a . < \lmulmmt[M] \tsubst[`b/`g] | `b > ~ \ByDef ~
\lmulmmt[m `a . `b M{\tsubst[`b/`g]}] &\ByDef& 
\lmulmmt[{`M`a . [`b] M \tsubst[`b/`g] }]
 \end{array} $

 \myitem[{$`M`a . [`a] M \rednclmuN M ~ (`a\notele M) $}]
\lmulmmt[`M`a . `a M] &\ByDef & 
\lmulmmt[m `a . `a M] &\redlmmtN (`h`m) & 
\lmulmmt[M]
 \end{array} $

 \myitem[$P \rednclmuN Q \Then `@ P R \rednclmuN `@ Q R $]
\lmulmmt[`@ P R ] &\ByDef& 
\lmulmmt[A P R ] &\crlmmtNn& \bsp (\IH) 
\lmulmmt[A Q R ] &\ByDef & 
\lmulmmt[`@ Q R ]. %
 \end{array} $

 \myitem[$P \rednclmuN Q \Then `M `a . `b P \rednclmuN `M `a . `b Q $]
\lmulmmt[`M `a . `b P ] &\ByDef & 
\lmulmmt[m `a . `b P ] &\crlmmtNn& \bsp (\IH)
\lmulmmt[m `a . `b Q ] &\ByDef & %
	\\
\lmulmmt[`M `a . `b Q ].
 \end{array} $
\arrayQED
 \end{description}

\observe
Notice that reduction only takes places in terms, never in contexts.

 \end{proof}

By all the observations we have made above on the details of the proofs and the permissibility of the reduction steps involved in \CBN, we also have the following result.

 \begin{corollary}[Preservation of $\redlmmtn$] \label{CBN reduction preserved}
$M \rtcrednclmun N$, then $\lmulmmt[M] \crlmmtn `N$.
 \end{corollary}

 \def \crlmmtVmo {\crlmmtV^{\UL{`m}}}
 \def \redlmmtVmo {\redlmmtV^{\UL{`m}}}

 \subsection{Results for \CBV}

We start by observing that the proof of \Thm \ref{interpretation respects beta reduction} also holds for \CBV:

 \begin{corollary} [{$\lmulmmt[`.]$ respects $\slmu$'s {\CBV} $`b$-reduction}]
 \label {interpretation respects CBV beta reduction}
$\lmulmmt[ `@ (`Lz.M) V ] \tcredlmmtV \lmulmmt[M\tsubst V/z ] $.
 \end{corollary}

\Comment{
In the proof of \Lmm \ref{interpretation respects mur substitution CBV}, we observed that the implementation of right structural substitution requires unrestricted reduction, but that all reduction steps are $\UL{`m}$-steps, permitted in $\redlmmtv$.
Using that result, we will now show that $\slmu$'s {\CBV} $(\mur)$ reduction gets respected by the interpretation, permitting for the additional $\UL{`m}$ expansions necessitated by the implementation of right structural substitution.

By the proof of \Thm \ref{interpretation respects mur reduction}, we can state that our interpretation maps a {$\redlmuV$} reduction in $\slmu$ into a {$\redlmmtV$} reduction in $\lmmt$, but for reductions that implement substitutions; these are still {$\redlmmtv$} reductions, but can take place in continuations, and are not arbitrary, but specifically generated by the interpretation.
%
In fact, we can postpone those reductions, until they become executable when moved into the executable position by reduction, and accept that when modelling $\redlmmtV$, we are left with some $\redlmmtv$-redexes still to contract, that are needed to complete the substitution.
\Comment{
\Mo `a . < \lmulmmt[{P\mursubst[N.`g/`d]}] | {\mt x . < \lmulmmt[{Q\mursubst[N.`g/`d]}] | {\mt y. <x|y`.`a>} >} > 
	&\crlmmtN& (\IH) \\
\Mo `a . < `P \tsubst [\Ng/`d] | {\mt x . < \lmulmmt[{Q\mursubst[N.`g/`d]}] | {\mt y. <x|y`.`a>} >} > 
	&\redlmmtN& (x) \\
\Mo `a . < \lmulmmt[{Q\mursubst[N.`g/`d]}] | {\mt y. < `P\tsubst[\Ng/`d] |y`.`a>} > 
	&\crlmmtN& (\IH) \\
\Mo `a . < `Q\tsubst[\Ng/`d] | {\mt y. < `P\tsubst[\Ng/`d] |y`.`a>} > 
	&\lmmtderN& (x) \\
}
For example, when we can run $`P \tsubst [\Ng/`d] $ in $\redlmmtV$ to a value $V$, we would get
 \[ \begin{array}{lcl}
\Mo `a . < `P \tsubst [\Ng/`d] | {\mt x . < \lmulmmt[{Q\mursubst[N.`g/`d]}] | {\mt y. <x|y`.`a>} >} > 
	&\rtcredlmmtV& \\
\Mo `a . < V | {\mt x . < \lmulmmt[{Q\mursubst[N.`g/`d]}] | {\mt y. <x|y`.`a>} >} > 
	&\redlmmtV& (x) \\
\Mo `a . < \lmulmmt[{Q\mursubst[N.`g/`d]}] | {\mt y. < V |y`.`a>} > 
 \end{array} \]
If running $P\mursubst[N.`g/`d] $ never runs to a value, of course we can never $\mutilde$ reduce over $x$, and the ($\redlmmtv$) reductions needed to run $ \lmulmmt[{Q\mursubst[N.`g/`d]}] $ will never be executed.

As to $(\mul)$, to show the relation between $\lmu$'s $\mur$ {\CBV} reduction and that of $\lmmt$, we can only show a result `modulo $\UL{`m}$-expansion', so $`m$-expansion of redexes generated by the interpretation to deal with the fact that $\lmu$'s reduction can create new applications. 
%
We will write $\redlmmtVmo$ for $\redlmmtV$, while allowing for extra $\redlmmtmo$ renaming expansion steps. 
}

Since the proofs of \Lmm \ref{interpretation respects mur substitution CBV} and \Thm \ref{interpretation respects mur reduction} are also valid for $\redlmmtV$, by their structure and the observations made after the proof of \Thm \ref{interpretation respects mur reduction}, we can now state:

 \begin{corollary} [{$\lmulmmt[`.]$} respects $\mur$-reduction] 
 \label {interpretation respects mur reduction CBV} 

 \begin{enumerate}

 \firstitem $ 
\lmulmmt[{ `M`g .[`g]`@ M\mursubst[N.`g/`d] N }] 
\crlmmtVv 
\lmulmmt[{ `@ (`M`b.[`b]M) N }] )
$.

 \item $ 
\lmulmmt[{`M`g .[`b] (M \mursubst[N.`g/`d]})] 
\crlmmtVv 
\lmulmmt[{ `@ (`M`d.[`b]M) N }] )
$, $`b \not=`d$.
 \end{enumerate}
 \end{corollary}

As to $(\mul)$, since the proofs of \Lmm \ref{interpretation respects lmu substitution} and \Thm \ref{interpretation respects lmu reduction} are also valid for $\redlmmtV$, by their structure we can now state:

 \begin{corollary} [{$\lmulmmt[`.]$ respects $\slmu$'s {$\redlmmtV$} $\mul$}]  
 \label {interpretation respects mul reduction CBV} 
 \begin{enumerate}

 \firstitem 
 $ \lmulmmt[`M`g . `g {`@ V (\mulsubst N.`g/`a M )}] 
\crlmmtVv 
 \lmulmmt[ `@ V (`M`a. `a M) ] $.

 \item 
 $ \lmulmmt[`M`g . `d {\mulsubst V.`g/`a } M ]  
\crlmmtVv 
 \lmulmmt[ `@ V (`M`a. `d M) ] $, 
 with $`a \not= `d$.
 
 \end{enumerate}

 \end{corollary}

We can now state:

 \begin{theorem}[Preservation of \CBV-strategy] 
If $M \rednclmuV N$, then  $`N \crlmmtVv \SemM  $.
 \end{theorem}
 \begin{proof}
 \begin{description}

 \item[{$ `@ (`Lz.M) V \rednclmuV M\tsubst V/z $}]
By \Cor \ref{interpretation respects CBV beta reduction}.

 \item[{$ `@ (`M `a . \Cmd) V \rednclmuV `M `g . \Cmd \mursubst V.`g/`a $}]
By \Cor \ref{interpretation respects mur reduction CBV}.

 \item[{$ `@ V (`M `a . \Cmd) \rednclmuV `M `g . \mulsubst V.`g/`a \Cmd $}]
By \Cor \ref{interpretation respects mul reduction CBV}.

 \myitem[{$`M`a . [`b] { `M`g .[`d] M } \rednclmuV `M`a . [`d] M \tsubst[`b/`g], ~ `g \not= `d $}]
\lmulmmt[`M`a . `b { `M`g . `d M }] &\ByDef & 
\lmulmmt[m `a . `b {m `g . `d M }] &\redlmmtV & \\
`M `a . < \SemM \tsubst[`b/`g] | `d > ~ \ByDef ~
\lmulmmt[m `a . `d M{\tsubst[`b/`g]}] &\ByDef& 
\lmulmmt[{`M`a . [`d] M \tsubst[`b/`g] }]
 \end{array} $

 \myitem[{$`M`a . [`b] { `M`g .[`g] M } \rednclmuV `M`a . [`b] M \tsubst[`b/`g] $}]
\lmulmmt[`M`a . `b { `M`g . `g M }] &\ByDef & 
\lmulmmt[m `a . `b {m `g . `g M }] &\redlmmtV & \\
`M `a . < \SemM \tsubst[`b/`g] | `b > ~ \ByDef ~
\lmulmmt[m `a . `b M{\tsubst[`b/`g]}] &\ByDef& 
\lmulmmt[{`M`a . [`b] M \tsubst[`b/`g] }]
 \end{array} $

 \myitem[{$`M`a . [`a] M \rednclmuV M ~ (`a\notele M) $}]
\lmulmmt[`M`a . `a M] &\ByDef & 
\lmulmmt[m `a . `a M] &\redlmmtV (`h`m) & 
\lmulmmt[M]
 \end{array} $

 \myitem[$P \rednclmuV Q \Then `@ P R \rednclmuV `@ Q R $]
\lmulmmt[`@ P R ] &\ByDef& 
\lmulmmt[A P R ] & \crlmmtVv & \bsp (\IH) 
\lmulmmt[A Q R ] &\ByDef & 
\lmulmmt[`@ Q R ]. 
 \end{array} $

 \myitem[$P \rednclmuV Q \Then `@ V P \rednclmuV `@ V Q $]
\lmulmmt[`@ V P ] &\ByDef & 
\lmulmmt[a V P ] & \redlmmtV & \bsp (x) \\
\lmulmmt[b V P ] & \crlmmtVv & \bsp (\IH)
\lmulmmt[b V Q ] &\lmmtderV & \bsp (x) \\
\lmulmmt[a V Q ] &\ByDef & 
\lmulmmt[`@ V Q ]. 
 \end{array} $

 \myitem[$P \rednclmuV Q \Then `M `a . `b P \rednclmuV `M `a . `b Q $]
\lmulmmt[`M `a . `b P ] &\ByDef & 
\lmulmmt[m `a . `b P ] &\crlmmtVv& \bsp (\IH)
\lmulmmt[m `a . `b Q ] &\ByDef & \\
\lmulmmt[`M `a . `b Q ].
 \end{array} $
\arrayQED
 \end{description}

\observe
Notice that the $(\MuT)$-contractions over $x$ are allowed in the penultimate part, since $`V$ is a value.
 \end{proof}

By all the observations we have made above on the details of the proofs and the permissibility of the reduction steps involved in \CBV, we also have the following result.

 \begin{corollary}[Preservation of $\redlmmtv$] \label{CBV reduction preserved}
If $M \rtcrednclmuV N$, then $\SemM \crlmmtv `N$.
 \end{corollary}
 
This concludes our results for the interpretation of $\slmu$ into $\lmmt$. 
We have shown that there exists a single interpretation from $\nclmu$ to $\lmmt$ that respects reduction, and the $\redlmmtN$ and $\redlmmtV$ reduction strategies (modulo $`m$ and $\mutilde$-expansions of redexes generated by the interpretation), thus establishing a strong relation between $\nclmu$, $\lmu$, and $\lmuV$ on one side, and $\lmmt$ on the other, as well as between the respective {\CBN} and {\CBV} sub-reduction systems and strategies in $\nclmu$ and $\lmmt$.
We have seen that, due to the fact that structural reduction in $\slmu$ adds applications, we cannot map substitutions to substitutions, and therefore not strategies to strategies, but also that the damage is limited.


%

\section{On Left-\TeXorPDFstring{$`m$}{Mu} Contraction in {\lmmtPDF}} 
    \label{add `m_l to lmmt}

Currently, there is no analogue of $\lmu$'s $(\mul)$-reduction in $\lmmt$.
We have remarked above that the term $ `M `g . < `Lx.t | {`M`a.<t'|`b>} `. `g > $ is not a $\lmmt$ critical pair, whereas its $\nclmu$-counterpart $ `@ (`Lx.M) (`M`a.[`b]N) $ \emph{is} a $\nclmu$ critical pair.
Notice that 
 \[ \begin{array}{rcl} 
`@ (`Lx.M) (`M`a.[`b]N) &\rednclmuV& `M`g.[`b]\lmusubst[{(`Lx.M)}.`g/`a] N ,
 \end{array} \]
but $ \SEMcbv{`@ (`Lx.M) (`M`a.[`b]N) } = `M `g . < `Lx.\SEMcbv{M} | {`M`a.<\SEMcbv{N}|`b>} `. `g > 
$ is not reducible under {\CBN} or {\CBV}.
We have shown that we can represent (implement) this reduction step under $\lmulmmt[`.]$, but there is an alternative path to this.

 \begin{definition}
We extend $\lmmt$'s notion of reduction by adding the reduction rule: 
 \[ \begin{array}{rrcl}
(`m_l) : & \cell < t | (`M`a. c ) `. `g > & \reduc & c \tsubst[ \mt y.< t | y`.`g > / `a ] 
 \end{array} \]
This rule is added for full reduction, but excluded from {\CBN} and restricted in the normal way for \CBV.
 \end{definition}
 
We can show that this reduction makes logical sense: first we show that the substitution $ {`.} \tsubst[ \mt y.< t | y`.`g > / `a ] $ is type sound.

 \def \Subst {\,\textsl{S}}

 \begin{lemma} \label{rule m_l lemma}
Let $ \Subst = \tsubst[ {\mt y.< t | y`.`g >} / `a ] $, and assume $ \derlmmt `G |- t : A\arr B | `g`:B,`D $.
Then:
 \begin{enumerate}
 \item
If $ \derlmmt c : `G |- `a`:A,`D $ then $ \derlmmt { c \Subst } : `G |- `g`:B,`D $. 

 \item
If $ \derlmmt `G |- t' : C | `a`:A,`D $ then $ \derlmmt `G |- {t' \Subst} : C | `g`:B,`D $. 

 \item
If $ \derlmmt `G | e : C |- `a`:A,`D $ then $ \derlmmt `G | {e \Subst} : C |- `g`:B,`D $. 

 \end{enumerate}
 \end{lemma}
 \begin{Proof}{By induction on the definition of type assignment.}{\Lmm \ref{rule m_l lemma}}
By induction on the definition of type assignment.

 \begin{description}

 \item[$\Cut$]
Then $ c = \cell<t|v> $, and both $ \derLmmt `G |- t : A | `g`:B,`D $ and $ \derLmmt `G | e : A |- `g`:B,`D $.
Then by induction, both $ \derLmmt `G |- { t \Subst } : A | `g`:B,`D $ and $ \derLmmt `G | { e \Subst }
 : A |- `g`:B,`D $.
By rule $(\Cut)$ we have 

 \[ \def \TurnLmmt {\Turn}
\Inf	{ \derLmmt `G |- {t \Subst} : A | `g`:B,`D \quad \derLmmt `G | {e \Subst} : A |- `g`:B,`D 
	}{\derLmmt { \cell< t \Subst | e \Subst>} : `G |- `D } 
 \]
and $ \cell< t \Subst | e \Subst> = \cell< t | e > \Subst $.

 \item[$\AxT$]
Then $t = x$; since $ x`:C \ele `G $, by rule $(\AxT)$ also $ \derlmmt `G |- x : C | `g`:B,`D $
 
 \item[$\AxC$]
Then $e = `b$, and $ `b`:C \ele `D $.
We have two cases:

 \begin{description}
 \item[$`a = `b$]
Then $C = A$; we can construct: 
 \[ \def \Turnlmmt {\Turn}
\Inf	[\MuT]
	{\Inf	[\Cut]
		{\Inf	[\Weak]{\InfBox { \derlmmt `G |- t : A\arr B | `g`:B,`D }
			}{ \derlmmt `G,y`:A |- t : A\arr B | `g`:B,`D }
		 \Inf	[\arrL]
			{\Inf	[\AxT]{ \derlmmt `G,y`:A |- y : A | `g`:B,`D }
		 	 \Inf	[\AxC]{ \derlmmt `G,y`:A | `g : B |- `g`:B,`D }
			}{ \derlmmt `G,y`:A | y`.`g : A\arr B |- `g`:B,`D }
		}{ \derlmmt {\cell<t|y`.`g>} : `G,y`:A |- `g`:B,`D }
	}{ \derlmmt `G | \mt y.<t|y`.`g> : A |- `g`:B,`D }
 \]
and $ `a \Subst = \mt y.< t | y`.`g > $.

 \item[$`a\not=`b$]
Since $ `b`:C \ele `D $, by rule $(\AxC)$ also $ \derlmmt `G | `b : C |- `g`:B,`D $.
 
 \end{description}
 \item[$\arrL)$, $(\arrR)$, $(`m)$, $(\MuT$]
By induction.
\QED

 \end{description}
 \end{Proof}

With this result, we can now show:

 \begin{theorem}[Soundness for rule $(`m_l)$]
If $ \derLmmt { \cell < t | { `M`a. c} `. `g > } : `G |- `g`:B,`D $, then $ \derLmmt { c \tsubst[ {\mt y.< t | y`.`g >} / `a ] } : `G |- `g`:B,`D $.
 \end{theorem}
 \begin{proof}
If $ \derLmmt { \cell < t | { `M`a. c} `. `g > } : `G |- `D $, then the derivation is shaped like:
 \[ \def \Turnlmmt {\Turn} 
\Inf	{\InfBox { \derlmmt `G |- t : A\arr B | `g`:B,`D }
	\quad
	 \Inf	{\Inf	{\InfBox{ \derlmmt c : `G |- `a`:A,`g`:B,`D }
			}{ \derlmmt `G |- `M`a.c : A | `g`:B,`D }
	 	 \Inf	{ \derlmmt `G | `g : B |- `g`:B,`D }
		}{ \derlmmt `G | `M`a.c`.`g : A\arr B |- `g`:B,`D }
	}{ \derlmmt {\cell<t|{`M`a.c}`.`g>} : `G |- `g`:B,`D }
 \]
Then, in particular, we have $ \derlmmt `G |- t : A\arr B | `g`:B,`D $ and $ \derlmmt c : `G |- `a`:A,`g`:B,`D $, and by \Lmm \ref{rule m_l lemma}, we have $ \derLmmt { c \tsubst[ {\mt y.< t | y`.`g >} / `a ] } : `G |- `g`:B,`D $.
\QED
 \end{proof}

So this new rule is sound.
We have shown above that we have full representation of the three notions of reductions we focus on for standard $\lmmt$, so do not need to extend reduction on $\lmmt$ to achieve this.
It will nonetheless be interesting to see if adding this reduction step to $\lmmt$ would yield different interpretations of $\nclmu$; we leave this for future work.

\Comment{
And in fact, we can also show that our interpretation respects this additional reduction step.

 \[ \begin{array}{lcl}
\lmulmmt[{`@ (`Lx.M) (`M`a. [`b] N ) }] & \ByDef & \\
`M `g . < `Lx.\SemM | { `M`a. < `N | `b > } `. `g > & \redlmmt & \\
`M`g . < \lmusubst [{(`Lx.\SemM)}.`g/`a] `N | `b > & \ByDef & \\
\lmulmmt[{`M`g . [`b] \lmusubst [{(`Lx.M)}.`g/`a] N }]
 \end{array} \]

 \begin{theorem}
 \end{theorem}

 \[ \begin{array}{lcl}
\lmulmmt[{`@ M (`M`a. [`b] N ) }] = \\
\lmulmmt[a <`g,x> M {m `a . `b N}] &\reduc (x)& \\
`M`g.<{`M`a.<\SEMcbn{N}|`b>} | {\mt y . < \SemM | y`.`g>} > &\reduc& \\
 \end{array} \]

$ \SEMcbn{`@ M (`M`a. [`b] N ) } = `M`g.< \SEMcbn{M} | {`M`a.<\SEMcbn{N}|`b>}`.`g > 
\reduc(`m_l) 
`M`g.<\SEMcbn{N}|`b> \tsubst[ {\mt y.< \SEMcbn{M} | y`.`g >} / `a ] 
$

$ \SEMcbv{`@ V (`M`a. [`b] N ) } = `M`a.< {`M`a.<\SEMcbn{N}|`b>} | {\mt x . < \SEMcbv{V} | x`.`a >} > $

Notice that
 \[ \def \Turnlmmt {\Turn} \begin{array}{c}
\Inf	{\InfBox { \derlmmt `G |- t : A\arr B | `g`:B,`D }
	\quad
	 \Inf	{\Inf	{\InfBox{ \derlmmt c : `G |- `a`:A,`g`:B,`D }
			}{ \derlmmt `G |- `M`a.c : A | `g`:B,`D }
	 	 \Inf	{ \derlmmt `G | `g : B |- `g`:B,`D }
		}{ \derlmmt `G | `M`a.c`.`g : A\arr B |- `g`:B,`D }
	}{ \derlmmt {\cell<t|{`M`a.c}`.`g>} : `G |- `g`:B,`D }
\\ [4mm]
\Inf	{\Inf	{\InfBox { \derlmmt `G,y`:A |- t : A\arr B | `g`:B,`D }
		\quad
		 \Inf	{\Inf	{ \derlmmt `G,y`:A |- y : A | `g`:B,`D }
		 	 \Inf	{ \derlmmt `G,y`:A | `g : B |- `g`:B,`D }
			}{ \derlmmt `G,y`:A | y`.`g : A\arr B |- `g`:B,`D }
		}{ \derlmmt {\cell<t|y`.`g>} : `G,y`:A |- `g`:B,`D }
	}{ \derlmmt `G | \mt y.<t|y`.`g> : A |- `g`:B,`D }
 \end{array} \]

$ \SEMcbn{MN} = `M`a.< \SEMcbn{M} | \SEMcbn{N}`.`a > $

$ \SEMcbv{MN} = `M`a.< \SEMcbv{N} | {\mt x . < \SEMcbv{M} | x`.`a >} > $

$ \lmulmmt[`@ M N ] = \lmulmmt[a M N] $

 \[ \begin{array}{lcl}
\lmulmmt[{ `@ (`Lx.M) (`M`a . [`b] N ) }] = 
\lmulmmt[a {l x . M} {m `a . `b N }] 
 \end{array} \]

 \[ \begin{array}{lcl}
`@ (`Lx.M) (`M`a. [`b] N ) & \rednclmu & 
`M`g . [`b] \lmusubst [{`Lx.t}.`g/`a] N 
 \end{array} \]

 \[ \begin{array}{lcl}
\SEMcbn{`@ (`Lx.M) (`M`a. [`b] N ) } & \ByDef & \\
`M `g . < `Lx.\SEMcbn{M} | { `M`a. < \SEMcbn{N} | `b > } `. `g > & \redlmmt & \\
`M`g . < \lmusubst [{(`Lx.\SEMcbn{M})}.`g/`a] \SEMcbn{N} | `b > & \ByDef & \\
\SEMcbn{`M`g . [`b] \lmusubst [{(`Lx.M)}.`g/`a] N }
 \end{array} \]
}


%


\setbox112=\hbox{ \large $\lmmt$ in $`X$}
\section{Interpreting {\lmmtPDF} in {\XiCalcPDF}}
 \label{lmmt in X}
 \label{sec:InterpLmmtX}

\Comment{
There exists an obvious translation from $`X$ into $\lmmt$:

 \[ \begin{array}{rcl}
\SemX { \caps <x,`a>} & \ByDef & \cell<x|`a> \\
\SemX { \exp x P `a . `b } & \ByDef & \wcell < `L x.`M`a . \SemX{P} | `b >\\
\SemX { \imp P `a [x] y Q } & \ByDef & \wcell <x | `M`a. \SemX{P}`. \mt y . \SemX{Q } > \\
\SemX { \cut P `a + x Q } & \ByDef & \wcell < `M `a.\SemX{P} | \mt x . \SemX{Q } >
 \end{array} \]
}

\citet{Lengrand'03} presents a translation of $\lmmt$ into $\X$ (called $`l`x$ there) which preserves the typing and shows that it respects reduction; since here we use a slight variation, that uses substitution rather than a renaming cut when dealing with $`M`a.c$ or $\mt x.c$, we need to give new proofs for this result.

We could say that, in $\lmmt$, not all inputs (normally present as term variables) and outputs (context variables) are explicitly named; for example, a term like $`Lx.t$ has no named output, which we can make explicit through $`h$-expanding it into $ `M`b . < `Lx.t | `b > $; similarly, $t`.e$ has no named input, which we can make explicit through $ \mt x . < x | t`.e > $.
This is the essence of the translation we define now: it interprets terms under an output, and environments under an input, making the implicit names explicit. 
The first version was defined by \citet{Lengrand'03} through:
%



 \[ \begin{array}{c}
 \begin {array}[t]{rcl}
\lmmtOrg (\cell<t|e>) & = & \lmmtOrg (C <t|e>) ~ (`a,x \textit{ fresh}) \hspace*{2cm} \\ 
 \end{array} 
 \\ 
 \begin {array}[t]{rcll}
\lmmtOrg (x){`a} & = & \lmmtOrg (V x){`a} \\
\lmmtOrg (`Lx.t){`a} & = & \lmmtOrg (L x . t){`a} ~ (`b \textit{ fresh}) \\ 
\lmmtOrg (`M`b.c){`a} & = & \lmmtOrg (M `b . c){`a} \\ 
 \end{array}
 \dquad 
 \begin {array}[t]{rcll}
\lmmtOrg (`a){x}	& = & \lmmtOrg (N `a){x} \\
\lmmtOrg (t`.e){x} & = & \lmmtOrg (X t . e){x} ~ (`g,z \textit{ fresh}) \\
\lmmtOrg (\mt y.c ){x} & = & \lmmtOrg (T y . c){x}
 \end{array}
 \end {array} \]
In \cite{Bakel-Lescanne-MSCS'08}, results were obtained for an interpretation where the last two alternatives are defined through:

 \[ \begin {array}[t]{rcll}
\lmmtTran{`M`b.c}{`a} & = & \cut { \lmmtTran{c}} `b + x { \caps<x,`a>} \\
\lmmtTran{ \mutilde y.c}{x} & = & \cut \caps<x,`b> `b + y { \lmmtTran{c}}
 \end {array} \]
showing that even here implicit substitutions are not needed.
Since we can show:
 \[ \begin {array}[t]{rclclclcl}
\cut { \lmmtTran{c}} `b + x { \caps<x,`a>} & \redXs & 
\cutL { \lmmtTran{c}} `b + x { \caps<x,`a>} & \rtcredXs & \bsp (\ref{renaming lemma}) 
\lmmtTran {`M`b.c}{`a} 
	\\
\cut \caps<x,`b> `b + y { \lmmtTran{c}} & \redXs & 
\cutR \caps<x,`b> `b + y { \lmmtTran{c}} & \rtcredXs & \bsp (\ref{renaming lemma})  
\lmmtTran {\mutilde y.c}{x}
 \end {array} \]
and in \Lmm\ref{mu simulation} and \ref {mut simulation} we want to be able to simulate $\lmmt$-substitution through $`X$-substitution, here we prefer to avoid adding the extra cuts, but use implicit substitution in the interpretation.

A disadvantage of this encoding is that it treats $\cell<x|`a>$ as a cut, rather than as $\caps<x.`a>$. 
However, we can show:
 \[ \begin{array}{rclclclcl}
\lmmtTran {\cell<y|`b>} 
&\ByDef& \lmmtTran {C <y|`b>} 
&\ByDef& \lmmtTran {C <{V y}|{N `b}>} 
& \redX& \bsp (\Cap) \caps<y,`b> 
	\\
\lmmtTran {\cell<y|t`.e>} 
&\ByDef& \lmmtTran {C <y|t`.e>} 
&\ByDef& \lmmtTran {C <{V y}|{X t . e}>} 
&\redX& \bsp (\Imp) \lmmtTran {X t . e}{y} 
&\ByDef& \lmmtTran {t`.e}{y} 
	\\ 
\lmmtTran {\cell<y|\mt z.c>} 
&\ByDef& \lmmtTran {C <y|\mt z.c>} 
&\ByDef& \lmmtTran {C <{V y}|{T z . c}>} 
&\redX& \bsp (\ref{renaming lemma}) \lmmtTran {T z . c}{y} 
&\ByDef& \lmmtTran {\mt z.c}{y} 
	\\ 
\lmmtTran {\cell<`Ly.t|`b>} 
&\ByDef& \lmmtTran {C <`Ly.t|`b>} 
&\ByDef& \lmmtTran {C <{L y . t}|{N `b}>} 
&\redX& \bsp (\Exp) \lmmtTran {L y . t}{`b} 
&\ByDef& \lmmtTran {`ly.t}{y} 
	\\ 
\lmmtTran {\cell<`M`g.c|`b>} 
&\ByDef& \lmmtTran {C <`M`g.c|`b>} 
&\ByDef& \lmmtTran {C <{M `g . c}|{N `b}>} 
&\redX& \bsp (\ref{renaming lemma}) \lmmtTran {M `g . c}{`b} 
&\ByDef& \lmmtTran {`M`g.c}{`b} 
 \end{array} \]
 
The interpretation we will consider here is an optimised version, that uses these observations and will induce less reduction in the image: this will yield stronger results.

 \begin {definition}[Translation of $\lmmt$ into $`X$] \label{from lmmt to x} ~
 \[ \begin{array}{c@{ \quad}c@{ \quad}c}
 \begin {array}[t]{rcll}
\lmmtX(\cell<y|`b>) & = & \caps<y,`b> \\ 
\lmmtX(\cell<y|e>) & = & \lmmtX(e){y} \\ 
\lmmtX(\cell<t|`b>) & = & \lmmtX(t){`b} \\
\lmmtX(\cell<t|e>) & = & \lmmtX(C <t|e>) & \\ && \quad (`a,x \textit{ fresh, otherwise}) &
 \end{array} 
 &
 \begin {array}[t]{rcll}
\lmmtX(x){`a} & = & \lmmtX(V x){`a} \\
\lmmtX(`M`b.c){`a} & = & \lmmtX(M `b . c){`a} \\ 
\lmmtX(`Lx.t){`a} & = & \lmmtX(L x . t){`a} \\ && \quad (`b \textit{ fresh}) \\ 
 \end{array}
 & 
 \begin {array}[t]{rcll}
\lmmtX(`a){x}	& = & \lmmtX(N `a){x} \\
\lmmtX(\mt y.c){x} & = & \lmmtX(T y . c){x} \\
\lmmtX(t`.e){x} & = & \lmmtX(X t . e){x} \\ && \quad (`g,z \textit{ fresh}) \\
 \end{array}
 \end {array} \]
 \end {definition}

There are overlapping cases in this definition.
However, we have 
 \[ \begin {array}[t]{rcccccl}
\lmmtX(\cell<y|`b>) & = & \lmmtX(`b){y} & = & { \lmmtX(y){`b} } &=& \caps<y,`b> 
 \end{array} \]
so this does not inconvenience the proofs.
This translation will only create a cut when interpreting a command that does not involve a variable or name.


 \begin{remark}
If $t'$ is a $\lmmt$-value that does not contain $`a$, then $ \lmmtX(t'){`a} $ introduces $`a$, and if $e'$ is a $\lmmt$-stack that does not contain $x$, then $ \lmmtX(e'){x} $ introduces $x$:
 \[ \begin{array}{rcl}
\lmmtX(x){`a} &=& \lmmtX(V x){`a} \\
\lmmtX(`lx.t){`a} &=& \lmmtX(L x . t){`a}
 \end{array} 
 \dquad
 \begin{array}{rcl}
\lmmtX(`a){x} &=& \lmmtX(N `a){x} \\
\lmmtX(t`.e){x} &=& \lmmtX(X t . e){x} 
 \end{array} \]
We will need this observation later.
 \end{remark}

This interpretation preserves typeability:

 \begin {lemma} 
 \begin {flatenumerate}
 \firstitem 
 If $\derlmmt `G | t : A |- `D $, then $ \derX \lmmtX(t){`a} : `G |- `a`:A,`D $;
 \item If $\derlmmt `G |- e : A | `D $, then $\derX \lmmtX(e){x} : `G,x`:A |- `D $; and
 \item If $\derlmmt c : `G |- `D $, then $\derX { \lmmtX(c) } : `G |- `D $.
 \end {flatenumerate}
 \end {lemma}
 \begin{proof}
This follows from a similar result for the original interpretation as shown in \cite{Bakel-Lescanne-MSCS'08}, the type preservation result shown there, and the observation we made above about the relation between the interpretations. \QED
 \end{proof}

We will now show the relation between implicit substitution in $`X$ and in $\lmmt$.
A similar result was shown in \cite{Lengrand'03}, but with respect to $\X$.
We show the proof in full since we need to know if reduction will stay within {\CBN} or \CBV.

First we show that $`X$ successfully encodes $\lmmt$'s context substitution $ `. \tsubst e/`a $ through right substitution $\cutL {`.} `a + x { \lmmtX(e){x} } $.

 \begin {lemma} \label {mu simulation} 
 \begin {flatenumerate}
 \firstitem $\cutL { \lmmtX(c) } `a + x {\lmmtX(e'){x} } \substo \lmmtX(c{ \tsubst[e'/`a] }) $; 
 \item If $`a \not= `b$, then $\cutL {{ \lmmtX(t){`b} } } `a + x {\lmmtX(e'){x} }  \substo \lmmtX(t{ \tsubst[e'/`a] }){`b} $; and 
 \item $\cutL \lmmtX(e){y} `a + x {\lmmtX(e'){x} }  \substo \lmmtX(e{ \tsubst[e'/`a] } ){y} $.
 \end {flatenumerate}
 \end {lemma}

 \begin{Proof}{By simultaneous induction on the structure of nets.}{\Lmm \ref{mu simulation}}
By simultaneous induction on the structure of nets.
 \begin{enumerate}
 \item 
 \begin{description}

 \pushmyitem{2mm}[$c = \cell<y|`a>$] && \kern-5mm 
\cutL { \lmmtX(\cell<y|`a>) } `a + x \lmmtX(e'){x} &\ByDef& 
\cutL \caps<y.`a> `a + x \lmmtX(e'){x} &\substo& \bsp (\deactL) 
\lmmtX(e'){x} \tsubst[y/x] &=&
\lmmtX(e'){y} &\ByDef& \\
\lmmtX(\cell<y|e'>) &=& 
\lmmtX(\cell<y|`a>{ \tsubst[e'/`a] }) 
 \end{array}$

 \myitem[$c = \cell<y|`b>, ~`a \not= `b$]
\cutL { \lmmtX(\cell<y|`b>) } `a + x \lmmtX(e'){x} &\ByDef& 
\cutL \caps<y.`b> `a + x \lmmtX(e'){x} &\substo& 
\caps<y.`b> &\substo& \bsp (\Li) 
\lmmtX(\cell<y|`b>) &=& \\
\lmmtX(\cell<y|`b>{ \tsubst[e'/`a] }) 
 \end{array}$

\Comment
{
 \myitem[$c = \cell<t|`a>$]
\cutL \lmmtX(\cell<t|`a>) `a + x \lmmtX(e'){x} &\ByDef& 
\cutL { \lmmtX(t){`a} } `a + x \lmmtX(e'){x} & \substo& \bsp (\IH) \\
\lmmtX(t{ \tsubst[e'/`a] }){`a} &\ByDef& 
\\ 
\\ 
\lmmtX(\cell<t{ \tsubst[e'/`a] }|e'>) &\ByDef& 
\lmmtX(\cell<t|`a>{ \tsubst[e'/`a] }) 
 \end{array}$
}

 \myitem[$c = \cell<`lz.t|`a>$]
\cutL \lmmtX(\cell<`lz.t|`a>) `a + x \lmmtX(e'){x} &\ByDef& 
\cutL {( \lmmtX(L z . t){`a} )} `a + x \lmmtX(e'){x} &\substo& \bsp (\Lii) \\
\cut { \exp z { \cutL {{ \lmmtX(t){`b} } } `a + x \lmmtX(e'){x} } `b . `g } `g + x \lmmtX(e'){x} &\substo& \bsp (\IH,`a\not=`b) 
\cut { \exp z { \lmmtX(t{ \tsubst[e'/`a] }) } `b . `g } `g + x \lmmtX(e'){x} &\ByDef& \\
\lmmtX(\cell<`lz.t{ \tsubst[e'/`a] }|e'>) &\ByDef& 
\lmmtX(\cell<`lz.t|`a>{ \tsubst[e'/`a] }) 
 \end{array}$

 \myitem[$c = \cell<`M`b.c|`a>$]
\cutL \lmmtX(\cell<`M`b.c|`a>) `a + x \lmmtX(e'){x} &\ByDef& 
\cutL \lmmtX(`M`b.c){`a} `a + x \lmmtX(e'){x} &\ByDef& \\
\cutL {\lmmtX(M `b . c){`a} } `a + x \lmmtX(e'){x} &=& 
\cutL { \lmmtX(c\tsubst[`a/`b]) } `a + x \lmmtX(e'){x} &=& \\
\cutL {\omt  \cutL { \lmmtX(c) } `a + x \lmmtX(e'){x} } `b + x \lmmtX(e'){x} &=& 
\cutL {\omt \cutL { \lmmtX(c) } `a + x \lmmtX(e'){x} \tsubst `a/`b } `a + x \lmmtX(e'){x} &\substo& \bsp (\IH) && \\
\cutL { \lmmtX(M `b . {c{ \tsubst[e'/`a] }}){`a} } `a + x \lmmtX(e'){x} &\ByDef& 
\cutL { \lmmtX(`M`b.c{ \tsubst[e'/`a] }){`a} } `a + x \lmmtX(e'){x} &\ByDef& \\
\cutL { \lmmtX({ \cell<`M`b.c{ \tsubst[e'/`a] }|`a>}) } `a + x \lmmtX(e'){x} &\substo&
\Span{5}{ \bsp (\IH)
\lmmtX(\cell<{`M`b.c \tsubst[e'/`a] }|e'>) ~=~ 
\lmmtX(\cell<`M`b.c|`a>{ \tsubst[e'/`a] }) 
} \end{array}$ 

 \myitem[$c = \cell<y|e>$]
\cutL \lmmtX(\cell<y|e>) `a + x \lmmtX(e'){x} &\ByDef& 
\cutL \caps<y.e> `a + x \lmmtX(e'){x} &\ByDef& 
\cutL \lmmtX(e){y} `a + x \lmmtX(e'){x} & \substo& \bsp (\IH) 
\lmmtX(e{ \tsubst[e'/`a] }){y} &\ByDef& \\
\lmmtX(\cell<y|e{ \tsubst[e'/`a] }>) &\ByDef& 
\lmmtX(\cell<y|e>{ \tsubst[e'/`a] }) 
 \end{array}$

\Comment{
 \myitem[$c = \cell<y|t`.e>$] 
\cutL { \lmmtX(\cell<y|t`.e>) } `a + x \lmmtX(e'){x} &\ByDef& 
\cutL \lmmtX(t`.e){y} `a + x \lmmtX(e'){x} &\ByDef& \\ 
\cutL { \lmmtX(X t . e){y}} `a + x \lmmtX(e'){x} &\ByDef& 
\imp { \cutL \lmmtX(t){`g} `a + x \lmmtX(e'){x} } `g [y] z { \cutL { \lmmtX(e){z} } `a + x \lmmtX(e'){x} } & \substo& \bsp (\IH) \\ 
\imp { \lmmtX(t{ \tsubst[e'/`a] }){`g} } `g [y] z { \lmmtX(e{ \tsubst[e'/`a] }){z} } &\ByDef& 
\lmmtX(\cell<y|t{ \tsubst[e'/`a] }`.e{ \tsubst[e'/`a] }>) &=& \\
\lmmtX(\cell<y|t`.e>{ \tsubst[e'/`a] }) 
 \end{array}$

 \myitem[$c = \cell<y|\mt z.c>$] 
\cutL \lmmtX(\cell<y|\mt z.c>) `a + x \lmmtX(e'){x} &\ByDef& 
\cutL \lmmtX(\mt z.c){y} `a + x \lmmtX(e'){x} &\ByDef& 
\cutL { \lmmtX(T z . c){y}} `a + x \lmmtX(e'){x} &\substo& \\ 
\cutL { \lmmtX(c) } `a + x \lmmtX(e'){x} \tsubst[y/z] & \substo& \bsp (\IH) 
\lmmtX(c{ \tsubst[e'/`a] }) \tsubst z/y &\ByDef& 
\lmmtX(\cell<y|\mt z.c\tsubst[e'/`a]>) &=& \\
\lmmtX(\cell<y|\mt z.c>{ \tsubst[e'/`a] }) 
 \end{array}$
}

 \myitem[$c = \cell<t|e>$]
\cutL { \lmmtX(\cell<t|e>) } `a + x \lmmtX(e'){x} &\ByDef& 
\cutL { \lmmtX(C |`b,y| <t|e>) } `a + x \lmmtX(e'){x} &\substo& \bsp (\Lv) \\
\cut { \cutL { \lmmtX(t){`b} } `a + x \lmmtX(e'){x} } 
	`b + y { \cutL \lmmtX(e){y} `a + x \lmmtX(e'){x} } & \substo& \bsp (\IH,`a\not=`b) 
\lmmtX(C |`b,y| <t{ \tsubst[e'/`a] }|e{ \tsubst[e'/`a] }>) &\ByDef& \\
\lmmtX(\cell<t{ \tsubst[e'/`a] }|e{ \tsubst[e'/`a] }>) &\ByDef& %
\lmmtX(\cell<t|e>{ \tsubst[e'/`a] }) 
 \end{array}$
 \end{description}

  \item 
 \begin{description}
 \pushmyitem{2mm}[$t = y$]
\cutL { \lmmtX(y){`b} } `a + x \lmmtX(e'){x} &\ByDef& 
\cutL { \lmmtX(V y){`b} } `a + x \lmmtX(e'){x} &\substo& \bsp (\Li) 
\lmmtX(V y){`b} &\ByDef& 
\lmmtX(y){`b} &\ByDef&
\lmmtX({y \tsubst[e'/`a] }){`b}
 \end{array}$
 
 \myitem[$t = `Ly.t$]
\cutL { \lmmtX(`Ly.t){`b} } `a + x \lmmtX(e'){x} &\ByDef& 
\cutL {( \lmmtX(L |`d| y . t){`b} )} `a + x \lmmtX(e'){x} &\substo& \bsp (\Lii) \\ 
\exp y { \cutL { \lmmtX(t){`d} } `a + x \lmmtX(e'){x} } `d . `b & \substo& \bsp (\IH, `a \not = `d) 
\lmmtX(L |`d| y . {t\tsubst[e'/`a] }){`b} &\ByDef& 
\lmmtX(`Ly.t{\tsubst[e'/`a]}){`b} &\ByDef& 
\lmmtX({}{(`Ly.t) \tsubst[e'/`a] }){`b}
 \end{array}$

 \myitem[$t = `M`g.c$]
\cutL { \lmmtX(`M`g.c){`b} } `a + x \lmmtX(e'){x} &\ByDef& 
\cutL { \lmmtX(M `g . c){`b} } `a + x \lmmtX(e'){x} & \substo& \bsp (\IH, `a \not = `b,`g) 
\lmmtX(c{ \tsubst[e'/`a] }) \tsubst[`b/`g] &\ByDef& \\
\lmmtX(`M`g.c{ \tsubst[e'/`a] }){`b} &\ByDef& 
\lmmtX({(`M`g.c) \tsubst[e'/`a] }){`b}
 \end{array}$

 \end{description}

 \item 

 \begin{description}

 \pushmyitem{2mm}[$e = `a$]
\cutL { \lmmtX(`a){y} } `a + x \lmmtX(e'){x} &\ByDef& 
\cutL \caps<y.`a> `a + x \lmmtX(e'){x} &\substo& \bsp (\deactL) 
\lmmtX(e'){x} \tsubst[y/x] &\ByDef& 
\lmmtX(e'){y} &\ByDef & 
\lmmtX(`a{ \tsubst[e'/`a] }){y}
 \end{array}$
 
 \myitem[$e = `b \not= `a$]
\cutL { \lmmtX(`b){y} } `a + x \lmmtX(e'){x} &=& 
\cutL \caps<y.`b> `a + x \lmmtX(e'){x} &\substo& \bsp (\Li) 
\caps<y.`b> &\ByDef& 
\lmmtX(`b){y} &\ByDef& 
\lmmtX(`b{ \tsubst[e'/`a] }){y}
 \end{array}$
 
 \myitem[$e = t`. e''$]
\cutL { \lmmtX(t`.e''){y} } `a + x \lmmtX(e'){x} &\ByDef& 
\cutL { \lmmtX(X t . e''){y} } `a + x \lmmtX(e'){x} ~ \substo (\Liv) \\ 
\imp { \cutL \lmmtX(t){`g} `a + x \lmmtX(e'){x} } `g [y] z { \cutL { \lmmtX(e''){z} } `a + x \lmmtX(e'){x} } & \substo& \bsp (\IH, `a \not= `g) 
\imp { \lmmtX(t{ \tsubst[e'/`a] }){`g} } `g [y] z { \lmmtX(e''{ \tsubst[e'/`a] }){z} } &\ByDef& \\
\lmmtX(X {t{ \tsubst[e'/`a] }} . {e''{ \tsubst[e'/`a] }}){y} &\ByDef& 
\lmmtX({}{(t`.e'') \tsubst[e'/`a] }){y} 
 \end{array}$
 
 \myitem[$e = \mt y.c$]
\cutL { \lmmtX(\mt z.c){y} } `a + x \lmmtX(e'){x} &\ByDef& 
\cutL { \lmmtX(T z . c){y} } `a + x \lmmtX(e'){x} &\ByDef& 
\lmmtX(\mt z.c{ \tsubst[e'/`a] }){y} &\ByDef& 
\lmmtX({}{(\mt z.c) \tsubst[e'/`a] }){y}
 \end{array}$
 \qed 

 \end{description}
 \end{enumerate}
 \observe
Notice that no reduction steps are used in this proof.
 \end{Proof}

Likewise, we can show that $`X$ successfully encodes $\lmmt$'s term substitution $ `. \tsubst t/x $ through left substitution $\cutR { \lmmtX(t){`a} } `a + x `. $.

 \begin {lemma} \label {mut simulation}
 \begin {flatenumerate}
 \firstitem $\cutR \lmmtX(t'){`a} `a + x { \lmmtX(c) } = \lmmtX(c\tsubst[t'/x]) $;
 \item $\cutR \lmmtX(t'){`a} `a + x { \lmmtX(t){`b} } = \lmmtX(t\tsubst[t'/x]){`b} $; and
 \item If $z \not= x$, then $\cutR \lmmtX(t'){`a} `a + x { \lmmtX(e){z} } = \lmmtX(e\tsubst[t'/x]){z} $.
 \end {flatenumerate}
 \end {lemma}

 \begin{Proof}{By simultaneous induction on the structure of nets.}{\Lmm \ref{mut simulation}} 
By simultaneous induction on the structure of nets.
 \begin{enumerate}
 \item 
 \begin{description}

 \pushmyitem{2mm}[$c = \cell<y|`a>$] \kern-8mm && 
\cutR \lmmtX(t'){`a} `a + x \lmmtX(\cell<x|`b>) &\ByDef& 
\cutR \lmmtX(t'){`a} `a + x \caps<x.`b> &\substo& \bsp (\deactR) 
\lmmtX(t'){`a} \tsubst[`b/`a] &=&
\lmmtX(t'){`b} &=& \\
\lmmtX(\cell<t'|`b>) &=& 
\lmmtX(\cell<x|`b>\tsubst[t'/x]) 
 \end{array}$

 \myitem[$c = \cell<y|`a>$, $y \not= x$] \kern-3mm && 
\cutR \lmmtX(t'){`a} `a + x \lmmtX(\cell<y|`b>) &\ByDef& 
\cutR \lmmtX(t'){`a} `a + x \caps<y.`b> &\substo& \bsp (\Ri) 
\caps<y.`b> &=& 
\lmmtX(\cell<y|`b>) &=& \\
\lmmtX(\cell<y|`b>\tsubst[t'/x]) 
 \end{array}$

 \myitem[$c = \cell<t|`a>$]
\cutR \lmmtX(t'){`a} `a + x \lmmtX(\cell<t|`b>) &\ByDef& 
\cutR \lmmtX(t'){`a} `a + x \lmmtX(t){`b} & \substo& \bsp (\IH) 
\lmmtX(t\tsubst[t'/x]){`b} &\ByDef& 
\lmmtX(\cell<t\tsubst[t'/x]|`b>) &\ByDef& \\
\lmmtX(\cell<t|`b>\tsubst[t'/x]) 
 \end{array}$

\myitem[$c = \cell<y|e>$]
\cutR \lmmtX(t'){`a} `a + x { \lmmtX(\cell<y|e>) } &\ByDef& 
\cutR \lmmtX(t'){`a} `a + x \caps<y.e> &\ByDef& 
\cutR \lmmtX(t'){`a} `a + x \lmmtX(e){y} & \substo& \bsp (\IH) 
\lmmtX(e\tsubst[t'/x]){y} &\ByDef& \\
\lmmtX(\cell<y|e\tsubst[t'/x]>) &\ByDef& 
\lmmtX(\cell<y|e>\tsubst[t'/x]) 
 \end{array}$

 \myitem[$c = \cell<y|t`.e>$] 
\cutR \lmmtX(t'){`a} `a + x \lmmtX(\cell<y|t`.e>) &\ByDef& 
\cutR \lmmtX(t'){`a} `a + x \lmmtX(t`.e){y} &\ByDef& \\ 
\cutR \lmmtX(t'){`a} `a + x { \lmmtX(X t . e){y} } &\substo& 
\imp { \cutR \lmmtX(t'){`a} `a + x \lmmtX(t){`g} } `g [y] z { \cutR \lmmtX(t'){`a} `a + x { \lmmtX(e){z} } } & \substo& \bsp (\IH) \\ 
\imp { \lmmtX(t\tsubst[t'/x]){`g} } `g [y] z { \lmmtX(e\tsubst[t'/x]){z} } &\ByDef& 
\lmmtX(\cell<y|t\tsubst[t'/x]`.e\tsubst[t'/x]>) ~=~ 
\lmmtX(\cell<y|t`.e>\tsubst[t'/x]) 
 \end{array}$

 \myitem[$c = \cell<x|\mt z.c>$] 
\cutR \lmmtX(t'){`a} `a + x { \lmmtX(\cell<x|\mt z.c>) } &\ByDef& 
\cutR \lmmtX(t'){`a} `a + x { \lmmtX(\mt z.c){x} } &\ByDef& \\ 
\cutR \lmmtX(t'){`a} `a + x {(\lmmtX(T z . c){x} )} &=& 
\cutR \lmmtX(t'){`a} `a + z {\omt \cutR \lmmtX(t'){`a} `a + x \lmmtX(c) } & \substo& \bsp (\IH) \\ 
\cutR \lmmtX(t'){`a} `a + x {( \lmmtX(c\tsubst[t'/x]) \tsubst x/z )} &\ByDef& 
\cutR \lmmtX(t'){`a} `a + x { \lmmtX(\mt z.{c\tsubst[t'/x] }){x} } &\ByDef& \\ 
\cutR \lmmtX(t'){`a} `a + x { \lmmtX(\cell<x|\mt z.{c\tsubst[t'/x] }>) } & \substo& 
\Span{4}{ \bsp (\IH) 
\lmmtX(\cell<t'|\mt z.c\tsubst[t'/x]>) ~=~ 
\lmmtX(\cell<x|\mt z.c>\tsubst[t'/x]) 
} \end{array}$

 \myitem[$c = \cell<y|\mt z.c>$, $y \not= x$] && \kern-8mm 
\cutR \lmmtX(t'){`a} `a + x { \lmmtX(\cell<y|\mt z.c>) } &\ByDef& 
\cutR \lmmtX(t'){`a} `a + x { \lmmtX(\mt z.c){y} } &\ByDef& \\ 
\cutR \lmmtX(t'){`a} `a + x { \lmmtX(T z . c){y} } &=& 
\cutR \lmmtX(t'){`a} `a + x { \lmmtX(c) \tsubst[y/z] } & \substo& \bsp (\IH) 
\lmmtX(c\tsubst[t'/x]) \tsubst[z/y]&\ByDef& \\ 
\lmmtX(\cell<y|\mt z.c\tsubst[t'/x]>) &=& 
\lmmtX(\cell<y|\mt z.c>\tsubst[t'/x]) 
 \end{array}$

\Comment{

 \myitem[$c = \cell<t|e>$]
\cutR \lmmtX(t'){`a} `a + x \lmmtX(\cell<t|e>) &\ByDef& 
\cutR \lmmtX(t'){`a} `a + x { \lmmtX(C <R `b,y t|e>) } &=& \bsp (\Rv) \\ 
\Span{3}{
	\cut { \cutR \lmmtX(t'){`a} `a + x { \lmmtX(t){`b} } } `b + y { \cutR \lmmtX(t'){`a} `a + x \lmmtX(e){y} } 
}	& \substo& \bsp (\IH, x \not= y) \\ 
\cut \lmmtX(t\tsubst[t'/x]){`b} `b + y \lmmtX(e\tsubst[t'/x]){y} &=& 
\lmmtX(\cell<t\tsubst[t'/x]|e\tsubst[t'/x]>) &\ByDef& 
\lmmtX(\cell<t|e>\tsubst[t'/x]) 
 \end{array}$
}

 \end{description}

 \item

 \begin{description}
 \pushmyitem{2mm}[$t = x$] 
\cutR \lmmtX(t'){`a} `a + x { \lmmtX(x){`b} } &\ByDef& 
\cutR \lmmtX(t'){`a} `a + x { \lmmtX(V x){`b} } &=& \bsp (\deactR) 
\lmmtX(t'){`a} \tsubst[`b/`a] &=& 
\lmmtX(t'){`b} &\ByDef& 
\lmmtX(x\tsubst[t'/x]){`b}
 \end{array}$

 \myitem[$t = y \not= x$]
\cutR \lmmtX(t'){`a} `a + x { \lmmtX(y){`b} } &\ByDef& 
\cutR \lmmtX(t'){`a} `a + x { \lmmtX(V y){`b} } &\substo& \bsp (\Ri) 
\lmmtX(V y){`b} &\ByDef& 
\lmmtX(y){`b} &\ByDef& 
\lmmtX(y\tsubst[t'/x]){`b}
 \end{array}$

 \myitem[$t = `Ly.t''$]
\cutR \lmmtX(t'){`a} `a + x { \lmmtX(`Ly.t''){`b} } &\ByDef& 
\cutR \lmmtX(t'){`a} `a + x { \lmmtX(L |`d| y . t''){`b} } &\substo& \bsp (\Rii) \\ 
\exp y { \cutR \lmmtX(t'){`a} `a + x { \lmmtX(t''){`d} } } `d . `b & \substo& \bsp (\IH) 
\exp y { \lmmtX(t''\tsubst[t'/x]){`d} } `d . `b &\ByDef& 
\lmmtX(`Ly.t''\tsubst[t'/x]){`b} &\ByDef& \\
\lmmtX((`Ly.t'')\tsubst[t'/x]){`b}
 \end{array}$

 \myitem[$t = `M`g.c$]
\cutR \lmmtX(t'){`a} `a + x { \lmmtX(`M`g.c){`b} } &\ByDef& 
\cutR \lmmtX(t'){`a} `a + x { \lmmtX(c) \tsubst `b/`g } & \substo& \bsp (\IH) 
\lmmtX(c\tsubst[t'/x]) \tsubst[`b/`g] &\ByDef& \\
\lmmtX(`M`g.c\tsubst[t'/x]){`b} &\ByDef& 
\lmmtX({(`M`g.c) \tsubst[t'/x] }){`b}
 \end{array}$
 \end{description}

 \item 

 \begin{description}

 \pushmyitem{2mm}[$e = `b$]
\cutR \lmmtX(t'){`a} `a + x { \lmmtX(`b){z} } &\ByDef& 
\cutR \lmmtX(t'){`a} `a + x \caps<z.`b> &\substo& \bsp (\Ri) 
\caps<z.`b> &\ByDef& 
\lmmtX(z){`b} &\ByDef& 
\lmmtX(z\tsubst[t'/x]){`b}
 \end{array}$
 
 \myitem[$e = t`. e$]
\cutR \lmmtX(t'){`a} `a + x { \lmmtX(t`.e){z} } &\ByDef& 
\Span{3}{
	\cutR \lmmtX(t'){`a} `a + x { \imp { \lmmtX(t){`b} } `b [z] y \lmmtX(e){y} } ~\substo~ \bsp (\Riv)
}	\\ 
\Span{3}{
	\imp { \cutR \lmmtX(t'){`a} `a + x { \lmmtX(t){`b} } } `b [z] y { \cutR \lmmtX(t'){`a} `a + x \lmmtX(e){y} }
}	& \substo& \bsp (\IH, x\not=y) \\ 
\imp { \lmmtX(t\tsubst[t'/x]){`b} } `b [z] y { \lmmtX(e\tsubst[t'/x]){y} } &\ByDef& 
\lmmtX(t\tsubst[t'/x] `.~ e\tsubst[t'/x]){z} &\ByDef& 
\lmmtX({(t`. e) \tsubst[t'/x] }){z} 
 \end{array}$ 
 
 \myitem[$e = \mt y.c$]
\cutR { \lmmtX(t'){x} } `a + x { \lmmtX(\mt y.c){z} } &\ByDef& 
\cutR { \lmmtX(t'){x} } `a + x { \lmmtX(c) \tsubst[z/y] } &\ByDef& 
( \cutR { \lmmtX(t'){x} } `a + x { \lmmtX(c) } ) \tsubst[z/y] & \substo& \bsp (\IH) \\
\lmmtX(c\tsubst[t'/x]) \tsubst[z/y] &\ByDef& 
\lmmtX(\mt y.c\tsubst[t'/x]){z} &\ByDef& 
\lmmtX({(\mt y.c) \tsubst[t'/x] }){z}
 \end{array}$
 \arrayqed 

 \end{description}
 \end{enumerate}

 \observe
As above, no reduction steps are used in this proof. 

 \end{Proof}

We now strengthen these results by stating that this encoding preserves reduction:

 \begin {theorem}[Simulation of $\redlmmt$]\label{simul:lmmt}
 \begin{flatenumerate}
 \firstitem If $c\redlmmt c' $ then $\lmmtX(c) \tcredXs \lmmtX(c') $; 
 \item If $t\redlmmt t' $ then $\lmmtX(t){`a} \tcredXs \lmmtX(t'){`a} $; and 
 \item If $e\redlmmt e' $ then $\lmmtX(e){x} \tcredXs \lmmtX(e'){x} $.
 \end{flatenumerate}
 \end {theorem}

 \begin{proof}
Simultaneously by induction of the definition of $\redlmmt$.

 \begin{description}

 \myitem [$`l$] 
\lmmtX(\cell<`Lx.t_1|t_2`.e>) &\ByDef& 
\lmmtX(C <{<L y . t_1>}|{X |`g,z| {t_2} . e}>) &\redX& (\Ins) \\
\cut { \lmmtX(t_2){`g} } `g + y {\cut { \lmmtX(t_1){`b} } `b + z { \lmmtX(e){z} }  } & \ByDef &
\lmmtX(C |`g,y| <t_2|{T x . {C |`b,z| <t_1|e>}}>) & \ByDef & \\
\lmmtX(C |`g,y| <t_2|{ \mt x.<t_1|e>}>) & \ByDef & 
\lmmtX(\cell<t_2|{ \mt x.<t_1|e>}>) 
 \end{array} $


 \myitem [$`m$] 
\lmmtX(\cell<`M`b.c|e>) & \ByDef &
\lmmtX(C <`M`b.c|e>) & \ByDef & 
\lmmtX(C <{M `b . c}|e>) & =_{`a} & \bsp(`a \textit{ fresh}) \\
\lmmtX(C |`b,x| <c|e>) & \redX & \bsp (\actL) 
\cutL { \lmmtX(c) } `b + x {\lmmtX(e){x} } & = & \bsp (\ref {mu simulation}) 
\lmmtX(c{ \tsubst[e/`b] }) 
 \end{array} $


 \myitem [$\MuT$] 
\lmmtX(\cell<t|\mt y.c>) & \ByDef &
\lmmtX(C <t|\mt y.c>) & \ByDef & 
\lmmtX(C <t|{T y . c}>) & =_{`a} & \bsp(x \textit{ fresh}) \\
\cut {\lmmtX(t){`a} } `a + y { \lmmtX(c) } & \redX & \bsp (\actR)
\cutR {\lmmtX(t){`a} } `a + y { \lmmtX(c) } & = & \bsp (\ref {mut simulation}) 
\lmmtX(c{ \tsubst[t/y] }) 
 \end{array} $



 \myitem [$`h`m$]
\lmmtX(`M`a.<t|`a>){`b} &\ByDef&
\lmmtX(M `a . {\cell <t|`a>}){`b} &\ByDef& 
\lmmtX(\cell<t|`b>) &\ByDef& 
\lmmtX(t){`b} 
  \end{array} $

 \myitem [$`h\mutilde$]
\lmmtX(\mt x.<x|e>){y} &\ByDef&
\lmmtX(T x . {\cell<x|e>}){y} &\ByDef& 
\lmmtX(\cell<y|e>) &\ByDef& 
\lmmtX(e){y} 
 \end{array} $

\Comment{

 \myitem[$ t \redlmmt t' \Implies `Lx.t \redlmmt `Lx.t' $]
\lmmtX(`Lx.t){`a} &\ByDef& 
\lmmtX(L x.t){`a} &\rtcredXs& \bsp (\IH) 
\lmmtX(L x.t'){`a} &\ByDef& 
\lmmtX(`Lx.t'){`a} 
 \end{array} $


 \myitem[$ t \redlmmt t' \Implies \cell<t|e> \redlmmt \cell<t'|e> $]
\lmmtX(\cell<t|e>) &\ByDef& 
\lmmtX(C <t|e>) &\rtcredXs& \bsp (\IH) 
\lmmtX(C <t'|e>) &\ByDef& 
\lmmtX(\cell<t'|e>) 
 \end{array} $
 

 \myitem[$ t \redlmmt t' \Implies t`.e \redlmmt t'`.e $]
\lmmtX(t`.e){x} &\ByDef& 
\lmmtX(X t . e){x} &\rtcredXs& \bsp (\IH) 
\lmmtX(X t' . e){x} &\ByDef& 
\lmmtX(t'`.e){x} 
 \end{array} $


 \myitem[$ e \redlmmt e' \Implies \cell<t|e> \redlmmt \cell<t|e'> $]
\lmmtX(\cell<t|e>) &\ByDef& 
\lmmtX(C <t|e>){x} &\rtcredXs& \bsp (\IH) 
\lmmtX(C <t|e'>){x} &\ByDef& 
\lmmtX(\cell<t|e'>) 
 \end{array} $


 \myitem[$ e \redlmmt e' \Implies t`.e \redlmmt t`.e' $]
\lmmtX(t`.e){x} &\ByDef& 
\lmmtX(X t . e){x} &\rtcredXs& \bsp (\IH) 
\lmmtX(X t . e'){x} &\ByDef& 
\lmmtX(t`.e'){x} 
 \end{array} $


 \myitem[$ c \redlmmt c' \Implies `M`a.c \redlmmt `M`a.c' $]
\lmmtX(`M`a.c){`b} &\ByDef& 
\lmmtX(M `a . c){`b} &\rtcredXs& \bsp (\IH) 
\lmmtX(M `a . c'){`b} &\ByDef& 
\lmmtX(`M`a.c'){`b} 
 \end{array} $


 \myitem[$ c \redlmmt c' \Implies \mt x.c \redlmmt \mt x.c' $]
\lmmtX(\mt y.c){x} &\ByDef& 
\lmmtX(T y . c){x} &\rtcredXs& \bsp (\IH) 
\lmmtX(T y . c'){x} &\ByDef& 
\lmmtX(\mt y.c'){x} 
 \end{array} $
 \arrayqed[-4\point] 


}

 \end{description}

 \observe
The contextual rules all follow by straightforward induction. \QED

 \end{proof}

\newqsymbol{`S}{\textbf{S}}

We cannot model the $`h$-reduction rule $`L x . `M`b.< t | x`.`b > \redX t ~ (x,`b \notele \FV{t})$, since the `surrounding' $`l$-abstraction produces an export term that can never be removed; $`X$ itself is not extensional.

\Comment{
 $ \begin{array}{lcl}
\lmmtX(`L z . `M`s.<t|z`.`s>){`g} &=& \\
\lmmtX(L |`d| z . {`M`s.<t|z`.`s>}){`g} &=& \\
\lmmtX(L |`d| z . {M `s . { \cell <t|z`.`s>}}){`g} &=& \\
\lmmtX(L |`d| z . {M `s . {<t|z`.`s>}}){`g} &=& \\
\lmmtX(L |`d| z . {M `s . {<t|{X z . `s}>}}){`g} &=& \\
\lmmtX(L |`d| z . {M `s . {<t|{X {V z} . {N `s}}>}}){`g} &=& \\
\lmmtX(L |`s| z . {<t|{X {V z} . {N `d}}>}){`g} &=& \\
 \end{array} $}

We can also show that the {\CBN} strategy is respected.
We need to check that $(\actLN)$ gets correctly applied, and no reduction takes place in the environment.

 \begin {theorem}[Simulation of $\redlmmtN$] \label{simul:lmmtN}
If $c \redlmmtN c' $ then $\lmmtX(c) \tcredXsN \lmmtX(c') $; and 
if $t \redlmmtN t' $ then $\lmmtX(t){`a} \tcredXsN \lmmtX(t'){`a} $.
 \end {theorem}

 \begin{proof}
Simultaneously by induction of the definition of reduction.
The proof is mostly as that of \Thm\ref{simul:lmmt}, which showed this result for full reduction; we will highlight the differences.

 \begin{description}

 \item [$`l$]
No {\CBN} issues.

 \myitem [$`m$] 
\lmmtX(\cell<`M`b.c|`S>) & \ByDef &
\lmmtX(C <`M`b.c|`S>) & \ByDef & 
\lmmtX(C <{M `b . c}|`S>) & = & \bsp (`a \textit{ fresh})
\cut { \lmmtX(c) } `b + x { \lmmtX(`S){x} } 
 \end{array} $

\noindent
Notice that $\lmmtX(`S){x} $ introduces $x$; also, $ c = \cell<t|e> $, so $\lmmtX(c){`b} $ is a cut and does not introduce $`b$, and we have
 \[ \begin{array}{@{}lclclc}
\lmmtX(C |`b,x| <c|`S>) & \redX & \bsp (\actLN) 
\cutL { \lmmtX(c) } `b + x {\lmmtX(`S){x} } & = & \bsp (\ref {mu simulation}) 
\lmmtX(c{ \tsubst[S/`b] }) 
 \end{array} \]
Since $\lmmtX(`S){x} $ introduces $x$, $(\actLN)$ is permitted.

 \item [$\MuT$] No {\CBN} issues.

 \item [$`h`m$]
The $(\actL)$ step in the proof of \Thm\ref{simul:lmmt} becomes $(\actLN)$ since $\caps<z.`b>$ introduces $z$.

 \item [$`h\MuT$] No {\CBN} issues.




 \item[$ c \redlmmtN c' \Implies `M`a.c \redlmmtN `M`a.c' $]
Reduction inside a cut is allowed in \CBN.


 \end{description}
The other rules are not part of $\redlmmtN$.
\QED

 \end{proof}

We can also show that the {\CBV} strategy is respected.
Now we need to check that now $(\actRV)$ gets correctly applied, and again no reduction takes place in the environment.

 \begin {theorem}[Simulation of $\redlmmtV$]\label{simul:lmmtV}
If $c \redlmmtV c' $ then $\lmmtX(c) \tcredXsV \lmmtX(c') $, and
if $ t \redlmmtV t' $ then $\lmmtX(t){`a} \tcredXsV \lmmtX(c'){`a} $.
 \end {theorem}

 \begin{proof}
Simultaneously by induction of the definition of reduction.
The proof is mostly as that for \Thm\ref{simul:lmmt}; we will highlight the differences.

 \begin{description}

 \item [$`l$]
No {\CBV} issues.

 \item [$`m$] 
No {\CBV} issues.

\myitem [$\MuT$] 
\lmmtX(\cell<V|\mt y.c>) & \ByDef &
\lmmtX(C <{}V|\mt y.c>) & \ByDef &
\lmmtX(C <{}V|{T y . c}>) & = & \bsp (x \textit{ fresh}) 
\cut {\lmmtX({}V){`a} } `a + y { \lmmtX(c) } 
 \end{array} $

\noindent
Notice that $\lmmtX({}V){`a} $ introduces $`a$;
also, $ c = \cell<t|e> $, so $\lmmtX(c)$ is a cut that does not introduce $y$, and we have%
 \[ \begin{array}{lclclc}
\cut { \lmmtX({}V){`a} } `a + y { \lmmtX(c) } & \redX & \bsp (\actRV) 
\cutR { \lmmtX({}V){`a} } `a + y { \lmmtX(c) } & = & \bsp (\ref {mut simulation}) 
\lmmtX(c{ \tsubst[V/y] }) 
 \end{array} \]

 \item [$`h`m$]
No {\CBV} issues.

 \item [$`h\MuT$]
The $(\actR)$ step in the proof of \Thm\ref{simul:lmmt} becomes $(\actRV)$ since $\caps<y.`a>$ introduces $`a$.

 \item[$ c \redlmmt c' \Implies `M`a.c \redlmmt `M`a.c' $]
Allowed in \CBV.

 \myitem[$`M`b.<t|`b> \redlmmtV t$, $ `b \notele \FV{t} $] 
\lmmtX(`M`b.<t|`b>){`a} & \ByDef & 
\lmmtX(M `b . {\cell<t|`b>}){`a} & \ByDef & 
\lmmtX(t){`b} \tsubst[`a/`b] &\ByDef&  
\lmmtX(t){`a} 
 \end{array} $

 \end{description}
The other rules are not part of $\redlmmtV$.
\QED

  \end{proof}

So $\lmmtX(`.) $, the natural encoding of $\lmmt$ in $\X$ (and $`X$), that names the implicit term and context variables, strongly connects the two calculi: not only the standard $\lmmt$ reduction is embedded into the reduction of $`X$, but also the {\CBN} and {\CBV} reduction strategies of $\lmmt$ are respected by their $`X$ counterparts.

In a similar way, we can show that the sub-reduction systems $\redlmmtn$ and $\redlmmtv$ can be faithfully simulated.

 \begin{corollary}
 \begin{enumerate}
 \firstitem If $c \redlmmtn c' $ then $\lmmtX(c) \rtcredXsn \lmmtX(c') $, and 
if $t \redlmmtn t' $ then $\lmmtX(t){`a} \rtcredXsn \lmmtX(t'){`a} $.
 \item If $c \redlmmtv c' $ then $\lmmtX(c) \rtcredXsv \lmmtX(c') $, and
if $ t \redlmmtv t' $ then $\lmmtX(t){`a} \rtcredXsv \lmmtX(c'){`a} $.
 \end{enumerate}
 \end{corollary}
 


 \def \LKmm{\textsf{LK}_{`m\mutilde}}
 \def \lmmtE {\lmmt{}^{\kern-.5\point \mbox{\scriptsize \textsc{e}}}}
 \def \eqetapar(#1){\mathrel{{=_{`h}}\,(#1)}}
 \def \eqeta {\ifnextchar(
 	{\eqetapar}{=_{`h}}}
 \def \etamu {`h`m}
 \def \etamt {`h\MuT}
 
 \def \eqmmtpar(#1){\mathrel{{=_{`m\MuT}}\,(#1)}}
 \def \eqmmt {\ifnextchar(
 	{\eqmmtpar}{=_{`m\MuT}}}
 
\section{Embedding {\XiCalcPDF} in {\lmmtPDF}}
 \label{sec:EmbedXiLMMT}
 \label {Xi in lmmt}

In this section we will study the reverse of the previous section, and investigate if the natural interpretation of $\X$ into $\lmmt$, as first suggested in \cite{Curien-Herbelin'00}, respects the notions of reduction we focus on in this paper.
 
 \begin{definition}[Translation of $\X$ into $\lmmt$ \cite{Curien-Herbelin'00,Lengrand'03}] 
The interpretation of terms of $\X$ into commands of $\lmmt$ is defined by:
 \[ \begin{array}{rcl}
\Xlmmt[ \caps <x,`a> ] & \ByDef & \Xlmmt[<x.`a>] \\
\Xlmmt[ \exp x P `a . `b ] & \ByDef & \Xlmmt[e x P `a . `b ] \\
\Xlmmt[{ \imp P `a [y] x Q }] & \ByDef & \Xlmmt[i P `a y x Q ] \\
\Xlmmt[ \cut P `a + x Q ] & \ByDef & \Xlmmt[c P `a + x Q ]
 \end{array} \]
 \end{definition}

In fact, this is the origin of $\X$: in Remark 4.1 of \cite{Curien-Herbelin'00}, Curien and Herbelin give a hint on a way to connect $\LKmm$ (as presented there) and {\LK}.
The proofs of {\LK} embed in $\LKmm$ by considering the following sub-syntax of $\lmmt$:
 \[ \begin{array}{rcl}
c &::=& \cell<x | `a> \mid \cell < `l x.`m `a.c | `b > \mid \cell <y| `m `a.c `. \mt x.c > \mid \cell < `m `a.c | \mt x.c >
 \end{array} \]
Later it was discovered that this corresponded closely to Urban's idea in \cite{Urban'00}; however, since Urban focussed on normalisation, the approaches to the definition of reduction differ.

The interpretations from $\lmmt$ to $\Xs$ and back are strongly related, as that they act as each other's inverse, albeit with some reductions involved, as can be expected.
These mainly deal with converting implicit names to explicit names, as discussed above.

First we look at $ \lmmt \mapsto `X \mapsto \lmmt $, and show that $\Xlmmt[`.]$ is $ \lmmtX (`.) $'s left-inverse up to extensionality.
We will write $\eqeta$ for the equivalence relation on $\lmmt$ terms generated by rules $(\etamu)$ and $(\etamt)$.
 
 \begin{theorem} \label{left inverse}
 \begin{flatenumerate}
 \item $ \Xlmmt[ \lmmtX (c) ] \eqeta c $, 
 \item $ `M`a.\Xlmmt[ \lmmtX (t){`a} ] \eqeta t $, and 
 \item $ \mt x.\Xlmmt[ \lmmtX (e){x} ] \eqeta e $.
 \end{flatenumerate}
 \end{theorem}
 \begin{proof}
Simultaneous by induction\Shorter{; we only show the interesting cases}.

 \begin{enumerate} \itemsep 2pt
 \item
 \begin{description}

\Longer{
 \pushitem[$ c = \cell<y|`b> $]
$ \Xlmmt[ \lmmtX (\cell<y|`b>) ] \ByDef \Xlmmt[\caps<y.`b>] \ByDef \cell<y|`b> $

 \item[$ c = \cell<`ly.t|`b> $]
$ \Xlmmt[ \lmmtX (\cell<`ly.t|`b>) ] \ByDef 
\Xlmmt[ \lmmtX (`ly.t){`b} ] \ByDef 
\cell<`ly.`m`g.\Xlmmt[\lmmtX (t){`g} ]|`b> \eqeta (\IH) 
\cell<`ly.t|`b> $
}

 \Longer{\item}
 \Shorter{\pushitem}[$ c = \cell<`m`a.c'|`b> $]
$ \Xlmmt[ \lmmtX (\cell<`m`a.c'|`b>) ] \ByDef 
\Xlmmt[ \lmmtX (`m`a.c'){`b} ] \ByDef 
\Xlmmt[\lmmtX (M `a . c'){`b} ] \ByDef 
\Xlmmt[\lmmtX (c') ] \tsubst `b/`a \eqeta (\IH) \\
c' \tsubst `b/`a \cuder(\etamu) 
\cell<`m`a.c'|`b> $

\Longer{
 \item[$ c = \cell<y|t`.e> $]
$ \Xlmmt[ \lmmtX (\cell<y|t`.e>) ] \ByDef 
\Xlmmt[\lmmtX (t`.e){y} ] \ByDef 
\Xlmmt[\lmmtX (X |`a,z| t . e){y} ] \ByDef \\
\cell< y | `M `a . \Xlmmt[ \lmmtX (t){`a} ] `. \mt z.\Xlmmt[\lmmtX (e){z} ] > \eqeta (\IH) ~
\cell<y|t`.e> $
}

 \item[$ c = \cell<y|\mt z.c'> $]
$ \Xlmmt[ \lmmtX (\cell<y|\mt z.c'>) ] \ByDef 
\Xlmmt[\lmmtX (\mt z.c'){y} ] \ByDef 
\Xlmmt[\lmmtX (T z . c'){y} ] \ByDef 
\Xlmmt[\lmmtX (c') ] \tsubst y/z \eqeta (\IH) \\
c' \tsubst y/z \cuder(\etamt) 
\cell<y|\mt z.c'> $

\Longer{
 \item[$ c = \cell<t|e> $]
$ \Xlmmt[ \lmmtX (\cell<t|e>) ] \ByDef \Xlmmt[ \lmmtX (C <t|e>) ] \ByDef \cell < `M `a . \Xlmmt [ \lmmtX (t){`a} ] | \mt x . \Xlmmt [\lmmtX (e){x} ] > \eqeta (\IH) \cell<t|e> $
}

 \end{description}

 \item
 \begin{description}

 \pushitem[$ t = x $]
$ `M`a.\Xlmmt[ \lmmtX (x){`a} ] \ByDef `M`a.\Xlmmt[{ \lmmtX (V x){`a} }] \ByDef `M`a.<x|`a> \redlmmt (\etamu) x $

 \item[$ t = `Lx.t' $]
$ `M`a.\Xlmmt[ \lmmtX (`Lx.t'){`a} ] \ByDef 
`M`a.\Xlmmt[ \lmmtX (L x . t'){`a} ] \ByDef 
`M`a.\Xlmmt[e x {\lmmtX (t'){`b}} `b . `a ] \eqeta (\IH) \\
`M`a.< `Lx.t' | `a > \redlmmt (\etamu) `Lx.t' $.

\Longer{
 \item[$ t = `M`b.c $]
$ `M`a.\Xlmmt[ \lmmtX (`M`b.c){`a} ] \ByDef 
`M`a.\Xlmmt[ \lmmtX (M `b . c){`a} ] \ByDef
`M`a.\Xlmmt[ \lmmtX ({c\tsubst[`b/`a]}) ] \eqeta (\IH) \\
`M`a.c \tsubst `b/`a =_{`a}
`M`b.c $
}

 \end{description}

 \item
 \begin{description}

 \pushitem[$ e = `a $]
$ \mt x.\Xlmmt[ \lmmtX (`a){x} ] \ByDef 
\mt x.\Xlmmt[ \lmmtX (N `a){x} ] \ByDef 
\mt x.<x|`a> \redlmmt (\etamt) `a $

 \item[$ e = t`.e' $]
$ \mt x.\Xlmmt[ \lmmtX (t`.e'){x} ] \ByDef 
\mt x.\Xlmmt[ \lmmtX (X t . e'){x} ] \ByDef 
\mt x.\Xlmmt[i {\lmmtX (t){`g} } `g x z { \lmmtX (e'){z} }] \eqeta (\IH) \\
\mt x.<x|t`.e'> 
\redlmmt (\etamt) ~ t`.e' $
\Longer{

 \item[$ e = \mt y.c $]
$ \mt x.\Xlmmt[ \lmmtX (\mt y.c){x} ] \ByDef 
\mt x.\Xlmmt[ \lmmtX (T y . c){x} ] \ByDef
\mt x.\Xlmmt[ \lmmtX ({c\tsubst[x/y]}) ] \eqeta(\IH) 
\mt x.c \tsubst x/y =_{`a} 
\mt y.c $
}
\QED

 \end{description}
 \end{enumerate}

 \observe
Notice that the only reduction steps needed here are $(\etamu)$ and $(\etamt)$, in both directions, so the compositions of encodings gives identity modulo $`h$-reduction, \emph{i.e.} extensional equality. 

 \end{proof}

\Comment{
 Of course this result does not come as a surprise, as 
 \[ \begin{array}{rclclclcl}
`M`a.\Xlmmt[{ \lmmtX (t){`a} }] &\ByDef& 
`M`a.\Xlmmt[{ \lmmtX (\cell<t|`a>) }] &\ByDef& 
`M`a.\Xlmmt[{ \caps<t.`a> }] &\ByDef& 
`M`a.<t|`a> &\reduc& t \\
\mt x.\Xlmmt[{ \lmmtX (e){x} }] &\ByDef& 
\mt x.\Xlmmt[{ \lmmtX (\cell<x|e>) }] &\ByDef& 
\mt x.\Xlmmt[{ \caps<x.e> }] &\ByDef& 
\mt x.<x|e> &\reduc& e 
 \end{array} \] 
}

Had we stuck to the interpretation $\lmmtOrg(.)$ as defined in \cite{Lengrand'03}, then reduction would have been involved.

 \begin{example}

 \begin{enumerate}
 \firstitem
$ \lmmtOrg (\Xlmmt[\caps<y,`b>]) = \lmmtOrg (\cell<y|`b>) = 
\lmmtOrg (C <y|`b>) = \cut \caps<y.`a> `a + x \caps<x.`b> \redXs \caps<y.`b> $

 \item
$ \lmmtOrg (\Xlmmt[{\exp y P `g . `b }]) = 
 \lmmtOrg ({}{ \Xlmmt[e y P `g . `b ] }) = 
\cut { \exp y { \lmmtOrg({`M`g.{\Xlmmt [P]} }){`d} } `d . `a } `a + x \caps<x.`b> \redXs 
\exp y { \lmmtOrg(`M`g.{\Xlmmt [P]}){`d} } `d . `b = \\
\exp y { \lmmtOrg (\Xlmmt[P]) } `g . `b \rtcredXs (\IH) ~ 
\exp y P `g . `b $

 \item
$ \lmmtOrg ({ \Xlmmt [{ \imp P `a [y] x Q }] }) = 
\lmmtOrg ({}{ \Xlmmt [i P `a y x Q ] }) = 
\cut \caps<y.b> `b + z { \imp \lmmtOrg (\Xlmmt [P]) `a [z] x \lmmtOrg (\Xlmmt [Q]) } \redXs \\
\imp \lmmtOrg (\Xlmmt [P]) `a [y] x \lmmtOrg (\Xlmmt [Q]) \rtcredXs (\IH) ~
\imp P `a [y] x Q $

 \end{enumerate}

 \observe
Notice that, in these cases, the rules $(\Cap)$, $(\Exp)$, and $(\Imp)$ are used.
 \end{example}

Now we look at $ `X \mapsto \lmmt \mapsto `X $, and show that $\Xlmmt[`.]$ is $ \lmmtX (`.) $'s right-inverse.

\Comment{
 \[ \begin{array}{c}
 \begin {array}[t]{rcl}
\lmmtX (\cell<t|e>) & = & \lmmtX (<t|e>) ~ (`a,x \textit{ fresh}) \\ 
 \end{array} 
 \\ 
 \begin {array}[t]{rcll}
\lmmtX (x){`a} & = & \lmmtX (V x){`a} \\
\lmmtX (`Lx.t){`a} & = & \lmmtX (L x . t){`a} ~ (`b \textit{ fresh}) \\ 
\lmmtX (`M`b.c){`a} & = & \lmmtX (M `b . c){`a} \\ 
 \end{array}
 \dquad 
 \begin {array}[t]{rcll}
\lmmtX (`a){x}	& = & \lmmtX (N `a){x} \\
\lmmtX (t`.e){x} & = & \lmmtX (X t.e){x} ~ (`g,z \textit{ fresh}) \\
\lmmtX (\mt y.c){x} & = & \lmmtX (T y . c){x}
 \end{array}
 \end {array} \]
}

 \begin{theorem} \label{identity}
$ \lmmtX ({ \Xlmmt[P] }) = P $.
 \end{theorem}
 \begin{proof} 
By induction.

 \begin{enumerate}
 \item
$ \lmmtX ({ \Xlmmt[\caps <y,`b>] }) \ByDef 
\lmmtX ({ \Xlmmt[<y.`b>] }) \ByDef 
\caps<y.`b> $

 \item
$ \begin{array}[t]{lclclclcl}
\lmmtX ({ \Xlmmt[{ \exp y P `g . `b }] }) &\ByDef& 
\lmmtX ({}{ \Xlmmt[e y P `g . `b] }) &\ByDef& 
\lmmtX (`ly.{`M`g.\Xlmmt[P]}){`b} &\ByDef& 
\exp y P `g . `b 
 \end{array} $

 \item
$ \begin{array}[t]{lclclcl}
\lmmtX ({ \Xlmmt[{ \imp P `a [y] x Q }] }) &\ByDef& 
\lmmtX (\cell<y|{`M`a.\Xlmmt[P] `. \mt x.\Xlmmt[Q]}>) &\ByDef& 
\lmmtX (X |`g,z| {M `a . \Xlmmt[P]} . {T x . \Xlmmt[Q]}){y} &=_{`a}& \\
\imp \lmmtX(\Xlmmt[P]) `a [y] x \lmmtX(\Xlmmt[Q]) &=& \bsp (\IH) 
\imp P `a [y] x Q 
 \end{array} $

 \item
$ \begin{array}[t]{lclclcl}
\lmmtX ({ \Xlmmt[{ \cut P `a + x Q }] }) ~\ByDef~
\lmmtX ({}{ \Xlmmt[c P `a + x Q ] }) ~\ByDef~ 
\cut { \lmmtX (`m`a.{ \Xlmmt[P] }){`b} } `b + z { \lmmtX (\mt x.{ \Xlmmt[Q] }){z} } ~\ByDef~ \\ 
\cut {\lmmtX ({\Xlmmt[P]\tsubst[`b/`a]})} `b + z {\lmmtX ({\Xlmmt[Q]\tsubst[z/x]})} ~=_{`a}~ 
\cut \lmmtX(\Xlmmt[P]) `a + x \lmmtX(\Xlmmt[Q]) ~=~ \bsp (\IH) 
\cut P `a + x Q 
 \end{array} $
\arrayQED

 \end{enumerate}
 
 \observe
Notice that reduction is not used here.
 \end{proof}


We will now show that reduction in $`X$ is respected by the interpretation $\Xlmmt[`.]$, for which, as suggested above, we need to extend $\lmmt$.
 
 \begin{definition}[Extended $\lmmt$] \label{lmmtE}
We define $\lmmtE$ by adding the rule:
 \[ \begin{array}{rrcl@{\quad}l}
(`l'): & \cell<`Ly.t | t'`.e > &\reduc& 
\cell <{`M`g . < t' | {\mt y.<t|`g> } > }|e> & (`g \textit{ fresh}) 
 \end{array} \]
 \end{definition}

We will first show that the two implicit substitutions of $\Xs$ are respected by the interpretation; we need to involve reduction for these results as well, but in \emph{both} directions.

We will write $\eqmmt$ for the equivalence relation on $\lmmtE$ generated by rule $(`m)$ and $(\MuT)$.
 
 \resetqsymbol{`Q}{ \Xlmmt[Q] }
 \resetqsymbol{`P}{ \Xlmmt[P] }
 \resetqsymbol{`R}{\Xlmmt[R]}
 \resetqsymbol{`S}{\Xlmmt[S]}

 \begin{lemma} \label{sub respected}
 \begin{enumerate}
 \firstitem
$ \Xlmmt[\cutL Q `a + x P ] \eqmmt {`Q} \tsubst[\mt x.`P/`a] $. 
 \item
$ \Xlmmt[\cutR P `a + x Q ] \eqmmt {`Q} \tsubst[`m`a.`P/x] $. 
 \end{enumerate}
 \end{lemma}

 \begin{proof}{By induction on the structure of nets.}{\Lmm \ref{sub respected}}
 \begin{enumerate} \itemsep3pt 
 \firstitem
By induction on the structure of terms in $`X$\Shorter{; we only show the interesting cases}.

 \begin{description}
 
 \item [$ Q = \caps<y.`b> $] 
We have two cases:
 \begin{description}
 \myitem[$`a = `b$]
\Xlmmt[ \cutL \caps<y.`a> `a + x P ] &=& 
\Xlmmt[P \tsubst y/x ] ~=~ 
`P \tsubst[y/x] &\lmmtder (x) & 
\cell<y|\mt x.`P> &=& \\
\cell<y|`a> \tsubst[\mt x.`P/`a] &\ByDef& 
\Xlmmt[\caps<y.`a>] \tsubst[\mt x.`P/`a] 
 \end{array} $

 \myitem[$`a \not= `b$]
\Xlmmt[ \cutL \caps<y.`b> `a + x P ] &=& 
\Xlmmt[\caps<y.`b>] &=& 
\Xlmmt[\caps<y.`b>] \tsubst[\mt x.`P/`a] 
 \end{array} $

 \end{description}

\newqsymbol{`Z}{R\tsubst[`d/`a]}
 \item [$ Q = \exp y R `b . `g $] We have two cases:
 \begin{description}
 \myitem[$`g = `a$]
\Xlmmt[\cutL {\exp y R `b . `a } `a + x P ] &=& 
\Xlmmt[ \cut { \exp y { \cutL R `a + x P } `b . `d } `d + x P ] &\ByDef& \\
\Xlmmt[c {e y { \cutL R `a + x P } `b . `d } `d + x P ] &\redlmmt& \bsp (`d) 
\cell < `ly.`m`b.\Xlmmt[\cutL R `a + x P ] | \mt x.`P > &\eqmmt& \bsp (\IH) \\
\cell < `ly.`m`b.\Xlmmt[R]\tsubst[\mt x.`P/`a] | \mt x.`P > &=& 
\Xlmmt[e y R `b . `a ] \tsubst[\mt x.`P/`a] &\ByDef& \\
\Xlmmt[\exp y R `b . `a ] \tsubst[\mt x.`P/`a] 
 \end{array} $

 \myitem[$`g \not= `a$]
\Xlmmt[\cutL {\exp y R `b . `g } `a + x P ] &\ByDef& 
\Xlmmt[ \exp y { \cutL R `a + x P } `b . `g ] ~\ByDef~ 
\Xlmmt[e y { \cutL R `a + x P } `b . `g ] \\ 
\eqmmt (\IH) ~
\cell<`ly.`m`b.\Xlmmt[R]\tsubst[\mt x.`P/`a] | `g > &\ByDef& 
	\Xlmmt[e y R `b . `g] \tsubst[\mt x.`P/`a]
 \end{array} $

 \end{description}

\Longer{
 \myitem[{$ Q = \imp Q `b [z] y R $}] 
\Xlmmt[\cutL {\imp Q `b [z] y R } `a + x P ] &\ByDef& 
\Xlmmt[{ \imp { \cutL Q `a + x P } `b [z] y { \cutL R `a + x P } }] &\ByDef& \\
\Xlmmt[i { \cutL Q `a + x P } `b z y { \cutL R `a + x P }] &\eqmmt& \bsp (\IH) 
\Span{3}{
	\cell < z | `m`b.\Xlmmt[Q] \tsubst[\mt x.`P/`a] `. \mt y . \Xlmmt[R] \tsubst[\mt x.`P/`a] > 
	} \\ 
= ~ \cell < z | `m`b.\Xlmmt[Q] `. \mt y . \Xlmmt[R] > \tsubst \mt x.`P/`a &\ByDef& 
\Xlmmt[{\imp Q `b [z] y R }] \tsubst[\mt x.`P/`a]
 \end{array} $

 \myitem[{$ Q = \cut Q `b + y R $}] 
\Xlmmt[\cutL {\cut Q `b + y R } `a + x P ] &\ByDef& 
\Xlmmt[{ \cut { \cutL Q `a + x P } `b + y { \cutL R `a + x P } }] &\ByDef& \\
\Xlmmt[c { \cutL Q `a + x P } `b + y { \cutL R `a + x P }] &\eqmmt& \bsp (\IH) 
\cell < `m`b.\Xlmmt[Q] \tsubst \mt x.`P/`a | \mt y . \Xlmmt[R] \tsubst \mt x.`P/`a > &=& \\
\cell < `m`b.\Xlmmt[Q] | \mt y . \Xlmmt[R] > \tsubst \mt x.`P/`a &\ByDef& 
\Xlmmt[\cut Q `b + y R ] \tsubst[\mt x.`P/`a]
 \end{array} $
}

 \end{description}
\Shorter{The other cases follow, as the last one, by induction.}

 \item
By induction on the structure of terms\Shorter{; we only show the interesting cases}.
 \begin{description}
 \item [$ Q = \caps<y.`b> $] 
We have two cases:
 \begin{description}
 \myitem[$y = x$]
\Xlmmt[ \cutR P `a + x \caps<x.`b> ] &\ByDef& 
\Xlmmt[P \tsubst `b/`a ] ~=~ 
`P \tsubst[`b/`a] &\lmmtder (`a) & 
\cell<`M `a.{`P} | `b > &=& \\
\cell<x|`b> \tsubst[`m`a.`P/x] &\ByDef& 
\Xlmmt[\caps<x.`b>] \tsubst[`m`a.`P/x] 
 \end{array} $

 \myitem[$y \not= x$]
\Xlmmt[ \cutR P `a + x \caps<y.`b> ] &\ByDef& 
\Xlmmt[\caps<y.`b>] &\ByDef& 
\Xlmmt[\caps<y.`b>] \tsubst[`m`a.`P/x] 
 \end{array} $
 
 \end{description}

 \myitem[{$Q = \imp R `b [x] y S $}]
\Xlmmt[{\cutR P `a + x { \imp R `b [x] y S } }] &\ByDef& \\
\Xlmmt[\cut P `a + z {\imp {\cutR P `a + x R } `b [z] y {\cutR P `a + x S } } ] &\ByDef& \\ 
\Xlmmt[c P `a + z {i {\cutR P `a + x R } `b z y {\cutR P `a + x S }}] &\redlmmt& \bsp (z) \\
\cell < `m`a.`P | `M`b.\Xlmmt[{\cutR P `a + x R }] `.~ \mt y.\Xlmmt[{\cutR P `a + x S }] > &\eqmmt& \bsp (\IH) \\
\cell < `m`a.`P | `M`b.{ `R \tsubst[`m`a.`P/x] } `.~ \mt y.`S{\tsubst `m`a.`P/x } > &=& \\
\Xlmmt[i R `b x y S ] \tsubst[`m`a.`P/x] &=& 
\Xlmmt[{\imp R `b [x] y S }] \tsubst[`m`a.`P/x] 
 \end{array} $
\Longer{\arrayQED}

 \end{description}
\Shorter{The other cases follow by induction.\QED}

 \end{enumerate}

 \observe
Notice that in the first part reduction is limited to two steps, in both directions; $(\mt)$ in the first, and $(`m)$ in the second part.
The latter is not a valid step in \CBN.
In the second part reduction is limited to two steps, using rules $(\MuT)$ and $(`m)$, but in opposite direction with respect to the previous proof.
Notice that the $(\mt)$ step (over $z$) in the last part is not allowed in \CBV.

 \end{proof}

We can now show that the interpretation preserves reduction.

 \begin{theorem} \label{redX respected}
If $ P \redX Q $, then $ `P \eqlmmt {`Q} $.
 \end{theorem}

 \begin{proof} 
 \begin{description}
 \prooffirstmyitem [$\Cap$]
\Xlmmt[\cut \caps<y.`a> `a + x \caps<x.`b> ] \ByDef 
\Xlmmt[c {<y.`a>} `a + x {<x.`b>}]	\reduc (`m) 
\cell <y | {\mt x . < x | `b > } > 	 \reduc (\MuT) 
\Xlmmt[<y.`b>] ~ \ByDef ~ 
\Xlmmt[\caps<y.`b>]
 \end{array} $

 \myitem[$\Exp$]
\Xlmmt[ \cut{ \exp y P `b . `a } `a + x { \caps<x.`g> } ] &\ByDef& 
\Xlmmt[c { e y P `b . `a } `a + x {<x.`g>}] 
	&\reduc& \bsp (`m) \\
\Xlmmt[e y P `b . {\mt x . < x | `g > } ] 
	&\reduc& \bsp (\MuT) 
\Xlmmt[e y P `b . `g ] &\ByDef& 
\Xlmmt[ \exp y P `b . `g ]
 \end{array} $

 \myitem[$\Imp$]
\Xlmmt[ \cut { \caps<y.`a> } `a + x { \imp Q `b [x] z R } ] &\ByDef& 
\Xlmmt[c {<y.`a>} `a + x {i Q `b x z R}]
	&\reduc& \bsp (`m) \\
\cell< y | \mt x . \Xlmmt[i Q `b x z R] > 
	&\reduc& \bsp (\MuT) 
\Xlmmt[i Q `b y z R] 
	&\ByDef&
\Xlmmt[{ \imp Q `b [y] z R }]
 \end{array} $

 \myitem[$\Ins$]
\Xlmmt[{\cut { \exp y P `b . `a } `a + x { \imp Q `g [x] z R } }] ~\ByDef~ 
\Xlmmt[c {e y P `b . `a } `a + x {i Q `g x z R}] \reduc (`m) \\
\Xlmmt[e y P `b . { \mt x . \Xlmmt[i Q `g x z R] } ] ~\reduc (\MuT) ~ 
\cell < `ly.`m`b.\Xlmmt[P] | `M`g.\Xlmmt[Q] `. \mt z . {`R} > \reduc (`l) \\
\Xlmmt[c Q `g + y {c P `b + z R}] ~\ByDef~
\Xlmmt[ \cut Q `g + y { \cut P `b + z R } ]
 \end{array} $

\def \backleft {\kern-14mm}

 \item[$\Ins$] 
 $ \begin{array}[t]{lclcl}
 \multicolumn{3}{l}{
\Xlmmt[{\cut { \exp y P `b . `a } `a + x { \imp Q `g [x] z R } }] 
	~\ByDef~ 
\Xlmmt[c {e y P `b . `a } `a + x {i Q `g x z R}] } 
	 &\reduc& \bsp (`m) \\
\backleft
\Xlmmt[e y P `b . { \mt x . \Xlmmt[i Q `g x z R]}]
	&\reduc& \bsp (\MuT) 
\cell < `ly.`m`b.`P | `m`g.`Q `. \mt z . `R > 
	&\reduc& \bsp (`l') \\
\backleft
\cell <{`M`d . < `m`g.`Q | {\mt y.< `m`b.`P |`d> } > } | \mt z . `R > 
	&\reduc& \bsp (`m) 
\cell <{`M`d . < `m`g.`Q | {\mt y. `P \tsubst[`d/`b] } > }| \mt z . `R > 
	&=_{`a}& \\
\backleft
\Xlmmt[c {c Q `g + y P } `b + z R ] 
	&\ByDef& 
\Xlmmt[ \cut { \cut Q `g + y P } `b + z R ]
 \end{array} $

 \myitem[$\actL$] 
\Xlmmt[ \cut P `a + x Q ] &\ByDef& 
\Xlmmt[c P `a + x Q ] &\reduc& \bsp (`m) 
`P \tsubst[\mt x.`Q/`a] &\eqmmt& \bsp (\ref{sub respected}) 
\Xlmmt[\cutL P `a + x Q ]
 \end{array} $

 \myitem[$\actR$] 
\Xlmmt[ \cut P `a + x Q ] &\ByDef& 
\Xlmmt[c P `a + x Q ] &\reduc& \bsp (\MuT) 
`Q \tsubst[`m`a.`P/x] &\eqmmt& \bsp (\ref{sub respected}) 
\Xlmmt[\cutR P `a + x Q ] && \qed
 \end{array} $

 \end{description}
 \observe
Notice that in the first four cases we can swap the $(`m)$ and $(\MuT)$ reduction steps.

 \end{proof}

Since non-\CBN/\CBV{} reduction is used in the simulation of substitution (\Lmm\ref{sub respected}), we cannot show a similar result for the {\CBN} and {\CBV} reduction strategies or even sub-reduction systems. 
Also, the $(`m)$-reduction step in the second $(\Ins)$ case (over $`b$) takes place in the environment, which would not be allowed in either {\CBN} or {\CBV}, which strengthens our choice to exclude the second alternative of rule $(\Ins)$ for both those strategies on $`X$.

This means that, even when changing the active cuts of $\X$ into substitution, and the strong relationship we have established between $\lmmt$ and $`X$, these calculi are fundamentally different.
The absence of implicit variables and names gives $`X$ a more direct control over cut-elimination, and mapping $`X$'s substitution onto $\lmmt$ creates additional $(`m)$ and $(\MuT)$ redexes that do not adhere to either {\CBN} or {\CBV} reduction or strategies.

\Comment{
}


\setbox112=\hbox{\large $\lmmt$ in $`X$}
\section{Interpreting {\slmuPDF} in {\XiCalcPDF}}
 \label{slmu in X}

In this section we will, encouraged by the results we achieved in the previous sections, study an interpretation of $\slmu$ into $\X$.
In view of the detailed study in \Sect \ref{nclmu in lmmt}, and the observations we made of the impact of the fact that $\slmu$ adds application had on the proofs and the exact formulation of the result, it will be clear that we could do the same here, and focus on {\CBN} and \CBV.
However, the restrictions on the results will be very similar (using notions like $\crlmmtNn$), and we will therefore concentrate on full reduction $\redlmmt$, $\redlmmtn$, and $\redlmmtv$.

We started investigating if we can improve on the results shown in \cite{Bakel-Lescanne-MSCS'08}.
That paper defined the following encoding of $\lmu$ into $\X$:

 \def \leftSem {\LeftTop}
 \def \rightSem {\Rightbot}
 \def \SemEl#1{\mathord{\LeftTop{#1}\RightBot}}

 \[ \begin {array}{rcl}
\SemLmu{x}{`a} & \ByDef & \caps<x,`a>
	\\
\SemLmu{`lx.M}{`a} & \ByDef & \exp x \SemLmu{M}{`b} `b . `a 
	\\
\SemLmu{MN}{`a} & \ByDef & \cut \SemLmu{M}{`g} `g + x { \imp \SemLmu{N}{`b} `b [x] y \caps<y,`a> }
	\\
\SemLmu{`m`d.[`g]M}{`a} & \ByDef & \cut \SemLmu{M}{`g} `d + x \caps<x,`a>
 \end {array} \]
which is an extension of an interpretation of the $`l$-calculus, by adding the last alternative.
The way the application is dealt with is inspired by how the natural deduction $(\arrE)$ rule gets translated into the sequent calculus.

That paper goes on to show:

 \begin{lemma} [\cite{Bakel-Lescanne-MSCS'08}] \label{mu substitution cr}
 \begin{enumerate}
 \firstitem $ \SemLmu{M}{`d} \innercutL{`d}{x} {( \imp \SemLmu{N}{`b} `b [x] y \caps<y,`g> )} \crX \SemLmu{M [ N{`.}`g / `d ] N }{`g}$.
 \item $ \SemLmu{M}{`n} \innercutL{`d}{x} {( \imp \SemLmu{N}{`b} `b [x] y \caps<y,`g> )} \crX \SemLmu{M [ N{`.}`g / `d ] }{`n}$, if $`d \not= `n$.
 \end{enumerate}
 \end{lemma}
from which it shows:
 \begin {theorem} [Simulation of {\CBN} for $`l`m$ \cite{Bakel-Lescanne-MSCS'08}] 
 \label{simul:lm:cbn cr}
If $M \redCBN N $ then $\SemLmu{M}{`a} \crXCBN \SemLmu{N}{`a} $.
 \end {theorem}

Trying to make this interpretation also work with $(\mul)$ and left-structural substitution proved impossible:

 \[ \begin{array}{lcl}
\Sem{`@ V (`M`a.[`a]N)}_{`d} & \ByDef & 
\cut \Sem{V}_{`g} `g + x { \imp 
	{\cut \Sem{N}_{`a} `a + z \caps<z,`b> } 
 `b [x] y \caps<y,`d> }
 \end{array} \]
Notice that both $`g$ and $x$ are introduced here, so we could only apply a logical rule, and would not be able to express left-structural substitution.

So our attention shifted to combining the interpretation $\lmulmmt[`.]$ and $\lmmtX(`.)$, which results in: 

 \[ \begin {array}{rclclcl}
\lmmtX(\lmulmmt[x]){`d} &\ByDef& 
\lmmtX({}{\lmulmmt[v x]}){`d} ~\ByDef~ 
\slmuX (V x){`d} 
	\\
\lmmtX(\lmulmmt[`lx.t]){`d} &\ByDef& 
\lmmtX({}{\lmulmmt[l x . t]}){`d} ~\ByDef~ 
\exp x { \lmmtX(\lmulmmt[t])`b } `b . `d 
	\\
\lmmtX(\lmulmmt[PQ]){`d} &\ByDef& 
\lmmtX({}{\lmulmmt[a P Q]}){`d} \\ &\ByDef& 
\cut { \lmmtX(\lmulmmt[P]){`g} } `g + x { \cut {\lmmtX(\lmulmmt[Q]){`b} } `b + y { \cut \caps <x,`h> `h + z { \imp \caps<y,`v> `v [z] v \caps<v,`d> } } } 
	\\
\lmmtX(\lmulmmt[{`M`b.[`g]P}])`d &\ByDef&  
\lmmtX({}{\lmulmmt[m `b . `g P]}){`d} 
	~\ByDef~
\lmmtX(\lmulmmt[P]){`g} \tsubst [`d/`b] 
 \end {array} \]

It might be tempting to improve on this in the case for application, and pre-contract some of the cuts.
But that would produce:
 \[ \begin{array}{lcl}
\cut { \lmmtX(\lmulmmt[P]){`g} } `g + x { \cut {\lmmtX(\lmulmmt[Q]){`b} } `b + y { \cut \caps <x,`e> `e + e { \imp \caps<y,`v> `v [e] v \caps<v,`d> } } } 
	&\redXs& \\
\cut { \lmmtX(\lmulmmt[P]){`g} } `g + x { \cut {\lmmtX(\lmulmmt[Q]){`b} } `b + y { \imp \caps<y,`v> `v [x] v \caps<v,`d> } } 
	&\redXs& \\
\cut { \lmmtX(\lmulmmt[P]){`g} } `g + x { \imp {\lmmtX(\lmulmmt[Q]){`b} } `b [x] v \caps<v,`d> } 
 \end{array} \]
where we end up with exactly the interpretation of \cite{Bakel-Lescanne-MSCS'08}; as suggested above, in the proofs we need the fact that $x$ is not introduced, and that is the case for the second term.
When we interpreted $\slmu$ into $\lmmt$ in \Sect \ref{nclmu in lmmt}, we could also have optimised the interpretation of application,
 \[ \begin{array}{c@{~}c@{~}c@{~}c@{~}c@{~}c@{~}c}
\lmulmmt[`@ P Q ] & \ByDef &
\lmulmmt[a P Q] & \reduc & 
\lmulmmt[b P Q] & \reduc &
`M`a . < `P | `Q `. `a > 
 \end{array} \] 
but we have shown through the proofs of that section that we need the additional cuts, and, more than anything, the \emph{choice} to run them.

Moreover, notice that the combination of interpretations yields a different encoding for $`M`b.[`g]P$ under $`d$ from the one used in \cite{Bakel-Lescanne-MSCS'08}.
We have a choice between
$ \Sem[P]_{`g} \tsubst[`d/`b] $, $ \cut \Sem{P}_{`g} `b + x \caps<x,`d> $, $ \cut { \Sem[P]_{`g} \tsubst[`d/`b]} `b + x \caps<x,`d> $ or even 
$ \cut { \Sem[{P \tsubst [`d/`b]}]_{`g} } `b + x \caps<x,`d> $, all semantically equivalent, and each with their own advantages and disadvantages.
For example, using the former would give us
 \[ \begin{array}{lcl}
\Sem{`@ (`M`a.[`a]P) Q }{`d} & \ByDef & 
\cut { \Sem[P]_{`a} \tsubst[`g/`a] } `g + x { \cut {\Sem[Q]_{`b} } `b + y { \cut \caps <x,`e> `e + e { \imp \caps<y,`v> `v [e] v \caps<v,`d> } } } 
 \end{array} \]
where $`g$ might be introduced in $ \Sem[P]_{`a} \tsubst[`g/`a] $ or not, affecting which reduction rules can be applied, complicating the structure of the proofs.
It is not in $ \cut \Sem{P}_{`a} `a + x \caps<x,`g> $, which seems the better choice; as in $\lmulmmt[{`M`a.[`a]P}] \ByDef \lmulmmt[m `a . `a P]$, we need the name $`a$ to be accessible. 
So we cannot just take the composition of the two interpretations $\lmulmmt[`.]{`.}$ and $\lmmtX(`.){`.} $, but have to tweak that.
This led us to:

 \begin {definition}[Translation of $\slmu$ into $\X$] \label{from slmu to x}
 \[ \begin {array}{rcll}
\slmuX (x)`d &\ByDef& \slmuX (V x)`d \\
\slmuX (`lx.P)`d &\ByDef& \slmuX (L x . P)`d \\ 
\slmuX (PQ)`d &\ByDef& \slmuX (b |`t,t| P Q)`d \\
\slmuX (`M`b.[`g]P)`d &\ByDef& \slmuX (m `b . `g P)`d 
 \end {array} \]
 
 \end {definition}

The term $ \slmuX (m `b . `g P)`d $ is of course reducible to $ \slmuX(P\tsubst[`d/`b])`g $, but, as in other situations, we need the extra cuts for expressivity. 
Notice that in the former, $`d$ is not introduced, whereas it is not clear if it is in the latter.

 \begin{example} \label{reduce in continuation}
Although under the reduction strategies we consider here, we cannot reduce in the continuation, sometimes we can simulate it:
 \[ \begin{array}{lclcl}
\slmuX (b |`t,t| P Q)`d 
	&\redXsN& \bsp (\actR) 
\subL { \slmuX (P)`t } `t + t { \cut { \slmuX (Q)`a } `a + a {\imp \caps<a,`v> `v [t] v \caps<v,`g> } }
	&\substo& \\
\cut { \slmuX (Q)`a } `a + a { \cut { \slmuX (P)`t } `t + t {\imp \caps<a,`v> `v [t] v \caps<v,`g> } }
	&\redXsN& \bsp (\actR) 
\subL { \slmuX (Q)`a } `a + a { \cut { \slmuX (P)`t } `t + t {\imp \caps<a,`v> `v [t] v \caps<v,`g> } }
	&\substo& \\
\cut { \slmuX (P)`t } `t + t {\imp { \slmuX (Q)`a } `a [t] v \caps<v,`g> } 
 \end{array} \]
Notice that $ \cut { \slmuX (Q)`a } `a + a {\imp \caps<a,`v> `v [t] v \caps<v,`g> } $ has been replaced by $ \imp { \slmuX (Q)`a } `a [t] v \caps<v,`g> $, as if that cut-elimination was performed directly.
As in \Exm \ref{lmmt prep}, we write $ \redXsS $ for this reduction.
 \end{example}
 
What follows is rather syntax heavy, so we will use some abbreviations of sub-terms that are inactive: we will write:
 \[ \begin{array}{r@{\textit{ for }}l}
\yxd{y}{x}{`d} & \yxdF{y}{x}{`d} \\
\Vxd{Q}{x}{`d} & \vxd{Q}{x}{`d} 
 \end{array} \]
when convenient.

We start by showing that the interpretation respects $`b$-reduction; first we show that term substitution is preserved under our interpretation.

 \begin{lemma}\label{slmuX respects term substitution}
$ \subL { \slmuX (Q)`r } `r + x {\slmuX (P)`g } = \slmuX (P\tsubst[Q/x])`g $.
 \end{lemma}
 \begin{proof}
By induction on the structure of terms.
 \begin{description} 

 \myitem [$ P = x $]
\subL { \slmuX (Q)`r } `r + x {\slmuX (x)`g }
	&\ByDef& 
\subL { \slmuX (Q)`r } `r + x {\slmuX (V x)`g }
	&\substo& 
\slmuX (Q)`g
	&=& 
\slmuX (x\tsubst[Q/x])`g 
 \end{array} $

 \myitem [$ P = y \not= x $]
\subL { \slmuX (Q)`r } `r + x {\slmuX (y)`g }
	&\ByDef& 
\subL { \slmuX (Q)`r } `r + x {\slmuX (V y)`g }
	&\substo& 
\slmuX (V y)`g
	&=& 
\slmuX (y\tsubst[Q/x])`g 
 \end{array} $

 \myitem [$ P = `Ly.R $]
\subL { \slmuX (Q)`r } `r + x {\slmuX ({}{(`Ly.R)})`g }
	&\ByDef& 
\subL { \slmuX (Q)`r } `r + x {(\slmuX (L |`d| y . R)`g )}
	&\substo& 
\exp y {\subL { \slmuX (Q)`r } `r + x {\slmuX (R)`d } } `d . `g
	&=& (\IH) \\
\exp y {\slmuX (R\tsubst[Q/x])`d } `d . `g
	&\ByDef& 
\slmuX (`Ly.{R\tsubst[Q/x]})`g
	&=& 
\slmuX ({(`Ly.R)}\tsubst[Q/x])`g
 \end{array} $

 \item [{$ P = `@ M N $, $`M `a . [`d] M $}]
By induction.
\QED

 \end{description}
 \observe
Norice that reduction does not play a role in this proof.
 \end{proof}
 
With this result we can now show that $`b$-reduction is preserved.

 \begin{theorem} [$\slmuX(`.)`. $ respects beta reduction] \label{slmuX respects beta reduction}
$ \slmuX ({}{`@ (`lx.P) Q })`g \rtcredX \slmuX (P\tsubst[Q/x])`g $.
 \end{theorem}

 \begin{proof}
$ \kern-7mm \begin{array}[t]{lclcl} \kern7mm
\slmuX ({}{`@ (`lx.P) Q })`g
	&\ByDef& 
\cut {\exp x {\slmuX (P)`b } `b . `d } `d + d {\cut {\slmuX (Q)`a } `a + a \yxd{a}{d}{`g} }
	&\redXs& \bsp (\actR) \\ 
\subL {\exp x {\slmuX (P)`b } `b . `d } `d + d {\cut {\slmuX (Q)`a } `a + a \yxd{a}{d}{`g} }
	&\substo& 
\cut {\slmuX (Q)`a } `a + a {\cut {\exp x {\slmuX (P)`b } `b . `d } `d + d \yxd{a}{d}{`g} }
	&\redXs& \bsp (\actR) \\ 
\subL {\slmuX (Q)`a } `a + a 
	{\cut {\exp x {\slmuX (P)`b } `b . `d } `d + d \yxd{a}{d}{`g} }
 	&\substo& 
\cut {\exp x {\slmuX (P)`b } `b . `d } `d + d 
	{\imp { \slmuX (Q)`a } `a [d] v \caps<v,`g> }
	&\redXs& \bsp (\Ins) \\
\cut { \slmuX (Q)`a } `a + x {\cut { \slmuX (P)`b } `b + v \caps<v,`g> }
	&\redXs& \bsp (\actR) 
\subL { \slmuX (Q)`a } `a + x {\cut { \slmuX (P)`b } `b + v \caps<v,`g> }
	&\substo& \\
\cut { \subL { \slmuX (Q)`a } `a + x {\slmuX (P)`b } } `b + v \caps<v,`g>
	&\redXs& \bsp (\actL) 
\subR { \subL { \slmuX (Q)`a } `a + x {\slmuX (P)`b } } `b + v \caps<v,`g>
	&\substo& \bsp (`b \textit{ fresh}) \\
\subL { \slmuX (Q)`a } `a + x {\slmuX (P)`g }
	&=& \bsp (\ref{slmuX respects term substitution}) 
\slmuX (P\tsubst[Q/x])`g 
 \end{array} $ %
\arrayQED
 \observe
Notice that all steps are permitted in $\redXsn$ and $\redXsv$ (since $ \exp x {\slmuX (P)`b } `b . `d $ introduces $`d$, and then $Q$ is a value so $\lmuX(Q)`a $ introduces $`a$).

 \end{proof}
Notice that no $`h$ like step is needed here so $\slmuX(`.)`. $ provides a more direct interpretation than $\lmulmmt[`.]$.

In proofs that follow, we will sometimes directly apply the substitutions resulting from the reduction steps $(\actL)$ and $(\actR)$.

We will now show that $\slmuX (`.)`. $ respects $\mur$-substitution; since it interprets terms under a name, different from \Lmm \ref{interpretation respects mur substitution CBV}, this is not a single result.
We will write $ M \RightSub[Q.`g/`b] $ for $ M \RightSubF[Q.`g/`b] $.

 \begin{lemma} [Simulation of $\slmu$' right substitution through $\slmuX(`.)`. $] ~
 \label{lmuX mu right substitution} \label{lmuX Right substitution preservation}
 \begin{enumerate}

 \item 
$ \slmuX ({}{`@ (M\mursubst [Q.`g/`b]) Q })`g ~\rtcredXs~ \slmuX (D {}M [Q.`g/`b])`b ~\ByDef ~ \slmuX (d {}M [Q.`g/`b])`b $.

 \item $ \slmuX ({}{M\mursubst [Q.`g/`b]})`d ~\rtcredXs~ \slmuX (D {}M [Q.`g/`b])`d ~\ByDef~ \slmuX (d {}M [Q.`g/`b])`d $, if $`d \not= `b$.

 \end{enumerate}
 \end{lemma}

 \begin{proof} 
By simultaneous induction on the structure of terms; we only reduce the top cut in each step, but apply induction to sub-terms.

 \begin{enumerate} \itemsep 3pt
 \item

 \begin{description} 

 \pushmyitem{2mm}[$M = y$] 
\slmuX ({}{`@ (y\mursubst [Q.`g/`b]) Q })`g
	~=~
\slmuX ({}{`@ y Q })`g
	&\ByDef& 
\slmuX (a {V y} Q)`g
	&\redXs& \bsp (\actR) \\
\subL \caps<y,`s> `s + s \vxd{Q}{s}{`g} 
	&\substo& 
\vxd{Q}{y}{`g} 
	&\fromsubs& \\
\slmuX (D {V y} [Q.`g/`b])`b 
	&\ByDef& 
\slmuX (d y [Q.`g/`b])`b 
 \end{array} $

 \myitem[$M = `lz.P$]
\slmuX ({}{(`lz.P)\mursubst [Q.`g/`b] \, Q })`g
	&\ByDef& 
\cut { \slmuX (`lz.{P\mursubst[Q.`g/`b]})`t } `t + x \Vxd{Q}{x}{`g} 
	&\ByDef& \\
{ \cut {( \slmuX (L |`s| z . {I P  [Q.`g/`b]})`t )} `t + x \Vxd{Q}{x}{`g} }
	& \rtcredXs & \bsp ({\IH (2)}) 
{ \cut { \exp z {\slmuX (d P [Q.`g/`b])`s } `s . `t } `t + x \Vxd{Q}{x}{`g} }
	&\fromsubs& \\
\slmuX (D {<L |`s| z . P>} [Q.`g/`b])`b
	&\ByDef& 
\slmuX (d {`lz.P} [Q.`g/`b])`b 
 \end{array} $

 \myitem [$M = PR$] 
\slmuX ({}{(PR)\mursubst [Q.`g/`b]\,Q})`g
	~\ByDef~ 
\slmuX (A |`t,x| {P\mursubst[Q.`g/`b]\,R\mursubst [Q.`g/`b]} Q)`g
	&\ByDef& \\
\slmuX (A |`t,x| {a |`s,s| {I P  [Q.`g/`b]} {I R [Q.`g/`b]}} Q)`g 
	& \rtcredXs & \bsp ({\IH (2)}) \\
\slmuX (A |`t,x| {a |`s,s| {d P [Q.`g/`b]} {d R [Q.`g/`b]}} Q)`g 
	&\fromsubs& \\
\firstrenamefalse \slmuX (d {<a |`s,s| P R>} [Q.`g/`b])`b 
	~\ByDef~ 
\slmuX (d PR [Q.`g/`b])`b
 \end{array} $
 
 \myitem [${M = `m`s . [`b] P}$]{@{}lcl@{}cl}
\slmuX ({}{(`m`s.[`b]P)\mursubst[Q.`g/`b]\,Q})`g 
	&\ByDef& 
\slmuX ({}{(`m`s.[`g](P\mursubst[Q.`g/`b]\,Q))\,Q})`g 
	&\ByDef& \\
\slmuX (A |`t,t| {{}{`m`s.[`g]P\mursubst[Q.`g/`b]\,Q}} Q)`g 
	&\ByDef& 
\cut { \cut { \slmuX (P\mursubst[Q.`g/`b]\,Q)`g } `s + z \caps<z,`t> } `t + t \Vxd{Q}{t}{`g}  
	&\redXs& \bsp (\actL) \\
\Span{3}{
\subR { \cut { \slmuX (P\mursubst[Q.`g/`b]\,Q)`g } `s + z \caps<z,`t> } `t + t \Vxd{Q}{t}{`g}  
	~\substo~ 
\cut { \slmuX (P\mursubst[Q.`g/`b]\,Q)`g } `s + t \Vxd{Q}{t}{`g} 
}
	& \rtcredXs & \bsp ({\IH (1)}) \\
\cut { \slmuX (d P [Q.`g/`b])`b } `s + t \Vxd{Q}{t}{`g} 
	&\fromsubs& 
\slmuX (f {m `s . `b P} [Q.`g/`b])`b 
	&\ByDef& \\
\slmuX (d {`m`s.[`b]P} [Q.`g/`b])`b 
 \end{array} $

 \myitem [${M = `m`s . [`n] P}$]
\slmuX ({}{(`m`s.[`n]P)\mursubst[Q.`g/`b]\,Q})`g 
	&\ByDef& 
\slmuX ({}{(`m`s.[`n](P\mursubst[Q.`g/`b]))\,Q})`g 
	&\ByDef& \\
\slmuX (A {m `s . `n {I P [Q.`g/`b]}} Q)`g 
	&\redXs& \bsp (\actL) 
\subR { \cut { \slmuX(I P [Q.`g/`b])`n } `s + x \caps<x,`a> } `a + a \Vxd{Q}{a}{`g} 
	&\substo& \\
\cut { \slmuX ({}{P\mursubst[Q.`g/`b]})`n } `s + x \Vxd{Q}{x}{`g} 
	&\rtcredXs& \bsp ({\IH (2)}) 
\cut { \slmuX (d P [Q.`g/`b])`n } `s + x \Vxd{Q}{x}{`g} 
	&\fromsubs& \\
\slmuX (D {<m `s . `n P>} [Q.`g/`b])`b 
	&\ByDef& 
\slmuX (d {`m`s.[`n]P} [Q.`g/`b])`b 
 \end{array} $

 \end{description}

 \item %
  \Shorter{Straightforward by induction; no reduction is needed here.\QED} \Longer
{
 \begin{description}


 \pushmyitem{2mm}[$M = y$]
 \slmuX (y\mursubst [Q.`g/`b])`d
	&=& 
\slmuX (y)`d
	&\ByDef& 
\slmuX (V y)`d
	&=& 
\slmuX (d {V y} [Q.`g/`b])`d 
	&\ByDef& 
\slmuX (d {}y [Q.`g/`b])`d 
 \end{array} $

 \myitem[$M = `lz.P$]
\slmuX ({}{(`lz.P)\mursubst [Q.`g/`b]})`d
	&\ByDef& 
\slmuX ({}{`lz.I P  [Q.`g/`b]})`d
	&\ByDef& 
\slmuX (L |`s| z . {I P  [Q.`g/`b]})`d 
	& \rtcredXsv & \bsp ({\IH (2)}) \\
\exp z {\slmuX (d P [Q.`g/`b])`s } `s . `d
	&=& 
\slmuX (d {<L |`s| z . P>} [Q.`g/`b])`d
	&\ByDef& 
\slmuX (d {`lz.P} [Q.`g/`b])`d 
 \end{array} $

 \myitem [$M = PR$]{l@{}cl} 
\slmuX ({}{(PR)\mursubst [Q.`g/`b]})`d
	~\ByDef~ 
\slmuX (a |`s,s| {I P  [Q.`g/`b]} {I R  [Q.`g/`b]})`d 
	~ \rtcredXsv ({\IH (2)}) \\
\slmuX (a |`s,s| {d P [Q.`g/`b]} {d R [Q.`g/`b]})`d 
	~=~ 
\slmuX (d {<a P R>} [Q.`g/`b])`d 
	~\ByDef~ \\
\slmuX (d PR [Q.`g/`b])`d
 \end{array} $
 
 \myitem [${M = `m`s . [`b] P}$]
\slmuX ({}{(`m`s.[`b]P)\mursubst[Q.`g/`b]})`d 
	&\ByDef& 
\slmuX ({}{`m`s.[`g]P\mursubst[Q.`g/`b]\,Q})`d 
	&\ByDef& \\
\slmuX (m `s . `g {P\mursubst[Q.`g/`b]\,Q})`d 
	&\rtcredXsv& \bsp ({\IH (1)}) 
\slmuX (m `s . `b {d P [Q.`g/`b]})`d 
	&=& 
\slmuX (d {`m`s.[`d]P} [Q.`g/`b])`b 
 \end{array} $

 \myitem [${M = `m`s . [`n] P}$]
\slmuX ({}{(`m`s.[`n]P)\mursubst[Q.`g/`b]})`d 
	&\ByDef& 
\slmuX ({}{`m`s.[`n](P\mursubst[Q.`g/`b])})`d 
	&\ByDef& 
\slmuX (m `s . `n {I P [Q.`g/`b]})`d 
	&\rtcredXsv& \bsp ({\IH (2)}) \\
\slmuX (m `s . `n {d P [Q.`g/`b]})`d 
	&=& 
\slmuX (d {`m`s.[`n]P} [Q.`g/`b])`d 
 \end{array} $
 \arrayQED

 \end{description}
}

 \end{enumerate}
 
 \observe
The only reduction steps take place in three\footnote{Notice that this is not so in the proof of \Lmm \ref{interpretation respects mur substitution CBV}, where reduction takes place in only one part.} cases in the first part, $M=y$, $ M = `m `s . [`n] P$, and $M = `m`s . [`b] P $, which are valid in \CBV; in the latter two cases, this $(\actL)$ step is not valid in $\redXn$, since $x$ is not introduced in $\vxdfull{Q}{x}{`d} $; it is valid in $\redXv$, since $`s$ is not introduced.
Induction gets applied also in the continuation and inside exports, so this is not a proof for \CBV, but for $\redXsv$.

 \end{proof}


The following is needed below:

 \begin{lemma} \label{rescue lemma}
\resetqsymbol{`V}{V} 

$ \begin{array}[t]{lclclcl}
\cut { \slmuX (P)`b } `b + t \Vxd{Q}{t}{`d} 
	&\rtcredXsv& 
\slmuX(f P [Q.`d/`b])`b 
 \end{array} $.
 \end{lemma}

 \begin{proof}
If $`b$ is not introduced, then the activation of the substitution is permitted (notice that $t$ is not introduced in $\Vxd{Q}{t}{`d}$). 
Otherwise, $P$ is a value and $`b$ does not occur in $P$; we have two cases:

 \begin{description}
 
 \myitem [$V = y$]
\cut { \slmuX (y)`b } `b + t \vxd{Q}{t}{`d} 
	&\ByDef& 
\cut { \slmuX (V y)`b } `b + t \vxd{Q}{t}{`d} 
	&\redXsv& \bsp (\actR) \\
\subL { \slmuX (V y)`b } `b + t \vxdfull{Q}{t}{`d} 
	&\substo& 
\vxdfull{Q}{y}{`d} 
	&\fromsubs& \\
\slmuX(F {V y} [Q.`d/`b])`b 
 \end{array} $

 \myitem [$V = `ly.P$]
\cut {\slmuX (`Ly.P)`b } `b + t \Vxd{Q}{t}{`d} 
	&\ByDef& 
\cut {\slmuX (<L |`g| y . P>)`b } `b + t \Vxd{Q}{t}{`d} 
	&\fromsubs& \bsp ({`b \notele \slmuX(P)`g }) \\
\slmuX(f {<L |`g| y . P>} [Q.`d/`b])`b 
	&\ByDef&
\slmuX(f `ly.P [Q.`d/`b])`b \\ [2mm]
 \end{array} $
\arrayqed

 \end{description}
 \end{proof}

We can now show that the interpretation preserves $\mur$ reduction.

  \begin{theorem} [Preservation of $\mur$ contraction through $\slmuX(`.)`. $] 
 \label{lmuX Right reduction preservation} ~
 \begin{enumerate}
 \item $ \slmuX ({}{`@ (`M`b.[`b]P) Q })`d \crXs \slmuX ({}{`M`g.[`g] `@ {P\mursubst[Q.`g/`b]} Q })`d $.
 \item $ \slmuX ({}{`@ (`M`b.[`n]P) Q })`d \crXs \slmuX ({}{`M`g.[`b] P\mursubst[Q.`g/`b]})`d $.
 \end{enumerate}
 \end{theorem}

 \begin{proof} 
 \begin{enumerate}
 
 \firstitem $ \kern-7.5mm \begin{array}[t]{lclcl} \kern7.5mm 
\slmuX({}{`@ (`M`b.[`b]P) Q })`d 
	&\ByDef& 
\cut { \slmuX (U `b . `b P)`t } `t + t \Vxd{Q}{t}{`d}
	&\redXs& \bsp (\actL) \\ 
\subR { \slmuX (U `b . `b P)`t } `t + t \Vxd{Q}{t}{`d}
	&\substo& 
\cut { \slmuX (f {}P [Q.`d/`t])`b } `b + t \Vxd{Q}{t}{`d} 
	&\substo& (`t \notin \slmuX (P)`b ) \\
\cut { \slmuX (P)`b } `b + t \Vxd{Q}{t}{`d} 
	&\rtcredXs& \bsp (\ref{rescue lemma}) 
\slmuX(f P [Q.`d/`b])`b 
	&\worraXs& \bsp ({\ref{lmuX Right substitution preservation}\,(1)}) \\
\slmuX(P\mursubst[Q.`d/`b]\,Q)`d 
	&\fromsubs& 
\subR { \slmuX(P\mursubst[Q.`g/`b]\,Q)`g } `g + z \caps<z,`d> 
	&\worraXs& \bsp (\actL) \\
\slmuX(m `g . `g {P\mursubst[Q.`g/`b]\,Q})`d 
	&\ByDef& 
\slmuX({}{`M`g.[`g]P\mursubst[Q.`g/`b]\,Q})`d 
 \end{array} $

   \item
 $ \begin{array}[t]{lclclcl} 
\slmuX ({}{`@ (`M`b.[`n]P) Q })`d
	&\ByDef& 
\slmuX (A |`t,t|  {m `b . `n P} Q)`d
	&\redXs& \bsp (\actL) \\ 
\subR { \slmuX (m `b . `n P)`t } `t + t \Vxd{Q}{t}{`d}
	&\substo& 
		\bsp ({`t \notin \slmuX (P)`n }) 
\cut { \slmuX (P)`n } `b + t \Vxd{Q}{t}{`d} 
	&\redXs& \bsp (\actL) (`b \textit{ not introduced}) \\ 
\subR {\slmuX(P)`n } `b + t \Vxd{Q}{t}{`d} 
	&\worraXs& \bsp ({\ref{lmuX Right substitution preservation}\,(2)}) 
\slmuX (I P [Q.`d/`b])`n  
	&=& \\ 
\subR {\slmuX (I P [Q.`g/`b])`d } `g + z \caps<z.`d> 
	&\worraXs& \bsp (\actL) 
\slmuX (m `g . `n {I P [Q.`g/`b]})`d
	&\ByDef& 
\slmuX ({}{`M`g.[`n] P\mursubst[Q.`g/`b]})`d
 \end{array} $
\arrayQED 
 \end{enumerate}
 
 \observe
The reduction steps in this proof are part of \CBV; notice that both $ \redXs (\actL) $ that pull in $\Vxd{Q}{t}{`d} = \vxdfull{Q}{t}{`d} $ are not permitted in \CBN, since $t$ is not introduced,

 \end{proof}

\def \RightSub[#1.#2/#3]{\def\sysid{`m} \sk \ToX({#1{`.}#2/#3})_{\ftn}}
\def \RightSubF[#1.#2/#3]
 {\def\sysid{`m} \omtrue \cutL {} #3 + x {\vxdo{#1}{x}{#2} } }
\def \RightSubFull[#1.#2/#3]
 {\def\sysid{`m} \omtrue \cutL {} #3 + x {\imp {\slmuX(#1)`a } `a [x] v \caps<v,#2> } } 
 
We will now show that the interpretation also respects \CBN; as in \Sect \ref{nclmu in lmmt}, we need a different proof for that, that prepares the term for a {\CBN} reduction.
First we show that result for right substitution.
The proofs will work with 
$\vxdofull{Q}{x}{`d}$
rather than $\vxdfull{Q}{x}{`d} $, so, as in \Sect \ref{nclmu in lmmt}, we need a new proof for the lemma as well.
We will now write $ M \RightSub[Q.`g/`b] $ for $ M \RightSubFull[Q.`g/`b] $.

 \begin{lemma} [Simulation of $\slmu$' {\CBN} right substitution through $\slmuX(`.)`. $] ~
 \label{lmuX Right CBN substitution preservation}
 \begin{enumerate}

 \item  
$ \slmuX ({}{`@ (M\mursubst [Q.`g/`b]) Q })`g \rtcredXsn \slmuX (F {}M [Q.`g/`b])`b \ByDef \slmuX (d {}M [Q.`g/`b])`b $

 \item $ \slmuX ({}{M\mursubst [Q.`g/`b]})`d \rtcredXsn \slmuX (F {}M [Q.`g/`b])`d \ByDef \slmuX (d {}M [Q.`g/`b])`d $, if $`d \not= `b$.

 \end{enumerate}
 \end{lemma}

 \begin{proof} 
By simultaneous induction on the structure of terms. 
The proof is as that for \Lmm \ref{lmuX mu right substitution}, but for the parts where reduction takes place that is not valid in $\redXsn$, to wit:

 \begin{enumerate} \itemsep 3pt
 \item

 \begin{description}

 \pushmyitem{2mm}[${M = `m`s . [`b] P}$]
\slmuX ({}{(`m`s.[`b]P)\mursubst[Q.`g/`b]\,Q})`g 
	&\ByDef& %
	\\
\slmuX ({}{(`m`s.[`g](P\mursubst[Q.`g/`b]\,Q))\,Q})`g 
	&\ByDef& \\
\slmuX (a |`t,x| {{}{`m`s.[`g]P\mursubst[Q.`g/`b]\,Q}} Q)`g 
	&\ByDef& \bsp (`t,x \textit{ fresh}) \\
\slmuX (a |`t,x| {m `s . `g {P\mursubst[Q.`g/`b]\,Q}} Q)`g 
	&\redXsS& \bsp (\ref{reduce in continuation}) \\
\cut { \cut { \slmuX (P\mursubst[Q.`g/`b]\,Q])`b } `s + z \caps<z,`t> } `t + x {\imp {\slmuX(Q)`a } `a [x] v \caps<v,`g> } 
	&\redXsn& \bsp (\actL) (`t \textit{ not introduced}) \\ 
\subR { \cut { \slmuX (P\mursubst[Q.`g/`b]\,Q)`b } `s + z \caps<z,`t> } `t + x {\imp {\slmuX(Q)`a } `a [x] v \caps<v,`g> } 
	&\substo& \\
\cut { \slmuX (P\mursubst[Q.`g/`b]\,Q)`g } `s + x {\imp {\slmuX(Q)`a } `a [x] v \caps<v,`g> } 
	&\rtcredXsn& \bsp ({\IH (1)}) \\
\cut { \slmuX (d P [Q.`g/`b])`b } `s + x \vxdofull{Q}{x}{`g}
	&=& \bsp (\textit{subst}) \\
\slmuX (F {<m `s . `b P>} [Q.`g/`b])`b 
	&\ByDef& %
	\\ 
\slmuX (d {`m`s.[`b]P} [Q.`g/`b])`b
 \end{array} $ 
 \arrayqed

 \myitem [${M = `m`s . [`n] P}$]
\slmuX ({}{(`m`s.[`n]P)\mursubst[Q.`g/`b]\,Q})`g 
	&\ByDef& 
\slmuX ({}{(`m`s.[`n](P\mursubst[Q.`g/`b]))\,Q})`g 
	&\ByDef& \\
\slmuX (A {m `s . `n {I P [Q.`g/`b]}} Q)`g 
	&\redXsS& \bsp (\ref{reduce in continuation}) \\
\Span{3}{ \cut { \cut { \slmuX (I P [Q.`g/`b])`n } `s + z \caps<z,`t> } `t + x {\imp {\slmuX(Q)`a } `a [x] v \caps<v,`g> } 
	~\redXs~ (\actL) } \\
\cut { \slmuX ({}{P\mursubst[Q.`g/`b]})`n } `s + x \Vxd{Q}{x}{`g} 
	&\rtcredXs& \bsp ({\IH (2)}) 
\cut { \slmuX (d P [Q.`g/`b])`n } `s + x \Vxd{Q}{x}{`g} 
	&\fromsubs& \\
\slmuX (d {<m `s . `n P>} [Q.`g/`b])`b 
	&\ByDef& 
\slmuX (d {`m`s.[`n]P} [Q.`g/`b])`b 
 \end{array} $

 \end{description}
 \end{enumerate}
 
 \observe
Reduction steps $\redXsS$ (which are in $\redXsn$) are needed to prepare the $(\actL)$ step, so all reductions are in $\redXsn$; notice that $ \imp {\slmuX(Q)`a } `a [x] v \caps<v,`g> $ introduces $x$, so the right-substitution can now be activated.
Since we show the property for $\redXsn$, induction can get applied inside sub-terms, like in the cases for abstraction and application.

 \end{proof}

We can now show a preservation result under $\redXsn$ for the interpretation. 
As with the previous lemma, some steps have to be executed to rewrite the term so that the previous lemma can be applied.

 \begin{theorem} [Preservation of $\mur$ $\redXsn$ reduction through $\slmuX(`.)`. $] 
 \label{lmuX Right CBN reduction preservation} ~
 \begin{enumerate}
 \item $ \slmuX ({}{`@ (`M`b.[`b]P) Q })`d \crXsn \slmuX ({}{`M`g.[`g] `@ {P\mursubst[Q.`g/`b]} Q })`d $.
 \item $ \slmuX ({}{`@ (`M`b.[`n]P) Q })`d \crXsn \slmuX ({}{`M`g.[`n] P\mursubst[Q.`g/`b]})`d $.
 \end{enumerate}
 \end{theorem}

 \begin{proof} 
 \begin{enumerate}
 
 \firstitem $ \kern-7.5mm \begin{array}[t]{lclcl} \kern7.5mm 
\slmuX({}{`@ (`M`b.[`b]P) Q })`d 
	&\ByDef& 
\slmuX(a |`t,t| {m `b . `b P} Q)`d 
	&\redXsS& \bsp (\ref{reduce in continuation}) \\ 
\cut { \slmuX (U `b . `b P)`t } `t + t \vxdofull{Q}{t}{`d} 
	&\redXsn& \bsp (\actL) 
\subR { \slmuX (U `b . `b P)`t } `t + t \vxdofull{Q}{t}{`d} 
	&\substo& \\
\cut { \slmuX (P)`b } `b + t \vxdofull{Q}{t}{`d} 
	&\rtcworraXsn& \bsp (\ref{lmuX Right CBN substitution preservation}) 
\slmuX(P\mursubst[Q.`d/`b]\,Q)`d 
	&\worraXsn& \bsp (\actL) \\
\slmuX(U `g . `g {P\mursubst[Q.`g/`b]\,Q})`d 
	&\ByDef& 
\slmuX({}{`M`g.[`g]P\mursubst[Q.`g/`b]\,Q})`d 
 \end{array} $

 \item
 $ \begin{array}[t]{lclcl}
\slmuX ({}{`@ (`M`b.[`n]P) Q })`d
	&\ByDef& 
\slmuX (a {m `b . `n P} Q)`d
	&\redXsn& \bsp (\actL) \\ 
\subR { \slmuX (m `b . `n P)`s } `s + s { \vxd{Q}{s}{`d} }
	&=& 
\cut {\slmuX(P)`n } `b + t \vxd{Q}{t}{`d} 
	&\redXsn& \bsp (\actL) \\
\subR { \slmuX(P)`n } `b + t \vxd{Q}{t}{`d} 
	&\rtcworraXsn& \bsp (\ref{lmuX Right CBN substitution preservation}) 
\slmuX (I P [Q.`d/`b])`n 
	&=& \\ 
\Span{3}{
\subR {\slmuX (I P [Q.`g/`b])`n } `g + z \caps<z.`d> 
	~\worraXs (\actL) ~ %
\slmuX (m `g . `n {I P [Q.`g/`b]})`d
	~ \ByDef~ %
\slmuX ({}{`M`g.[`n] P\mursubst[Q.`g/`b]})`d
}
 \end{array} $
\arrayQED 
 \end{enumerate} 
 \end{proof}

 \resetqsymbol{`V}{Q} 
We will now turn to the interpretation of left structural reduction.
First we will show that left structural substitution is preserved;
we will write $ \vxdv{Q}{a}{`g} $ for $ \vxdvfull{Q}{a}{`g} $.
Notice that $ 
\vxdv{Q}{a}{`g} 
\not= \Vxd{Q}{a}{`g} = \cut { \slmuX(Q)`d } `d + d { \imp \caps<d,`v> `v [a] v \caps<v,`g> } $: inside the import, $d$ and $a$ have swapped position.

 \begin{lemma} [$\slmuX(`.)`. $ preserves $\slmu$'s left structural substitution] \label{mul substitution in X} ~
 \begin{enumerate}
 \item $ 
 \slmuX (Q\,{\mulsubst[Q.`g/`b]M})`g 
	\rtcredXs 
 \subR { \slmuX ({}M)`b } `b + a \vxdV{Q}{a}{`g} 
 	\ByDef
 \LeftSub [Q.`g/`b] \slmuX({}M)`b $.
 
 \item $ 
\slmuX (\mulsubst[Q.`g/`b]M)`d 
\rtcredXs 
\subR { \slmuX ({}M)`d } `b + a \vxdV{Q}{a}{`g}  
\ByDef 
\LeftSub [Q.`g/`b] \slmuX({}M)`d 
$.
 \end{enumerate}

 \end{lemma}

 \begin{proof}
Simultaneously by induction on the structure of terms; as before, we only reduce the top cut in each step, but apply induction to sub-terms.

 \begin{enumerate} \itemsep 3pt
 \item 

 \begin{description}

 \pushmyitem{2mm}[$M = y$]
\slmuX (Q\,{\mulsubst[Q.`g/`b]y})`g 
	~=~
\slmuX (Qy)`g 
	&\ByDef& 
\cut {\slmuX (Q)`t } `t + t { \cut \caps<y,`a> `a + a \yxdfull{a}{t}{`g} } 
	&\redXs& \bsp (\actR) \\ 
\subL {\slmuX (Q)`t } `t + t { \cut \caps<y,`a> `a + a \yxdfull{a}{t}{`g} } 
	&\substo& 
\cut \caps<y,`a> `a + a {\cut {\slmuX (Q)`t } `t + t \yxdfull{a}{t}{`g} } 
	&\redXs& \bsp (\actR) \\ 
\subL \caps<y,`a> `a + a {\cut {\slmuX (Q)`t } `t + t \yxdfull{a}{t}{`g} } 
	&\substo& 
\vxdV{Q}{y}{`g} 
	\hfill \fromsubs& \\
{\subR { \slmuX (V y)`b } `b + a \vxdV{Q}{a}{`g} } 
	&\ByDef& 
\LeftSub [Q.`g/`b] \slmuX (y)`b 
 \end{array} $

 \myitem[$M = `ly.P$]
\slmuX (Q\,\mulsubst[Q.`g/`b]{(`ly.P)})`g 
	&\ByDef& 
\cut {\slmuX (Q)`t } `t + t { \cut { \slmuX (J {{(`ly.P)}} [Q.`g/`b])`r } `r + r \yxd{r}{t}{`g} }
	&\redXs& \bsp (\actR) \\ 
\subL {\slmuX (Q)`t } `t + t { \cut { \slmuX (J {{(`ly.P)}} [Q.`g/`b])`r } `r + r \yxd{r}{t}{`g} }
	&\substo& %
\cut { \slmuX (J {{(`ly.P)}} [Q.`g/`b])`s } `s + r {\cut {\slmuX (Q)`t } `t + t \yxd{r}{t}{`g} } 
 	&\ByDef& \\
\cut { \slmuX (J {{(`ly.P)}} [Q.`g/`b])`s } `s + r \vxdv{Q}{r}{`g} 
	&\rtcredXs& \bsp ({\IH (2)}) 
\cut { \slmuX (s {`ly.P} [Q.`g/`b])`s } `s + r \vxdv{Q}{r}{`g} 
	&\rtcredXs& \\
\Span{3}{ \begin{array}{@{}lclcl}
\cut { \exp y { \LeftSub [Q.`g/`b] {\slmuX(P)`n } } `n . `s } `s + r { \vxdv 
{Q}{r}{`g} }
	&\fromsubs& 
\subR {({ \slmuX (L |`n| y . P)`b })} `b + r \vxdv{Q}{r}{`g} 
	&\ByDef& 
\LeftSub [Q.`g/`b] \slmuX (`ly.P)`b 
 \end{array} }
 \end{array} $

 \myitem [$M = PR$] 
\slmuX (Q\,{\mulsubst[Q.`g/`b]{(P R)}})`g 
	~\ByDef~ 
\slmuX (a |`s,s| Q {\mulsubst[Q.`g/`b]{(P R)}})`g 
	&\ByDef& \\
\slmuX (a |`s,s| Q {a |`t,t| {\mulsubst[Q.`g/`b]P} {\mulsubst[Q.`g/`b]R}})`g 
	& \redXs & \bsp (\actR) \\ 
\subL { \slmuX(Q)`s } `s + s { \cut { \slmuX (a |`t,t| {\mulsubst[Q.`g/`b]P} {\mulsubst [Q.`g/`b]R})`a } `a + a \yxd{a}{s}{`g} }
	&\substo& \\
\cut { \cut { \slmuX(\mulsubst[Q.`g/`b]P)`t } 
	`t + t { \cut { \slmuX(\mulsubst[Q.`g/`b]R)`a } 
			`a + a \yxd{a}{t}{`b} } 
		} `b + z \vxdv{Q}{z}{`g} 
	& \rtcredXs & \bsp ({\IH (2)}) \\
\cut { \cut { \LeftSub [Q.`g/`b] \slmuX(P)`t } 
	`t + t { \cut { \LeftSub [Q.`g/`b] \slmuX(R)`a } 
			`a + a \yxd{a}{t}{`b} } 
		} `b + z \vxdv{Q}{z}{`g} 
	&\fromsubs& \\
\subR {\firstrenamefalse \slmuX (a |`t,t| P R)`b } `b + z \vxdv{Q}{z}{`g} 
	&\ByDef& \\
\LeftSubF [Q.`g/`b] {\slmuX (PR)`b }
	~\ByDef~ 
\LeftSub [Q.`g/`b] \slmuX (PR)`b 
 \end{array} $

 \myitem [${M = `m`s . [`b] P}$]
\slmuX (Q\,{\mulsubst[Q.`g/`b](`m`s.[`b]P)})`g 
	&=& 
\slmuX (Q\,{(`m`s.[`g]Q\mulsubst[Q.`g/`b]P)})`g 
	~\ByDef~ \\
\Span{3}{ \begin{array}{@{}lclcl}
\slmuX (a Q {m `s . `g {Q\mulsubst[Q.`g/`b]P}})`g 
	&\redXs& \bsp (\actR) \\ 
\slmuX (l |`d,d| Q {m `s . `g {Q\mulsubst[Q.`g/`b]P}})`g 
 	&\substo& \\
\cut { \slmuX (m `s . `g {Q\,{\mulsubst[Q.`g/`b]P}})`r } `r + r { \cut { \slmuX(Q)`d } `d + d \yxd{r}{d}{`g} }
	&\ByDef& 
  \end{array} } \\
\cut { \slmuX (m `s . `g {Q\,{\mulsubst[Q.`g/`b]P}})`r } `r + r \vxdv{Q}{r}{`g} 
	&\redXs& \bsp (\actL) 
\subR { \slmuX (m `s . `g {Q\,{\mulsubst[Q.`g/`b]P}})`r } `r + r \vxdv{Q}{r}{`g} 
	&\substo& \\
\cut { \slmuX (Q\,\mulsubst[Q.`g/`b]P)`g } `s + r \vxdv{Q}{r}{`g}
	&\rtcredXs& \bsp ({\IH (1)}) 
\cut { \slmuX (s P [Q.`g/`b])`b } `s + r \vxdv{Q}{r}{`g}
	&\fromsubs& \\
\LeftSubF [Q.`g/`b] { \slmuX (m `s . `b P)`b }
	&\ByDef& 
\LeftSub [Q.`g/`b] \slmuX (`m`s.[`b]P)`b 
 \end{array} $
 

 \myitem [${M = `m`s . [`n] P}$]
\slmuX (Q\,{\mulsubst[Q.`g/`b](`m`s.[`n]P)})`g 
	&=& 
\slmuX ({Q\,(`m`s.[`n]\mulsubst[Q.`g/`b]P)})`g 
	&=& \\
\Span{3}{ \begin{array}{@{}lclcl}
\slmuX (a Q {m `s . `n {\mulsubst[Q.`g/`b]P}})`g 
	&\redXs& \bsp (\actR) \\
  \end{array} } \\
\cut { \cut { \slmuX (\mulsubst[Q.`g/`b]P)`n } `s + z \caps<z,`r> } `r + r \vxdv{Q}{r}{`g}
	&\redXs& \bsp (\actL) 
\cut { \slmuX (\mulsubst[Q.`g/`b]P)`n } `r + r \vxdv{Q}{r}{`g}
	& \rtcredXs & \bsp ({\IH (2)}) \\
\cut { \slmuX (s P [Q.`g/`b])`n } `s + r \vxdv{Q}{r}{`g}
	&\fromsubs& 
\LeftSubF [Q.`g/`b] { \slmuX (m `s . `n P)`b }
	&\ByDef& \\
\LeftSubF [Q.`g/`b] { \slmuX (`m`s.[`n]P)`b }
	&\ByDef& 
\LeftSub [Q.`g/`b] \slmuX (`m`s.[`n]P)`b 
 \end{array} $

 \end{description}

 \item Much like that for the previous part, but easier since the $\mur$ substitution does not insert terms. \QED 
 \Comment
 { 
 
 \begin{description}

 \pushmyitem{2mm}[$M = z$]
\slmuX (\mulsubst[Q.`g/`b]y)`d 
	&=& 
\slmuX (y)`d 
	&\ByDef& 
\slmuX (V y)`d 
	&=& 
\subR { \slmuX (V y)`d } `b + z {\cut {\slmuX (Q)`a } `a + a \yxd{a}{z}{`g} } 
	&\ByDef& 
\LeftSub [Q.`g/`b] \slmuX(y)`d 
 \end{array} $

 \myitem[$M = `lz.P$]
\slmuX (\mulsubst[Q.`g/`b]{(`lz.P)})`d
	&\ByDef& 
\exp y { \slmuX (\mulsubst[Q.`g/`b]P)`n } `n . `d 
	&\rtcredXs& \bsp ({\IH (2)}) 
\exp y {\slmuX(S P [Q.`g/`b])`n } `n . `d 
	&\substo& \\ 
\subR {({ \slmuX (L |`n| y . P)`d })} `b + z \vxdv {Q}{z}{`g} 
	&\ByDef& 
\LeftSub [Q.`g/`b] \slmuX(`lz.P)`d 
 \end{array} $

 \myitem [$M = PR$] 
\LeftSubF [Q.`g/`b] {\slmuX (PR)`d }
	&\ByDef& 
\subR { \slmuX (a |`t,t| P R)`d } `b + z \vxdv{Q}{z}{`g} 
	&\substo& \\
\Span{3}{ \begin{array}{@{}lclcl}
\slmuX(a {s P [Q.`g/`b]} {s R [Q.`g/`b]})`d }
	&\rtcworraXs& \bsp ({\IH (2)}) \\
\slmuX (a {\mulsubst[Q.`g/`b]P} {\mulsubst[Q.`g/`b]R})`g 
	&\ByDef& 
\slmuX (\mulsubst[Q.`g/`b]PR)`d 
 \end{array} $

 \myitem [${M = `m`s . [`b] P}$]
\LeftSub [Q.`g/`b] \slmuX (`m`s.[`b]P)`d 
	&\ByDef& 
\LeftSubF [Q.`g/`b] { \slmuX (m `s . `b P)`d } 
	&=& \\
\cut { \slmuX (S P [Q.`g/`b])`b } `s + z \caps<z,`d> 
	&\rtcworraXs& \bsp ({\IH (1)}) \\
\slmuX (m `s . `g {a Q \mulsubst[Q.`g/`b]P})`d 
	&\ByDef& \\
\slmuX ({`m`s.[`g]Q\,\mulsubst[Q.`g/`b]P})`d 
	~\ByDef~ 
\slmuX (\mulsubst[Q.`g/`b]{(`m`s.[`b]P)})`d 
 \end{array} $

 \myitem [${M = `m`s . [`n] P}$]
\LeftSub [Q.`g/`b] \slmuX (`m`s.[`n]P)`d 
	&\ByDef& 
\LeftSubF [Q.`g/`b] { \slmuX (m `s . `n P)`d } 
	~ = \\
\cut { \slmuX (s P [Q.`g/`b])`n } `s + z \caps<z,`d> 
	&\rtcworraXs& \bsp ({\IH (2)}) 
\cut { \slmuX (\mulsubst[Q.`g/`b]P)`n } `s + z \caps<z,`d> 
 	~\ByDef~ 
\slmuX (\mulsubst[Q.`g/`b]{(`m`s.[`n]P)})`d 
 \end{array} $

 \end{description}
}

 \end{enumerate}

 \observe
All cases need cut-eliminations, that are all valid in $\redXsv$; notice that if $Q$ is a value, then $\slmuX(Q)`s $ introduces $`s$.

 \end{proof}

With this result we can now show that our encoding respects $\mul$ contraction.
First we need to show:

 \begin{lemma} \label{second rescue lemma}
\resetqsymbol{`V}{V} 

$ \begin{array}[t]{lclclcl}
\cut { \slmuX (P)`b } `b + r \vxdv{Q}{r}{`d} 
	&\rtcredXsv& 
\slmuX(S P [Q.`d/`b])`b 
 \end{array} $.
 \end{lemma}

 \begin{proof}
If $`b$ is not introduced, then the activation of the substitution is permitted (notice that $t$ is not introduced in $\vxdv{Q}{t}{`d} = \vxdvfull{Q}{t}{`d} $). 
Otherwise, $P$ is a value and $`b$ does not occur in $P$; we have two cases:

 \begin{description}
 
 \myitem [$V = y$]
\cut { \slmuX (V y)`b } `b + r \vxdv{Q}{r}{`d} 
	&\redXsv& \bsp (\actR) 
\subL { \slmuX (V y)`b } `b + r \vxdvfull{Q}{r}{`d} 
	&\substo& \\
\vxdvfull{Q}{y}{`d} 
	&\fromsubs& 
\slmuX(G {V y} [Q.`d/`b])`b 
 \end{array} $

 \myitem [$V = `ly.P$]
\cut {\slmuX (<L |`g| y . P>)`b } `b + t \vxdv{Q}{t}{`d} 
	&\fromsubs& \bsp ({`b \notele \slmuX(P)`g }) 
\slmuX(g {<L |`g| y . P>} [Q.`d/`b])`b 
	&\ByDef& \\
\slmuX(g `ly.P [Q.`d/`b])`b 
 \end{array} $
\arrayqed

 \end{description}
 \end{proof}

With this result we can now show:

 \begin{theorem} [Preservation of $\mul$ contraction through $\slmuX(`.)`. $] \label{lmuX Left reduction preservation} ~
 \begin{enumerate}
 \item $ \slmuX ({}{`@ Q (`M`b.[`b]P)})`d \crXs \slmuX ({}{`M`g.[`g] `@ Q \mulsubst[Q.`g/`b]P })`d $.
 \item $ \slmuX ({}{`@ Q (`M`b.[`h]P)})`d \crXs \slmuX ({}{`M`g.[`h] \mulsubst[Q.`g/`b]P })`d $.
 \end{enumerate}
 \end{theorem}

 \begin{proof} 
 \begin{enumerate}
 
 \firstitem $ \kern-7.5mm \begin{array}[t]{lclcl} 
\kern7.5mm 
\slmuX ({}{`@ Q (`M`b.[`b]P)})`d
	&\ByDef& 
\slmuX (a |`t,t| Q {m `b . `b P})`d
	&\redXs& \bsp (\actR) \\
\slmuX (l |`t,t| Q {m `b . `b P})`d
	&\substo& 
\cut { \slmuX (m `b . `b P)`r } `r + r {\cut {\slmuX (Q)`t } `t + t \yxd{r}{t}{`g} } 
	&\ByDef& \\
\cut { \slmuX (m `b . `b P)`r } `r + r \vxdv{Q}{r}{`d} 
	&\redXs& \bsp (\actL) 
\subR { \slmuX (m `b . `b P)`r } `r + r \vxdv{Q}{r}{`d}  
	&\substo& \\
\cut { \slmuX (P)`b } `b + r \vxdv{Q}{r}{`d} 
	  &\rtcworraXs& \bsp (\ref{second rescue lemma}) 
\slmuX (S P [Q.`d/`b])`b 
	&\rtcworraXs& \bsp (\ref{mul substitution in X}) \\ 
\Span{3}{ \begin{array}{@{}lclcl}
\slmuX ({}{`@ Q \mulsubst[Q.`d/`b]P })`d
	~\worraXs (\actL) ~ 
\slmuX (m `g . `g {{}{`@ Q \mulsubst[Q.`g/`b]P }})`d
	~\ByDef~ 
\slmuX ({}{`M`g.[`g] `@ Q \mulsubst[Q.`g/`b]P })`d
 \end{array} } 
 \end{array} $ 
 
 \item
$ \begin{array}[t]{lclclcl} 
\slmuX ({}{`@ Q (`M`b.[`h]P)})`d
	&\ByDef& 
\Span{3}{
\slmuX (a |`a,a| Q {m `b . `h P})`d ~ \redXs (\actR)
}	\\ 
\cut {\slmuX(m `b . `h P)`r } `r + r \vxdv{Q}{r}{`d} 
	&\redXs& \bsp (\actL) 
\cut { \slmuX (P)`d } `b + r \vxdv{Q}{r}{`d} 
	&\redXs& \bsp (\actL) 
\subR { \slmuX (P)`d } `b + r \vxdv{Q}{r}{`d} 
	\\ 
\rtcworraXs (\ref{mul substitution in X}) ~ \slmuX (\mulsubst[Q.`d/`b]P)`d 
	&\worraXs& \bsp (\actL) 
\slmuX (m `g . `h {\mulsubst[Q.`g/`b]P })`d
	&\ByDef& 
\slmuX ({}{`M`g.[`h] \mulsubst[Q.`g/`b]P })`d
 \end{array} $ 
 
 \end{enumerate}
 \end{proof}

We conclude the results of this paper with the following:

 \begin{theorem}
 \begin{enumerate}
 \firstitem If $ P \rednclmu Q $, then $ \slmuX(P)`d \crXs \slmuX(Q)`d $. 
 \item If $ P \rednclmun Q $, then $ \slmuX(P)`d \crXsn \slmuX(Q)`d $. 
 \item If $ P \rednclmuv Q $, then $ \slmuX(P)`d \crXsv \slmuX(Q)`d $. 
 \end{enumerate}
 \end{theorem}

 \begin{proof} 
 \begin{enumerate}
 \firstitem By induction on the definition of $\rednclmu$, using \Thm \ref{slmuX respects beta reduction}, \ref{lmuX Right reduction preservation}, and \ref{lmuX Left reduction preservation}.
 \item By induction on the definition of $\rednclmun$, using \Thm \ref{slmuX respects beta reduction} and \ref{lmuX Right CBN reduction preservation}.
 \item By induction on the definition of $\rednclmuv$, using \Thm \ref{slmuX respects beta reduction}, \ref{lmuX Right reduction preservation}, and \ref{lmuX Left reduction preservation}.
 \end{enumerate}
 \end{proof}



\section*{Conclusion and Future Work}

This paper has presented mappings from $\nclmu$ to $\lmmt$, and from $\lmmt$ to $\Xs$, which preserve the {\CBN} and {\CBV} reduction and strategies, but with small caveats.
Furthermore, these mappings are strict in the sense that if $M \reduc N$, then $\Sem[M]$ and $\Sem[N]$ have a common reduct in $t$ such that $\Sem[M] \tcredlmmt t$; otherwise put, no reductions are `lost' in the mapping.
It follows from these results that there are mappings from $\LC$ into $\lmmt$ and $\X$ which preserve {\CBN} and {\CBV}.

\paragraph{Other reduction disciplines}
Our focus has been on preserving the {\CBN} and {\CBV} reduction and strategies, as these are the most commonly considered.
It should be interesting to see if these mappings preserve other evaluation disciplines, such as call-by-need or even call-by-co-need.
Of course, this would first require defining call-by-need and co-need for $\X$, although these definitions already exist for $\nclmu$ and $\lmmt$.
Such definitions may even follow from our translation into $\X$.

\paragraph{$`m$-reductions}
It may be surprising that a translation from $\lmu$ to $\lmmt$ such as the one we defined here exists; as mentioned in \Sect\ref{sec:TranslateLmuLmmt}, the nature of $`m$ reductions in both is distinct.
Although $`m$-reductions in $\nclmu$ are often understood to capture their context, they normally do not capture their \textit{entire} context.
Consider the term $`@ {`@ M (`M`a.\Cmd)} (`M`b.\Cmd')$, and the reduction sequence
 \[ \begin{array}{rcl}
    `@ {`@ M (`M`a.\Cmd)} (`M`b.\Cmd') 
    &\reduc (`a) & `@ (`M`g.\Cmd\lmusubst[M.`g/`a]) (`M`b.\Cmd') \\
    &\reduc (`b) & (`M`d.\Cmd' \lmusubst[{(`M`g.\Cmd[\lmusubst[M.`g/`a]])} . `d / `b] )
 \end{array} \]
The context of the subterm $`M`a.\Cmd$ is $\Appl {(M {{\EmptyCont}})} (`M`b.\Cmd')$, yet only $M$ is captured by $`M`a.\Cmd$.
The $`m$-reductions in $\lmu$ are performed term-by-term, which allows for the continuation to be captured piecewise instead of at-once;
this is in contrast with $`m$ in $\lmmt$, which necessarily substitutes its entire continuation in a command.

\Comment{
Furthermore, there is no direct analogue of the $(\mul)$ reduction in $\lmmt$; the term $`M`a.< t | `M`b.c `. `a >$ does not have a (head) redex, whereas its $\nclmu$ analogue $`M`a.[`a] M( `M`b.C )$ is reducible;
 \[ \begin{array}{ccccccc}
`M`a.[`a] M( `M`b.\Cmd ) & \reduc &
`M`a.[`a] `M`g.\lmusubst[M.`g/`b] \Cmd &\reduc& 
`M`a.\lmusubst[M.`g/`b] \Cmd \tsubst[`a/`g] &=& 
`M`a.\lmusubst[M.`a/`b] \Cmd  
 \end{array} \]
(notice that $`g$ does not appear in $\Cmd$).
}

Furthermore, there is no direct analogue of the $(\mul)$ reduction in $\lmmt$; the term $`M`a.< t | `M`b.c `. `a >$ does not have a (head) redex, whereas its $\nclmu$ analogue $`M`a.[`a] M( `M`b.\Cmd )$ is reducible;
 \[ \begin{array}{ccccccc}
`M`a.[`a] M( `M`b.\Cmd ) & \reduc &
`M`a.[`a] `M`g.\lmusubst[M.`g/`b] \Cmd &\reduc& 
`M`a.\lmusubst[M.`g/`b] \Cmd \tsubst[`a/`g] &=& 
`M`a.\lmusubst[M.`a/`b] \Cmd  
 \end{array} \]
(notice that $`g$ does not appear in $\Cmd$).
	
The translation of \Def\ref{from lmmt to x} circumvents this by introducing extra redexes around the applicands so that its image reduces:
\[\begin{array}{rcl}
    \lmmtSem[{[`a] `@ M ( `M`b.C) }] 
    &=& 
	\cell< {\Mo `g . 
		< \lmmtSem[M] 
			| { \mt x . 
				< { `M`b.\lmmtSem[C] } 
					| {\mt y .
						 < x  | y `. `g > } 
				> } 
		> } 
	| `a >
\\ &\rtcreduc&
	\cell< { `M`b.\lmmtSem[C] } 
	    | {\mt y .
		< \lmmtSem[M] | y `. `a > } 
	    >, 
\end{array}\]
which then allows for the $(\mul)$ reduction to be simulated.

An alternative solution could be to add a $(\mul)$ reduction to $\lmmt$,
\[\begin{array}{rcl}
    \cell< t | (`M`b.c) `. e >  
    &\reduc&
    c\tsubst[{\mt x . < t | x `. e >} / `b]
\end{array}\]
This would then allow for $(\mul)$ to be directly reflected, rather than just simulated.

\paragraph{Relating Semantics per-discipline}
The translations could provide a unified way to relate the by-name and by-value semantics of each calculi.
For example, the operational semantics of $\lmu$ in an abstract machine (KAM) could be directly translated into $\lmmt$ (and thus $\X$), making $\lmmt$ and $\X$ capable of simulating {\CBN} and {\CBV} $\lmu$-machines in a uniform way.
This could furthermore be used to inspect the behaviour of classical realizability models of $\lmu$ through $\lmmt$ and $\X$.

In a similar vein, one obtains CPS translations from $\lmu$ into the $\LC$ by combining $\Sem[`.]$ with the CPS translations of $\lmmt$ given by \citet{Curien-Herbelin'00}.
The resulting translations agree up to equality with the CPS translations of $\lmu$ given in \cite{Curien-Herbelin'00}, however that given by our mapping has extra redexes that are `clerical' in nature.
These precisely come from the extra $\mutilde$ redex given in $\Sem[MN]$ when compared with $\SEMcbv{MN}$ and $\SEMcbn{MN}$.
This would give the operational semantics obtained as suggested subtly different intensional properties.

\paragraph{Type Systems}
It was relatively simple to show the simple type systems of $\nclmu$ and $\lmmt$ are respected by the translations.
The case for polymorphic types would be more subtle; in the presence of control, restrictions are needed to determine when a term can have a polymorphic type \cite{Harper-Lillibridge'91,Herbelin'05}.
The translation must be ensured to preserve the appropriate restrictions.

The existence of a mapping of $\nclmu$ into $\lmmt$ means the former probably cannot be given a sound and complete notion of intersection type assignment \cite{Bakel-APAL'10}.
Nonetheless, $\lmmt$ and $\X$ is known to enjoy such a system once restricted to {\CBN} or {\CBV} \cite{Bakel-APAL'10,Bakel-FI'12}.
Intersection types for {\CBN} $\lmu$ have been explored before \cite{Bakel-ITRS'10}, and the mapping preserving {\CBV} suggests one is also possible for $\lmuV$.

\Comment{
\begin{itemize}
    \item Relate control models?
    \item Other connectives?
\end{itemize}
}

\subsection*{Acknowledgements}
We would like to thank Yoram Boccia and Nicolas Wu for many fruitful discussions.


 \bibliographystyle{plainnat}
 \bibliography{references}
 \Journal{ %
 \immediate\closeout\proofsappendix
 \small
  \par \par \topmargin -10mm \parindent 10.5\point \leftmargini 2\parindent \leftmarginii 2\parindent \leftmarginiii 1.1\parindent \leftmarginiv .9\parindent \itemcorrection \parindent \appendix \section {Proofs} 
\begin {trivlist} \item [{\textbf {~Proof of \Lmm \ref {lmu substitution lemma}.~}}]  \begin {enumerate} \itemsep 4\point \par \firstitem By induction on the definition of term substitution. \par \Longer {\begin {description} \itemsep 2pt \par \item [{$ x\tsubst [L/x] = L $}] If $ \derlmu `G,x`:B |- x : A | `D $, then $B = A$, so $ \derlmu `G |- L : A | `D $, so also $ \derlmu `G |- x[L/x] : A | `D $. \par \item [{$ y \tsubst [L/x] \same y $ $(y \not = x)$}] Then $y`:A \ele `G$ and by rule $(\Ax )$ we have $ \derlmu `G |- y : A | `D $. \par \item [{$ `ly.(N\tsubst [L/x]) = (`ly.N)\tsubst [L/x] $}] Then $A = C\arrow D$. If $ \derlmu `G,x`:B |- `ly.N : C\arrow D | `D $, then by rule $(\arrI )$, $ \derlmu `G,x`:B,y`:C |- N : D | `D $. Then by induction, $ \derlmu `G,y`:C |- N\tsubst [L/x] : D | `D $, so by $(\arrI )$, $ \derlmu `G |- `ly.(N\tsubst [L/x]) : C\arrow D | `D $. \par \item [{$ (PQ)\tsubst [L/x] \equiv `@ P\tsubst [L/x] Q\tsubst [L/x] $}] If $ \derlmu `G,x`:B |- PQ : A | `D $, then, by rule $(\arrE )$ there exist $C$ such that both $\derlmu `G,x`:B |- P : C\arrow A | `D $ and $\derlmu `G,x`:B' |- Q : C | `D $. Then, by induction, $\derlmu `G |- P\tsubst [L/x] : C\arrow A | `D $ and $\derlmu `G |- Q\tsubst [L/x] : C | `D $; the result follows by rule $(\arrE )$. \par \item [{$ (`M`a . [`b] N)\tsubst [L/x] \same `M`a . [`b] N\tsubst [L/x] $}] If $ \derlmu `G,x`:B |- `M`a.[`b]N : A | `D $, then, by rule $(`m)$ there exist $C$ such that $\derlmu `G,x`:B |- N : C | `a`:A,`b`:C,`D' $ with $`D = `b`:C,`D'$. Then, by induction, $\derlmu `G |- N\tsubst [L/x] : C | `a`:A,`b`:C,`D' $, and by rule $(`m)$ we have $\derlmu `G |- N\tsubst [L/x] : C | `a`:A,`b`:C,`D' $. \par \end {description} }\par \item By induction on the definition of right-structural substitution\Shorter {; we only show the interesting cases}. \par \begin {description} \par \Longer {\item [{$ x\mursubst [L.`g/`a] \ByDef x $}] Then $x`:A \ele `G$, and by rule $(\Ax )$ we have $\derlmu `G |- x : A | `g`:C,`D $. \par \item [{$ (`Lx.N)\mursubst [L.`g/`a] \ByDef `Lx . (N\mursubst [L.`g/`a] ) $}] Then $A = D\arrow E$ and, by rule $(\arrI )$, $\derlmu `G,x`:D |- N : E | `a`:B\arrow C,`D $. Then by induction we have $\derlmu `G,x`:D |- {N\mursubst [L.`g/`a]} : E | `g`:C,`D $, so by rule $(\arrI )$ also $\derlmu `G |- { `Lx . {N\mursubst [L.`g/`a]} } : D\arrow E | `g`:C,`D $. \par \item [{$ (`@ P Q )\mursubst [L.`g/`a] \ByDef `@ P\mursubst [L.`g/`a] Q\mursubst [L.`g/`a] $}] Then by rule $(\arrE )$ there exists $D$ such that $ \derlmu `G |- P : D\arrow A | `a`:B\arrow C,`D $ and $ \derlmu `G |- Q : D | `a`:B\arrow C,`D $. Then by induction we can assume both $ \derlmu `G |- P\mursubst [L.`g/`a] : D\arrow A | `g`:C,`D $ and $ \derlmu `G |- Q\mursubst [L.`g/`a] : D | `g`:C,`D $; the result follows by rule $(\arrE )$. }\par \item [{$ `M`d.[`a]N\mursubst [L.`g/`a] \ByDef `M`d.[`g](N\mursubst [L.`g/`a] L) $}] Then by rule $(`m)$ $ \derlmu `G |- N : B\arrow C | `d`:A,`a`:B\arrow C,`D $, and by induction $ \derlmu `G |- N\mursubst [L.`g/`a] : B\arrow C | `d`:A,`g`:C,`D $. Since $`a$, $`d$ and $`g$ all do not occur (free) in $N$, we can construct \[ \def \Turnlmu {\Turn } \Inf [`m] {\Inf [\arrE ] {\InfBox { \derlmu `G |- N\mursubst [L.`g/`a] : B\arrow C | `d`:A,`g`:C,`D } \quad \Inf [\Weak ] {\InfBox { \derlmu `G |- L : B | `D } }{ \derlmu `G |- L : B | `d`:A,`g`:C,`D } }{ \derlmu `G |- { `@ (N\mursubst [L.`g/`a]) L } : C | `d`:A,`g`:C,`D } }{ \derlmu `G |- { `M`d.[`g] `@ (N\mursubst [L.`g/`a]) L } : A | `g`:C,`D } \] \par \item [{$ (`M`d.[`b]N)\mursubst [L.`g/`a] \ByDef `M`d.[`b] ( N\mursubst [L.`g/`a] ) ~ (`b \not = `a) $}] Then by rule $(`m)$ there exists $D$ such that $`D = `b`:D,`D'$, and $ \derlmu `G |- N : D | `d`:A,`b`:D,`a`:B\arrow C,`D' $, and by induction $ \derlmu `G |- N\mursubst [L.`g/`a] : D | `d`:A,`b`:D,`g`:C,`D' $. But then, by rule $(`m)$, also $ \derlmu `G |- `M`d.[`b] N\mursubst [L.`g/`a] : A | `b`:D,`g`:C,`D' $. \par \end {description} \par \item By induction on the definition of left-structural substitution\Shorter {; we only show the interesting cases}. \par \begin {description} \par \Longer {\item [{$ \mulsubst [L.`g/`a] x \ByDef x $}] Then $x`:A \ele `G$, and by rule $(\Ax )$ we have $\derlmu `G |- x : A | `g`:C,`D $. \par \item [{$ \mulsubst [L.`g/`a] (`Lx.N) \ByDef `Lx . (\mulsubst [L.`g/`a] N) $}] Then $A = D\arrow E$ and, by rule $(\arrI )$, $\derlmu `G,x`:D |- N : E | `a`:B,`D $. Then by induction we have $\derlmu `G,x`:D |- {\mulsubst [L.`g/`a] N} : E | `g`:C,`D $, so by rule $(\arrI )$ also $\derlmu `G |- { `Lx . \mulsubst [L.`g/`a] N } : D\arrow E | `g`:C,`D $. \par \item [{$ \mulsubst [L.`g/`a] (`@ P Q ) \ByDef {\mulsubst [L.`g/`a] P} \, {\mulsubst [L.`g/`a] Q} $}] Then by rule $(\arrE )$ there exists $D$ such that $ \derlmu `G |- P : D\arrow A | `a`:B,`D $ and $ \derlmu `G |- Q : D | `a`:B,`D $. Then by induction both $ \derlmu `G |- {\mulsubst [L.`g/`a] P} : D\arrow A | `g`:C,`D $ and $ \derlmu `G |- {\mulsubst [L.`g/`a] Q} : D | `g`:C,`D $; the result follows by rule $(\arrE )$. }\par \item [{$ \mulsubst [L.`g/`a ] `M`d.[`a]N \ByDef `M`d.[`g] `@ L ({\mulsubst [L.`g/`a] N}) $}] Then by rule $(`m)$ $ \derlmu `G |- N : B | `d`:A,`a`:B,`D $, and by induction $ \derlmu `G |- {\mulsubst [L.`g/`a] N} : B | `d`:A,`g`:C,`D $. Since $`d$ and $`g$ do not occur (free) in $L$, we can construct \[ \def \Turnlmu {\Turn } \Inf [`m] {\Inf [\arrE ] {\Inf [\Weak ] {\InfBox { \derlmu `G |- L : B\arrow C | `D } }{ \derlmu `G |- L : B\arrow C | `d`:A,`g`:C,`D } \InfBox { \derlmu `G |- {\mulsubst [L.`g/`a] N} : B | `d`:A,`g`:C,`D } }{\derlmu `G |- { `@ L {\mulsubst [L.`g/`a] N} } : C | `d`:A,`g`:C,`D } }{ \derlmu `G |- { `M`d.[`g] `@ L {\mulsubst [L.`g/`a] N} } : A | `g`:C,`D } \] \par \item [{$ \mulsubst [L.`g/`a] (`M`d.[`b]N) \ByDef `M`d.[`b] ( {\mulsubst [L.`g/`a] N} ) ~ (`b \not = `a) $}] Then by rule $(`m)$ there exists $D$ such that $`b`:D,`D' = `D$, and $ \derlmu `G |- N : D | `d`:A,`a`:B,`b`:D,`D' $. Then we have $ \derlmu `G |- {\mulsubst [L.`g/`a] N} : D | `d`:A,`g`:C,`b`:D,`D' $ by induction. But then, by rule $(`m)$, also $ \derlmu `G |- `M`d.[`b]{\mulsubst [L.`g/`a] N} : A | `g`:C,`b`:D,`D' $. \qed \par \end {description} \end {enumerate} \par  \end {trivlist}
 
\begin {trivlist} \item [{\textbf {~Proof of \Thm \ref {soundness lmu}.~}}]  By induction on the definition of $\rednclmuN $. \par \begin {description} \par \item [{$ `@ ( `Lx . M ) N \rednclmuN M \tsubst [N/x] $}] The derivation for $ \derlmu `G |- {`@ (`L x . M ) N } : A | `D $ is shaped like \[ \def \Turnlmu {\Turn } \Inf [\arrE ] {\Inf [\arrI ] {\InfBox { \derlmu `G,x`:B |- M : A | `D } }{ \derlmu `G |- `Lx.M : B\arrow A | `D } \InfBox { \derlmu `G |- N : B | `D } }{ \derlmu `G |- { `@ (`Lx . M ) N } : A | `D } \] \par Then, by \Lmm \,\ref {term substitution lemma}, we have $ \derlmu `G |- M\tsubst [N/x] : A | `D $. \par \item [{$ `@ (`M`a.[`a]M) N \rednclmuN `M`g . [`g] `@ (M\mursubst [N.`g/`a]) N $}] The derivation for $ `@ (`M`a.[`a]M) N $ is shaped like \[ \def \Turnlmu {\Turn } \Inf [\arrE ] {\Inf [`m] {\InfBox { \derlmu `G |- M : B\arrow A | `a`:B\arrow A,`D } }{ \derlmu `G |- `M`a.[`a]M : B\arrow A | `D } \InfBox { \derlmu `G |- N : B | `D } }{ \derlmu `G |- {`@ (`M`a.[`a]M) N } : A | `D } \] Then by \Lmm \,\ref {term substitution lemma}, we have $\derlmu `G |- M\mursubst [N.`g/`a] : B\arrow A | `g`:A,`D $. Since $`g$ is fresh, by weakening also $ \derlmu `G |- N : B | `g`:A,`D $, and we can construct \[ \def \Turnlmu {\Turn } \Inf [`m] {\Inf [\arrE ] {\InfBox { \derlmu `G |- M\mursubst [N.`g/`a] : B\arrow A | `g`:A,`D } \quad \InfBox { \derlmu `G |- N : B | `g`:A,`D } }{ \derlmu `G |- {`@ M\mursubst [N.`g/`a] N } : A | `g`:A,`D } }{ \derlmu `G |- `M`g.[`g]{`@ (M\mursubst [N.`g/`a]) N } : A | `D } \] \par \item [{$ `@ (`M`a.[`d]M) N \rednclmuN `M`g . [`d] M\mursubst [N.`g/`a] $, with $`a \not = `d $}] The derivation for $ `@ (`M`a.[`d]M) N $ is shaped like \[ \def \Turnlmu {\Turn } \Inf [\arrE ] {\Inf [`m] {\InfBox { \derlmu `G |- M : C | `a`:B\arrow A,`d`:C,`D' } }{ \derlmu `G |- `M`a.[`d]M : B\arrow A | `d`:C,`D' } \InfBox { \derlmu `G |- N : B | `d`:C,`D' } }{ \derlmu `G |- {`@ (`M`a.[`d]M) N } : A | `d`:C,`D' } \] with $`D = `d`:C,`D'$. Then by \Lmm \,\ref {term substitution lemma}, we have $ \derlmu `G |- M\mursubst [N.`g/`a] : C | `g`:A,`d`:C,`D' $, and we can construct \[ \def \Turnlmu {\Turn } \Inf [`m] {\InfBox { \derlmu `G |- M\mursubst [N.`g/`a] : C | `g`:A,`d`:C,`D' } }{ \derlmu `G |- `M`g.[`d]M\mursubst [N.`g/`a] : A | `d`:C,`D' } \] \par \item [{$ `@ M (`M`a.[`a]N) \rednclmuN `M`g.[`g]`@ M (\mulsubst M.`g/`a N) $}] The derivation for $ `@ M (`M`a.[`a]N) $ is shaped like \[ \def \Turnlmu {\Turn } \Inf [\arrE ] {\InfBox { \derlmu `G |- M : B\arrow A | `D } \quad \Inf [`m] {\InfBox { \derlmu `G |- N : B | `a`:B,`D } }{ \derlmu `G |- { `M`a.[`a]N } : B | `D } }{ \derlmu `G |- { `@ M (`M`a . [`a] N ) } : A | `D } \] Then by \Lmm \ref {term substitution lemma}, we have $ \derlmu `G |- {\mulsubst M.`g/`a N} : B | `g`:A,`D $, and we can construct \[ \def \Turnlmu {\Turn } \Inf [`m] {\Inf [\arrE ] {\Inf [\Weak ] {\InfBox { \derlmu `G |- M : B\arrow A | `D } }{ \derlmu `G |- M : B\arrow A | `g`:A,`D } \InfBox { \derlmu `G |- {\mulsubst [M.`g/`a] N} : B | `g`:A,`D } }{ \derlmu `G |- {`@ M (\mulsubst [M.`g/`a] N) } : A | `g`:A,`D } }{ \derlmu `G |- { `M`g . `g `@ M (\mulsubst [M.`g/`a] N ) } : A | `D } \] \par \item [{$ `@ M (`M`a.[`d]N) \rednclmuN `M`g . [`d] {\mulsubst M.`g/`a N }$, with $`a \not = `d $}] The derivation for $ `@ M (`M`a.[`d]N) $ is shaped like \[ \def \Turnlmu {\Turn } \Inf [\arrE ] {\InfBox { \derlmu `G |- M : B\arrow A | `d`:C,`D' } \quad \Inf [`m] {\InfBox { \derlmu `G |- N : C | `a`:B,`d`:C,`D' } }{ \derlmu `G |- { `M`a.[`d]N } : B | `d`:C,`D' } }{ \derlmu `G |- {`@ M (`M`a.[`d]N) : A } | `d`:C,`D' } \] with $`D = `d`:C,`D'$. Then by \Lmm \ref {term substitution lemma}, we have $ \derlmu `G |- {\mulsubst [M.`g/`a] N} : C | `g`:A,`d`:C,`D' $, and by rule $(`m)$ we have $ \derlmu `G |- { `M`g . [`d] \mulsubst [M.`g/`a] N } : A | `d`:C,`D' $. \par \Comment {\item [{$ `M`a.[`a]M \rednclmuN M $}] The derivation for $ `M`a.[`a]M $ is shaped like \[ \def \Turnlmu {\Turn } \Inf [`m] {\InfBox { \derlmu `G |- M : A | `a`:A,`D } }{ \derlmu `G |- `M`a.[`a]M : A | `D } \] Since $`a$ does not occur in $M$, we can thin $`a`:A,`D$ and obtain $ \derlmu `G |- M : A | `D $. }\par \item [{$ `M`a.[`b] `M`g.[`d]M \rednclmuN `M`a.([`d]M) \tsubst [`b/`g] $}] The derivation for $ `@ (`M`a.[`d]M) N $ is shaped like \[ \def \Turnlmu {\Turn } \Inf [`m] {\Inf [`m] {\InfBox { \derlmu `G |- M : D | `a`:A,`b`:B,`g`:B,`d`:D,`D' } }{ \derlmu `G |- `M`g.[`d]M : B | `a`:A,`b`:B,`d`:D,`D' } }{ \derlmu `G |- `M`a.[`b] `M`g.[`d]M : A | `b`:B,`d`:D,`D' } \] So in particular, replacing all occurrences of $`g$ by $`b$, we obtain a derivation for $ \derlmu `G |- M\tsubst [`b/`g] : D | `a`:A,`b`:B,`d`:D,`D' $. Now either: \par \begin {description} \item [$`d \not = `g$] Then we can construct: \[ \def \Turnlmu {\Turn } \Inf [`m] {\InfBox { \derlmu `G |- M\tsubst [`b/`g] : D | `a`:A,`b`:B,`d`:D,`D' } }{ \derlmu `G |- `M`a.[`d]M\tsubst [`b/`g] : A | `b`:B,`d`:D,`D' } \] \par \item [$`d = `g$] Then $D = B$ as well, and we can construct: \[ \def \Turnlmu {\Turn } \Inf [`m] {\InfBox { \derlmu `G |- M\tsubst [`b/`g] : B | `a`:A,`b`:B,`D' } }{ \derlmu `G |- `M`a.[`b]M\tsubst [`b/`g] : A | `b`:B,`D' } \] \par \end {description} \par \item [{$`M`a . [`a] M \reduc M $}, with $a\notele M$] The derivation for $ `M`a . [`a] M $ is shaped like \[ \def \Turnlmu {\Turn } \Inf [`m] {\InfBox { \derlmu `G |- M : A | `a`:A,`D } }{ \derlmu `G |- `M`a.[`a]M : A | `D } \] Since $a\notele M$, by thinning we get $ \derlmu `G |- M : A | `D $. \par \end {description} The contextual rules follow by induction.\QED  \end {trivlist}
 
\begin {trivlist} \item [{\textbf {~Proof of \Thm \ref {Subject reduction nclmu}.~}}]  \par Simultaneous by induction on the definition of $\rtcredlmmt $; we will only show the base cases. \begin {description} \item [$ \cell <`L x.t_1 | t_2`.e > \reduc \cell < t_2 | {\mt x.< t_1|e>} > $] If $ \derLmmt \cell <`Lx.t_1|t_2`.e> : `G |- `D $, then the derivation is shaped like on the left; regrouping the sub-derivations, we can construct the one on the right. \[ \kern -5mm \def \Turnlmmt {\Turn } \Inf {\Inf {\InfBox { \derLmmt `G,x`:A |- t_1 : B | `D } }{ \derLmmt `G |- `Lx.t_1 : A\arr B | `D } \Inf {\InfBox { \derLmmt `G |- t_2 : A | `D } \quad \InfBox { \derLmmt `G | e : B |- `D } }{ \derLmmt `G | t_2`.e : A\arr B |- `D } }{ \derLmmt \cell <`Lx.t_1|t_2`.e> : `G |- `D } \dquad \Inf {\InfBox { \derLmmt `G |- t_2 : A | `D } \kern -10mm \Inf {\Inf {\InfBox { \derLmmt `G,x`:A |- t_1 : B | `D } \quad \Inf [\Weak ] {\InfBox { \derLmmt `G | e : B |- `D } }{ \derLmmt `G,x`:A | e : B |- `D } }{ \derLmmt \cell <t_1|e> : `G,x`:A |- `D } }{ \derLmmt `G | \mt x . <t_1|e> : A |- `D } }{ \derLmmt \cell <t_2|{\mt x.<t_1|e>}> : `G |- `D } \] \par \item [$ \cell < `M`b.c | e > \reduc c\tsubst e/`b $] If $ \derLmmt \cell <`M`b.c|e> : `G |- `D $, then the derivation is shaped like \[ \def \Turnlmmt {\Turn } \Inf {\Inf {\InfBox { \derLmmt c : `G |- `b`:A,`D } }{ \derLmmt `G |- `M`b.c : A | `D } \quad \InfBox { \derLmmt `G | e : A |- `D } }{ \derLmmt \cell <`M`b.c|e> : `G |- `D } \] By \Lmm \ref {lmmt subst lemma} we get $ \derLmmt c\tsubst [e/`b] : `G |- `D $. \par \item [$ \cell <t|\mt x.c> \reduc c\tsubst t/x $] If $ \derLmmt \cell <t|\mt x.c> : `G |- `D $, then the derivation is shaped like \[ \def \Turnlmmt {\Turn } \Inf {\InfBox { \derLmmt `G |- t : A | `D } \quad \Inf {\InfBox { \derLmmt c : `G,x`:A |- `D } }{ \derLmmt `G | \mt x.c : A |- `D } }{ \derLmmt \cell <t|\mt x.c> : `G |- `D } \] By \Lmm \ref {lmmt subst lemma} we get $ \derLmmt c\tsubst [t/x] : `G |- `D $. \par \item [$ `L x . `M`b.\cell < t | x`.`b > \reduc t $, $ x,`b \notele \FV {t} $] If $ \derLmmt `G |- `Lx.`M`b.<t|x`.`b> : A | `D $, then the derivation is shaped like \[ \def \Turnlmmt {\Turn } \Inf {\Inf {\Inf {\InfBox { \derLmmt `G,x`:A |- t : A\arr B | `b`:B,`D } \quad \Inf {\Inf { \derLmmt `G,x`:A |- x : A | `b`:B,`D } \Inf { \derLmmt `G,x`:A | `b : B |- `b`:B,`D } }{ \derLmmt `G,x`:A | x`.`b : A\arr B |- `b`:B,`D } }{ \derLmmt \cell <t|x`.`b> : `G,x`:A |- `b`:B,`D } }{ \derLmmt `G,x`:A | `M`b.<t|x`.`b> : B |- `D } }{ \derLmmt `G |- `Lx.`M`b.<t|x`.`b> : A\arr B | `D } \] From $ \derLmmt `G,x`:A |- t : A\arr B | `b`:B,`D $ and $ x,`b \notele \FV {t} $, by Thinning we get $ \derLmmt `G |- t : A\arr B | `D $. \par \item [$ `M`a.< t | `a > \reduc t $, $`a \notele \FV {t} $] If $ \derLmmt `G |- `M`a.<t|`a> : A | `D $, then the derivation is shaped like \[ \def \Turnlmmt {\Turn } \Inf {\Inf {\InfBox { \derLmmt `G |- t : A | `a`:A,`D } \quad \Inf { \derLmmt `G | `a : A |- `a`:A,`D } }{ \derLmmt \cell <t|`a> : `G |- `a`:A,`D } }{ \derLmmt `G |- `M`a.<t|`a> : A | `D } \] From $ \derLmmt `G |- t : A | `a`:A,`D $ and $ `a \notele \FV {t} $, by Thinning we get $ \derLmmt `G |- t : A | `D $. \par \item [$ \mt x.< x | e > \reduc e $, $ x \notele \FV {e} $] If $ \derLmmt `G |- `M`a.<t|`a> : A | `D $, then the derivation is shaped like \[ \def \Turnlmmt {\Turn } \Inf {\Inf {\Inf { \derLmmt `G,x`:A |- x : A | `D } \quad \InfBox { \derLmmt `G,x`:A | e : A |- `D } }{ \derLmmt \cell <x|e> : `G,x`:A |- `D } }{ \derLmmt `G | \mt x . <x|e> : A |- `D } \] From $ \derLmmt `G,x`:A | e : A |- `D $ and $ x \notele \FV {e} $, by Thinning we get $ \derLmmt `G | e : A |- `D $. \QED \end {description} \observe We can extend this last result also for the second alternative to rule $(`l)$, since we can derive: \[ \def \Turnlmmt {\Turn } \Inf {\Inf {\Inf {\InfBox { \derLmmt `G |- t' : A | `D } \kern -7mm \Inf {\Inf {\InfBox { \derLmmt `G,y`:A |- t : B | `D } \Inf { \derLmmt `G,y`:A |- `g : B | `g`:B,`D } }{ \derLmmt {\cell <t|`g>} : `G,y`:A |- `g`:B,`D } }{ \derLmmt `G | {\mt y.<t|`g>} : A |- `g`:B,`D } }{ \derLmmt { \cell < t' | {\mt y.<t|`g>} >} : `G |- `g`:B,`D } }{ \derLmmt `G |- {`M`g . < t' | {\mt y.<t|`g> } > } : B | `D } \kern -10mm \InfBox { \derLmmt `G | e : B |- `D } }{ \derLmmt \cell <{`M`g . < t' | {\mt y.<t|`g> } > }|e> : `G |- `D } \]  \end {trivlist}
 
\begin {trivlist} \item [{\textbf {~Proof of \Lmm \ref {garbage collection}.~}}]  By induction on the structure of nets. \begin {enumerate} \itemsep 4\point \item \begin {description} \par \firstmyitem [$P = \caps <y.`a>$] \cutL \caps <y.`a> `a + x \caps <x.`b> &\substo & \bsp (\deactL ) \Tran [V y]{`b} &\ByDef & \caps <y.`a> \tsubst [`b/`a] \end {array}$ \par \myitem [$P = \caps <y.`g>$, $`g\not =`a$] \cutL \caps <y.`g> `a + x \caps <x.`b> &\substo & \bsp (\Li ) \caps <y.`g> &\ByDef & \caps <y.`g> \tsubst [`b/`a] \end {array}$ \par \myitem [$P = \exp y Q `g . `a $] \cutL { \exp y Q `g . `a } `a + x \caps <x.`b> &\substo & \bsp (\Lii ) \exp y { \cutL Q `a + x \caps <x.`b> } `g . `b &\substo & \bsp (\IH ) \exp y { Q \tsubst [`b/`a]} `g . `b &\ByDef & \\ ( \exp y Q `g . `a ) \tsubst [`b/`a] \end {array}$ \par \myitem [$P = \exp y Q `g . `d $, $`d\not =`a$] \cutL { \exp y Q `g . `d } `a + x \caps <x.`b> &\substo & \bsp (\Liii ) \exp y { \cutL Q `a + x \caps <x.`b> } `g . `d &\substo & \bsp (\IH ) \exp y Q\tsubst [`b/`a] `g . `d &\ByDef & \\ ( \exp y Q `g . `d ) \tsubst [`b/`a] \end {array}$ \par \myitem [{$P = \imp Q `g [y] z R $}] \cutL { \imp Q `g [y] z R } `a + x \caps <x.`b> &\substo & \bsp (\Liv ) \imp { \cutL Q `a + x \caps <x.`b> } `g [y] z {\cutL R `a + x \caps <x.`b> } &\substo & \bsp (\IH ) \\ \imp Q\tsubst [`b/`a] `g [y] z R\tsubst [`b/`a] &\ByDef & ( \imp Q `g [y] z R ) \tsubst [`b/`a] \end {array}$ \par \myitem [$P = \cut Q `g + z R $] \cutL { \cut Q `g + z R } `a + x \caps <x.`b> &\substo & \bsp (\Lv ) \cut { \cutL Q `a + x \caps <x.`b> } `g + z {\cutL R `a + x \caps <x.`b> } &\substo & \bsp (\IH ) \\ \cut Q\tsubst [`b/`a] `g + z R\tsubst [`b/`a] &\ByDef & ( \cut Q `g + z R ) \tsubst [`b/`a] \end {array}$ \\ \par \end {description} \par \noindent For the second part, if $`a$ is introduced in $P$, the result follows by rules $(\Cap )$ or $(\Exp )$; otherwise $ \cut P `a + x \caps <x.`b> \redX (\actL ) \cutL P `a + x \caps <x.`b> $, and the result follows by the first part. \par \item \begin {description} \par \firstmyitem [$P = {\Tran [V x]{`b}} $] \cutR \caps <y.`a> `a + x {\Tran [V x]{`b}} &\substo & \bsp (\deactR ) {\Tran [V x]{`b}} &\ByDef & \Tran [V y]{`b} \tsubst [y/x] \end {array}$ \par \myitem [$P = \caps <z.`b>$, $z\not =x$] \cutR \caps <y.`a> `a + x \caps <z.`b> &\substo & \bsp (\Ri ) \caps <z.`b> &\ByDef & \caps <z.`b> \tsubst [y/x] \end {array}$ \par \myitem [$P = \exp z Q `g . `b $] \cutR \caps <y.`a> `a + x { \exp z Q `g . `b } &\substo & \bsp (\Rii ) \exp z { \cutR \caps <y.`a> `a + x Q } `g . `b &\substo & \bsp (\IH ) \exp z Q\tsubst [y/x] `g . `b &\ByDef & \\ ( \exp z Q `g . `b ) \tsubst [y/x] \end {array}$ \par \myitem [{$P = \imp Q `g [x] z R $}] \cutR \caps <y.`a> `a + x { \imp Q `g [x] z R } &\substo & \bsp (\Riii ) \imp { \cutR \caps <y.`a> `a + x Q } `g [y] z {\cutR \caps <y.`a> `a + x R } \\ ~=(\IH ) ~ \imp Q\tsubst [y/x] `g [y] z R\tsubst [y/x] &\ByDef & ( \imp Q `g [x] z R ) \tsubst [y/x] \end {array}$ \par \myitem [{$P = \imp Q `g [v] z R $, $v \not = x$}] \cutR \caps <y.`a> `a + x { \imp Q `g [v] z R } &\substo & \bsp (\Riv ) \\ \imp { \cutR \caps <y.`a> `a + x Q } `g [v] z {\cutR \caps <y.`a> `a + x R } &\substo & \bsp (\IH ) \\ \imp Q\tsubst [y/x] `g [v] z R\tsubst [y/x] &\ByDef & ( \imp Q `g [v] z R ) \tsubst [y/x] \end {array}$ \par \myitem [$P = \cut Q `g + z R $] \cutR \caps <y.`a> `a + x { \cut Q `g + z R } &\substo & \bsp (\Rv ) \cut { \cutR \caps <y.`a> `a + x Q } `g + z {\cutR \caps <y.`a> `a + x R } &\substo & \bsp (\IH ) \\ \cut Q\tsubst [y/x] `g + z R\tsubst [y/x] &\ByDef & ( \cut Q `g + z R ) \tsubst [y/x] \end {array}$ \\ \par \end {description} \par \noindent For the second part, if $x$ is introduced in $P$, the result follows by rules $(\Cap )$ or $(\Imp )$; otherwise $ \cut \caps <y.`a> `a + x P \redX \cutR \caps <y.`a> `a + x P $, and the result follows by the first part. \par \item \begin {description} \par \firstmyitem [$Q = \caps <y.`g>$, $`g\not =`a$] \cutL \caps <y.`g> `a + x P &\substo & \bsp (\Li ) \caps <y.`g> \end {array}$ \par \myitem [$Q = \exp y R `g . `d $, $`d\not =`a$] \cutL { \exp y R `g . `d } `a + x P &\substo & \bsp (\Liii ) \exp y { \cutL R `a + x P } `g . `d &\substo & \bsp (\IH ) \exp y R `g . `d \end {array}$ \par \myitem [{$Q = \imp R `g [y] z S $}] \cutL { \imp R `g [y] z S } `a + x P &\substo & \bsp (\Liv ) \imp { \cutL R `a + x P } `g [y] z {\cutL S `a + x P } &\substo & \bsp (\IH ) \imp R `g [y] z S \end {array}$ \par \myitem [$Q = \cut R `g + z S $] \cutL { \cut R `g + z R } `a + x P &\substo & \bsp (\Lv ) \cut { \cutL R `a + x P } `g + z {\cutL S `a + x P } &\substo & \bsp (\IH ) \cut R `g + z S \end {array}$ \par \end {description} \par \item \begin {description} \par \firstmyitem [$P = \caps <z.`b>$, $z\not =x$] \cutR P `a + x \caps <z.`b> &\substo & \bsp (\Ri ) \caps <z.`b> \end {array}$ \par \myitem [$P = \exp z R `g . `b $] \cutR P `a + x { \exp z R `g . `b } &\substo & \bsp (\Rii ) \exp z { \cutR P `a + x R } `g . `b &\substo & \bsp (\IH ) \exp z R `g . `b \end {array}$ \par \myitem [{$P = \imp Q `g [v] z R $, $v \not = x$}] \cutR P `a + x { \imp Q `g [v] z R } &\substo & \bsp (\Riv ) \imp { \cutR P `a + x Q } `g [v] z {\cutR P `a + x R } &\substo & \bsp (\IH ) \\ \imp Q `g [v] z R \end {array}$ \par \myitem [$P = \cut Q `g + z R $] \cutR P `a + x { \cut Q `g + z R } &\substo &\bsp (\Rv ) \cut { \cutR P `a + x Q } `g + z {\cutR P `a + x R } &\substo & \bsp (\IH ) \cut Q `g + z R \end {array}$ \qed \par \end {description} \end {enumerate} \par \observe Notice that substitution activation plays no role in this proof. \par  \end {trivlist}
 
\begin {trivlist} \item [{\textbf {~Proof of \Thm \ref {thm:TranslateLmuLmmt:TranslationWellTyped}.~}}]  \begin {description} \item [$\Ax $] Then $ M \same x $ and $ x`:A \ele `G $; since $ \lmulmmt [x] = x $, also $ \derLmmt `G |- x : A | `D $ by rule $(\Ax )$. \par \item [$\arrI $] Then $ M \same `L x . N $, $ A \same B \arr C $, and $ \derLmu `G,x`:B |- N : C | `D $. By induction, $ \derLmmt `G,x`:B |- \lmulmmt [N] : C | `D $, and by rule $(\arrI )$, $ \derLmmt `G |- `Lx.\lmulmmt [N] : A | `D $, and $`Lx.\lmulmmt [N] = \lmulmmt [`Lx.N] $. \par \item [$\arrE $] Then $ M \same `@ P Q $, and there exists $B$ such that $ \derLmu `G |- P : B\arr A | `D $ and $ \derLmu `G |- Q : B | `D $. Then by induction, $ \derLmmt `G |- \lmulmmt [P] : B\arr A | `D $ and $ \derLmmt `G |- \lmulmmt [Q] : B | `D $; we also have $ \derLmmt `G |- \lmulmmt [P] : B\arr A | `a`:A,`D $ and $ \derLmmt `G,x`:B\arr A |- \lmulmmt [Q] : B | `a`:A,`D $ by weakening, and we can construct (where $`G' = `G,x`:B\arr A,y`:B$): \[ \def \Turnlmmt {\Turn } \Inf {\Inf {\InfBox { \derLmmt `G |- \lmulmmt [P] : B\arr A | `a`:A,`D } \kern -12mm \Inf {\Inf {\InfBox { \derLmmt `G,x`:B\arr A |- \lmulmmt [Q] : B | `a`:A,`D } \kern -10mm \Inf {\Inf {\Inf { \derLmmt `G' |- x : B\arr A | `a`:A,`D } \Inf {\Inf { \derLmmt `G' |- y : B | `a`:A,`D } \Inf { \derLmmt `G' | `a : A |- `a`:A,`D } }{ \derLmmt `G' | { y `. `a } : B\arr A |- `a`:A,`D } }{ \derLmmt {\cell < x | y `. `a > } : `G' |- `a`:A,`D } }{ \derLmmt `G,x`:B\arr A | {\mt y . < x | y `. `a > } : B |- `a`:A,`D } }{ \derLmmt {\cell < \lmulmmt [Q] | {\mt y . < x | y `. `a > } > } : `G,x`:B\arr A |- `a`:A,`D } }{ \derLmmt `G | {\mt x . < \lmulmmt [Q] | {\mt y . < x | y `. `a > } > } : B\arr A |- `a`:A,`D } }{ \derLmmt { \cell < \lmulmmt [P] | {\mt x . < \lmulmmt [Q] | {\mt y . < x | y `. `a > } > } > } : `G |- `a`:A,`D } }{ \derLmmt `G |- \lmulmmt [a P Q] : A | `D } \] and $ \lmulmmt [a P Q] = \lmulmmt [{`@ P Q }] $ \par \item [$`m$] We have two cases: $ M \same `M`a.[`b]N $, $`D = `b`:B,`D'$, and $ \derlmu `G |- N : B | `a`:A,`b`:B,`D' $; then by induction we have $ \derLmmt `G |- \lmulmmt [N] : B | `a`:A,`b`:B,`D' $. We can construct: \[ \def \Turnlmmt {\Turn } \Inf {\Inf {\InfBox { \derLmmt `G |- \lmulmmt [N] : B | `a`:A,`b`:B,`D' } ~ \Inf { \derLmmt `G | `b : B |- `a`:A,`b`:B,`D' } }{ \derLmmt {\cell <\lmulmmt [N] | `b > } : `G |- `a`:A,`b`:B,`D' } }{ \derLmmt `G |- \lmulmmt [m `a . `b N] : A | `b`:B,`D' } \] and $ \lmulmmt [m `a . `b N] = \lmulmmt [`M `a . `b N] $. \par Or $ M \same `M`a.[`a]N $ and $ \derlmu `G |- N : A | `a`:A,`D $; then by induction we have $ \derLmmt `G |- \lmulmmt [N] : A | `a`:A,`D $. We can construct: \[ \def \Turnlmmt {\Turn } \Inf {\Inf {\InfBox { \derLmmt `G |- \lmulmmt [N] : A | `a`:A,`D } ~ \Inf { \derLmmt `G | `a : A |- `a`:A,`D } }{ \derLmmt {\cell <\lmulmmt [N] | `a > } : `G |- `a`:A,`D } }{\derLmmt `G |- \lmulmmt [m `a . `a N] : A | `D } \] and $ \lmulmmt [m `a . `a N] = \lmulmmt [`M `a . `a N] $. \QED \par \end {description}  \end {trivlist}
 
\begin {trivlist} \item [{\textbf {~Proof of \Lmm \ref {tsubst lemma}.~}}]  By induction on the structure of terms. \addqsymbol {`N}{\lmmtSem [N]\!} \begin {description} \myitem [$M = z$] \lmulmmt [z] \tsubst [`N/z] &\ByDef & \lmulmmt [v z] \tsubst [`N/z] &\ByDef & `N \ByDef \lmulmmt [{z \tsubst N/z }] \end {array} $ \par \myitem [$M = y$, $y\not =z$] \lmulmmt [y] \tsubst [`N/z] &\ByDef & \lmulmmt [v y] \tsubst [`N/z] &\ByDef & y \ByDef \lmulmmt [y] \ByDef \lmulmmt [{y \tsubst N/z }] \end {array} $ \par \myitem [$M = `Ly.P$] \lmulmmt [`Ly.P] \tsubst [`N/z] &\ByDef & \lmulmmt [l y . P] \tsubst [`N/z] = \! (\IH ) ~ \lmulmmt [l y . {P\tsubst N/z }] &\ByDef & \lmulmmt [`L y . {P\tsubst N/z }] \end {array} $ \par \myitem [$M = PQ$] \lmulmmt [`@ P Q ] \tsubst [`N/z] ~\ByDef ~ \lmulmmt [a P Q] \tsubst [`N/z] ~\ByDef ~ \\ `M`a.<`P \tsubst [`N/z] | {\mt x . < `Q \tsubst [`N/z] | { \mt y.< x | y`.`a >} >} > =\! (\IH ) ~\\ \lmulmmt [a {P \tsubst N/z } {Q \tsubst N/z }] ~\ByDef ~ \lmulmmt [{(P \tsubst [N/z]) (Q \tsubst [N/z])}] ~\ByDef ~ \lmulmmt [{(PQ) \tsubst N/z }] \end {array} $ \par \myitem [{$M = `M`b.[`g]P$}] \lmulmmt [`M`b. `g P] \tsubst [`N/z] &\ByDef & \lmulmmt [m `b . `g P] \tsubst [`N/z] &\ByDef & `M`b.< `P \tsubst [`N/z] | `g > &=& \bsp (\IH ) \\ \lmulmmt [m `b . `g {P \tsubst N/z }] &\ByDef & \lmulmmt [{`M`b.[`g]P\tsubst N/z }] &\ByDef & \lmulmmt [{(`M`b.[`g]P)\tsubst N/z }] \end {array} $ \par \arrayqed [2\point ] \end {description}  \end {trivlist}
 
\begin {trivlist} \item [{\textbf {~Proof of \Lmm \ref {rule m_l lemma}.~}}]  By induction on the definition of type assignment. \par \begin {description} \par \item [$\Cut $] Then $ c = \cell <t|v> $, and both $ \derLmmt `G |- t : A | `g`:B,`D $ and $ \derLmmt `G | e : A |- `g`:B,`D $. Then by induction, both $ \derLmmt `G |- { t \Subst } : A | `g`:B,`D $ and $ \derLmmt `G | { e \Subst } : A |- `g`:B,`D $. By rule $(\Cut )$ we have \par \[ \def \TurnLmmt {\Turn } \Inf { \derLmmt `G |- {t \Subst } : A | `g`:B,`D \quad \derLmmt `G | {e \Subst } : A |- `g`:B,`D }{\derLmmt { \cell < t \Subst | e \Subst >} : `G |- `D } \] and $ \cell < t \Subst | e \Subst > = \cell < t | e > \Subst $. \par \item [$\AxT $] Then $t = x$; since $ x`:C \ele `G $, by rule $(\AxT )$ also $ \derlmmt `G |- x : C | `g`:B,`D $ \par \item [$\AxC $] Then $e = `b$, and $ `b`:C \ele `D $. We have two cases: \par \begin {description} \item [$`a = `b$] Then $C = A$; we can construct: \[ \def \Turnlmmt {\Turn } \Inf [\MuT ] {\Inf [\Cut ] {\Inf [\Weak ]{\InfBox { \derlmmt `G |- t : A\arr B | `g`:B,`D } }{ \derlmmt `G,y`:A |- t : A\arr B | `g`:B,`D } \Inf [\arrL ] {\Inf [\AxT ]{ \derlmmt `G,y`:A |- y : A | `g`:B,`D } \Inf [\AxC ]{ \derlmmt `G,y`:A | `g : B |- `g`:B,`D } }{ \derlmmt `G,y`:A | y`.`g : A\arr B |- `g`:B,`D } }{ \derlmmt {\cell <t|y`.`g>} : `G,y`:A |- `g`:B,`D } }{ \derlmmt `G | \mt y.<t|y`.`g> : A |- `g`:B,`D } \] and $ `a \Subst = \mt y.< t | y`.`g > $. \par \item [$`a\not =`b$] Since $ `b`:C \ele `D $, by rule $(\AxC )$ also $ \derlmmt `G | `b : C |- `g`:B,`D $. \par \end {description} \item [$\arrL )$, $(\arrR )$, $(`m)$, $(\MuT $] By induction. \QED \par \end {description}  \end {trivlist}
 
\begin {trivlist} \item [{\textbf {~Proof of \Lmm \ref {mu simulation}.~}}]  By simultaneous induction on the structure of nets. \begin {enumerate} \item \begin {description} \par \pushmyitem {2mm}[$c = \cell <y|`a>$] && \kern -5mm \cutL { \lmmtX (\cell <y|`a>) } `a + x \lmmtX (e'){x} &\ByDef & \cutL \caps <y.`a> `a + x \lmmtX (e'){x} &\substo & \bsp (\deactL ) \lmmtX (e'){x} \tsubst [y/x] &=& \lmmtX (e'){y} &\ByDef & \\ \lmmtX (\cell <y|e'>) &=& \lmmtX (\cell <y|`a>{ \tsubst [e'/`a] }) \end {array}$ \par \myitem [$c = \cell <y|`b>, ~`a \not = `b$] \cutL { \lmmtX (\cell <y|`b>) } `a + x \lmmtX (e'){x} &\ByDef & \cutL \caps <y.`b> `a + x \lmmtX (e'){x} &\substo & \caps <y.`b> &\substo & \bsp (\Li ) \lmmtX (\cell <y|`b>) &=& \\ \lmmtX (\cell <y|`b>{ \tsubst [e'/`a] }) \end {array}$ \par \Comment {\myitem [$c = \cell <t|`a>$] \cutL \lmmtX (\cell <t|`a>) `a + x \lmmtX (e'){x} &\ByDef & \cutL { \lmmtX (t){`a} } `a + x \lmmtX (e'){x} & \substo & \bsp (\IH ) \\ \lmmtX (t{ \tsubst [e'/`a] }){`a} &\ByDef & \\ \\ \lmmtX (\cell <t{ \tsubst [e'/`a] }|e'>) &\ByDef & \lmmtX (\cell <t|`a>{ \tsubst [e'/`a] }) \end {array}$ }\par \myitem [$c = \cell <`lz.t|`a>$] \cutL \lmmtX (\cell <`lz.t|`a>) `a + x \lmmtX (e'){x} &\ByDef & \cutL {( \lmmtX (L z . t){`a} )} `a + x \lmmtX (e'){x} &\substo & \bsp (\Lii ) \\ \cut { \exp z { \cutL {{ \lmmtX (t){`b} } } `a + x \lmmtX (e'){x} } `b . `g } `g + x \lmmtX (e'){x} &\substo & \bsp (\IH ,`a\not =`b) \cut { \exp z { \lmmtX (t{ \tsubst [e'/`a] }) } `b . `g } `g + x \lmmtX (e'){x} &\ByDef & \\ \lmmtX (\cell <`lz.t{ \tsubst [e'/`a] }|e'>) &\ByDef & \lmmtX (\cell <`lz.t|`a>{ \tsubst [e'/`a] }) \end {array}$ \par \myitem [$c = \cell <`M`b.c|`a>$] \cutL \lmmtX (\cell <`M`b.c|`a>) `a + x \lmmtX (e'){x} &\ByDef & \cutL \lmmtX (`M`b.c){`a} `a + x \lmmtX (e'){x} &\ByDef & \\ \cutL {\lmmtX (M `b . c){`a} } `a + x \lmmtX (e'){x} &=& \cutL { \lmmtX (c\tsubst [`a/`b]) } `a + x \lmmtX (e'){x} &=& \\ \cutL {\omt \cutL { \lmmtX (c) } `a + x \lmmtX (e'){x} } `b + x \lmmtX (e'){x} &=& \cutL {\omt \cutL { \lmmtX (c) } `a + x \lmmtX (e'){x} \tsubst `a/`b } `a + x \lmmtX (e'){x} &\substo & \bsp (\IH ) && \\ \cutL { \lmmtX (M `b . {c{ \tsubst [e'/`a] }}){`a} } `a + x \lmmtX (e'){x} &\ByDef & \cutL { \lmmtX (`M`b.c{ \tsubst [e'/`a] }){`a} } `a + x \lmmtX (e'){x} &\ByDef & \\ \cutL { \lmmtX ({ \cell <`M`b.c{ \tsubst [e'/`a] }|`a>}) } `a + x \lmmtX (e'){x} &\substo &\Span {5}{ \bsp (\IH ) \lmmtX (\cell <{`M`b.c \tsubst [e'/`a] }|e'>) ~=~ \lmmtX (\cell <`M`b.c|`a>{ \tsubst [e'/`a] }) } \end {array}$ \par \myitem [$c = \cell <y|e>$] \cutL \lmmtX (\cell <y|e>) `a + x \lmmtX (e'){x} &\ByDef & \cutL \caps <y.e> `a + x \lmmtX (e'){x} &\ByDef & \cutL \lmmtX (e){y} `a + x \lmmtX (e'){x} & \substo & \bsp (\IH ) \lmmtX (e{ \tsubst [e'/`a] }){y} &\ByDef & \\ \lmmtX (\cell <y|e{ \tsubst [e'/`a] }>) &\ByDef & \lmmtX (\cell <y|e>{ \tsubst [e'/`a] }) \end {array}$ \par \Comment {\myitem [$c = \cell <y|t`.e>$] \cutL { \lmmtX (\cell <y|t`.e>) } `a + x \lmmtX (e'){x} &\ByDef & \cutL \lmmtX (t`.e){y} `a + x \lmmtX (e'){x} &\ByDef & \\ \cutL { \lmmtX (X t . e){y}} `a + x \lmmtX (e'){x} &\ByDef & \imp { \cutL \lmmtX (t){`g} `a + x \lmmtX (e'){x} } `g [y] z { \cutL { \lmmtX (e){z} } `a + x \lmmtX (e'){x} } & \substo & \bsp (\IH ) \\ \imp { \lmmtX (t{ \tsubst [e'/`a] }){`g} } `g [y] z { \lmmtX (e{ \tsubst [e'/`a] }){z} } &\ByDef & \lmmtX (\cell <y|t{ \tsubst [e'/`a] }`.e{ \tsubst [e'/`a] }>) &=& \\ \lmmtX (\cell <y|t`.e>{ \tsubst [e'/`a] }) \end {array}$ \par \myitem [$c = \cell <y|\mt z.c>$] \cutL \lmmtX (\cell <y|\mt z.c>) `a + x \lmmtX (e'){x} &\ByDef & \cutL \lmmtX (\mt z.c){y} `a + x \lmmtX (e'){x} &\ByDef & \cutL { \lmmtX (T z . c){y}} `a + x \lmmtX (e'){x} &\substo & \\ \cutL { \lmmtX (c) } `a + x \lmmtX (e'){x} \tsubst [y/z] & \substo & \bsp (\IH ) \lmmtX (c{ \tsubst [e'/`a] }) \tsubst z/y &\ByDef & \lmmtX (\cell <y|\mt z.c\tsubst [e'/`a]>) &=& \\ \lmmtX (\cell <y|\mt z.c>{ \tsubst [e'/`a] }) \end {array}$ }\par \myitem [$c = \cell <t|e>$] \cutL { \lmmtX (\cell <t|e>) } `a + x \lmmtX (e'){x} &\ByDef & \cutL { \lmmtX (C |`b,y| <t|e>) } `a + x \lmmtX (e'){x} &\substo & \bsp (\Lv ) \\ \cut { \cutL { \lmmtX (t){`b} } `a + x \lmmtX (e'){x} } `b + y { \cutL \lmmtX (e){y} `a + x \lmmtX (e'){x} } & \substo & \bsp (\IH ,`a\not =`b) \lmmtX (C |`b,y| <t{ \tsubst [e'/`a] }|e{ \tsubst [e'/`a] }>) &\ByDef & \\ \lmmtX (\cell <t{ \tsubst [e'/`a] }|e{ \tsubst [e'/`a] }>) &\ByDef & \lmmtX (\cell <t|e>{ \tsubst [e'/`a] }) \end {array}$ \end {description} \par \item \begin {description} \pushmyitem {2mm}[$t = y$] \cutL { \lmmtX (y){`b} } `a + x \lmmtX (e'){x} &\ByDef & \cutL { \lmmtX (V y){`b} } `a + x \lmmtX (e'){x} &\substo & \bsp (\Li ) \lmmtX (V y){`b} &\ByDef & \lmmtX (y){`b} &\ByDef & \lmmtX ({y \tsubst [e'/`a] }){`b} \end {array}$ \par \myitem [$t = `Ly.t$] \cutL { \lmmtX (`Ly.t){`b} } `a + x \lmmtX (e'){x} &\ByDef & \cutL {( \lmmtX (L |`d| y . t){`b} )} `a + x \lmmtX (e'){x} &\substo & \bsp (\Lii ) \\ \exp y { \cutL { \lmmtX (t){`d} } `a + x \lmmtX (e'){x} } `d . `b & \substo & \bsp (\IH , `a \not = `d) \lmmtX (L |`d| y . {t\tsubst [e'/`a] }){`b} &\ByDef & \lmmtX (`Ly.t{\tsubst [e'/`a]}){`b} &\ByDef & \lmmtX ({}{(`Ly.t) \tsubst [e'/`a] }){`b} \end {array}$ \par \myitem [$t = `M`g.c$] \cutL { \lmmtX (`M`g.c){`b} } `a + x \lmmtX (e'){x} &\ByDef & \cutL { \lmmtX (M `g . c){`b} } `a + x \lmmtX (e'){x} & \substo & \bsp (\IH , `a \not = `b,`g) \lmmtX (c{ \tsubst [e'/`a] }) \tsubst [`b/`g] &\ByDef & \\ \lmmtX (`M`g.c{ \tsubst [e'/`a] }){`b} &\ByDef & \lmmtX ({(`M`g.c) \tsubst [e'/`a] }){`b} \end {array}$ \par \end {description} \par \item \par \begin {description} \par \pushmyitem {2mm}[$e = `a$] \cutL { \lmmtX (`a){y} } `a + x \lmmtX (e'){x} &\ByDef & \cutL \caps <y.`a> `a + x \lmmtX (e'){x} &\substo & \bsp (\deactL ) \lmmtX (e'){x} \tsubst [y/x] &\ByDef & \lmmtX (e'){y} &\ByDef & \lmmtX (`a{ \tsubst [e'/`a] }){y} \end {array}$ \par \myitem [$e = `b \not = `a$] \cutL { \lmmtX (`b){y} } `a + x \lmmtX (e'){x} &=& \cutL \caps <y.`b> `a + x \lmmtX (e'){x} &\substo & \bsp (\Li ) \caps <y.`b> &\ByDef & \lmmtX (`b){y} &\ByDef & \lmmtX (`b{ \tsubst [e'/`a] }){y} \end {array}$ \par \myitem [$e = t`. e''$] \cutL { \lmmtX (t`.e''){y} } `a + x \lmmtX (e'){x} &\ByDef & \cutL { \lmmtX (X t . e''){y} } `a + x \lmmtX (e'){x} ~ \substo (\Liv ) \\ \imp { \cutL \lmmtX (t){`g} `a + x \lmmtX (e'){x} } `g [y] z { \cutL { \lmmtX (e''){z} } `a + x \lmmtX (e'){x} } & \substo & \bsp (\IH , `a \not = `g) \imp { \lmmtX (t{ \tsubst [e'/`a] }){`g} } `g [y] z { \lmmtX (e''{ \tsubst [e'/`a] }){z} } &\ByDef & \\ \lmmtX (X {t{ \tsubst [e'/`a] }} . {e''{ \tsubst [e'/`a] }}){y} &\ByDef & \lmmtX ({}{(t`.e'') \tsubst [e'/`a] }){y} \end {array}$ \par \myitem [$e = \mt y.c$] \cutL { \lmmtX (\mt z.c){y} } `a + x \lmmtX (e'){x} &\ByDef & \cutL { \lmmtX (T z . c){y} } `a + x \lmmtX (e'){x} &\ByDef & \lmmtX (\mt z.c{ \tsubst [e'/`a] }){y} &\ByDef & \lmmtX ({}{(\mt z.c) \tsubst [e'/`a] }){y} \end {array}$ \qed \par \end {description} \end {enumerate} \observe Notice that no reduction steps are used in this proof.  \end {trivlist}
 
\begin {trivlist} \item [{\textbf {~Proof of \Lmm \ref {mut simulation}.~}}]  By simultaneous induction on the structure of nets. \begin {enumerate} \item \begin {description} \par \pushmyitem {2mm}[$c = \cell <y|`a>$] \kern -8mm && \cutR \lmmtX (t'){`a} `a + x \lmmtX (\cell <x|`b>) &\ByDef & \cutR \lmmtX (t'){`a} `a + x \caps <x.`b> &\substo & \bsp (\deactR ) \lmmtX (t'){`a} \tsubst [`b/`a] &=& \lmmtX (t'){`b} &=& \\ \lmmtX (\cell <t'|`b>) &=& \lmmtX (\cell <x|`b>\tsubst [t'/x]) \end {array}$ \par \myitem [$c = \cell <y|`a>$, $y \not = x$] \kern -3mm && \cutR \lmmtX (t'){`a} `a + x \lmmtX (\cell <y|`b>) &\ByDef & \cutR \lmmtX (t'){`a} `a + x \caps <y.`b> &\substo & \bsp (\Ri ) \caps <y.`b> &=& \lmmtX (\cell <y|`b>) &=& \\ \lmmtX (\cell <y|`b>\tsubst [t'/x]) \end {array}$ \par \myitem [$c = \cell <t|`a>$] \cutR \lmmtX (t'){`a} `a + x \lmmtX (\cell <t|`b>) &\ByDef & \cutR \lmmtX (t'){`a} `a + x \lmmtX (t){`b} & \substo & \bsp (\IH ) \lmmtX (t\tsubst [t'/x]){`b} &\ByDef & \lmmtX (\cell <t\tsubst [t'/x]|`b>) &\ByDef & \\ \lmmtX (\cell <t|`b>\tsubst [t'/x]) \end {array}$ \par \myitem [$c = \cell <y|e>$] \cutR \lmmtX (t'){`a} `a + x { \lmmtX (\cell <y|e>) } &\ByDef & \cutR \lmmtX (t'){`a} `a + x \caps <y.e> &\ByDef & \cutR \lmmtX (t'){`a} `a + x \lmmtX (e){y} & \substo & \bsp (\IH ) \lmmtX (e\tsubst [t'/x]){y} &\ByDef & \\ \lmmtX (\cell <y|e\tsubst [t'/x]>) &\ByDef & \lmmtX (\cell <y|e>\tsubst [t'/x]) \end {array}$ \par \myitem [$c = \cell <y|t`.e>$] \cutR \lmmtX (t'){`a} `a + x \lmmtX (\cell <y|t`.e>) &\ByDef & \cutR \lmmtX (t'){`a} `a + x \lmmtX (t`.e){y} &\ByDef & \\ \cutR \lmmtX (t'){`a} `a + x { \lmmtX (X t . e){y} } &\substo & \imp { \cutR \lmmtX (t'){`a} `a + x \lmmtX (t){`g} } `g [y] z { \cutR \lmmtX (t'){`a} `a + x { \lmmtX (e){z} } } & \substo & \bsp (\IH ) \\ \imp { \lmmtX (t\tsubst [t'/x]){`g} } `g [y] z { \lmmtX (e\tsubst [t'/x]){z} } &\ByDef & \lmmtX (\cell <y|t\tsubst [t'/x]`.e\tsubst [t'/x]>) ~=~ \lmmtX (\cell <y|t`.e>\tsubst [t'/x]) \end {array}$ \par \myitem [$c = \cell <x|\mt z.c>$] \cutR \lmmtX (t'){`a} `a + x { \lmmtX (\cell <x|\mt z.c>) } &\ByDef & \cutR \lmmtX (t'){`a} `a + x { \lmmtX (\mt z.c){x} } &\ByDef & \\ \cutR \lmmtX (t'){`a} `a + x {(\lmmtX (T z . c){x} )} &=& \cutR \lmmtX (t'){`a} `a + z {\omt \cutR \lmmtX (t'){`a} `a + x \lmmtX (c) } & \substo & \bsp (\IH ) \\ \cutR \lmmtX (t'){`a} `a + x {( \lmmtX (c\tsubst [t'/x]) \tsubst x/z )} &\ByDef & \cutR \lmmtX (t'){`a} `a + x { \lmmtX (\mt z.{c\tsubst [t'/x] }){x} } &\ByDef & \\ \cutR \lmmtX (t'){`a} `a + x { \lmmtX (\cell <x|\mt z.{c\tsubst [t'/x] }>) } & \substo & \Span {4}{ \bsp (\IH ) \lmmtX (\cell <t'|\mt z.c\tsubst [t'/x]>) ~=~ \lmmtX (\cell <x|\mt z.c>\tsubst [t'/x]) } \end {array}$ \par \myitem [$c = \cell <y|\mt z.c>$, $y \not = x$] && \kern -8mm \cutR \lmmtX (t'){`a} `a + x { \lmmtX (\cell <y|\mt z.c>) } &\ByDef & \cutR \lmmtX (t'){`a} `a + x { \lmmtX (\mt z.c){y} } &\ByDef & \\ \cutR \lmmtX (t'){`a} `a + x { \lmmtX (T z . c){y} } &=& \cutR \lmmtX (t'){`a} `a + x { \lmmtX (c) \tsubst [y/z] } & \substo & \bsp (\IH ) \lmmtX (c\tsubst [t'/x]) \tsubst [z/y]&\ByDef & \\ \lmmtX (\cell <y|\mt z.c\tsubst [t'/x]>) &=& \lmmtX (\cell <y|\mt z.c>\tsubst [t'/x]) \end {array}$ \par \Comment {\par \myitem [$c = \cell <t|e>$] \cutR \lmmtX (t'){`a} `a + x \lmmtX (\cell <t|e>) &\ByDef & \cutR \lmmtX (t'){`a} `a + x { \lmmtX (C <R `b,y t|e>) } &=& \bsp (\Rv ) \\ \Span {3}{ \cut { \cutR \lmmtX (t'){`a} `a + x { \lmmtX (t){`b} } } `b + y { \cutR \lmmtX (t'){`a} `a + x \lmmtX (e){y} } } & \substo & \bsp (\IH , x \not = y) \\ \cut \lmmtX (t\tsubst [t'/x]){`b} `b + y \lmmtX (e\tsubst [t'/x]){y} &=& \lmmtX (\cell <t\tsubst [t'/x]|e\tsubst [t'/x]>) &\ByDef & \lmmtX (\cell <t|e>\tsubst [t'/x]) \end {array}$ }\par \end {description} \par \par \item \par \begin {description} \pushmyitem {2mm}[$t = x$] \cutR \lmmtX (t'){`a} `a + x { \lmmtX (x){`b} } &\ByDef & \cutR \lmmtX (t'){`a} `a + x { \lmmtX (V x){`b} } &=& \bsp (\deactR ) \lmmtX (t'){`a} \tsubst [`b/`a] &=& \lmmtX (t'){`b} &\ByDef & \lmmtX (x\tsubst [t'/x]){`b} \end {array}$ \par \myitem [$t = y \not = x$] \cutR \lmmtX (t'){`a} `a + x { \lmmtX (y){`b} } &\ByDef & \cutR \lmmtX (t'){`a} `a + x { \lmmtX (V y){`b} } &\substo & \bsp (\Ri ) \lmmtX (V y){`b} &\ByDef & \lmmtX (y){`b} &\ByDef & \lmmtX (y\tsubst [t'/x]){`b} \end {array}$ \par \myitem [$t = `Ly.t''$] \cutR \lmmtX (t'){`a} `a + x { \lmmtX (`Ly.t''){`b} } &\ByDef & \cutR \lmmtX (t'){`a} `a + x { \lmmtX (L |`d| y . t''){`b} } &\substo & \bsp (\Rii ) \\ \exp y { \cutR \lmmtX (t'){`a} `a + x { \lmmtX (t''){`d} } } `d . `b & \substo & \bsp (\IH ) \exp y { \lmmtX (t''\tsubst [t'/x]){`d} } `d . `b &\ByDef & \lmmtX (`Ly.t''\tsubst [t'/x]){`b} &\ByDef & \\ \lmmtX ((`Ly.t'')\tsubst [t'/x]){`b} \end {array}$ \par \myitem [$t = `M`g.c$] \cutR \lmmtX (t'){`a} `a + x { \lmmtX (`M`g.c){`b} } &\ByDef & \cutR \lmmtX (t'){`a} `a + x { \lmmtX (c) \tsubst `b/`g } & \substo & \bsp (\IH ) \lmmtX (c\tsubst [t'/x]) \tsubst [`b/`g] &\ByDef & \\ \lmmtX (`M`g.c\tsubst [t'/x]){`b} &\ByDef & \lmmtX ({(`M`g.c) \tsubst [t'/x] }){`b} \end {array}$ \end {description} \par \item \par \begin {description} \par \pushmyitem {2mm}[$e = `b$] \cutR \lmmtX (t'){`a} `a + x { \lmmtX (`b){z} } &\ByDef & \cutR \lmmtX (t'){`a} `a + x \caps <z.`b> &\substo & \bsp (\Ri ) \caps <z.`b> &\ByDef & \lmmtX (z){`b} &\ByDef & \lmmtX (z\tsubst [t'/x]){`b} \end {array}$ \par \myitem [$e = t`. e$] \cutR \lmmtX (t'){`a} `a + x { \lmmtX (t`.e){z} } &\ByDef & \Span {3}{ \cutR \lmmtX (t'){`a} `a + x { \imp { \lmmtX (t){`b} } `b [z] y \lmmtX (e){y} } ~\substo ~ \bsp (\Riv ) } \\ \Span {3}{ \imp { \cutR \lmmtX (t'){`a} `a + x { \lmmtX (t){`b} } } `b [z] y { \cutR \lmmtX (t'){`a} `a + x \lmmtX (e){y} } } & \substo & \bsp (\IH , x\not =y) \\ \imp { \lmmtX (t\tsubst [t'/x]){`b} } `b [z] y { \lmmtX (e\tsubst [t'/x]){y} } &\ByDef & \lmmtX (t\tsubst [t'/x] `.~ e\tsubst [t'/x]){z} &\ByDef & \lmmtX ({(t`. e) \tsubst [t'/x] }){z} \end {array}$ \par \myitem [$e = \mt y.c$] \cutR { \lmmtX (t'){x} } `a + x { \lmmtX (\mt y.c){z} } &\ByDef & \cutR { \lmmtX (t'){x} } `a + x { \lmmtX (c) \tsubst [z/y] } &\ByDef & ( \cutR { \lmmtX (t'){x} } `a + x { \lmmtX (c) } ) \tsubst [z/y] & \substo & \bsp (\IH ) \\ \lmmtX (c\tsubst [t'/x]) \tsubst [z/y] &\ByDef & \lmmtX (\mt y.c\tsubst [t'/x]){z} &\ByDef & \lmmtX ({(\mt y.c) \tsubst [t'/x] }){z} \end {array}$ \arrayqed \par \end {description} \end {enumerate} \par \observe As above, no reduction steps are used in this proof. \par  \end {trivlist}

 }

\end{document}